\documentclass[11pt]{memoir}

\usepackage{amssymb}

\usepackage{hyperref}
\usepackage{memhfixc} % Fix conflicts with the memoir class.
\usepackage[numbered]{bookmark} % Allows to place a bookmark, see the title. Shows section numbers.
\usepackage[all]{hypcap} % After hyperref, to anchor floats correctly.

\usepackage[algo2e,algosection,tworuled,noend,noline]{algorithm2e}
\usepackage[charter]{gacs} % After hyperref.
\usepackage{gacs-algo} % After hyperref.

\setlrmargins{*}{*}{1} %equal size margins left and right
\checkandfixthelayout

\makeevenfoot {ruled}{}{\thepage}{} % page numbers in the middle
\makeoddfoot  {ruled}{}{\thepage}{}
\makeevenhead {ruled}{\sffamily\leftmark}{}{} 
\counterwithout{exercise}{section}
\chapterstyle{section}

\firmlists

\newcommand{\longv}[1]{##1}
\newcommand{\shortv}[1]{}

\newcommand{\longversion}{\renewcommand{\longv}[1]{##1}\renewcommand{\shortv}[1]{}}

\longversion
\numberwithin{theorem}{section}

\numberwithin{equation}{section} % in amsmath

\def\m{\mathbf{m}} % so that it also works with XeTeX
\newcommand{\nn}{\nonumber}
\newcommand{\x}{{\cX}}
\newcommand\cnv{\txt{cnv}}

\newcommand{\card}[1]{|#1|}
\newcommand{\ptwise}{\preceq}
\newcommand{\prefix}{\sqsubseteq}
\newcommand{\postfix}{\sqsupseteq}
\newcommand{\ltri}{\triangleleft}
\newcommand{\frto}[2]{#1:#2}
\newcommand{\len}[1]{l(#1)}

\renewcommand{\pair}[2]{\mathopen(#1,#2\mathclose)}
\renewcommand{\tup}[1]{(#1)}

\newcommand{\tupcod}[1]{\mathopen\langle #1\mathclose\rangle}
\newcommand{\tupdec}[1]{\mathopen[#1\mathclose]}

\newcommand{\dom}{\operatorname{dom}}
\newcommand{\Cl}[1]{\overline{#1}}
\newcommand{\Int}[1]{#1^{\txt{o}}}
\newcommand{\s}{s}
\renewcommand{\S}{S}
\renewcommand{\t}{\mathbf{t}}

\renewcommand{\d}{\mathbf{d}}
\newcommand{\En}{\operatorname{En}}
\newcommand{\Cf}{\operatorname{Cf}}

\newcommand{\F}{\operatorname{F}}

\newcommand{\Lip}{\operatorname{Lip}}
\newcommand{\id}{\operatorname{id}}

\renewcommand{\u}{u}

\begin{document}
\frontmatter
\title{Lecture notes on descriptional complexity and randomness}
\author{ Peter G\'acs\\ Boston University\\ gacs@bu.edu}
\bookmark[page=1,level=0]{Lecture notes on descriptional complexity and randomness}
% \author{ Peter G\'acs}
% \address{Boston University }
%\email{gacs@bu.edu}
% \thanks{Supported, over the years, by a number of NSF grants, as well as
% employment at the University of Rochester, Boston University, DIMACS, Center for
% Wiskunde en Informatica, Ecole Normal Superieure, and the University of Provence.} 
% \subjclass[2000]{?},
% \keywords{}

\date{}

\maketitle
\renewcommand{\abstractname}{}
 \begin{abstract}
A didactical survey of the foundations of Algorithmic Information Theory.
These notes are short on motivation, history and background but introduce
some of the main techniques and concepts of the field.

The ``manuscript'' has been evolving over the years.
Please, look at ``Version history'' below to see what has changed when.
\end{abstract}
% To avoid page number on the title page.  It is still counted in order not to
% confuse hyperref.
\thispagestyle{empty} 

\cleardoublepage

\tableofcontents

\section*{Version history}

May 2021: changed the notation to the now widely accepted one: Kolmogorov's original complexity
is denoted \( C(x) \) (instead of the earlier \( K(x) \)), while the prefix complexity is denoted \( K(x) \)
instead of the earlier \( H(x) \).

June 2013: corrected a slight mistake in the section on the section on randomness via algorithmic
probability.

February 2010: chapters introduced, and recast using the {\sffamily memoir} class.

April 2009: besides various corrections, a section is added
on infinite sequences.  This is not new material, just places the
most classical results on randomness of infinite sequences before the more
abstract theory.

January 2008: major rewrite.
 \begin{itemize}
  \item Added formal definitions throughout.
  \item Made corrections in the part on uniform tests and generalized
    complexity, based on remarks of Hoyrup, Rojas and Shen.
  \item Rearranged material.
  \item Incorporated material on uniform tests from the work of Hoyrup-Rojas.
  \item Added material on class tests.
 \end{itemize}

\mainmatter
\maxsecnumdepth{subsection}
\setsecnumdepth{subsection}

\chapter{Complexity}
% \begin{quote}
% \raggedleft
% Ainsi, au jeu de \emph{croix ou pile}, 
% \\ l'arriv\'ee de croix cent fois de suite, 
% \\ nous para\^it extraordinaire; parce que 
% \\ le nombre presque infini des combinaisons 
% \\ qui peuvent arriver en cent coups
% \\ \'etant partag\'e en s\'eries reguli\`eres, 
% \\ ou dans lesquelles nous voyons 
% \\ r\'egner un ordre facile \`a saisir, 
% \\ et en s\'eries irreguli\`eres; celles-ci sont 
% \\ incomparablement plus nombreuses.
% \\(Laplace) 
% \end{quote}

\epigraph{
  Je n'ai fait celle-ci plus longue que parce que
  je n'ai pas eu le loisir de la faire plus courte.
}{Pascal}

\epigraph{
%\raggedleft
Ainsi, au jeu de \emph{croix ou pile}, 
 l'arriv\'ee de croix cent fois de suite, 
nous para\^it extraordinaire; parce que 
le nombre presque infini des combinaisons 
qui peuvent arriver en cent coups
\'etant partag\'e en s\'eries reguli\`eres, 
ou dans lesquelles nous voyons 
r\'egner un ordre facile \`a saisir, 
et en s\'eries irreguli\`eres; celles-ci sont 
incomparablement plus nombreuses.
}{Laplace}

\section{Introduction}

The present section can be read as an independent survey on the
problems of randomness.  It serves as some motivation for the dryer
stuff to follow.

If we toss a coin 100 times and it shows each time Head, we feel
lifted to a land of wonders, like Rosencrantz and Guildenstern 
in~\cite{StoppardRosencr}. 
The argument that 100 heads are just as probable as any
other outcome convinces us only that the axioms of Probability
Theory, as developed in~\cite{KolmFound}, do not solve all mysteries
they are sometimes supposed to. 
We feel that the sequence consisting
of 100 heads is \emph{not random}, though others with the same 
probability are. 
Due to the somewhat philosophical character of this 
paradox, its history is marked by an amount of controversy unusual for
ordinary mathematics. 
Before the reader hastes to propose a solution,
let us consider a less trivial example, due to L.~A.~Levin. 

Suppose that in some country, the share of votes for the ruling party
in 30 consecutive elections formed a sequence \(  0.99 x_{i}  \) where for
every even \( i \), the number \( x_{i} \) is the \( i \)-th digit of 
\( \pi=3.1415\ldots \). Though many of us would feel that the election
results were manipulated, it turns out to be surprisingly difficult
to prove this by referring to some general principle.

In a sequence of \( n \) fair elections, every sequence \( \og \) of \( n \)
digits has approximately the probability \( Q_{n}(\og)=10^{-n} \) to appear
as the actual sequence of third digits. Let us fix \( n \). We are given
a particular sequence \( \og \) and want to \emph{test} the validity of
the government's claim that the elections were fair. We interpret the
assertion ``\( \og \) is random with respect to \( Q_{n} \)'' as a synonym for
``there is no reason to challenge the claim that \( \og \) arose from the
distribution \( Q_{n} \)''. 

How can such a claim be challenged at all? The government, just like
the weather forecaster who announces 30\%{} chance of rain, does not
guarantee any particular set of outcomes. However, to stand behind
its claim, it must agree to any \emph{bet} based on the announced
distribution. Let us call a \emph{payoff function} with respect the
distribution \( P \) any nonnegative function \( t(\og) \) with \( \sum_\og
P(\og) t(\og) \le 1 \). If a ``nonprofit'' gambling casino asks 1
dollar for a game and claims that each outcome has probability 
\( P(\og) \) then it must agree to pay \( t(\og) \) dollars on outcome \( \og \).
We would propose to the government the following payoff function \( t_{0} \)
with respect to \( Q_{n} \): let \( t_{0}(\og)=10^{n/2} \) for all sequences \( \og \)
whose even digits are given by \( \pi \), and 0 otherwise. This bet would
cost the government \( 10^{n/2}-1 \) dollars. 

Unfortunately, we must propose the bet \emph{before} the elections
take place and it is unlikely that we would have come up exactly with
the payoff function \( t_{0} \). Is then the suspicion unjustifiable? 

No. 
Though the function \( t_{0} \) is not as natural as to guess it in
advance, it is still highly ``regular''. 
And already Laplace assumed
in~\cite{LaplacePhil} that the number of ``regular'' bets is so
small we can afford to make them \emph{all} in advance and still win by
a wide margin.

\begin{sloppypar}
Kolmogorov discovered in~\cite{Kolm65} and~\cite{KolmI3E68}
(probably without knowing about~\cite{LaplacePhil} but following a four
decades long controversy on von Mises' concept of randomness, 
see~\cite{MisesGeiringer}) that to make this approach work we must define
``regular''  or ``simple'' as ``having a short description'' (in some
formal sense to be specified below). 
There cannot be many objects
having a short description because there are not many short strings
of symbols to be used as descriptions. 
  \end{sloppypar}

We thus come to the principle saying that on a random outcome, all
sufficiently simple payoff functions must take small values. It turns
out below that this can be replaced by the more elegant principle
saying that \emph{a random outcome} itself \emph{is not too simple}. If
descriptions are written in a 2-letter alphabet then a typical
sequence of \( n \) digits takes \( n\log 10 \) letters to describe (if not
stated otherwise, all logarithms in these notes are to the base 2).
The digits of \( \pi \) can be generated by an algorithm whose
description takes up only some constant length.  Therefore the
sequence \( x_{1} \ldots x_{n} \) above can be described with approximately
\( (n/2)\log 10 \) letters, since every other digit comes from \( \pi \).  It
is thus significantly simpler than a typical sequence and can be
declared nonrandom.

\subsection{Formal results}

The theory of randomness is more impressive for infinite sequences
than for finite ones, since sharp distinction can be made between
random and nonrandom infinite sequences. 
For technical simplicity,
first we will confine ourselves to finite sequences, especially a
\df{discrete sample space} \( \Og \), which we identify with the set of
natural numbers. 
Binary strings will be identified with the natural numbers they denote. 

 \begin{definition}
A Turing machine is an imaginary computing device consisting of the 
following. 
A \df{control state} belonging to a finite set \( A \) of possible control states. 
A fixed number of infinite (or infinitely
extendable) strings of cells called \df{tapes}. 
Each cell contains a symbol belonging to a finite \df{tape alphabet} \( B \). 
On each tape, there is a read-write head observing one of the tape cells. 
The machine's \df{configuration} (global state) is determined at each
instant by giving its control state, the contents of the tapes and
the positions of the read-write heads. 
The ``hardware program'' of
the machine determines its configuration in the next step as a
functon of the control state and the tape symbols observed. 
It can change the control state, the content of the observed cells and the
position of the read-write heads (the latter by one step to the left
or right). 
Turing machines can emulate computers of any known design
(see for example~\cite{Yasuhara71}). 
The tapes are used for storing programs, input data, output and as memory. 
 \end{definition}

The Turing machines that we will use to interpret descriptions will be
somewhat restricted.

 \begin{definition}
Consider a Turing machine \( F \) which from a binary string \( p \)
and a natural number \( x \) computes the output \( F(p,x) \) (if anything at
all).  
We will say that \( F \) \df{interprets} \( p \) as a \df{description} of 
\( F(p,x) \) in the presence of the side information \( x \).
We suppose that if \( F(p,x) \) is defined then \( F(q,x) \) is not defined
for any  prefix \( q \) of \( p \). Such machines are called \df{self-delimiting}
(s.d.).

The \df{conditional complexity} \( K_{F}(x \mvert y) \) of the number \( x \) with
respect to the number \( y \) is the length of the shortest description 
\( p \) for which \( F(p,y)=x \). 
 \end{definition}

Kolmogorov and Solomonoff observed that the
function \( K_{F}(x \mvert y) \) depends only weakly on the machine \( F \), because
there are universal Turing machines capable of simulating the work of
any other Turing machine whose description is supplied.
More formally, the following theorem holds:

 \begin{theorem}[Invariance Theorem]
There is a s.d.~Turnig 
machine \( T \) such that for any s.d.~machine \( F \) a constant \( c_{F} \)
exists such that for all \( x,y \) we have 
\( K_{T}(x \mvert y) \le K_{F}(x \mvert y) + c_{F} \). 
 \end{theorem}

This theorem motivates the following definition.

 \begin{definition}
Let us fix \( T \) and define \( K(x \mvert y)=K_{T}(x \mvert y) \) and \( K(x) = K(x \mvert 0) \).  
 \end{definition}

The function \( K(x \mvert y) \) is not computable. 
We can compute a
nonincreasing, convergent sequence of approximations to \( K(x) \) (it is
\df{semicomputable} from above), but will not know how far to go in
this sequence for some prescribed accuracy.

If \( x \) is a binary string of length \( n \) then 
\( K(x) \le n + 2\log n + c \) for some constant \( c \). 
The description of \( x \)  with this length
gives \( x \) bit-for-bit, along with some information of length \( 2\log n \) 
which helps the s.d.~machine find the end of the description. 
For most binary strings of lenght \( n \), no significantly shorter
description exists, since the number of short descriptions is small.
Below, this example is generalized and sharpened for the case when
instead of counting the number of simple sequences, we measure their
probability. 

 \begin{definition}
Denote by \( x^{*} \) the first one among the shortest descriptions of \( x \).
 \end{definition}

The correspondence \( x \to x^{*} \) is a \df{code} in which no codeword
is the prefix of another one. 
This implies by an argument well-known in Information Theory the inequality
\begin{equation}\label{eq:Hkraft}
   \sum_{x} 2^{-K(x \mvert y)} \le 1, 
\end{equation}
  hence only a few objects \( x \) can have small complexity. 
The converse of the same argument goes as follows. 
Let \( \mu \) be a \df{computable}
probability distribution, one for which there is a binary program
computing \( \mu(\og) \) for each \( \og \) to any given degree of accuracy. 
Let \( K(\mu) \) be the length of the shortest one of these programs.
Then 
\begin{equation}  \label{eq:HbyMu}
   K(\og) \le -\log\mu(\og) + K(\mu) + c. 
 \end{equation}
 Here \( c \) is a universal constant. 
These two inequalities are the key
to the estimates of complexity and the characterization of randomness
by complexity. 

Denote
 \[
 d_{\mu}(\og) = -\log\mu(\og) - K(\og).
 \]
 Inequality~\eqref{eq:Hkraft} implies
 \[
  t_{\mu}(\og) = 2^{d_{\mu}(\og)}
 \]
can be viewed as a \df{payoff function}. 
Now we are in a position to solve the election paradox.
We propose the payoff function 
 \[
 2^{-\log Q_{n}(\og) - K(\og \mvert n)}
 \]
 to the government. 
(We used the conditional complexity \( K(\og \mvert n) \)
because the uniform distribution \( Q_{n} \) depends on \( n \).) 
If every
other digit of the outcome \( x \) comes from \( \pi \) then 
\( K(x \mvert n) \le (n/2)\log 10 + c_{0} \) hence we win a sum \( 2^{t(x)} \ge c_{1} 10^{n/2} \)
from the government (for some constants \( c_{0},c_{1}>0 \)), even though the
bet does not contain any reference to the number \( \pi \). 

The fact that \( t_{\mu}(\og) \) is a payoff function implies by Markov's
Inequality for any \( k>0 \) 
 \begin{equation}   \label{eq:Hstat1} 
 \mu \setof{ \og : K(\og) < -\log \mu(\og) - k } < 2^{-k}.
 \end{equation}
  Inequalities~\eqref{eq:HbyMu} and~\eqref{eq:Hstat1} say that with large
probability, the complexity \( K(\og) \) of a random outcome \( \og \) is
close to its upper bound \(  -\log \mu(\og)+K(\mu) \). 
This law occupies
distinguished place among the ``laws of probability'', because if
the outcome \( \og \) violates \df{any} such law, the complexity falls
far below the upper bound. 
Indeed, a proof of some ``law of
probability'' (like the law of large numbers, the law of iterated
logarithm, etc.) always gives rise to some simple computable payoff
function \( t(\og) \) taking large values on the outcomes violating the
law, just as in the election example. 
Let \( m \) be some large number,
suppose that \( t(\og) \) has complexity \( < m/2 \), and that 
\( t(\og_{0}) > 2^{m} \). 
Then inequality~\eqref{eq:HbyMu} can be applied to 
\( \nu(\og) = \mu(\og) t(\og) \), and we get 
 \begin{align*} 
      K(\og) &\le -\log \mu(\og) - m + K(\nu) + c_{0} 
 \\          &\le -\log \mu(\og) - m/2 + K(\mu) + c_{1}
 \end{align*}
 for some constants \( c_{0},c_{1} \). 

More generally, the payoff function \( t_{\mu}(\og) \) is \df{maximal} 
(up to a multiplicative constant) among all payoff functions that are
semicomputable (from below). 
Hence the quantity \(  -\log \mu(\og) - K(\og)  \) 
is a \df{universal test of randomness}.  
Its value measures
the \df{deficiency of randomness} in the outcome \( \og \) with respect
to the distribution \( \mu \), or the extent of justified suspicion
against the hypothesis \( \mu \) given the outcome \( \og \).

\subsection{Applications}

Algorithmic Information Theory (AIT) justifies the intuition of 
random sequences as nonstandard analysis justifies infinitely small 
quantities. Any statement of classical probability theory is provable 
without the notion of randomness, but some of them are easier to find 
using this notion. Due to the incomputability of the universal 
randomness test, only its approximations can be used in practice. 

\df{Pseudorandom sequences} are sequences generated by some 
algorithm, with \emph{some} randomness properties with respect to the
coin-tossing distribution. 
They have very low complexity (depending
on the strength of the tests they withstand, see for 
example~\cite{Daley75}),
hence are not random. 
Useful pseudorandom sequences can be defined
using the notion of \df{computational complexity}, for example 
the number
of steps needed by a Turing machine to compute a certain function.
The existence of such sequences can be proved using some difficult
unproven (but plausible) assumptions of computation theory.
See~\cite{BlumMic84}, \cite{YaoPseu82}, \cite{GoGoMic86}, 
\cite{LevinPseu87}.

\subsubsection*{Inductive inference}
 The incomputable ``distribution'' \( \m(\og)=2^{-K(\og)} \) has the 
remarkable property that, the test \( d(\og \mvert \m) \), shows all outcomes
\( \og \) ``random'' with respect to it. Relations~\eqref{eq:HbyMu} 
and~\eqref{eq:Hstat1} can be read as saying that if the real distribution is
\( \mu \) then \( \mu(\og) \) and \( \m(\og) \) are close to each other with
large probability. 
Therefore if we know that \( \og \) comes from some
unknown simple distribution \( \mu \) then we can use \( \m(\og) \) as an
estimate of \( \mu(\og) \). 
This suggests to call \( \m \) the ``apriori probability'' (but we will not use this
term much).
The randomness test \( d_{\mu}(\og) \) can be interpreted
in the framework of hypothesis testing: it is the likelihood ratio
between the hypothesis \( \mu \) and the fixed alternative
\hbox{hypothesis \( \m \).}

In ordinary statistical hypothesis testing, some properties of the 
unknown distribution \( \mu \) are taken for granted. 
The sample \( \og \) is generally a large independent sample:
\( \mu \) is supposed to be a product distribution. 
Under these conditions, the universal test
could probably be reduced to some of the tests used in statistical
practice. 
However, these conditions do not hold in other
applications: for example testing for independence in a proposed random 
sequence, predicting some time series of economics, or pattern 
recognition.  

If the apriori probability \( \m \) is a good estimate of the actual 
probability then we can use the conditional apriori probability for
prediction, without reference to the unknown distribution \( \mu \). 
For
this purpose, one first has to define apriori probability for the set
of infinite sequences of natural numbers, as done in~\cite{ZvLe70}.
We denote this function by \( M \). 
For any finite sequence \( x \), the
number \( M(x) \) is the apriori probability that the outcome is some
extension of \( x \).  
Let \( x,y \) be finite sequences. 
The formula
\begin{equation}\label{eq:infer}
  \frac{M(xy)}{M(x)} 
 \end{equation}
 is an estimate of the conditional probability that the next terms of
the outcome will be given by \( y \) provided that the first terms are
given by \( x \). 
It converges to the actual conditional probability 
\( \mu(xy)/\mu(x) \) with \( \mu \)-probability 1 for any computable
distribution \( \mu \) (see for example~\cite{Solomonoff78}). 
To understand the
surprising generality of the formula~\eqref{eq:infer}, suppose that some
infinite sequence \( z(i) \) is given in the following way. 
Its even
terms are the subsequent digits of \( \pi \), its odd terms are uniformly
distributed, independently drawn random digits. 
Let 
 \[
   z(\frto{1}{n}) = z(1) \dotsb z(n). 
 \]
Then \( M(z(\frto{1}{2i}) a / M(z(\frto{1}{2i})) \) converges to \( 0.1 \) for 
\( a=0,\dots, 9 \), while \( M(z(\frto{1}{2i+1}) a) / M(z(\frto{1}{2i+1})) \)
converges to 1 if \( a \) is the \( i \)-th digit of \( \pi \), and to 0 otherwise. 

The inductive inference formula using conditional apriori probability 
can be viewed as a mathematical form of ``Occam's Razor'': the advice 
to predict by the simplest rule fitting the data. 
It can also viewed as a realization of Bayes'Rule, with a universally applicable apriori 
distribution. 
Since the distribution \( M \) is incomputable, we view the 
main open problem of inductive inference to find maximally efficient 
approximations to it. Sometimes, even a simple approximation gives 
nontrivial results (see~\cite{BarzFreiv72}).

\subsubsection*{Information theory}
\begin{sloppypar}
 Since with large probability, \( K(\og) \) is close to \( -\log \mu(\og) \),
the \df{entropy} \( -\sum_\og \mu(\og) \log \mu(\og) \) of the 
distribution \( \mu \) is close to the \df{average complexity} 
\( \sum_\og\mu(\og) K(\og) \). 
The complexity \( K(x) \) of an object \( x \) can indeed
be interpreted as the distribution-free definition of \emph{information content}. 
The correspondence \( x \mapsto x^{*} \) is a sort of
\emph{universal code}: its average (even individual) ``rate'', or
codeword length is almost equally good for any simple computable 
distribution.
  \end{sloppypar}

It is of great conceptual help to students of statistical physics that 
entropy can be defined now not only for ensembles (probability 
distributions), but for individual states of a physical system. 
The notion of an individual \emph{incompressible sequence}, a sequence whose 
complexity is maximal, proved also extremely useful in finding 
information-theoretic lower bounds on the computing speed of certain 
Turing machines (see~\cite{PaulSeifSimon81}). 

\begin{sloppypar}
The conditional complexity \( K(x \mvert y) \) obeys identities analogous to 
the information-theoretical identities for conditional entropy, but
these identities are less trivial to prove in AIT.  
The information
\( I(x:y) = K(x)+K(y)-K(x,y) \) has several interpretations.  
It is known
that \( I(x:y) \) is equal, to within an additive constant, to 
\( K(y) - K(y \mvert x^{*}) \), the amount by which the object \( x^{*} \) (as
defined in the previous section) decreases our uncertainty about \( y \).
But it can be written as \( -\log\m^2(x,y)-K(x,y) = d((x,y) \mvert \m^2) \)
where \( \m^2=\m\times \m \) is the product distribution of \( \m \) with
itself. 
It is thus the deficiency of randomness of the pair \( (x,y) \)
with respect to this product distribution. 
Since any object is random
with respect to \( \m \), we can view the randomness of the pair \( (x,y) \)
with respect to the product \( \m^2 \) as the \emph{independence} of \( x \)
and \( y \) from each other. 
Thus ``information'' measures the ``deficiency of independence''. 
  \end{sloppypar}
 
\subsubsection*{Logic}
 Some theorems of Mathematical Logic (in particular, G\"odel's
theorem) have a strong quantitative form in AIT, with new
philosophical implications (see~\cite{Chaitin74}, \cite{ChaitinSciAm75},
\cite{Levin74}, \cite{LevinRandCons84}). 
Levin based a new system of
intuitionistic analysis on his Independence Principle (see below) 
in~\cite{LevinRandCons84}.

\subsection{History of the problem}
 P.~S.~Laplace thought that the number of  ``regular'' sequences 
(whatever ``regular'' means) is much smaller than the number of 
irregular ones (see~\cite{LaplacePhil}). 
In the first attempt at 
formalization hundred years later, R.~von Mises defined an infinite 
binary sequence as random (a ``Kollektiv'') if the relative
frequencies converge in any subsequence selected according to some 
(non-anticipating) ``rule'' (whatever ``rule'' means,
see~\cite{MisesGeiringer}). 
As pointed out by A.~Wald and others, Mises's 
definitions are sound only if a countable set of possible rules is
fixed. 
The logician A.~Church, in accordance with his famous thesis,
proposed to understand ``rule'' here as ``recursive (computable)
function''. 

The Mises selection rules can be considered as special randomness
tests. 
In the ingenious work~\cite{Ville39}, J.~Ville proved that they do 
not capture all relevant properties of random sequences. 
In particular, a 
Kollektiv can violate the law of iterated logarithm. 
He proposed to 
consider arbitrary payoff functions (a countable set of them), as 
defined on the set of infinite sequences---these are more commonly 
known as \emph{martingales}.  

For the solution of the problem of inductive inference, R.~Solomonoff 
introduced complexity and apriori probability in~\cite{Solomonoff64I} and 
proved the Invariance Theorem. 
A.~N.~Kolmogorov independently introduced
complexity as a measure of individual information content and 
randomness, and proved the Invariance Theorem (see~\cite{Kolm65}
and~\cite{KolmI3E68}).  
The incomputability properties of complexity
have noteworthy philosophical implications (see~\cite{BennSciAm79},
\cite{Chaitin74}, \cite{ChaitinSciAm75}.

P.~Martin-L\"of defined in~\cite{MLof66art} randomness for infinite
sequences. His concept is essentially the synthesis of Ville and
Church (as noticed in~\cite{SchnorrBook71}). 
He recognized the existence
of a universal test, and pointed out the close connection between the
randomness of an infinite sequence and the complexity of its initial
segments.  

L.~A.~Levin defined the apriori probability \( M \) as a maximal (to within a 
multiplicative constant) semicomputable measure.
With the help of a modified complexity definition, he gave a simple and
general characterization of random sequences by the behavior of the 
complexity of their initial segments (see~\cite{ZvLe70},
\cite{LevinRand73}). 
In~\cite{GacsSymm74} and~\cite{Levin74},
the information-theoretical properties of the self-delimiting complexity 
(as defined above) are exactly described. 
See also~\cite{Chaitin75}, \cite{Schnorr73} and~\cite{Willis70}.

In~\cite{LevinRandCons84}, Levin defined a deficiency of randomness 
\( d(\og \mvert \mu) \) in a uniform manner for all (computable or incomputable)
measures \( \mu \). 
He proved that all outcomes are random with respect
to the apriori probability \( M \). 
In this and earlier papers, he also
proved the Law of Information Conservation, stating that the
information \( I(\ag ; \bg) \) in a sequence \( \ag \) about a sequence \( \bg \)
cannot be significantly increased by any algorithmic processing of
\( \ag \) (even using random number generators). 
He derived this
law from a so-called Law of Randomness Conservation via the
definition of information \( I(\ag ; \bg) \) as deficiency of randomness
with respect to the product distribution \( M^2 \). 
Levin suggested the 
Independence Principle saying that any sequence \( \ag \) arising in
\emph{nature} contains only finite information \( I(\ag ; \bg) \) about
any sequence \( \bg \) defined by \emph{mathematical} means. 
With this
principle, he showed that the use of more powerful notions of
definability in randomness tests (or the notion of complexity) does
not lead to fewer random sequences among those arising in nature. 

The monograph~\cite{Fine73} is very useful as background information
on the various attempts in this century at solving the paradoxes of
probability theory. 
The work~\cite{LevinRandCons84} is quite comprehensive but
very dense; it can be recommended only to devoted readers. 
The work~\cite{ZvLe70} is comprehensive and readable but not quite
up-to-date. 
The surveys~\cite{ChaitinSurv77}, \cite{SchnorrBook71} 
and~\cite{SchnorrSurv75} can be used to complement it.

The most up-to-date and complete survey, which subsumes most of
these notes, is~\cite{LiViBook97}.

AIT created many interesting problems of its own. 
See for 
example~\cite{Chaitin75}, \cite{Chaitin76}, \cite{GacsSymm74},
\cite{GacsExact80}, \cite{GacsRel83}, \cite{GacsEnc86},
\cite{LevinVyugin77}, \cite{LevinRandCons84}, \cite{Loveland69},
\cite{SchnorrBook71}, \cite{SchnorrFuchs77}, \cite{SolovayManu}, and
the technically difficult results~\cite{Solovay3H77} and~\cite{Vyugin76}.

\subsubsection*{Acknowledgment}  The form of exposition of the results in
these notes and the general point of view represented in them
were greatly influenced by Leonid Levin.
More recent communication with Paul Vit\'anyi, Mathieu Hoyrup, Crist\'obal Rojas and
Alexander Shen has also been very important.

%%% Local Variables: 
%%% mode: latex
%%% TeX-master: "ait-notes"
%%% End: 

\section{Notation }

When not stated otherwise, \( \log \) means base 2 logarithm.  
The cardinality of a set \( A \) will be denoted by \( \card{A} \).
(Occasionally there will be inconsistency, sometimes denoting it
by \( |A| \), sometimes by \( \# A \).)
If \( A \) is a set then \( 1_{A}(x) \) is its indicator function:
 \begin{align*}
   1_{A}(x) = \begin{cases}
      1 & \txt{if \( x\in A \),}
\\  0 & \txt{otherwise.}
 \end{cases}
 \end{align*}
The empty string is denoted by \( \Lg \).
The set \( A^{*} \) is the set of all finite strings of elements of \( A \), including
the empty string.  
Let \( \len{x} \) denote the length of string \( x \).  
(Occasionally there will be inconsistency, sometimes denoting it
by \( |x| \), sometimes by \( l(x) \)).
For sequences \( x \) and \( y \), let \( x\prefix y \) denote that \( x \) is a prefix of \( y \).  
For a string \( x \) and a (finite or infinite) sequence \( y \), we
denote their concatenation by \( x y \). 
For a sequence \( x \) and \( n \le\len{x} \), the \( n \)-th element of \( x \) is  \( x(n) \), and 
 \[
  x(\frto{i}{j}) = x(i) \dotsb x(j).  
 \]
Sometimes we will also write
 \[
    x^{\le n} = x(\frto{1}{n}).
 \]
The string \( x_{1}\dotsb x_{n} \) will sometimes be written also as
\( \tup{x_{1},\dots,x_{n}} \).
For natural number \( m \) let \( \bg(m) \) be the binary
sequence denoting \( m \) in the binary notation. 
We denote by \( X^{\dN} \) the set of all infinite sequences of elements of \( X \).

The sets of natural numbers, integers, rational numbers, real numbers
and complex numbers
will be denoted respectively by \( \dN, \dZ, \dQ, \dR,\dC \).
The set of nonnegative real numbers will be denoted by \( \dR_{+} \).
The set of real numbers with \( -\infty,\infty \) added (with the appropriate
topology making it compact) will be denoted by \( \ol\dR \).
Let 
 \[
  \dS_{r} = \{0,\ldots,r-1\}^*,\quad
  \dS = \dN^*, \quad
  \dB = \{0,1\}.
 \]
We use \( \land \) and \( \lor \) to denote \( \min \) and \( \max \), further
 \[
   |x|^{+} = x\lor 0,\quad |x|^{-} = |-x|^{+}
 \]
for real numbers \( x \).

Let \( \tupcod{ \cdot } \) be some standard one-to-one encoding of \( \dN^* \)
to \( \dN \), with partial inverses \( \tupdec{ \cdot }_{i} \) where
 \( \tupdec{\ang{x}}_{i} = x(i) \) for \( i \le \len{x} \).
For example, we can have
 \[
   \tupcod{i, j} = \frac{1}{2}(i+1)(i+j+1) + j,
\quad \tupcod{n_{1},\dots,n_{k+1}} = \tupcod{\tupcod{n_{1},\dots,n_{k}},n_{k+1}}.
 \]
We will use \( \ang{\cdot,\cdot} \) in particular as a pairing function over \( \dN \).
Similar pairing functions
will be assumed over other domains, and they will also be denoted by
\( \ang{\cdot,\cdot} \).

Another use of the notation \( \tup{\cdot,\cdot} \) may arise when the usual
notation \( (x,y) \) of an ordered pair of real numbers could be confused with the
same notation of an open interval.
In this case, we will use \( \tup{x,y} \) to denote the pair.

The relations
 \[
  f \lea g,\quad f\lem g
 \]
mean inequality to within an additive constant and multiplicative constant
respectively.
The first is equivalent to \( f \le g + O(1) \), the second to \( f = O(g) \). 
The relation \( f\eqm g \) means \( f \lem g \) and \( f \gem g \).

\section{Kolmogorov complexity } \label{sec:kolm}

\subsection{Invariance }

It is natural to try to define the \emph{information content} of some 
text as the size of the smallest string (code) from which it can be 
reproduced by some decoder, interpreter. 
We do not want too much
information ``to be hidden in the decoder'' we want it to be a
``simple'' function.  
Therefore, unless stated otherwise, we require
that our interpreters be \emph{computable}, that is partial recursive
functions. 

 \begin{definition}
A partial recursive 
function \( A \) from \( \dS_{2} \times \dS \) to \( \dS \) will be called a (binary)
\df{interpreter}, or \df{decoder}. 
We will often refer to its first argument as the \df{program}, or \df{code}. 
 \end{definition}

Partial recursive functions are relatively simple; on the
other hand, the class of partial recursive functions has some
convenient closure properties. 

 \begin{definition}
For any binary interpreter \( A \) and strings \( x,y \in \dS \), the number
 \[
   C_A(x \mvert y) = \min \setof{ \len{p} : A(p,y) = x }
 \] 
is called the \df{conditional Kolmogorov-complexity of \( x \) given
\( y \), with respect to the interpreter \( A \) }. 
(If the set after the ``min'' is empty, then the minimum is \( \infty \)). 
We write  \( C_A(x)=C_A(x \mvert \Lg) \).  
 \end{definition}

The number \( C_A(x) \) measures the length of
the shortest description for \( x \) when the algorithm \( A \) is used to
interpret the descriptions. 

The value \( C_A(x) \) depends, of course, on the underlying function
\( A \).  
But, as Theorem~\ref{thm:invar} shows, if we restrict ourselves
to sufficiently powerful interpreters \( A \), then switching between
them can change the complexity function only by amounts bounded by
some additive constant. 
Therefore complexity can be considered an intrinsic characteristic of finite objects.

 \begin{definition}
A binary p.r.~interpreter \( U \) is called \df{optimal}
if for any binary p.r.~interpreter \( A \) there is a constant \( c < \infty \) such that 
for all \( x,y \) we have
 \begin{equation}  
C_U(x \mvert y) \le C_A(x \mvert y) + c. \label{eq:opt}
 \end{equation}
 \end{definition}

\begin{theorem}[Invariance Theorem] \label{thm:invar}
 There is an optimal p.r.~binary interpreter.
 \end{theorem}
\begin{proof} 
The idea is to use interpreters that come from universal partial recursive
functions.
However,  not all such functions can be used for universal interpreters.
Let us introduce an appropriate pairing function.
For \( x \in \dB^{n} \), let 
 \[
   x^{o} = x(1) 0 x(2) 0 \ldots x(n-1) 0 x(n) 1 
 \]
  where \( x^o = x1 \) for \( \len{x} =1 \). Any binary sequence \( x \) can be
uniquely represented in the form \( p = a^o b \). 

We know there is a p.r.~function
\( V : \dS_{2} \times \dS_{2} \times \dS \to \dS \) that
is \emph{universal}: for any
p.r.~binary interpreter \( A \), there is a string \( a \) such that for all
\( p,x \), we have \( A(p,x) = V(a,p,x) \). 
Let us define the function \( U(p,x) \) as follows. 
We represent the string \( p \) in the form 
\( p = u^o v \) and define \(  U(p,x) = V(u,v,x)  \). 
Then \( U \) is a p.r.~interpreter.
Let us verify that it is optimal. 
Let \( A \) be a p.r.~binary
interpreter, \( a \) a binary string such that \( A(p,x) = U(a,p,x) \) for
all \( p,x \). 
Let \( x,y \) be two strings. 
If \( C_A(x \mvert y)=\infty \), then~\eqref{eq:opt} holds trivially. 
Otherwise, let \( p \) be a binary string of
length \( C_A(x \mvert y) \) with \( A(p,y)=x \). 
Then we have
 \[
  U(a^o p, y) = V(a,p,y) = A(p,y) = x, 
 \]
  and
 \[
  C_U(x \mvert y) \le 2\len{a} + C_A(x \mvert y).
 \]
  \end{proof}

The constant \( 2\len{a} \) is in the above
proof a bound on the complexity of description of the interpreter \( A \)
for the optimal interpreter \( U \). 
Let us note that for any two optimal
interpreters \( U^{(1)},U^{(2)} \), there is a constant \( c \) such that for
all \( x,y \), we have
 \begin{equation}
  |C_{U^{(1)}}(x \mvert y) - C_{U^{(2)}}(x \mvert y)| < c .
 \label{eq:invar}  \end{equation} 
  Hence the complexity \( C_U(x) \) of description of an object \( x \) does
not depend strongly on the interpreter \( U \). 
Still, for every string
\( x \), there is an optimal interpreter \( U \) with \( C_U(x)=0 \).  
Imposing a
universal bound on the table size of the Turing machines used to
implement optimal interpreters, we can get a universal bound on the
constants in~\eqref{eq:invar}.

The theorem motivates the following definition.

 \begin{definition}
We fix an optimal binary p.r.~interpreter \( U \) and write 
\( C(x \mvert y) \) for \( C_U(x \mvert y) \). 
 \end{definition}

Theorem~\ref{thm:invar} (as well as other invariance theorems) is used
in AIT for much more than just to show that \( C_A(x) \) is a proper
concept. 
It is the principal tool to find upper bounds on \( C(x) \);
this is why most such upper bounds are proved to hold only to within
an additive constant. 

The optimal interpreter \( U(p,x) \) defined in the proof of 
Theorem~\ref{thm:invar} is obviously a universal partial recursive function.
Because of its convenient properties we will use it from now on as
our standard universal p.r.~function, and we will refer to an
arbitrary p.r.~function as \( U_{p}(x)=U(p,x) \).

 \begin{definition}
We define \( U_{p}(x_{1},\ldots,x_{k}) = U_{p}(\ang{ x_{1}\,\ldots, x_{k}}) \). 
Similarly,
we will refer to an arbitrary computably enumerable
set as the range \( W_{p} \) of some \( U_{p} \).  
We often do not need the second argument of the function \( U(p,x) \).  
We therefore define 
 \[ U(p) = U(p,\Lg).
 \]
 \end{definition}

It is of no consequence that we chose binary strings as as
descriptions.  
It is rather easy to define, (Exercise) for any two
natural numbers \( r,s \), a standard encoding \( \cnv_s^r \) of base \( r \)
strings \( x \) into base \( s \) strings with the property
 \[   \len{\cnv_s^r(x)} \le \len{x}\frac{\log r}{\log s} + 1.
 \] 
 Now, with \( r \)-ary strings as descriptions, we must define
\( C_A(x) \) as the minimum of \( \len{p}\log r \) over all programs \( p \)
in \( \dS_{r} \) with \( A(p)=x \).  
The equivalence of the definitions for
different values of \( r \) is left as an exercise.

\subsection{Simple quantitative estimates}

We found it meaningful to speak about the information content, 
complexity of a finite object. 
But how to compute this quantity? 
It turns out that complexity is not a computable function,
but much is known about its statistical behavior. 

The following notation will be useful in the future.  

 \begin{definition}
The relation 
\( f \lea g \) means inequality to within an additive constant, that
there is a constant \( c \) such that for all \( x \), \( f(x) \le g(x) \). 
We can write this also as \( f \le g + O(1) \). 
We also say that \( g \) \df{additively dominates} \( f \). 
The relation \( f \eqa g \) means \( f \lea g \)
and \( f \gea g \). 
The relation \( f \lem g \) among nonnegative functions
means inequality to within a multiplicative constant. 
It is
equivalent to \( \log f \lea \log g \) and \( f = O(g) \). 
We also say that
\( g \) \df{multiplicatively dominates} \( f \). 
The relation \( f\eqm g \) means \( f \lem g \) and \( f \gem g \).
 \end{definition}

With the above notation, here are the simple properties.

\begin{theorem} \label{thm:upbd}
The following simple upper bound holds.
 \begin{alphenum}  
  \item \label{thebd} For any natural number \( m \), we have 
 \begin{equation} C(m) \lea \log m. \label{eq:upbd}
 \end{equation}
 \item \label{howsharp}
  For any positive real number \( v \) and string \( y \), every finite set 
	\( E \) of size \( m \) has at least \( m(1-2^{-v+1}) \) elements \( x \) with 
	\( C(x \mvert y) \ge \log m - v  \).
 \end{alphenum}
\end{theorem}

\begin{corollary}  \( \lim_{n\to\infty} C(n) = \infty \).
\end{corollary}

Theorem~\ref{thm:upbd} suggests that if there are so few strings of low
complexity, then \( C(n) \) converges fast to \( \infty \).  In fact this
convergence is extremely slow, as shown in a later section.

\begin{proof}[Proof of Theorem \protect\ref{thm:upbd}]
 First we prove~\eqref{thebd}.  Let the interpreter \( A \) be defined such
that if \( \bg(m) = p \) then \( A(p,y) = m \). Since
 \( |\bg(m)| = \cei{ \log m} \), we have \( C(m) \lea C_A(m) < \log m + 1 \).

Part~\eqref{howsharp} says that the trivial estimate~\eqref{eq:upbd} is
quite sharp for {\em most} numbers under \( m \). The reason for this is
that since there are only few short programs, there are only few
objects of low complexity. For any string \( y \), and any positive real
natural number \( u \), we have 
 \begin{equation} \card{\setof{ x : C(x \mvert y) \le u }} < 2^{u+1}. 
\label{eq:howsharp} \end{equation} 
  To see this, let \( n=\flo{ \log u } \). The number \( 2^{n+1}-1 \) of 
different binary strings of length \( \le n \) is an upper 
bound on the number of different shortest programs of length \( \le u \). 
Now~\eqref{howsharp} follows immediately from~\eqref{eq:howsharp}.
 \end{proof}

The three properties of complexity contained in Theorems~\ref{thm:invar} 
and~\ref{thm:upbd} are the ones responsible for the
applicabiliy of the complexity concept.  
Several important later
results can be considered transformations, generalizations or 
analogons of these. 
						
Let us elaborate on the upper bound~\eqref{eq:upbd}. 
For any string \( x \in \dS_{r} \), we have 
 \begin{equation}  \label{eq:binupbd}
C(x) \lea n\log r + 2 \log r.
 \end{equation}
  In particular, for any binary string \( x \), we have 
 \begin{equation*} 
  C(x) \lea \len{x}.
 \end{equation*}
  Indeed, let the interpreter \( A \) be defined such that if \( p =
\bg(r)^o\cnv_{2}^r(x) \) then \( A(p,y)=x \). 
We have \(  C_A(x) \le
2^{\len{\bg(r)}}+n\log r +1 \le (n+2)\log r + 3 \).  
Since \( C(x) \lea C_A(x) \), we are done. 

We will have many more upper bound proofs of this form, and will not 
always explicitly write out the reference to 
\( C(x \mvert y) \lea C_A(x\mvert y) \), needed for the last step. 

Apart from the conversion, inequality~\eqref{eq:binupbd} says that since 
the beginning of a program can command the optimal interpreter to copy 
the rest, the complexity of a binary sequence \( x \) of length \( n \) is not 
larger than \( n \). 

Inequality~\eqref{eq:binupbd} contains the term \( 2\log r \) instead of \( \log r \) 
since we have to apply the encoding \( w^o \) to the string \( w=\bg(r) \); 
otherwise the interpreter cannot detach it from the program \( p \). 
We could use the code \( \bg(|w|)^o w \), which is also a prefix code,
since the first, detachable part of the codeword tells its length. 
								       
For binary strings \( x \), natural numbers \( n \) and real numbers \( u \ge 1 \) 
we define
 \begin{alignat*}{3} 
        &J(u)	&&= u + 2 \log u, 
  \\ &\ig(x)	&&= \bg(|x|)^o x, 			
  \\ &\ig(n)	&&= \ig(\bg(n)).				
 \end{alignat*}
  We have \( \len{\ig(x)} \lea  J(\len{x}) \) for binary sequences \( x \), and 
\( \len{\ig(r)} \lea J(\log r) \) for numbers \( r \). 
Therefore~\eqref{eq:binupbd} is true with \( J(\log r) \) in place of \( 2\log r \). 
Of course, \( J(x) \)
could be replaced by still smaller functions, for example \( x + J(\log x) \).
We return to this topic later.

\section{Simple properties of information } \label{sec:propert}

If we transform a string \( x \) by applying to it a p.r.~function, 
then we cannot gain information over the amount contained in \( x \) plus 
the amount needed to describe the transformation. 
The string \( x \) becomes 
easier to describe, but it helps less in describing other strings \( y \). 

\begin{theorem} \label{thm:rectransf}
 For any partial recursive function \( U_{q} \), over 
strings, we have
 \begin{align*} C(U_{q}(x) \mvert y)	&\lea  C(x \mvert y) + J(\len{q}),
  \\   C(y \mvert U_{q}(x))	&\gea  C(y \mvert x) - J(\len{q}).
 \end{align*}
\end{theorem}
					    
\begin{proof}
 To define an interpreter \( A \) with \( C_A(U_{q}(x) \mvert y) \lea C(x \mvert y )
+ J(\len{q}) \), we define \( A(\ig(q)p,y) = U_{q}(U(p,y)) \). 
To define an
interpreter \( B \) with \( C_B(y \mvert x) \lea C(y \mvert U_{q}(x)) + J(q) \), we
define \( B(\ig(q)p,x) = U(p,U_{q}(x)) \).
 \end{proof}
	
The following notation is natural.

 \begin{definition}
The definition of conditional complexity is extended to pairs, etc. by 
 \[
 C(x_{1},\ldots,x_{m} \mvert y_{1},\ldots, y_{n}) 
	= C(\ang{ x_{1},\ldots,x_{m}}  \mvert \ang{ y_{1},\ldots,y_{n}} ).
 \]
 \end{definition}

With the new notation, here are some new results.

\begin{corollary} \label{c:upbdrec}
 For any one-to-one p.r. function \( U_{p} \), we have 
 \[ |C(x) - C(U_{p}(x))| \lea J(\len{p}).
 \]
  Further,
 \begin{align} 	C(x \mvert y,z)	&\lea  C(x \mvert U_{p}(y),z)+J(\len{p}), \nn
  \\	C(U_{p}(x)) 	&\lea  C(x) + J(\len{p}),\nn	
  \\	C(x \mvert z)	&\lea  C(x,y \mvert z), \label{eq:x.xy} 		
  \\	C(x \mvert y,z)	&\lea  C(x \mvert y),   \nn		
  \\	C(x,x)		&\eqa  C(x),       \label{eq:xx.x}
  \\	C(x,y \mvert z)	&\eqa  C(y,x \mvert z),	\nn	
  \\	C(x \mvert y,z)	&\eqa  C(x \mvert z,y),	\nn	
  \\	C(x,y \mvert x,z)	&\eqa  C(y \mvert x,z),	\nn	
  \\	C(x \mvert x,z)	&\eqa  C(x \mvert x) \eqa 0.  \nn		
 \end{align}
 \end{corollary}

We made several times implicit use of a basic additivity 
property of complexity which makes it possible to estimate the joint 
complexity of a pair by the complexities of its constituents. 
As expressed in Theorem~\ref{thm:addit}, it says essentially that to
describe the pair of strings, it is enough to know the description of
the first member and of a method to find the second member using our
knowledge of the first one.

 \begin{theorem} \label{thm:addit}
 \[  C(x,y) \lea J(C(x)) + C(y \mvert x).
 \]
 \end{theorem}
									
\begin{proof}
 We define an interpreter \( A \) as follows. 
We decompose any binary
string \( p \) into the form \( p = \ig(w)q \), and let \( A(p,z) = U(q,U(w)) \).
Then \( C_A(x,y) \) is bounded, to within an additive constant, by the 
right-hand side of the above inequality.
 \end{proof}

\begin{corollary} \label{cy:funcofpair}
 For any p.r.~function \( U_{p} \) over pairs of strings, we have 
 \begin{align*} 
C(U_{p}(x,y)) &\lea  J(C(x)) + C(y \mvert x) + J(\len{p})
  \\		  &\lea  J(C(x)) + C(y) + J(\len{p}),
 \end{align*}
  and in particular,
 \[
 	C(x,y) \lea J(C(x)) + C(y).
 \]
 \end{corollary}

This corollary implies the following continuity property of 
the function \( C(n) \): for any natural numbers \( n,h \), we have 
 \begin{equation}  
|C(n+h)-C(n)| \lea J(\log h). \label{eq:continui}
 \end{equation}
 Indeed, \( n+h \) is a recursive function of \( n \) and \( h \).
The term \( 2 \log C(x) \) making up the difference between
\( C(x) \) and \( J(C(x)) \) in Theorem~\ref{thm:addit} is attributable to the
fact that minimal descriptions cannot be concatenated without loosing
an ``endmarker''. 
It cannot be eliminated, since there is a constant
\( c \) such that for all \( n \), there are binary strings \( x,y \) of
length \( \le n \) with
 \[ 
  C(x) + C(y) + \log n < C(x,y) + c.
 \]
  Indeed, there are \( n 2^{n} \) pairs of binary strings whose sum of
lengths is \( n \). 
Hence by Theorem~\ref{thm:upbd}~\ref{howsharp}, there
will be a pair \( (x,y) \) of binary strings whose sum of lengths is \( n \),
with \( C(x,y) > n + \log n - 1 \). 
For these strings, 
inequality~\eqref{eq:binupbd} implies \( C(x)+C(y) \lea \len{x} + \len{y} \). 
Hence \( C(x)+C(y)+\log n \lea n + \log n < C(x,y) + 1 \). 

Regularities in a string will, in general, decrease its complexity
radically. 
If the whole string \( x \) of length \( n \) is given by some
rule, that is we have \( x(k) = U_{p}(k) \), for some recursive function
\( U_{p} \), then 
 \[  C(x) \lea C(n) + J(\len{p}).
 \]
  Indeed, let us define the p.r.~function \( V(q,k)=U_{q}(1)\ldots U_{q}(k) \). 
Then \( x=V(p,n) \) and the above estimate follows from
Corollary~\ref{c:upbdrec}.

For another example, let \( x=y_{1}y_{1}y_{2}y_{2} \cdots y_{n}y_{n} \), and 
\( y=y_{1}y_{2} \cdots y_{n} \). 
Then \( C(x) \eqa C(y) \) even though the string
\( x \) is twice longer. This follows from Corollary~\ref{c:upbdrec}
since \( x \) and \( y \) can be obtained from each other by a simple
recursive operation.

Not all ``regularities'' decrease the complexity of a string, only 
those which distinguish it from the mass of all strings, putting it 
into some class which is both small and algorithmically definable. For
a binary string of length \( n \), it is nothing unusual to have its number
of 0-s between \( n/2-\sqrt n \) and \( n/2 \). 
Therefore such strings can 
have maximal or almost maximal complexity. 
If \( x \) has \( k \) zeroes then the inequality 
 \[  
  C(x) \lea \log \binom{n}{k} + J(\log n) + J(\log k)
 \]
  follows from Theorems~\ref{thm:upbd} and~\ref{thm:condupbd}.
									
\begin{theorem} \label{thm:condupbd}
 Let \( E=W_{p} \) be an enumerable set of pairs of strings defined
enumerated with the help of the program \( p \) for example as 
 \begin{equation} 
 W_{p}= \setof{ U(p,x): x\in \dN}. \label{eq:Wp}
 \end{equation}
  We define the \df{section} 
 \begin{equation} 
 E^a = \setof{ x :  \tupcod{a,x}  \in E }. \label{eq:sectiondef}
 \end{equation}
  Then for all \( a \) and \( x \in E^a \), we have 
 \[  
  C(x \mvert a) \lea \log\card{E^a} + J(\len{p}).
 \]
 \end{theorem}									
\begin{proof}
 We define an interpreter \( A(q,b) \) as follows. 
We decompose \( q \) into \( q = \ig(p) \bg(t) \). 
From a standard enumeration \( (a(k),x(k)) \)
\( (k=1,2,\ldots) \) of the set \( W_{p} \), we produce an enumeration \( x(k,b) \)
\( (k=1,2,\ldots) \), without repetition, of the set \( W_{p}^b \) for each
\( b \), and define \( A(q,b)=x(t,b) \). 
For this interpreter \( A \), we get
\( k_A(x \mvert a) \lea J(\len{p}) + \log \card{E^a} \). 
 \end{proof}

It follows from Theorem~\ref{thm:upbd} that whenever the set \( E^a \) is
finite, the estimate of Theorem~\ref{thm:condupbd} is sharp for most of
its elements \( x \).

The shortest descriptions of a string \( x \) seem to carry some extra 
information in them, above the description of \( x \). 
This is suggested by the following strange identities.

\begin{theorem} \label{thm:xKx}
We have \( C(x,C(x)) \eqa C(x) \).
 \end{theorem}
\begin{proof}
 The inequality \( \gea \) follows from~\eqref{eq:x.xy}. 
To prove \( \lea \), let \( A(p) =  \ang{ U(p,\len{p})}  \). 
Then \( C_A(x,C(x)) \le C(x) \).
 \end{proof}

More generally, we have the following inequality.

\begin{theorem} \label{thm:strange}
 \[	
C(y\mvert x,\ i-C(y\mvert x,i)) \lea C(y\mvert x,i).
 \]
 \end{theorem}
 The proof is exercise.

The above inequalities are in some sense ``pathological'', and do not
necessarily hold for all ``reasonable'' definitions of descriptional
complexity.

\section{Algorithmic properties of complexity}

The function \( C(x) \) is not computable, as it will be shown in this 
section. 
However, it has a property closely related to computability.
Let \( \dQ \) be the set of rational numbers.

\begin{definition}[Computability]
Let \( f:\dS\to\dR \) be a function.
It is \df{computable} if there is a recursive function \( g(x,n) \) with rational
values, and  \( |f(x)-g(x,n)| < 1/n \). 
\end{definition}

 \begin{definition}[Semicomputability]\label{def:semicomputability}
A function \( f : \dS \to \rint{-\infty}{\infty} \) is \df{(lower) semicomputable}
if the set 
 \[
	\setof{ \tup{x,r} : x \in \dS,\; r \in \dQ,\; r<f(x) }
 \] 
 is recursively enumerable. 
A function \( f \) is called \df{upper semicomputable} if \( -f \) is
lower semicomputable.
 \end{definition}

It is easy to show the following:

\begin{bullets}
\item A function \( f:\dS\to\dR \) is lower semicomputable iff there 
exists a recursive function with rational values, (or, equivalently, a 
computable real function) \( g(x,n) \) nondecreasing in \( n \), with 
\( f(x) = \lim_{n\to\infty}g(x,n) \). 
 \item A function \( f:\dS\to\dR \) is computable if it is both lower and upper
   semicomputable. 
\end{bullets}

The notion of semicomputibility is naturally
extendable over other discrete domains, like \( \dS \times \dS \).
It is also useful to extend it to functions \( \dR\to\dR \).
Let \( \cI \) denote the set of open rational intervals of \( \dR \), that is
 \begin{align*}
 \bg=\setof{\opint{p}{q}: p,q\in\dQ, p<q}.
 \end{align*}

\begin{definition}[Computability for real functions]
Let \( f:\dR\to\dR \) be a function.
It is \df{computable} if there is a recursively enumerable set 
\( \cF\sbsq \bg^{2} \)
such that denoting \( \cF_{J}=\setof{I: (I,J)\in \cF} \) we have
 \begin{align*}
   f^{-1}(J) = \bigcup_{I\in\cF_{J}}I.
 \end{align*}
\end{definition}
This says that if \( f(x)\in I \) then sooner or later we will find
an interval \( I\in\cF_{J} \) with the property that \( f(z)\in J \) for all \( z\in I \).
Note that computability implies continuity.

\begin{definition}
  We say that \( f:\dR\to\ol\dR \) is
\df{lower semicomputable} if there is a recursively enumerable set 
\( \cG\sbsq \cQ\times\bg \)
such that denoting \( \cG_{q}=\setof{I: (I,q)\in \cG} \) we have
 \begin{align*}
   f^{-1}(\opint{q}{\infty}) = \bigcup_{I\in\cF_{J}}I.
 \end{align*}
\end{definition}
This says that if \( f(x)>r \) then sooner or later we will find
an interval \( I\in\cF_{J} \) with the property that \( f(z)>q \) for all \( z\in I \).

The following facts are useful and simple to verify:

\begin{proposition}\label{propo:l-sc-properties}\ 
  \begin{alphenum}
  \item 
  The function \( \cei{\cdot}:\dR\to\dR \) is lower semicomputable.
  \item If \( f,g:\dR\to\dR \) are computable then their composition
    \( f(g(\cdot)) \) is, too.
  \item If \( f:\dR\to\dR \) is lower semicomputable and \( g:\dR\to\dR \) 
(or \( \dS\to\dR \)) is computable then \( f(g(\cdot)) \) is lower semicomputable.
  \item If \( f:\dR\to\dR \) is lower semicomputable and monotonic,
and \( g:\dR\to\dR \) (or \( \dS\to\dR \)) is lower 
semicomputable then \( f(g(\cdot)) \) is lower semicomputable.
  \end{alphenum}
\end{proposition}

 \begin{definition}
We introduce a \df{universal lower semicomputable function}. 
For every binary string \( p \), and \( x \in \dS \), let
 \[ 
  S_{p}(x)= \sup\setof{ y/z : y,z\in Z, z \not = 0, 
	\ang{\ang{ x } ,y,z }  \in W_{p} }
 \]
where \( W_{p} \) was defined in~\eqref{eq:Wp}.  
 \end{definition}

The function \( S_{p}(x) \) is lower
semicomputable as a function of the pair \( (p,x) \), and for different
values of \( p \), it enumerates all semicomputable functions of \( x \). 

 \begin{definition}
Let \( S_{p}(x_{1},\ldots,x_{k})= S_{p}(\ang{ x_{1},\ldots,x_{k} } ) \).  
For any lower semicomputable function \( f \), we call any binary string \( p \) with
\( S_{p}=f \) a \df{G\"odel number} of \( f \). 
There is no universal computable function, but for any computable 
function \( g \), we call any number \(  \tupcod{  \tupcod{ p } , \tupcod{ q } }  \) 
a G\"odel number of \( g \) if \( g=S_{p}=-S_{q} \).
 \end{definition}

With this notion, we claim the following.

\begin{theorem}	      \label{thm:Ksemicp}
 The function \( C(x \mvert y) \) is upper semicomputable.
 \end{theorem}
\begin{proof} 
 By Theorem~\ref{thm:upbd} there is a constant \( c \) such that for
any strings \( x,y \), we have \( C(x \mvert y) < \log \tupcod{  x  }  + c \).  
Let some Turing machine compute our optimal binary p.r.~interpreter. 
Let \( U^{t}(p,y) \) be defined as \( U(p,y) \) if this machine, when started on
input \( (p,y) \), gives an output in \( t \) steps, undefined otherwise. 
Let \( K^{t}(x \mvert y) \) be the smaller of \( \log  \tupcod{  x  }  + c \) and 
 \[  
   \min\setof{ \len{p} \le t : U^{t}(p,y)=x }.
 \]
Then the function \( K^{t}(x \mvert y) \) is computable, monotonic in \( t \) and 
\( \lim_{t\to\infty} K^{t}(x \mvert y) = C(x \mvert y) \). 
 \end{proof}

The function \( C(n) \) is not computable. Moreover, it has no nontrivial 
partial recursive lower bound.

\begin{theorem} \label{thm:nolb}
 Let \( U_{p}(n) \) be a partial recursive function whose values are numbers
smaller than \( C(n) \) whenever defined. Then 
 \[  
 U_{p}(n) \lea J(\len{p}).
 \]
 \end{theorem}
\begin{proof}
 The proof of this theorem resembles ``Berry's paradox '', which says:
``The least number that cannot be defined in less than 100 words'' is a 
definition of that number in 12 words. 
The paradox can be used to 
prove, whenever we agree on a formal notion of ``definition'', that the 
sentence in quotes is not a definition. 

Let \( x(p,k) \) for \( k=1,2,\ldots \) be an enumeration of the domain of 
\( U_{p} \). For any natural number \( m \) in the range of \( U_{p} \), let \( f(p,m) \) be 
equal to the first \( x(p,k) \) with \( U_{p}(x(p,k)) = m \). Then by definition, 
\( m < C(f(p,m)) \). On the other hand, \( f(p,m) \) is a p.r.~function, hence 
applying Corollary~\ref{cy:funcofpair} and~\eqref{eq:binupbd} we get 
 \begin{align*}
     m	&< C(f(p,m)) \lea C(p) + J(C(m))
   \\	&\lea \len{p} + 1.5\log m. 		
 \end{align*}
 \end{proof}

Church's first example of an undecidable recursively enumerable set 
used universal partial recursive functions and the diagonal technique 
of G\"odel and Cantor, which is in close connection to Russel's
paradox. 
The set he constructed is \df{complete} in the sense of 
``many-one reducibility''.

The first example of sets not complete in this sense were Post's 
simple sets. 

 \begin{definition}
A computably enumerable set is \df{simple} if its complement is infinite but
does not contain any infinite computably enumerable subsets. 
 \end{definition}

The paradox used in the
above proof is not equivalent in any trivial way to Russell's
paradox, since the undecidable computably enumerable
sets we get from Theorem~\ref{thm:nolb} 
are simple. 
Let \( f(n) \le \log n \) be any recursive 
function with \(  \lim_{n\to\infty} f(n) = \infty  \). 
We show that the set \( E=\setof{ n : C(n) \le f(n) } \) is simple. 
It follows from
Theorem~\ref{thm:Ksemicp} that \( E \) is recursively enumerable, and from
Theorem~\ref{thm:upbd} that its complement is infinite. 
Let \( A \) be any computably enumerable set disjoint from \( E \). 
The restriction \( f_A(n) \) of the function
\( f(n) \) to \( A \) is a p.r.~lower bound for \( C(n) \). 
It follows from Theorem~\ref{thm:nolb} that \( f(A) \) is bounded. 
Since \( f(n)\to\infty \), this is possible only if \( A \) is finite. 

G\"odel used the diagonal technique to prove the incompleteness of
any sufficiently rich logical theory with a recursively enumerable
axiom system. 
His technique provides for any sufficiently rich computably enumerable theory 
a concrete example of an undecidable sentence. 
Theorem~\ref{thm:nolb} leads to a new proof of this result, with essentially
different examples of undecidable propositions. Let us consider any
first-order language \( L \) containing the language of standard
first-order arithmetic. 
Let \(  \tupcod{ \cdot }  \) be a standard
encoding of the formulas of this theory into natural numbers.  
Let the computably enumerable theory \( T_{p} \) be the theory for which the set of codes 
\( \tupcod{ \Phi }  \) of its axioms \( \Phi \) is the computably enumerable set \( W_{p} \) of
natural numbers. 
									
\begin{corollary}
There is a constant \( c<\infty \) such that for any \( p \), if any sentences
with the meaning ``\( m < C(n) \)'' for \( m > J(\len{p}) + c \)  are 
provable in theory \( T_{p} \) then some of these sentences are false. 
 \end{corollary}
\begin{proof} For some \( p \), suppose that all sentences ``\( m<C(n) \)''
provable in \( T_{p} \) are true. 
For a sentence \( \Phi \), let \( T_{p} \vdash \Phi \) 
denote that \( \Phi \) is a theorem of the theory \( T_{p} \). 
Let us define the function 
 \[
 A(p,n) = \max\setof{ m : T_{p} \vdash \text{``}m < C(n)\text{''} }.
 \]
 This function is semicomputable since the set
\( \setof{(p,\tupcod{\Phi}): T_{p}\vdash\Phi} \) is recursively
enumerable.  
Therefore there is a binary string \( a \) such that
\( A(p,n)=S(a,\tupcod{ p,n}) \). 
It is easy to see that we have a constant string \( b \) such that 
 \[
   S(a, \tupcod{ p,n } ) = S(q, n)
 \]
  where \( q = \ol b \ol a p \). 
(Formally, this statement follows for example from the so-called \( S_{m}^{n} \)-theorem.)  
The function \( S_{q}(n) \) is a
lower bound on the complexity function \( C(n) \). 
By an immediate
generalization of Theorem~\ref{thm:nolb} for semicomputable functions),
we have \( S(q,n) \lea J(\len{q}) \lea  J(\len{p}) \).
 \end{proof}

We have seen that the complexity function \( C(x \mvert y) \) is not
computable, and does not have any nontrivial recursive lower bounds.
It has many interesting upper bounds, however, for example the ones in
Theorems~\ref{thm:upbd} and~\ref{thm:condupbd}.  
The following theorem
characterizes all upper semicomputable upper bounds (\df{majorants})
of \( C(x \mvert y) \).

\begin{theorem}[Levin]  \label{thm:Kmajor}
Let \( F(x,y) \) be a function of strings semicomputable from above. 
The relation \( C(x \mvert y) \lea F(x,y) \) holds for all \( x,y \) if and only if 
we have 
 \begin{equation} \label{eq:Kmajor}
  \log |\setof{ x : F(x,y ) < m } | \lea m  
 \end{equation}
  for all strings \( y \) and natural numbers \( m \). 
 \end{theorem}
\begin{proof}
 Suppose that \( C(x \mvert y) \lea F(x,y) \) holds. 
Then~\eqref{eq:Kmajor} follows from~\ref{howsharp} of Theorem~\ref{thm:upbd}.

Now suppose that~\eqref{eq:Kmajor} holds for all \( y,m \). 
Then let \( E \) be the computably enumerable set of triples \( (x,y,m) \) with \( F(x,y) < m \). 
It follows from~\eqref{eq:Kmajor} and Theorem~\ref{thm:condupbd} that for all
\( (x,y,m)\in E \) we have \( C(x \mvert y,m) \lea m \).  
By Theorem~\ref{thm:strange} then \( C(x \mvert y) \lea m \). 
 \end{proof}

The minimum of any finite number of majorants is again a majorant. 
Thus, 
we can combine different heuristics for recognizing patterns of sequences. 

\section{The coding theorem} \label{sec:codingth} 

\subsection{Self-delimiting complexity}

Some theorems on the addition of complexity do not have as simple a form as
desirable.
For example, Theorem~\ref{thm:addit} implies
 \[
   C(x, y) \lea J(C(x)) + C(y).
 \]
We would like to see just \( C(x) \) on the right-hand side, but 
the inequality \( C(x, y) \lea C(x) + C(y) \) does not always hold
(see Exercise~\ref{xrc.need-prefix}).
The problem is that if we compose a program for \( (x,y) \) from those of \( x \)
and \( y \) then we have to separate them from each other somehow.
In this section, we introduce a variant of Kolmogorov's 
complexity discovered by Levin and independently by Chaitin and Schnorr,
which has the property that ``programs'' are \emph{self-delimiting}:
this will free us of the necessity to separate them from each other.

 \begin{definition}
A set of strings is called \df{prefix-free} if for any pair \( x,y \) of
elements in it, \( x \) is not a prefix of \( y \).

A one-to-one function into a set of strings is a
\df{prefix code}, or \df{instantaneous code} if its domain of definition is
prefix free. 

An interpreter
\( f(p,x) \) is \df{self-delimiting (s.d.)} if for each \( x \), the set
\( D_{x} \) of strings \( p \) for which \( f(p,x) \) is defined is a prefix-free set.
 \end{definition}

 \begin{example}
The mapping \( p \to p^o \) is a prefix code.
 \end{example}

A self-delimiting p.r.~function \( f(p) \) can be imagined as a function 
computable on a special self-delimiting Turing machine.

\begin{definition}
A Turing machine \( \cT \) is \df{self-delimiting} if it
has no input tape, but can ask any time for a new input symbol. 
After some computation and a few such requests (if ever) the
machine decides to write the output and stop. 
\end{definition}

The essential 
difference between a self-delimiting machine \( \cT \)
and an ordinary Turing machine is that 
\( \cT \) does not know in advance how many input symbols suffice; she
must compute this information from the input symbols themselves,
without the help of an ``endmarker''. 

Requiring an interpreter \( A(p,x) \) to be self-delimiting in \( p \) has
many advantages. 
First of all, when we concatenate descriptions, the
interpreter will be able to separate them without endmarkers. 
Lost endmarkers were responsible for the additional logarithmic terms and
the use of the function \( J(n) \) in Theorem~\ref{thm:addit} and several
other formulas for \( K \). 

\begin{definition}
A self-delimiting partial recursive 
interpreter \( T \) is called \df{optimal} if for any other s.d.~p.r.~interpreter \( F \), 
we have 
 \begin{equation} 	C_{T} \lea C_{F}. \label{eq:sdopt}  
 \end{equation}
\end{definition}

\begin{theorem} \label{thm:sdopt} % 4.3
There is an optimal s.d.~interpreter.
\end{theorem}
\begin{proof}								
The proof of this theorem is similar to the proof of Theorem~\ref{thm:invar}.
We take the universal partial recursive function \( V(a,p,x) \). 
We transform each function 
\( V_a(p,x)=V(a,p,x) \) into a self-delimiting function
\( W_a(p,x)=W(a,p,x) \), but so as not to change the functions \( V_a \)
which are already self-delimiting. 
Then we form \( T \) from \( W \) just as we formed \( U \) from \( V \). 
It is easy to check that the function \( T \)
thus defined has the desired properties. 
Therefore we are done if we construct \( W \).

Let some Turing machine \( \cM \) compute \( y=V_a(p,x) \) in \( t \) steps.
We define \( W_a(p,x)=y \) if \( \len{p} \le t \) and if \( \cM \) does not
compute in \( t \) or fewer steps any output \( V_a(q,x) \) from any
extension or prefix \( q \) of length \( \le t \) of \( p \). 
If \( V_a \) is
self-delimiting then \( W_a=V_a \). But \( W_a \) is always self-delimiting.
Indeed, suppose that \( p_{0} \prefix p_{1} \) and that \( \cM \) computes
\( V_a(p_{i},x) \) in \( t_{i} \) steps. 
The value \( W_a(p_{i},x) \) can be defined only if \( \len{p_{i}} \le t_{i} \). 
The definition of \( W_a \) guarantees that if
\( t_{0} \le t_{1} \) then \( W_a(p_{1},x) \) is undefined, otherwise \( W_a(p_{0},x) \)
is undefined. 
 \end{proof}

 \begin{definition}
We fix an optimal self-delimiting p.r.~interpreter \( T(p,x) \) and write
 \[ 	
K(x \mvert y) = C_{T}(x \mvert y).
 \]
We call \( K(y \mvert x) \) the (self-delimiting) \df{conditional complexity} 
of \( x \) with respect to \( y \). 
 \end{definition}

We will use the expression
``self-delimiting'' only if confusion may arise. Otherwise, under
complexity, we generally understand the self-delimiting complexity.
Let \( T(p)=T(p,\Lg) \). 
All definitions for unconditional complexity,
joint complexity etc.~are automatically in force  for \( H=C_{T} \). 

Let us remark that the s.d.~interpreter defined in Theorem~\ref{thm:sdopt} 
has the following stronger property.  
For any other
s.d.~interpreter \( G \) there is a string \( g \) such that we have
 \begin{equation}  
 G(x)=T(gx) \label{eq:univsd}
 \end{equation}
  for all \( x \).  
We could say that the interpreter \( T \) is \df{universal}.

Let us show that the functions \( K \) and \( H \) are asymptotically equal. 
			      
\begin{theorem} \label{thm:HK} 
 \begin{equation} K \lea H \lea J(K). 
 \end{equation}
 \end{theorem}
\begin{proof} Obviously, \( K \lea H \). We define a self-delimiting
p.r.~function \( F(p,x) \). 
The machine computing \( F \) tries to decompose 
\( p \) into the form \( p=u^o v \) such that \( u \) is the number \( \len{v} \) in binary
notation. 
If it succeeds, it outputs \( U(v,x) \). 
We have 
 \[ 
K(y \mvert x) \lea C_{F}(y \mvert x) \lea C(y \mvert x) + 2\log C(y \mvert x).
 \]
 \end{proof}

Let us review some of the relations proved in Sections~\ref{sec:kolm}
and~\ref{sec:propert}. 
Instead of the simple estimate \( C(x) \lea n \) for
a binary string \( x \), of length \( n \), only the obvious consequence of
Theorem~\ref{thm:HK} holds that is
 \[
   K(x) \lea J(n).  
 \]
 We must similarly change Theorem~\ref{thm:condupbd}.  
We will use the
universal p.r.~function \( T_{p}(x) = T(p,x) \). The r.e.\ set \( V_{p} \) is
the range of the function \( T_{p} \).  Let \( E=V_{p} \) be an enumerable set of
pairs of strings. 
The section \( E^a \) is defined as in~\eqref{eq:sectiondef}.  
Then for all \( x \in E^a \), we have
 \begin{equation} 
K(x \mvert a) \lea J(\log \card{E^a}) + \len{p}. 
\label{eq:condHupbd} 
 \end{equation}
  We do not need \( J(\len{p}) \) since the function \( T \) is self-delimiting.
However, the real analogon of Theorem~\ref{thm:condupbd} 
is the Coding Theorem, to be proved later in this section.

Part~\ref{howsharp} of Theorem~\ref{thm:upbd} holds for \( H \), but is not
the strongest what can be said. 
The counterpart of this theorem is
the inequality~\eqref{eq:condHsum} below. 
Theorem~\ref{thm:rectransf} and
its corollary hold without change for \( H \). 
The additivity relations
will be proved in a sharper form for \( H \) in the next section.

\subsection{Universal semimeasure}

Self-delimiting complexity has an interesting characterization 
that is very useful in applications.

  \begin{definition}
A function \( w : S \to \clint{0}{1} \) is called a \df{semimeasure}
over the space \( S \) if 
 \begin{equation} \label{eq:dscrsemi} 
  \sum_{x} w(x) \le 1. 
 \end{equation}
 It is called a \df{measure} (a probability distribution) if
equality holds here. 

A semimeasure is called \df{constructive} if it
is lower semicomputable. 
A constructive semimeasure which
multiplicatively dominates every other constructive semimeasure is
called \df{universal}.
  \end{definition}

 \begin{remark}
Just as in an earlier remark, we warn that semimeasures will later be
defined over the space \( \dN^\dN \). 
They will be characterized by a
nonnegative function \( w \) over \( \dS=\dN^{*} \) but with a condition different
from~\eqref{eq:dscrsemi}.
 \end{remark}

The following remark is useful.

 \begin{proposition}
If a constructive semimeasure \( w \) is also a measure then it is
computable. 
 \end{proposition}
 \begin{proof}
If we compute an approximation \( w_{t} \) of the
function \( w \) from below for which \( 1-\eps < \sum_{x} w_{t}(x) \) then we
have \( | w(x)-w_{t}(x) | < \eps \) for all \( x \). 
 \end{proof}

 \begin{theorem} \label{thm:univsem} 
There is a universal constructive semimeasure.
 \end{theorem}
\begin{proof} 
 Every computable family \( w_{i} \) \( (i=1,2,\dotsc) \) of semimeasures is
dominated by the semimeasure \( \sum_{i} 2^{-i}w_{i} \). 
Since there are only
countably many constructive semimeasures there exists a semimeasure
dominating them all. 
The only point to prove is that the dominating
semimeasure can be made constructive. 

Below, we construct a function \( \mu_{p}(x) \) lower
semicomputable in \( p,x \) such 
that \( \mu_{p} \) as a function of \( x \) is a constructive semimeasure for
all \( p \) and all constructive semimeasures occur among the \( w_{p} \). 
Once we have \( \mu_{p} \) we are done because for any positive computable
function \( \dg(p) \) with \( \sum_{p} \dg(p) \le 1 \) the semimeasure 
\( \sum_{p} \dg(p) \mu_{p} \) is obviously lower
semicomputable and dominates all the
\( \mu_{p} \). 

Let \( S_{p}(x) \) be the lower semicomputable function with G\"odel number \( p \)
and let \( S_{p}^{t}(x) \) be a function recursive in \( p,t,x \) with rational
values, dondecreasing in \( t \), such that 
\( \lim_{t} S^{t}_{p}(x) = \max\{0,S_{p}(x)\} \) and for each \( t,p \), the function 
\( S_{p}^{t}(x) \) is
different from 0 only for \( x \le t \). Let us define \( \mu_{p}^{t} \)
recursively in \( t \) as follows. 
Let \( \mu_{p}^0(x) = 0 \). 
Suppose that \( \mu_{p}^{t} \) is already defined. 
If \( S_{p}^{t+1} \) is a semimeasure then we
define \( \mu_{p}^{t+1}=S_{p}^{t+1} \), otherwise \( \mu_{p}^{t+1} = \mu_{p}^{t} \).
Let \( \mu_{p} = \lim_{t} \mu_{p}^{t} \). 
The function \( \mu_{p}(x) \) is, by its
definition, lower semicomputable in \( p,x \) and a semimeasure for each fixed
\( p \). 
It is equal to \( S_{p} \) whenever \( S_{p} \) is a semimeasure, hence it
enumerates all constructive semimeasures.
 \end{proof}

The above theorem justifies the following notation.

\begin{definition}\label{def:m}
We choose a fixed universal constructive semimeasure and call it \( \m(x) \).
\end{definition}

With this notation, we can make it a little more precise in which sense
this measure dominates all other constructive semimeasures.

 \begin{definition}
For an arbitrary constructive semimeasure \( \nu \) let us define
 \begin{equation}\label{eq:m(nu)}
   \m(\nu) = 
\sum\setof{\m(p) : \txt{the (self-delimiting) program \( p \) computes \( \nu \)}}.
 \end{equation}
Similar notation will be used for other objects: for example
now \( \m(f) \) make sense for a recursive function \( f \).
 \end{definition}

Let us apply the above concept to some interesting cases.

 \begin{theorem}\label{thm:exact-domin} 
For all constructive semimeasures \( \nu \) and for all strings \( x \), we have 
 \begin{equation}\label{eq:exact-domin}
  \m(\nu) \nu(x) \lem \m(x).
 \end{equation}
 \end{theorem}
 \begin{proof}
Let us repeat the proof of Theorem~\ref{thm:univsem}, using \( \m(p) \) in place
of \( \dg(p) \).
We obtain a new constructive semimeasure \( \m'(x) \) with the property
\( \m(\nu)\nu(x) \le \m'(x) \) for all \( x \).
Noting \( \m'(x) \eqm \m(x) \) finishes the proof.
 \end{proof}

\subsection{Prefix codes}\label{subsec:prefix-codes}

The main theorem of this section says \( K(x) \eqa -\log \m(x) \). 
Before
proving it, we relate descriptional complexity to the classical coding
problem of information theory.

 \begin{definition}
A (binary) \df{code} is any
mapping from a set \( E \) of binary sequences to a set of objects. 
This mapping is the \df{decoding function} of the code. 
If the mapping is
one-to-one, then sometimes the set \( E \) itself is also called a code.

A code is a \df{prefix code} if \( E \) is a prefix-free set.
 \end{definition}

The interpreter \( U(p) \) is a decoding function. 

\begin{definition}
We call string \( p \) the first
shortest description of \( x \) if it is the lexicographically first
binary word \( q \) with \( |q| = C(x) \) and \( U(q)=x \). 
\end{definition}

The set of first shortest descriptions is a code.  
This imposes the implicit lower bound stated in~\eqref{eq:Kmajor} on \( C(x) \):
 \begin{equation}
	\card{\setof{ x : C(x)=n }} \le 2^{n}. % (4.2.4)
 \end{equation}

The least shortest descriptions in the definition of \( K(x) \) form moreover a 
\emph{prefix code}.
We recall a classical result of information theory.

 \begin{lemma}[Kraft's Inequality, see~\protect\cite{CsiszarKornerBook81}]
\label{l:kraft} 
For any sequence \( l_{1},l_{2},\ldots \) of
natural numbers, there is a prefix code with exactly these numbers as
codeword lengths, if and only if
 \begin{equation}	\sum_{i} 2^{-l_{i}} \le 1. \label{eq:kraft} 
 \end{equation}
 \end{lemma}
\begin{proof} 
Recall the standard correspondence \( x \leftrightarrow [x] \) between
binary strings and binary subintervals of the interval \( \clint{0}{1} \).  
A prefix code corresponds to a set of disjoint intervals, and the 
length of the interval \( [x] \) is \( 2^{-\len{x}} \). 
This proves 
that~\eqref{eq:kraft} holds for the lengths of codewords of a prefix code.

Suppose that \( l_{i} \) is given and~\eqref{eq:kraft} holds. 
We can assume that the sequence \( l_{i} \) is nondecreasing. 
Chop disjoint, adjacent 
intervals \( I_{1},I_{2},\ldots \) of length \( 2^{-l_{1}},2^{-l_{2}},\ldots \) from
the left end of the interval \( [0,1] \). 
The right end of \( I_{k} \) is \( \sum_{j=1}^k 2^{-l_{j}} \). 
Since 
the sequence \( l_{j} \) is nondecreasing, 
all intervals \( I_{j} \) are binary intervals. 
Take the binary string  corresponding to \( I_{j} \) as the \( j \)-th codeword. 
 \end{proof}
							
 \begin{corollary}
We have \( -\log \m(x) \lea K(x) \).
 \end{corollary}
 \begin{proof}
The lemma implies 
 \begin{equation} \sum_y 2^{-K(y \mvert x)} \le 1.  \label{eq:condHsum} 
 \end{equation}
  Since \( K(x) \) is an upper semicomputable function, the function
\( 2^{-K(x)} \) is lower semicomputable. 
Hence it is a constructive semimeasure, and we have \( -\log \m(x) \lea K(x) \).
 \end{proof}

The construction in the second part of the proof of Lemma~\ref{l:kraft}
 has some disadvantages.
We do not always want to rearrange the numbers \( l_{j} \),
for example because we want that the order of the codewords reflect the
order of our original objects. 
Without rearrangement, we can still
achieve a code with only slightly longer codewords.  

\begin{lemma}[Shannon-Fano code] (see~\cite{CsiszarKornerBook81})
 Let \( w_{1},w_{2},\ldots \) be positive numbers with \( \sum_{j} w_{j} \le 1 \).
There is a binary prefix code \( p_{1},p_{2},\ldots \) where the codewords
are in lexicographical order, such that 
 \begin{equation} 	|p_{j}| \le -\log w_{j} + 2. \label{eq:shfano} 
 \end{equation}
 \end{lemma}
 \begin{proof} We follow the construction above and cut off disjoint,
adjacent (not necessarily binary) intervals \( I_{j} \) of length \( w_{j} \) from 
the left end of \( [0,1] \). Let \( v_{j} \) be the length of the longest binary 
intervals contained in \( I_{j} \). 
Let \( p_{j} \) be the binary word corresponding to the first one of these. 
Four or fewer intervals of length \( v_{j} \) cover \( I_{j} \). 
Therefore~\eqref{eq:shfano} holds. 
\end{proof}

\subsection{The coding theorem for \texorpdfstring{\( K(x) \)}{\emph{K(x)}}}

The above results suggest \( K(x) \le -\log\m(x) \). 
But in the proof, 
we have to deal with the problem that \( \m \) is not computable. 

\begin{theorem}[Coding Theorem, see~\protect\cite{Levin74,GacsSymm74,Chaitin75}]
\label{thm:coding} 
  \begin{equation}  	K(x) \eqa -\log\m(x) \label{eq:codingth} 
  \end{equation}
 \end{theorem}
							
\begin{proof} 
 We construct a self-delimiting p.r.~function \( F(p) \) with the
property that \( C_{F}(x) \le -\log\m(x) + 4 \). The function \( F \) to be 
constructed is the decoding function of a prefix code hence the 
code-construction of Lemma~\ref{l:kraft} proposes itself. 
But since
the function \( \m(x) \) is only lower semicomputable, it is given only by a
sequence converging to it from below. 

Let \( \setof{ \pair{z_{t}}{k_{t}} : t=1,2,\ldots } \) be a recursive enumeration
of the set \( \setof{ \pair{x}{k} : k < \m(x) } \) without repetition. Then 
 \[ 
  \sum_{t} 2^{-k_{t}} = \sum_{x} \sum_{z_{t}=1} 2^{-k_{t}} \le \sum_{x} 2\m(x) < 2.
 \]
  Let us cut off consecutive adjacent, disjoint intervals \( I_{t} \) of
length \( 2^{-k_{t}-1} \) from the left side of the interval \( \clint{0}{1} \). 
We define \( F \) as follows. If \( [p] \) is a largest binary subinterval of
some \( I_{t} \) then \( F(p)=z_{t} \). Otherwise \( F(p) \) is undefined. 

The function \( F \) is obviously self-delimiting and partial recursive. It 
follows from the construction that for every \( x \) there is a \( t \) with 
\( z_{t}=x \) and \( 0.5\m(x) < 2^{-k_{t}} \). Therefore, for every \( x \) 
there is a \( p \) such that \( F(p)=x \) and \( |p| \le -\log\m(x) + 4 \). 
 \end{proof}

The Coding Theorem can be straightforwardly generalized as follows.
Let \( f(x,y) \) be a lower semicomputable nonnegative function.
Then we have
 \begin{equation}\label{eq:gencodth} 
 	K(y \mvert x) \lea -\log f(x,y)  
 \end{equation}
for all \( x \) with \( \sum_y f(x,y) \le 1 \). The proof is the same.

\subsection{Algorithmic probability}

We can interpret the Coding Theorem as follows. 
It is known in
classical information theory that for any probability distribution
\( w(x) \) \( (x \in S) \) a binary prefix code \( f_w(x) \) can be constructed
such that \( \len{f_w(x)} \le -\log w(x) + 1 \). 
We learned that for
computable, moreover, even for contstructive distributions \( w \) there
is a \df{universal code} with a self-delimiting partial-recursive
decoding function \( T \) independent of \( w \) such that for the codeword
length \( K(x) \) we have 
 \[
	K(x) \le -\log w(x) + c_w.
 \]
  Here, only the additive constant \( c_w \) depends on the distribution
\( w \). 

Let us imagine the self-delimiting Turing machine \( \cT \) computing
our interpreter \( T(p) \). 
Every time the machine asks for a new bit of
the description, we could toss a coin to decide whether to give 0 or 1.

\begin{definition}
Let \( P_{T}(x) \) be the probability that the self-delimiting
machine \( \cT \) gives out result \( x \) after receiving random bits as inputs.
\end{definition}

We can write \( P_{T}(x)=\sum W_{T}(x)  \) where \( W_{T}(x) \) is the set 
\( \setof{ 2^{-\len{p}} : T(p) = x } \).  
Since \( 2^{-K(x)} = \max W_{T}(x) \), we have \( -\log P_{T}(x) \le K(x) \).  
The semimeasure \( P_{T} \) is constructive, hence \( P \lem \m \), hence
 \[
	H \lea -\log\m \lea -\log P_{T} \le H,
 \]
  hence
 \[
	-\log P_{T} \eqa -\log\m \eqa H. 
 \]
  Hence the sum \( P_{T}(x) \) of the set \( W_{T}(x) \) is at most a constant
times larger than its maximal element \( 2^{-K(x)} \). 
This relation was not obvious in advance. 
The outcome \( x \) might have high probability
because it has many long descriptions. 
But we found that then it must have a short description too.
In what follows it will be convenient to fix the definition of \( \m(x) \) as
follows.

 \begin{definition}From now on let us define
 \begin{equation}\label{eq:m(x)}
   \m(x) = P_{T}(x).
 \end{equation}
 \end{definition}
%%% Local Variables: 
%%% mode: latex
%%% TeX-master: "ait-notes"
%%% End: 

\section{The statistics of description length } 

We can informally summarize the relation $H \eqa -\log P_{T}$ saying that
\emph{if an object has many long descriptions then it has a short one}.
But how many descriptions does an object really have?
With the Coding Theorem, we can answer several questions of this kind.
The reader can view this section as a series of exercises in the technique
acquired up to now.

\begin{theorem} \label{thm:Hstat} 
  Let $f(x,n)$ be the number of binary strings $p$ of length $n$ with
$T(p)=x$ for the universal s.d.~interpreter $T$ defined in 
Theorem~\ref{thm:sdopt}.
  Then for every $n \ge K(x)$, we have
  \begin{equation}
	\log f(x,n) \eqa n - K(x,n). \label{eq:Hstat} 
  \end{equation}
 \end{theorem}

  Using $K(x) \lea K(x,n)$, which holds just as~\eqref{eq:x.xy}, and
substituting $n=K(x)$ we obtain
$\log f(x,K(x)) \eqa K(x) - K(x,K(x)) \lea 0$, implying the following.
 
  \begin{corollary} The number of shortest descriptions of any object is
bounded by a universal constant.
 \end{corollary}

  Since the number $\log f(x,K(x))$ is nonnegative, we also derived the
identity
 \begin{equation} \label{eq:xHx} 
 	K(x, K(x)) \eqa K(x)
 \end{equation}
  which could have been proven also in the same way as
Theorem~\ref{thm:xKx}. 
					  
\begin{proof}[Proof of Theorem \protect\ref{thm:Hstat}]
 It is more convenient to prove equation~\eqref{eq:Hstat}
in the form 
 \begin{equation} \label{eq:Hstatmult} 
 	2^{-n}f(x,n) \eqm \m(x,n)
 \end{equation}
 where $\m(x,y) = \m(\ang{ x,y} )$.
First we prove $\lem$.
Let us define a p.r.~self-delimiting function $F$ by
 $F(p) = \tupcod{T(p), \len{p}}$.
Applying $F$ to a coin-tossing argument, $2^{-n}f(x,n)$ is the
probability that the pair $\ang{x, n} $ is obtained.
Therefore the left-hand side of~\eqref{eq:Hstatmult}, as a function of
$\ang{ x,n}$, is a constructive semimeasure, dominated by $\m(x,n)$.

  Now we prove $\gem$.
  We define a self-delimiting p.r.~function $G$ as follows.
  The machine computing $G(p)$ tries to decompose $p$ into
three segments $p=\bg(c)^o v w$ in such a way that $T(v)$ is a pair 
$\ang{ x, \len{p} + c }$.
  If it succeeds then it outputs $x$.
  By the universality of $T$, there is a binary string $g$ such that
$T(gp)=G(p)$ for all $p$.
  Let $r=\len{g}$.
 For an arbitrary pair $x,n$, let $q$ be a shortest binary string with
$T(q) = \tupcod{ x,n}$, and $w$ an arbitrary string of length
 \[
 	l = n - \len{q} - r - \len{\bg(r)^o}.
 \]
  Then $G(\bg(r)^o q w) = T(g \bg(r)^o q w) = x$.
  Since $w$ is arbitrary here, there are $2^l$ binary strings $p$ of length
$n$ with $T(p) = x$.
 \end{proof}

  How many objects are there with a given complexity $n$?
  We can answer this question with a good approximation.

 \begin{definition}
  Let $g_{T}(n)$ be the number of objects $x \in \dS$ with $K(x)=n$, and $D_{n}$
the set of binary strings $p$ of length $n$ for which $T(p)$ is
defined.
  Let us define the moving average
 \[
	h_{T}(n,c) = \frac{1}{2c+1}\sum_{i = -c}^c g_{T}(n+i).
 \]
 \end{definition}

Here is an estimation of these numbers.
							
\begin{theorem} \label{thm:exactstat} (\cite{SolovayManu}) 
 There is a natural number $c$ such that 
 \begin{align}
	\log\card{D_{n}}	&\eqa n-K(n),	\label{eq:Dn}
  \\	\log h_{T}(n,c)	&\eqa n-K(n).	\label{eq:aver}
 \end{align}
 \end{theorem}

Since we are interested in general only in accuracy up to additive
constants, we could omit the normalizing factor $1/(2c + 1)$ from  the
moving average $h_{T}$.
We do not know whether this average can be replaced by $g_{T}(n)$.
This might depend on the special universal
partial recursive function we use to construct the optimal 
s.d.~interpreter $T(p)$.
But equation~\eqref{eq:aver} is true for any \df{optimal} s.d.~interpreter
$T$ (such that the inequality~\eqref{eq:sdopt} holds for all $F$) while
there is an optimal s.d.~interpreter $F$ for which $g_{F}(n)=0$ for every
odd $n$.
Indeed, let $F(00p)=T(p)$ if $\len{p}$ is even, $F(1p)=T(p)$ if $\len{p}$ is odd
and let $F$ be undefined in all other cases.
Then $F$ is defined only for inputs of even length, while $C_{F} \le C_{T}+2$.

\begin{lemma} \label{l:projmeas} 
 \[
 	\sum_y \m(x,y) \eqm \m(x).
 \]
 \end{lemma}

\begin{proof}
  The left-hand side is a constructive semimeasure therefore it is
dominated by the right side.
  To show $\gem$, note that by equation~\eqref{eq:xx.x} as applied to $H$, we
have $K(x,x) \eqa K(x)$.
  Therefore $\m(x) \eqm \m(x,x) < \sum_y \m(x,y)$.
 \end{proof}

\begin{proof}[Proof of Theorem \protect\ref{thm:exactstat}:]
 Let $d_{n} = \card{D_{n}}$. Using Lemma~\ref{l:projmeas} and 
Theorem~\ref{thm:Hstat} we have 
 \[
	g_{T}(n) \le d_{n} = \sum_{x} f(x,n) \eqm 2^{n} \sum_{x} \m(x,n) 
			\eqm 2^{n} \m(n).
 \]
  Since the complexity $H$ has the same continuity 
property~\eqref{eq:continui} as $K$, we can derive from here
 $h_{T}(n,c) \le 2^{n}\m(n)O(2^c c^2)$.
  To prove $h_{T}(n,c) \gem d_{n}$ for an appropriate $c$, we will prove that
the complexity of at least half of the elements of $D_{n}$ is near $n$.

  For some constant $c_{0}$, we have $K(p) \le n + c_{0}$ for all $p \in
D_{n}$. 
  Indeed, if the s.d.~p.r.~interpreter $F(p)$ is equal to $p$
where $T(p)$ is defined and is undefined otherwise then $C_{F}(p) \le
n$ for $p \in D_{n}$.
  By $K(p,\len{p})\eqa K(p)$, and a variant of Lemma~\ref{l:projmeas}, we have
 \[
 \sum_{p \in D_{n}}\m(p) \eqm \sum_{p \in D_{n}} \m(p,n) \eqm \m(n).
 \]
  Using the expression~\eqref{eq:Dn} for $d_{n}$, we see that there is a
constant $c_{1}$ such that 
 \[
  d_{n}^{-1} \sum_{p \in D_{n}} 2^{-K(p)} \le 2^{-n+c_{1}}.
 \]
  Using Markov's Inequality, hence the number of elements $p$ of $D_{n}$
with $K(p) \ge n-c_{1}-1$ is at least $d_{n}/2$.
  We finish the proof making $c$ to be the maximium of $c_{0}$ and $c_{1}+1$.
 \end{proof}

  What can be said about the complexity of a binary string of length
$n$?
  We mentioned earlier the estimate $K(x) \lea n + 2\log n$. The same
argument gives the stronger estimate
 \begin{equation} 	K(x) \lea \len{x} + K(\len{x}). \label{eq:monbd}
 \end{equation}
  By a calculation similar to the proof of Theorem~\ref{thm:exactstat},
we can show that the estimate~\eqref{eq:monbd} is exact for most binary
strings.
  First, it follows just as in Lemma~\ref{l:projmeas} that there is a
constant $c$ such that
 \[
  	2^{-n}\sum_{\len{p} = n} 2^{-K(p)} \le c2^{-n}\m(n).
 \]
  From this, Markov's Inequality gives that the number of strings 
$p\in \dB^{n}$ with $K(p) < n-k$ is at most $c2^{n-k}$. 

  There is a more polished way to express this result. 

 \begin{definition}
Let us introduce the following function of natural numbers:
 \[	
  K^{+}(n) = \max_{k \le n} K(k).
 \]
 \end{definition}

  This is the smallest monotonic function above $K(n)$. 
It is upper semicomputable just as $H$ is, hence $2^{-K^+}$ is
universal among the monotonically decreasing constructive
semimeasures.
The following theorem expresses $K^{+}$ more directly in terms of $K$. 

\begin{theorem} \label{thm:monbd} 
We have $K^+(n) \eqa \log n + K(\flo{ \log n })$.
  \end{theorem}				
\begin{proof} 
 To prove $\lea$, we construct a s.d.~interpreter as follows.
  The machine computing $F(p)$ finds a decomposition $uv$ of $p$ such that
$T(u)=\len{v}$, then outputs the number whose binary representation (with
possible leading zeros) is the binary string $v$.
  With this $F$, we have
 \begin{equation} \label{eq:monbd1}   
   K(k) \lea C_{F}(k) \le \log n + H (\flo{ \log n }) 
 \end{equation}
  for all $k \le n$.
 The numbers between $n/2$ and $n$ are the ones whose binary representation
has length $\flo{ \log n} + 1$.
  Therefore if the bound~\eqref{eq:monbd} is sharp for most $x$ then the 
bound~\eqref{eq:monbd1} is sharp for most $k$.
 \end{proof}

  The sharp monotonic estimate $\log n + K(\flo{ \log n })$ for $K(n)$ is
less satisfactory than the sharp estimate $\log n$ for Kolmogorov's
complexity $C(n)$, because it is not a computable function.
  We can derive several computable upper bounds from it, for example
 $\log n + 2\log\log n$, $\log n + \log\log n + 2\log\log\log n$, etc. but
none of these is sharp for large $n$.
  We can still hope to find computable upper bounds of $H$ which are sharp
infinitely often.
  The next theorem provides one.

\begin{theorem}[see~\protect\cite{SolovayManu}] \label{thm:infsharp} 
  There is a computable upper bound $G(n)$ of the function $K(n)$ with the
property that $G(n)=K(n)$ holds for infinitely many $n$.
 \end{theorem}

  For the proof, we make use of a s.d.~Turing machine $\cT$ computing the
function $T$.
Similarly to the proof of Theorem~\ref{thm:Ksemicp}, we introduce a
time-restricted complexity. 

 \begin{definition}
We know that for some constant $c$, the
function $K(n)$ is bounded by $2\log n + c$.
Let us fix such a $c$.

For any number $n$ and s.d.~ Turing machine $\cM$, let $K_{\cM}(n;t)$ be
the minimum of $2\log n + c$ and the lengths of descriptions from which
$\cM$ computes $n$ in $t$ or fewer steps.
 \end{definition}

  At the time we proved the Invariance Theorem we were not interested in
exactly how the universal p.r.~ function $V$ was to be computed.
Now we need to be more specific.

 \begin{definition}
Let us agree now that the universal p.r.~function
is computed by a universal Turing
machine $\cV$ with the property that for any Turing machine $\cM$ there is
a constant $m$ such that the number of steps needed to simulate $t$ steps
of $\cM$ takes no more than $mt$ steps on $\cV$.
Let $\cT$ be the Turing machine constructed from $\cV$ computing the
optimal s.d.~interpreter $T$, and let us write 
 \begin{align*}
  K(n,t)=K_{\cT}(n,t).
 \end{align*}
 \end{definition}

With these definitions, 
for each s.d.~ Turing machine $\cM$ there are constants $m_{0},m_{1}$ with
 \begin{equation}  \label{eq:fastopt} 
 	K(n,m_{0}t) \le K_{\cM}(n,t) + m_{1}.
 \end{equation}
  The function $K(n;t)$ is nonincreasing in $t$. 

\begin{proof}[Proof of Theorem \protect\ref{thm:infsharp}]
 First we prove that there are constants $c_{0},c_{1}$ such that $K(n;t) <
K(n;t-1)$ implies 
 \begin{equation} \label{eq:timetoo} 
 	K(\ang{ n,t} ; c_{0} t) \le K(n;t) + c_{1}.
 \end{equation}
  Indeed, we can construct a Turing machine $\cM$ simulating the 
work of $\cT$ such that if $\cT$ outputs a number $n$ in $t$
steps then $\cM$ outputs the pair $\ang{n,t}$ in $2t$ steps:
 \[
	K_{\cM}(\tupcod{ n,t} ; 2t) \le K(n,t).
 \]
  Combining this with the inequality~\eqref{eq:fastopt} gives the
inequality~\eqref{eq:timetoo}.
  We define now a recursive function $F$ as follows.
  To compute $F(x)$, the s.d.~Turing machine first tries
to find $n,t$ such that $x = \ang{ n,t}$.
 (It stops even if it did not find them.)
  Then it outputs $K(x; c_{0} t)$.
  Since $F(x)$ is the length of some description of $x$, it is an upper
bound for $K(x)$.
  On the other hand, suppose that $x$ is one of the infinitely many
integers of the form $\ang{ n,t } $ with
 \[
	K(n) = K(n;t) < K(n; t-1). 
 \]
  Then the inequality~\eqref{eq:timetoo} implies $F(x) \le K(n) + c_{1}$ while
$K(n) \lea K(x)$ is known (it is the equivalent of~\ref{thm:xKx} for $H$), so
$F(x) \le K(n) + c$ for some constant $c$. Now if $F$ would not be equal to
$H$ at infinitely many places, only close to it (closer than $c$) then we
can decrease it by some additive constant less than $c$ and modify it at
finitely many places to obtain the desired function $G$.
 \end{proof}
 
%%% Local Variables: 
%%% mode: latex
%%% TeX-master: "ait-notes"
%%% End: 

%%% Local Variables: 
%%% mode: latex
%%% TeX-master: "ait-notes"
%%% End: 

%\include{coding} % \input into compl
% \include{dln-stat} % \input into compl

\chapter{Randomness}
\section{Uniform distribution} \label{sec:tests}

Speaking of a ``random binary string'', one generally understands 
randomness with respect to the coin-toss distribution (when the 
probability $P(x)$ of a binary sequence $x$ of length $n$ is $2^{-n}$). 
This is the only distribution considered in the present section. 

One can hope to distinguish sharply between random and nonrandom 	
\emph{infinite sequences}. 
Indeed, an infinite binary sequence whose 
elements are 0 with only finitely many exceptions, can be considered 
nonrandom without qualification. 
This section defines randomness for finite strings. 
For a sequence $x$ of length $n$, it would be
unnatural to fix some number $k$, declare $x$ nonrandom when its
first $k$ elements are 0, but permit it to be random if only the
first $k-1$ elements are 0. 
So, we will just declare some finite strings less random than others. 
For
this, we introduce a certain real-valued function $d(x) \ge 0$ measuring
the \df{deficiency of randomness} in the string $x$. 
The occurrence
of a nonrandom sequence is considered an  exceptional event,
therefore the function $d(x)$ can assume large values only with small
probability.
What is ``large'' and ``small'' depends only on our ``unit of measurement''.
We require for all $n,k$
 \begin{equation} \label{eq:testcond}
\sum\setof{ P(x): x\in \dB^{n},\ d(x) > k } < 2^{-k},
 \end{equation}
saying that there be at most $2^{n-k}$ binary sequences $x$ of length
$n$ with $d(x)>k$.
Under this condition, we even allow $d(x)$ to take the value $\infty$.

To avoid arbitrariness in the distinction between random and
nonrandom, the function $d(x)$ must be \emph{simple}. 
We assume
therefore that the set $\setof{ (n,k,x) : \len{x} = n,\; d(x)>k }$ is
recursively enumerable, or, which is the same, that the function 
 \begin{align*}
  \d : \Og\to\rint{-\infty}{\infty}   
 \end{align*}
is lower semicomputable.

 \begin{remark}
We do not assume that the set of strings $x \in \dB^{n}$ with \emph{small} 
deficiency of randomness is enumerable, because then it would be easy 
to construct a ``random'' sequence.
 \end{remark}

\begin{definition}
A lower semicomputable function $d : \dS_{2} \to \dR$ (where $\dR$ is the
set of real numbers) is called a \df{Martin-L\"of test (ML-test)}, or a
\df{probability-bounded test} if it satisfies~\eqref{eq:testcond}.

A ML-test $d_{0}(x)$ is \df{universal} if it additively dominates all
other ML-tests: for any other ML-test $d(x)$ there is a
constant $c<\infty$ such that for all $x$ we have 
$d(x) < d_{0}(x) + c$. 
\end{definition}

If a test $d_{0}$ is universal then any other test $d$ of randomness
can discover at most by a constant amount more deficiency of
randomness in any sequence $x$ than $d_{0}(x)$. 
Obviously, the
difference of any two universal ML-tests is bounded by some
constant.

The following theorem reveals a simple connnection between descriptional 
complexity and a certain randomness property.

 \begin{definition}
Let us define the following function for binary strings $x$:
 \begin{align}\label{eq:d0}
 d_{0}(x) = \len{x} - C(x \mvert\; \len{x}).
 \end{align}
 \end{definition}

\begin{theorem}[Martin-L\"of] \label{thm:MLcompl}
 The function $d_{0}(x)$ is a universal Martin-L\"of test.
 \end{theorem}
\begin{proof}
 Since $C(x \mvert y)$ is semicomputable from above, $d_{0}$ is
semicomputable from below. 
The property of semicomputability holds
for $d_{0}$ as a straightforward consequence of Theorem~\ref{thm:Ksemicp}
Therefore $d_{0}$ is a ML-test. 
We must show that it is larger (to
within an additive constant) than any other ML-test. Let $d$ be a
ML-test. 
Let us define the function $F(x,y)$ to be $y - d(x)$ for 
$y = \len{x}$, and $\infty$ otherwise. 
Then $F$ is upper semicomputable, and
satisfies~\eqref{eq:Kmajor}. 
Therefore by Theorem~\ref{thm:Kmajor}, we
have $C(x \mvert \len{x}) \lea \len{x} - d(x)$.  
 \end{proof}

Theorem~\ref{thm:MLcompl} says that under very general assumptions
about randomness, those strings $x$ are random whose descriptional
complexity is close to its maximum, $\len{x}$. 
The more ``regularities'' are
discoverable in a sequence, the less random it is. 
However, this is
true only of regularities which decrease the descriptional
complexity. 
``Laws of randomness'' are regularities whose probability
is high, for example the law of large numbers, the law of iterated
logarithm, the arcsine law, etc. 
A random sequence will satisfy \emph{all such laws}. 

 \begin{example}
The law of large numbers 
says that in most binary sequences of length $n$, the number of 0's
is close to the numer of 1's. 
We prove this law of probability theory
by constructing a ML-test $d(x)$ taking large values on sequences in
which the number of 0's is far from the number of 1's. 

Instead of requiring~\eqref{eq:testcond} we require a somewhat stronger
property of $d(x)$: for all $n$,
 \begin{equation}  \label{eq:llntestcd}
 \sum_{x\in \dB^{n}} P(x)2^{d(x)} \le 1.  
 \end{equation}
 
 \begin{sloppypar}
From this, the inequality~\eqref{eq:testcond} follows by the following
well-known inequality called the Markov Inequality which says the
following.  
Let $P$ be any probability distribution, $f$ any
nonnegative function, with expected value $E_{P}(f)=\sum_{x}P(x)f(x)$.  
For all $\lg\ge 0$ we have
 \begin{equation} \label{eq:Markov}
 \sum\setof{ P(x): f(x) > \lg E_{P}(f)} \le \lg.
 \end{equation}
    \end{sloppypar}

For any string $x\in \dS$, and natural number $i$, let $N(i \mvert x)$ 
denote the number of occurrences of $i$ in $x$. 
For a binary string $x$
of length $n$ define $p_{x} = N(1 \mvert x)/n$, and
 \begin{align*} 
   P_{x}(y)	&= p_{x}^{N(1 \mvert y)}(1-p_{x})^{N(0 \mvert y)}
 \\ d(x)	&= \log P_{x}(x) + n - \log(n+1).
 \end{align*}
  We show that $d(x)$ is a ML-test. 
It is obviously computable. 
We prove~\eqref{eq:llntestcd}. 
We have
 \begin{align*}
  \sum_{x} P(x) 2^{d(x)} &= \sum_{x} 2^{-n} 2^{n} P_{x}(x)\frac{1}{n+1}
	= \frac{1}{n+1}\sum_{x} P_{x}(x) 				
  \\ &= \frac{1}{n+1}\sum_{k=0}^{n}\binom{n}{k}(\frac{k}{n})^k
   (1-\frac{k}{n})^{n-k} < \frac{1}{n+1}\sum_{k=0}^{n} 1 = 1.
 \end{align*}
The test $d(x)$ expresses the (weak) law of large numbers in a rather 
strong version. We rewrite $d(x)$ as
 \[
 d(x) = n(1 - h(p_{x})) - \log(n+1)
 \] 
 where $h(p) = -p\log p - (1-p)\log(1-p)$. 
The entropy function
$h(p)$ achieves its maximum 1 at $p=1/2$. 
Therefore the test $d(x)$
tells us that the probability of sequences $x$ with $p_{x} < p < 1/2$
for some constant $p$ is bounded by the exponentially decreasing
quantitiy $(n+1)2^{-n(1-h(p))}$. 
We also see that since for some
constant $c$ we have
 \[
   1-h(p) > c(p-1/2)^2,
 \]
 % which c ?
 therefore if the difference $| p_{x} - 1/2 |$ is much larger than
 \[
   \sqrt\frac{\log n}{cn}
 \]
  then the sequence $x$ is ``not random'', and hence, by 
Theorem~\ref{thm:MLcompl} its complexity is significantly smaller than $n$. 
 \end{example}

\section{Computable distributions}
 \label{sec:discrand}

Let us generalize the notion of randomness to arbitrary discrete computable
probability distributions.

\subsection{Two kinds of test}

Let $P(x)$ be a probability distribution over some discrete countable 
space $\Og$, which we will identify for simplicity with the set $\dS$ of 
finite strings of natural numbers. 
Thus, $P$ is a nonnegative function with 
 \[
 \sum_{x} P(x) = 1. 
 \]
  Later we will consider probability distributions $P$ over the space
$\dN^\dN$, and they will also be characterized by a nonnegative function
$P(x)$ over $\dS$. 
However, the above condition will be replaced by a different one. 

We assume $P$ to be computable, with some G\"odel number $e$. 
We want to define a test $d(x)$ of randomness with respect to $P$. 
It will
measure how justified is the assumption that $x$ is the outcome of an
experiment with distribution $P$. 

 \begin{definition}
A function $d:\dS\to\dR$ is an \df{integrable test}, or 
\df{expectation-bounded test} of randomness
with respect to $P$ if it is lower semicomputable and satisfies the condition
 \begin{equation} 	\label{eq:newtestcond}
\sum_{x} P(x)2^{d(x)} \le 1. 
 \end{equation}
It is \df{universal} if it dominates all other integrable tests to within an
additive constant.
 \end{definition}

(A similar terminology is used in~\cite{LiViBook97} in order to distinguish
tests satisfying condition~\eqref{eq:newtestcond} from Martin-L\"of tests.)

\begin{remark}
  We must allow $d(x)$ to take negative values, even if only, say values 
$\ge -1$, since otherwise condition~\eqref{eq:newtestcond} could only be satisfied
with $d(x)=0$ for all $x$ with $P(x)>0$.
\end{remark}

 \begin{proposition}\label{propo:ml-integrable}
Let $c = \log(2\pi^{2}/6)$.
If a function $d(x)$ satisfies~\eqref{eq:newtestcond} then it 
satisfies~\eqref{eq:testcond}.
If $d(x)$ satisfies~\eqref{eq:testcond} then $d - 2\log d - c$ 
satisfies~\eqref{eq:newtestcond}.
 \end{proposition}
 \begin{proof}
The first statement follows by Markov's Inequality~\eqref{eq:Markov}.
The second one can be checked by immediate computation.
 \end{proof}

Thus, condition~\ref{eq:newtestcond} is nearly equivalent
to condition~\eqref{eq:testcond}: to each test $d$ according to one of them, 
there is a test $d'$ asymptotically equal to $d$ satisfying the other.
But condition~\eqref{eq:newtestcond} has the advantage
of being just one inequality instead of the infinitely many ones.

\subsection{Randomness via complexity}

We want to find a universal integrable test. 
Let us introduce the function $w(x) = P(x)2^{d(x)}$. 
The above conditions on $d(x)$
imply that $w(x)$ is semicomputable from below
and satisfies $\sum_{x} w(x) \le 1$: so, it is a constructive semimeasure.
With the universal constructive semimeasure of Theorem~\ref{thm:univsem}, 
we also obtain a universal test for any computable probability distribution $P$.

\begin{definition}
For an arbitrary measure $P$ over a discrete space $\Og$, let us denote
 \begin{equation}\label{eq:d_c}
  \ol\d_{P}(x) = \log\frac{\m(x)}{P(x)} = -\log P(x) - K(x).
 \end{equation}
\end{definition}

The following theorem shows that randomness can be tested by checking how close
is the complexity $K(x)$ to its upper bound $-\log P(x)$.

\begin{theorem}\label{thm:univ-integrable-test} 
The function $\ol\d_{P}(x)=-\log P(x)-K(x)$
is a universal integrable test for any fixed computable probability
distribution $P$.
More exactly, it is lower
semicomputable, satisfies~\eqref{eq:newtestcond} and
for all integrable tests $d(x)$ for $P$, we have
 \begin{equation*}%\label{eq:d_c-univ}
   d(x) \lea \ol\d_{P}(x) + K(d) + K(P).
 \end{equation*}
 \end{theorem} 
 \begin{proof}
Let $d(x)$ be an integrable test for $P$.
Then $\nu(x) = 2^{d(x)}P(x)$ is a constructive semimeasure, and it
has a self-delimiting program of length $\lea K(d) + K(P)$.
It follows (using the definition~\eqref{eq:m(nu)})
that $m(\nu) \lem 2^{K(d) + K(P)}$; hence inequality~\eqref{eq:exact-domin}
gives
 \[
   \nu(x) 2^{K(d) + K(P)} \lem \m(x).
 \]
Taking logarithms finishes the proof.
 \end{proof}

The simple form of our universal test suggests remarkable interpretations.  
It says that the outcome $x$ is random with respect to the distribution 
$P$, if the latter assigns to $x$ a large enough
probability---however, not in absolute terms, only relatively to the
universal semimeasure $\m(x)$. 
The relativization is essential, since
otherwise we could not distinguish between random and nonrandom
outcomes for the uniform distribution $P_{n}$ defined in 
Section~\ref{sec:tests}.

For any (not necessarily computable) distribution $P$, the
$P$-expected value of the function $\m(x)/P(x)$ is at most 1,
therefore the relation 
 \[
 	\m(x) \le k P(x)
 \] 
  holds with probability not smaller than $1-1/k$. 
If $P$ is computable then the inequality 
 \[
	\m(P) P(x) \lem \m(x)
 \]
  holds for all $x$. 
Therefore the following applications are obtained.

\begin{itemize} 

 \item If we assume $x$ to be the outcome of an experiment with some 
simple computable probability distribution $P$ then $\m(x)$ is a good
estimate of $P(x)$. 
The goodness depends on how simple $P$ is to
define and how random $x$ is with respect to $P$: how justified
is the assumption. 
Of course, we cannot compute $\m(x)$ but it is
nevertheless well defined.

 \item If we trust that a given object $x$ is random with respect to a
given distribution $P$ then we can use $P(x)$ as an estimate of $\m(x)$. 
The degree of approximation depends on the same two factors. 

 \end{itemize}

Let us illustrate the behavior of the universal semimeasure $\m(x)$
when $x$ runs over the set of natural numbers.  
We define the
semimeasures $v$ and $w$ as follows. 
Let $v(n)=c/n^2$ for an
appropriate constant $c$, let $w(n)=v(x)$ if $n=2^{2^k}$, and 0 
otherwise. 
Then $\m(n)$ dominates both $v(n)$ and $w(n)$. 
The function $\m(n)$ dominates $1/n\log^2n$, but jumps at many places
higher than that. 
Indeed, it is easy to see that $\m(n)$ converges to
0 slower than any positive computable function converging to 0: in
particular, it is not computable. 
Hence it cannot be a measure, that is $\sum_{x} \m(x) < 1$. 
We feel that we can make $\m(x)$ ``large'' whenever $x$ is ``simple''. 

We cannot compare the universal test defined in the present section
directly with the one defined in Section~\ref{sec:tests}.
There, we investigated tests for a whole family $P_{n}$ of
distributions, where $P_{n}$ is the uniform distribution on the set
$\dB^{n}$. 
That the domain of these distributions is a different one for
each $n$ is not essential since we can identify the set $\dB^{n}$ with
the set $\{2^{n},\ldots,2^{n+1}-1\}$. 
Let us rewrite the function $d_{0}(x)$ defined in~\eqref{eq:d0} as 
 \[
 	d_{0}(x) = n - C(x \mvert n)
 \]
 for $x \in \dB^{n}$, and set it $\infty$ for $x$ outside $\dB^{n}$. Now the 
similarity to the expression 
 \[
  	\ol\d_{P}(x) = -\log P(x) - \log \m(x)
 \]
  is striking because $n = \log P_{n}(x)$. 
Together with the remark
made in the previous paragraph, this observation suggests that the
value of $-\log \m(x)$ is close to the complexity $C(x)$.

\subsection{Conservation of randomness}

The idea of conservation of randomness has been expressed first by Levin,
who has published some papers culminating in~\cite{LevinRandCons84}
developing a general theory.
In this section, we address the issue in its simplest form only.
Suppose that we want to talk about the randomness of integers.
One and the same integer $x$
can be represented in decimal, binary, or in some
other notation: this way, every time a different string will represent the
same number.
Still, the form of the representation should not affect the question of
randomness of $x$ significantly.
How to express this formally?

Let $P$ be a computable measure and let $f : \dS \to \dS$ be a computable
function.
We expect that under certain conditions, the randomness of $x$ should imply
the randomness of $f(x)$---but, for which distribution?
For example, if $x$ is a binary string and $f(x)$ is the decimal notation
for the number expressed by $1x$ then we expect $f(x)$ to be random not
with respect to $P$, but with respect to the measure $f^{*}P$ that is
the image of $P$ under $f$.
This measure is defined by the formula
 \begin{equation}\label{eq:measure-image}
   (f^{*}P)(y) = P(f^{-1}(y)) = \sum_{f(x) = y} P(x).
 \end{equation}
 
 \begin{proposition}
If $f$ is a computable function and $P$ is a computable measure
then $f^{*}P$ is a computable measure.
 \end{proposition}
 \begin{proof}
It is sufficient to show that $f^{*}P$ is lower semicomputable, since we
have seen that a lower semicomputable measure is computable.
But, the semicomputability can be seen immediately from the definition.
 \end{proof}

Let us see that randomness is indeed preserved for the image.
Similarly to~\eqref{eq:m(nu)}, let us define for an arbitrary computable
function $f$:
 \begin{equation}\label{eq:m(f)}
   \m(f) = \sum\setof{\m(p) : \txt{program $p$ computes $f$}}.
 \end{equation}

 \begin{theorem}\label{thm:simple-rand-cons-det}
 Let $f$ be a computable function.
Between the randomness test for $P$ and
the test for the image of $P$, the following relation holds, for
all $x$:
  \begin{equation}\label{eq:simple-rand-cons-det}
  \ol\d_{f^{*} P }(f(x)) \lea \ol\d_{P}(x) + K(f) + K(P).
  \end{equation}
 \end{theorem}
 \begin{proof}
Let us denote the function on the left-hand side 
of~\eqref{eq:simple-rand-cons-det} by
 \[
   d_{P}(x) = \ol\d_{f^{*}P}(f(x)) = \log\frac{\m(f(x))}{(f^{*}P)(f(x))}.
 \]
It is lower semicomputable, with the help of a program of length
$K(f) + K(P)$, by its very definition.
Let us check that it is an integrable test, so it
satisfies the inequality~\eqref{eq:newtestcond}.
We have
 \[
 \sum_{x} P(x) 2^{d_{P}(x)} 
     = \sum_{y} (f^{*}P)(y) 2^{\ol\d_{f^{*}P}(y)} = 
       \sum_{y} \m(y) \le 1.
 \]
Hence $d_{P}(x)$ is a test, and hence it is 
$\lea \ol\d_{P}(x) + K(f) + K(P)$, by the universality of the test $\ol\d_{P}(x)$.
(Strictly speaking, we must modify the proof of
Theorem~\ref{thm:univ-integrable-test} and see that a program of length $K(f) +
K(P)$ is also sufficient to define the semimeasure $P(x) d(x \mvert P)$.)
 \end{proof}

The appearance of the term $K(f)$ on the right-hand side is understandable:
using some very complex function $f$, we can certainly turn any string
into any other one, also a random string into a much less random one.
On the other hand, the term $K(P)$ appears only due to the
imperfection of our concepts.
It will disappear in the more advanced theory presented in later sections,
where we develop \emph{uniform tests}.

\begin{example}
For each string $x$ of length $n \ge 1$, let 
 \[
  P(x) = \frac{2^{-n}}{n(n+1)}.
 \] 
This distribution is uniform on strings of equal length.
Let $f(x)$ be the function that erases the first $k$ bits.
Then for $n \ge 1$ and for any string $x$ of length $n$ 
we have 
 \[
  (f^{*}P)(x) = \frac{2^{-n}}{(n+k)(n+k+1)}.
 \]
(The empty string has the rest of the weight of $f^{*}P$.)
The  new distribution is still uniform on strings of equal length,
though 
the total probability of strings of length $n$ has changed somewhat.
We certainly expect randomnes conservation here:
after erasing the first $k$ bits of
a random string of length $n+k$, we should
still get a random string of length $n$.
\end{example}

The computable
function $f$ applied to a string $x$ can be viewed as some kind of
transformation $x \mapsto f(x)$.
As we have seen, it does not make $x$ less random.
Suppose now that we introduce randomization into the transformation
process itself: for example, the machine computing $f(x)$ can also toss
some coins.
We want to say that randomness is also conserved under such more general
transformations, but first of all, how to express such a transformation
mathematically?
We will describe it by a ``matrix'':
a computable probability transition function $T(x,y) \ge 0$ with the
property that
 \[
   \sum_{y} T(x, y) = 1.
 \]
Now, the image of the distribution $P$ under this transformation can be
written as $T^{*}P$, and defined as follows:
 \[
   (T^{*}P)(y) = \sum_{x} P(x) T(x, y).
 \]
How to express now randomness conservation?
It is certainly not true that every possible outcome $y$ is as random 
with respect to $T^{*}P$ as $x$ is with respect to $P$.
We can only expect that for each $x$, the conditional probability 
(in terms of the transition $T(x, y)$) of those $y$ whose non-randomness is
larger than that of $x$, is small.
This will indeed be expressed by the corollary below.
To get there, we upperbound the $T(x, \cdot)$-expected value
of $2^{\ol\d_{T^{*}P}(y)}$.
Let us go over to exponential notation:

 \begin{definition}
Denote
 \begin{align}\label{eq:e.ol-t}
    \ol\t_{P}(x) = 2^{\ol\d_{P}(x)}=\frac{\m(x)}{P(x)}.
  \end{align}
 \end{definition}

 \begin{theorem} We have
 \begin{equation}\label{eq:simple-rand-cons-rand}
 \log \sum_{y} T(x, y) 2^{\ol\d_{T^{*} P}(y)} \lea \ol\d_{P}(x) + K(T) + K(P).
 \end{equation}
 \end{theorem}
 \begin{proof}
Let us denote the function on the left-hand side
of~\eqref{eq:simple-rand-cons-rand} by $t_{P}(x)$.
It is lower semicomputable by its very construction, using a program of
length $\lea K(T) + K(P)$.
Let us check that it satisfies $\sum_{x} P(x) t_{P}(x) \le 1$ which,
in this notation, corresponds to inequality~\eqref{eq:newtestcond}.
We have
 \begin{align*}
   \sum_{x} P(x) t_{P}(x) &= \sum_{x} P(x) \sum_{y} T(x, y) t_{T^{*} P}(y) 
\\   &= \sum_{x} P(x) \sum_{y} T(x, y) \frac{\m(y)}{(T^{*}P)(y)}
      = \sum_{y} \m(y) \le 1.
 \end{align*}
It follows that $d_{P}(x) = \log\, t_{P}(x)$ is an integrable
test and hence $d_{P}(x) \lea \ol\d_{P}(x) + K(T) + K(P)$.
(See the remark at the end of the proof of
Theorem~\ref{thm:simple-rand-cons-det}.)
 \end{proof}

 \begin{corollary}
There is a constant $c_{0}$ 
such that for every integer $k \ge 0$, for all $x$ we have
 \[
 \sum\setof{ T(x, y) : \ol\d_{T^{*}P}(y) - \ol\d_{P}(x) > k + K(T) + K(P)} 
 \le 2^{-k + c_{0}}. 
 \]
 \end{corollary}
 \begin{proof}
The theorem says
 \[
  \sum_{y} T(x, y) \ol\t_{T^{*} P}(y) \lem t_{P}(x)2^{K(T) + K(P)}. 
 \]
Thus, it upperbounds the expected value of 
the function  $2^{\ol\d_{T^{*}P)}(y)}$ according to the distribution $T(x, \cdot)$.
Applying Markov's inequality~\eqref{eq:Markov}
to this function yields the desired result.
 \end{proof}

\section{Infinite sequences}

These lecture notes treat the theory of randomness over continuous spaces mostly
separately, starting in Section~\ref{sec:cont-spaces}.
But the case of computable measures over infinite sequences is particularly
simple and appealing, so we give some of its results here, even if most
follow from the more general results given later.

In this section, we will fix a finite or countable \df{alphabet}
\( \Sg=\set{s_{1},s_{2},\dots} \), and consider
probability distributions over the set 
\begin{align*}
   X = \Sg^{\dN}
 \end{align*}
of infinite sequences with members in \( \Sg \).
An alternative way of speaking of this is to consider a sequence of
\df{random variables} \( X_{1},X_{2},\dots \), where \( X_{i}\in\Sg \), with a joint
distribution.
The two interesting extreme special cases are \( \Sg=\dB=\{0,1\} \), giving the set
of infinite 0-1 sequences, and \( \Sg=\dN \), giving the set sequences of natural
numbers.

\subsection{Null sets}

Our goal is here to illustrate the measure theory developed in 
Section~\ref{sec:measures}, through Subsection~\ref{subsec:integral},
on the concrete example of the set
of infinite sequences, and then to develop the theory of randomness
in it.

We will distinguish a few simple kinds of subsets of the set \( \S \) of sequences,
those for which probability can be defined especially easily.

\begin{definition}\
  \begin{bullets}
  \item 
  For string \( z\in\Sg^{*} \), we will denote by \( zX \) the set of elements of \( X \)
with prefix \( z \).
Such sets will be called \df{cylinder sets}.
 \item
A subset of \( X \) is \df{open} if it is the union of any number (not necessarily
finite) of cylinder sets.
It is called \df{closed} if its complement is open.
 \item
An open set is called \df{constructive} if it is the union of a recursively
enumerable set of cylinder sets.
\item A set is called \( G_{\dg} \) if it is the intersection of a sequence of open sets.
 It is called \( F_{\sg} \) if it is the union of a sequence of closed sets.
\item
 We say that a set \( E\sbsq X \) is \df{finitely determined} if there is an \( n \)
and an \( \cE\sbsq \Sg^{n} \) such that
\begin{align*}
  E = \cE X=\bigcup_{s\in\cE}s X.
 \end{align*}
Let \( \cF \) be the class of all finitely determined subsets of \( X \).
  \item
A class of subsets of \( X \) is called an \df{algebra} if it is closed with respect
to finite intersections and complements (and then of course, also with respect
to finite unions).
It is called a \( \sg \)-algebra (sigma-algebra) 
when it is also closed with respect to countable
intersections (and then of course, also with respect to countable unions).
\end{bullets}
\end{definition}

\begin{example}
An example open set that is not finitely determined, is the set \( E \) of all sequences that
contain a substring \( 11 \).
\end{example}

The following observations are easy to prove.

\begin{proposition}\
  \begin{alphenum}
 \item
 Every open set \( G \) can be represented as the union of a sequence of \emph{disjoint}
cylinder sets.
 \item 
The set of open sets is closed with respect to finite intersection
and arbitrarily large union.
  \item
 Each finitely determined set is both open and closed.
 \item
The class \( \cF \) of finitely determined sets forms an algebra.
 \end{alphenum}
\end{proposition}

Probability is generally defined as a nonnegative, monotonic
function (a ``measure'', see later) on some subsets of the event
space (in our case the set \( X \)) that has a certain additivity property.
We do not need immediately the complete definition of measures, but let us say
what we mean by an additive set function.

\begin{definition}
Consider a nonnegative, function \( \mu \) is defined over some subsets of \( X \).
that is also \df{monotonic}, that is \( A\sbsq B \) implies  \( \mu(A)\sbsq\mu(B) \).
We say that it is \df{additive} if,
whenever it is defined on the sets
\( E_{1},\dots,E_{n} \), and on \( E=\bigcup_{i=1}^{n}E_{i} \)
and the sets \( E_{i} \) are mutually disjoint, then we have
\( \mu(E)=\mu(E_{1})+\dots+\mu(E_{n}) \).

We say that \( \mu \) is \df{countably additive}, or \( \sg \)-additive (sigma-additive),
if in addition whenever it is defined on the sets
\( E_{1},E_{2},\dots \), and on \( E=\bigcup_{i=1}^{\infty}E_{i} \)
and the sets \( E_{i} \) are mutually disjoint, then we have
\( \mu(E)=\sum_{i=1}^{\infty}\mu(E_{i}) \).
\end{definition}

First we consider only measures defined on cylinder sets.

\begin{definition}[Measure over \( \Sg^{\dN} \)]
Let \( \mu:\Sg^{*}\to\dR_{+} \) be a function assigning a nonnegative number to
each string in \( \Sg^{*} \).
We will call \( \mu \) a \df{measure} if it satisfies the condition
 \begin{align}\label{eq:extension-sum}
   \mu(x)=\sum_{s\in\Sg}\mu(xs).
 \end{align}
We will take the liberty to also write
 \begin{align*}
   \mu(s X)=\mu(s)
 \end{align*}
for all \( s\in\Sg^{*} \), and \( \mu(\emptyset)=0 \).
We call \( \mu \) a \df{probability measure} if \( \mu(\Lg)=1 \) and correspondingly,
\( \mu(X)=1 \).

A measure is called \df{computable} if it is computable as a function
\( \mu:\Sg^{*}\to\dR_{+} \) according to Definition~\ref{def:semicomputability}.
\end{definition}

The following observation lets us extend measures.

\begin{lemma}\label{lem:cyl-measure}
  Let \( \mu \) be a measure defined as above.
If a cylinder set \( s X \) is the disjoint union \( xX=x_{1}X\cup x_{2}X\cup\dots \)
of cylinder sets, then we have \( \mu(z)=\sum_{i}\mu(x_{i}) \).
\end{lemma}
\begin{proof}
Let us call an arbitrary \( y\in\Sg^{*} \) \df{good} if 
 \begin{align*}
   \mu(zy)=\sum_{i}\mu(zyX\cap x_{i}X)
 \end{align*}
holds.
We have to show that the empty string \( \Lg \) is good.

It is easy to see that \( y \) is good if \( zyX\sbsq x_{i}X \) for some \( i \).
Also, if \( ys \) is good for all \( s\in\Sg \) then \( y \) is good.
Now assume that \( \Lg \) is not good.
Then there is also an \( s_{1}\in\Sg \) that is not good.
But then there is also an \( s_{2}\in\Sg \) such that \( s_{1}s_{2} \) is not good.
And so on, we obtain an infinite sequence \( s_{1}s_{2}\dotsm \) such that
\( s_{1}\dotsm s_{n} \) is not good for any \( n \).
But then the sequence \( z s_{1}s_{2}\dotsm \) is not contained in
\( \bigcup_{i}x_{i}X \), contrary to the assumption .
\end{proof}

\begin{corollary}
  Two disjoint union representations of the same open set
 \begin{align*}
   \bigcup_{i}x_{i}X = \bigcup_{i}y_{i}X
 \end{align*}
imply \( \sum_{i}\mu(x_{i})=\sum_{i}\mu(y_{i}) \).
\end{corollary}
\begin{proof}
  The set \( \bigcup_{i}x_{i}X \) can also be written as
 \begin{align*}
  \bigcup_{i}x_{i}X=\bigcup_{i,j} x_{i}X\cap y_{j}X.
\end{align*}
An element \( x_{i}X\cap y_{j}X \) is nonempty only if one of the strings
\( x_{i},y_{j} \) is a continuation of the other.
Calling this string \( z_{ij} \) we have \( x_{i}X\cap y_{j}X=z_{ij}X \).
Now Lemma~\ref{lem:cyl-measure} is applicable to the union
\( x_{i}X=\bigcup_{j}(x_{i}X\cap y_{j}X) \),
giving
 \begin{align*}
        \mu(x_{i}) &=\sum_{j}\mu(x_{i}X\cap y_{j}X),
\\   \sum_{i}\mu(x_{i}) &=\sum_{i,j}\mu(x_{i}X\cap y_{j}X).
\end{align*}
The right-hand side is also equal similarly to \( \sum_{j}\mu(y_{j}) \).
\end{proof}

The above corollary allows the following definition:

\begin{definition}
For an open set \( G \) given as a disjoint union of cylinder sets \( x_{i}X \)
let \( \mu(G)=\sum_{i}\mu(x_{i}) \).
For a closed set \( F=X\xcpt G \) where \( G \) is open, let
\( \mu(F)=\mu(X)-\mu(G) \).
\end{definition}

The following is easy to check.

\begin{proposition}
The measure as defined on open sets is monotonic and
countably additive.
In particular, it is countably additive 
on the algebra of finitely determined sets.
\end{proposition}

Before extending measures to a wider range of sets, consider an important
special case.
With infinite sequences, there are certain events that are not impossible, but
still have probability 0.

\begin{definition}
Given a measure \( \mu \), a set \( N\sbsq X \) is called a \df{null set} with respect
to \( \mu \) 
if there is a sequence \( G_{1},G_{2},\dots \) of open sets with the property
\( N\sbsq \bigcap_{m}G_{m} \), and \( \mu(G_{m})\le 2^{-m} \).
We will say that the sequence \( G_{m} \) \df{witnesses} the fact that \( N \) is a null set.
\end{definition}

\begin{examples}\label{example:null-set}
Let \( \Sg=\{0,1\} \), let us define the measure \( \lg \) by
\( \lg(x)=2^{-n} \) for all \( n \) and for all \( x\in\Sg^{n} \).
  \begin{enumerate}
  \item\label{i:singleton}
For every infinite sequence \( \xi\in X \), the one-element set \( \{\xi\} \) is a null set
with respect to \( \lg \).
Indeed, for each natural number \( n \), let 
\( H_{n}=\setof{\xi\in X: \xi(0)=\xi(0),\xi(1)=\xi(1),\dots,\xi(n)=\xi(n)} \).
Then \( \lg(H_{n})=2^{-n-1} \), and \( \{s\}=\bigcap_{n}H_{n} \).
\item\label{i:even-places}
For a less trivial example, 
consider the set \( E \) of those elements \( t\in X \) that are 1 in each positive even
position, that is
 \begin{align*}
   E=\setof{t\in X: t(2)=1,t(4)=1,t(6)=1\dots}.
 \end{align*}
Then \( E \) is a null set.
Indeed, for each natural number \( n \), let 
\( G_{n}=\setof{t\in X: t(2)=1,t(4)=1,\dots,t(2n)=1} \).
This helps expressing \( E \) as \( E=\bigcap_{n}G_{n} \), where \( \lg(G_{n})=2^{-n} \).  
  \end{enumerate}
\end{examples}

\begin{proposition}\label{propo:infin-null-sets}
Let \( N_{1},N_{2},\dots \) be a sequence of null sets with respect to a measure
\( \mu \).
Their union \( N=\bigcup_{i}N_{i} \) is also a null set.
\end{proposition}
\begin{proof}
  Let \( G_{i,1},G_{i,2},\dots \) be the infinite sequence witnessing the fact that
\( N_{i} \) is a null set.
Let \( H_{m}=\bigcup_{i=1}^{\infty} G_{i,m+i} \).
Then the sequence \( H_{m} \) of open sets witnesses the fact that \( N \) is a null set.
\end{proof}

We would like to extend our measure to null sets \( N \) and say \( \mu(N)=0 \).
The proposition we have just proved shows that such an extension would be
countably additive on the null sets.
But we do not need this extension for the moment, so we postpone it.
Still, given a probability measure \( P \) over \( X \), it becomes meaningful to say
that a certain property holds with probability 1.
What this means is that the set of those sequences that do not have this
property is a null set.

Following the idea of Martin-L\"of, we would like to call a sequence
nonrandom, if it is contained in some \emph{simple} null set; that is it has
some easily definable property with probability 0.
As we have seen  in part~\ref{i:singleton} of Example~\ref{example:null-set}, 
it is important to insist on simplicity, otherwise (say
with respect to the measure \( \lg \)) every sequence might be contained in a null set,
namely the one-element set consisting of itself.
But most of these sets are not defined simply at all.
An example of a simply defined null set is given is
part~\ref{i:even-places} of Example~\ref{example:null-set}.
These reflections justify the following definition, in which ``simple'' is
specified as ``constructive''.

\begin{definition}
Let \( \mu \) be a computable measure.
A set \( N\sbsq X \) is called a \df{constructive null set} if there is 
recursively enumerable set \( \Gg\sbsq\dN\times\Sg^{*} \)
with the property that denoting \( \Gg_{m}=\setof{x: \tup{m,x}\in\Gg} \)
and \( G_{m}=\bigcup_{x\in\Gg_{m}}xX \)
we have \( N\sbsq\bigcap_{m}G_{m} \), and \( \mu(G_{m})\le 2^{-m} \).
\end{definition}

In words, the difference between the definition of null sets and constructive
null sets is that the sets \( G_{m}=\bigcup_{x\in\Gg_{m}}xX \) here
are required to be constructive open, moreover, in such a way that from \( m \) one
can compute the program generating \( G_{m} \).
In even looser words, a set is a constructive null set if for any \( \eps>0 \) one
can construct effectively
a union of cylinder sets containing it, with total measure \( \le\eps \).

Now we are in a position to define random infinite sequences.

\begin{definition}[Random sequence]
An infinite sequence \( \xi \) is \df{random} with respect to a probability measure
\( P \) if and only if \( \xi \) is not contained in any constructive null set with
respect to \( P \).  
\end{definition}

It is easy to see that the set of nonrandom sequences is a null set.
Indeed, there is only a countable number or constructive null sets,
so even their union is a null set.
The following theorem strengthens this observation significantly.

\begin{theorem}\label{thm:univ-seq-test}
Let us fix a \emph{computable} probability measure \( P \).
The set of all nonrandom sequences is a constructive null set.
\end{theorem}
Thus, there is a \df{universal} constructive null set, containing all other
constructive null sets.
A sequence is random when it does not belong to this set.
\begin{proof}[Proof of Theorem~\protect\ref{thm:univ-seq-test}]
The proof uses the projection technique that has appeared several times in this
book, for example in proving the existence of a universal lower semicomputable
semimeasure. 

We know that it is possible to list all recursively enumerable subsets of the
set \( \dN\times\Sg^{*} \) into a sequence, namely that there is a recursively
enumerable set \( \Dg\sbsq\dN^{2}\times\Sg^{*} \)
with the property that for every recursively enumerable
set \( \Gg\sbsq\dN\times\Sg^{*} \) there is an \( e \) with 
\( \Gg=\setof{\tup{m,x}: \tup{e,m,x}\in\Dg} \).
We will write
 \begin{align*}
  \Dg_{e,m}&=\setof{x:\tup{e,m,x}\in\Dg},
\\ D_{e,m}&=\bigcup_{x\in\Dg_{e,m}}xX.
\end{align*}
 We transform the set \( \Dg \) into another set \( \Dg' \)
with the following property.
\begin{bullets}
\item For each \( e,m \) we have \( \sum_{\tup{e,m,x}\in\Dg'}\mu(x) \le 2^{-m+1} \).
\item If  for some \( e \) we have 
\( \sum_{\tup{e,m,x}\in\Dg}\mu(x)\le 2^{-m} \) for all \( m \) then for all \( m \) we have
\( \setof{x: \tup{e,m,x}\in\Dg'}=\setof{x: \tup{e,m,x}\in\Dg} \).
\end{bullets}
This transformation is routine, so we leave it to the reader.
By the construction of \( \Dg' \), for every constructive null set \( N \) there is an \( e \)
with
 \begin{align*}
   N\sbsq\bigcap_{m}D'_{e,m}.
 \end{align*}
Define the recursively enumerable set
 \begin{align*}
   \hat\Gg=\setof{\tup{m,x}: \exist{e} \tup{e,m+e+2,x}\in\Dg'}.
 \end{align*}
Then \( \hat G_{m}=\bigcup_{e}D'_{e,m+e+2} \).
For all \( m \) we have
 \begin{align*}
\sum_{x\in\hat\Gg_{m}}\mu(x)&=\sum_{e}\sum_{\tup{e,m+e+2,x}\in\Dg'}\mu(x)
 \le \sum_{e}2^{-m-e-1} = 2^{-m}.   
 \end{align*}
This shows that \( \hat\Gg \) defines a constructive null set.
Let \( \Gg \) be any other recursively enumerable subset of \( \dN\times\Sg^{*} \) that
defines a constructive null set.
Then there is an \( e \) such that for all \( m \) we have \( G_{m}=D'_{e,m} \).
The universality follows now from
 \begin{align*}
\bigcap_{m}G_{m}=\bigcap_{m}D'_{e,m}\sbsq \bigcap_{m}D'_{e,m+e+2} \sbsq\bigcap_{m}\hat G_{m}.   
 \end{align*}
\end{proof}

\subsection{Probability space}

Now we are ready to extend measure to a much wider range of subsets of \( X \).

\begin{definition}
 Elements of the smallest \( \sg \)-algebra \( \cA \) containing the cylinder sets of \( X \) are
called the \df{Borel} sets of \( X \).
The pair \( (X,\cA) \) is an example of a \df{measureable space}, where the Borel
sets are called the \df{measureable sets}.

A nonnegative sigma-additive function \( \mu \) over \( \cA \) is called a \df{measure}.
It is called a \df{probability measure} if \( \mu(X)=1 \).
If \( \mu \) is fixed then the triple \( (X,\cA,\mu) \) is called a \df{measure space}.
If it is a probability measure then the space is called a \df{probability space}.
\end{definition}

In the Appendix we cited a central theorem of measure theory,
Caratheodory's extension theorem.
It implies that if a measure is defined on an algebra \( \cL \) in a sigma-additive way
then it can be extended uniquely
to the \( \sg \)-algebra generated by \( \cL \), that is the
smallest \( \sg \)-algebra containing \( \cL \).
We defined a measure \( \mu \) over \( X \) as a nonnegative function
\( \mu:\Sg^{*}\to\dR_{+} \)
satisfying the equality~\eqref{eq:extension-sum}.
Then we defined \( \mu(xX)=\mu(x) \), and further
extended \( \mu \) in a \( \sg \)-additive way to all elements of the algebra \( \cF \).
Now Caratheodory's theorem allows us to extend it uniquely to all Borel sets,
and thus to define a measureable space \( (X,\cA,\mu) \).

Of course, all null sets in \( \cA \) get measure 0.

Open,  closed and measureable sets can also be defined in the set of real numbers.

\begin{definition}
  A subset \( G\sbsq\dR \) is \df{open} if it is the union of a set of open
  intervals \( \opint{a_{i}}{b_{i}} \).
It is \df{closed} if its complement is open.

The set \( \cB \) of \df{Borel sets} of \( \dB \) is defined as the smallest
\( \sg \)-algebra containing all open sets.
\end{definition}

The pair \( (\dR,\cB) \) is another example of a measureable space.

\begin{definition}[Lebesgue measure]
Consider the set left-closed intervals of the line (including intervals
of the form \( \opint{-\infty}{a} \).
Let \( \cL \) be the set of finite disjoint unions of such intervals.
This is an algebra.
We define the function \( \lg \) over \( \cL \) as follows:
\( \lg(\bigcup_{i} \lint{a_{i}}{b_{i}}) = \sum_{i} b_{i} - a_{i} \). 
It is easy to check that this is a \( \sg \)-additive measure and therefore by
Caratheodory's theorem can be extended to the set \( \cB \) of all Borel sets.
This function is called the \df{Lebesgue measure} over \( \dR \), giving us the
measureable space \( (\dR,\cB,\lg) \).
\end{definition}

Finally, we can define the notion of a measureable function over \( X \).

\begin{definition}[Measureable functions]
A function \( f:X\to \dR \) is called \df{measureable} if and only if
\( f^{-1}(E)\in\cA \) for all \( E\in\cB \).
\end{definition}

The following is easy to prove.

\begin{proposition}
Function \( f:X\to\dR \) is measureable if and only if all sets of the form
\( f^{-1}(\opint{r}{\infty})=\setof{x: f(x)>r} \) are measureable, where \( r \) is a
rational number.
\end{proposition}

\subsection{Computability}

Measureable functions are quite general; it is worth introducing some more
resticted kinds of function over the set of sequences.

\begin{definition}
  A function \( f:X\to\dR \) is \df{continuous} if and only if for every \( x\in X \), for
every \( \eps>0 \) there is a cylinder set \( C\ni x \) such that \( |f(y)-f(x)|<\eps \) for all
\( y\in C \).

  A function \( f:X\to\dR \) is \df{lower semicontinuous} if for every \( r\in\dR \) the set
\( \setof{x\in X: f(x)>r} \) is open.
 It is \df{upper semicontinuous} if \( -f \) is lower semicontinuous.
\end{definition}

The following is easy to verify.

\begin{proposition}
A function \( f:X\to\dR \) is continuous if and only if it is both upper and lower semicontinuous.
\end{proposition}

\begin{definition}
  A function \( f:X\to\dR \) is \df{computable} if for every open rational interval
\( \opint{u}{v} \) the set \( f^{-1}(\opint{u}{v}) \) is a constructive open set of
\( X \), uniformly in \( u,v \).
\end{definition}

Informally this means that if for the infinite sequence \( \xi\in X \) we have
\( u<f(\xi)<v \) then from \( u,v \) sooner or later we will find a prefix \( x \)
of \( \xi \) with the property \( u<f(xX)<v \).

\begin{definition}
A function \( f:X\to\dR \) is \df{lower semicomputable} 
if for every rational \( r \) the set \( \setof{s\in X: f(s)>r} \) 
is a constructive open set of \( X \), uniformly in \( r \).
It is \df{upper semicomputable} if \( -f \) is lower semicomputable.
\end{definition}

The following is easy to verify.

\begin{proposition}
  A function \( X\to\dR \) is computable if and only if it is both lower and upper
  semicomputable. 
\end{proposition}

\subsection{Integral}

The definition of integral over a measure space is given in the
Appendix, in Subsection~\ref{subsec:integral}.
Here, we give an exposition specialized to infinite sequences.

 \begin{definition}
A measurable function \( f : X \to \dR \) is called a \df{step function} if
its range is finite.
The set of step functions will be called \( \cE \).

Given a step function \( f \) which takes values \( x_{i} \) on sets \( A_{i} \), and a
finite measure \( \mu \), we define 
 \[
 \mu(f) = \mu f = \int f\,d\mu = \int f(x) \mu(d x) 
   = \sum_{i} x_{i} \mu(A_{i}).
 \]
 \end{definition}

Proposition~\ref{propo:Riesz-extension}, when specified to our situation here, says
the following.

 \begin{proposition}\label{propo:Riesz-extension-seqs}
The functional \( \mu \) defined above on step functions 
can be extended to the set 
\( \cE_{+} \) of monotonic limits of nonnegative elements of \( \cE \), by
continuity.
The set \( \cE_{+} \) is the set of all nonnegative measurable functions.
 \end{proposition}

Now we extend the notion of integral to a wider class of functions.

 \begin{definition}
A measurable function \( f \) is called \df{integrable} with respect to a finite
measure \( \mu \) if \( \mu |f|^{+} < \infty \) and \( \mu |f|^{-} < \infty \).
In this case, we define \( \mu f = \mu |f|^{+} - \mu |f|^{-} \).
 \end{definition}

It is easy to see that the 
mapping \( f\mapsto \mu f \) is linear when \( f \) runs through the set of
measureable functions with \( \mu |f|<\infty \).

The following is also easy to check.

\begin{proposition}
Let \( \mu \) be a computable measure. 
\begin{alphenum}
   \item
If \( f \) is computable then a program to compute 
\( \mu f \) can be found from the program to compute \( f \).
   \item
If \( f \) is lower semicomputable then a program to lower semicompute 
\( \mu f \) can be found from a program to lower semicompute \( f \).
\end{alphenum}
\end{proposition}

\subsection{Randomness tests}

We can now define randomness tests similarly to Section~\ref{sec:discrand}.

 \begin{definition}
A function \( d:X\to\dR \) is an \df{integrable test}, or \df{expectation-bounded test}
of randomness with respect to the probability measure \( P \)
if it is lower semicomputable and satisfies the condition
 \begin{equation} 	\label{eq:newtestcond-seq}
   \int 2^{d(x)}P(dx) \le 1. 
 \end{equation}
It is called a \df{Martin-L\"of test}, or \df{probability-bounded test},
if instead of the latter condition
only the weaker one is satisfied saying \( P(d(x)>m)<2^{-m} \) for each 
positive integer \( m \).

It is \df{universal} if it dominates all other integrable tests to within an
additive constant.
Universal Martin-L\"of tests are defined in the same way.
 \end{definition}

Randomness tests and constructive null sets are closely related.

\begin{theorem}  Let \( P \) be a computable probability measure
over the set \( S=\Sg^{\infty} \).
There is a correspondence between constructive null sets and
randomness tests:
  \begin{alphenum}
    \item
  For every randomness tests \( d \), the set \( N=\setof{\xi: d(\xi)=\infty} \) is a
  constructive null set.
    \item
For every constructive null set \( N \) there is a randomness test \( d \) with
\( N=\setof{\xi: d(\xi)=\infty} \).
  \end{alphenum}
\end{theorem}
\begin{proof}
Let \( d \) be a randomness test: then for each \( k \) the set \( G_{k}=\setof{\xi: d(\xi)>k} \) is a
constructive open set, and by Markov's inequality
(wich is proved in the continuous case just as in the discrete case)
 we have \( P(G_{k})\le 2^{-k} \).
The sets \( G_{k} \) witness that \( N \) is a constructive null set.

Let \( N \) be a constructive null set with \( N\sbsq \bigcap_{k=1}^{\infty} G_{k} \),
where \( G_{k} \) is a uniform sequence of constructive open sets with
\( P(G_{k})=2^{-k} \).
Without loss of generality assume that the sequence \( G_{k} \) is decreasing.
Then the function \( d(\xi)=\sup\setof{k: \xi\in G_{k}} \) is lower semicomputable
and satisfies \( P\setof{\xi:d(\xi)\ge k} \le 2^{-k} \), so it is a Martin-L\"of test.
Just as in Proposition~\ref{propo:ml-integrable}, 
it is easy to check that \( d(x)-2\log d(x) -c \) is an integrable test for some constant \( c \).
\end{proof}

Just as there is a universal constructive null set, there are universal
randomness tests.

\begin{theorem}\label{thm:univ-integrable-test-seq} 
For a computable measure \( P \), there is a universal integrable test \( \ol\d_{P}(\xi) \).
More exactly, the function \( \xi\mapsto\ol\d_{P}(\xi) \) is lower semicomputable, satisfies~\eqref{eq:newtestcond} and
for all integrable tests \( d(\xi) \) for \( P \), we have
 \begin{equation*}%\label{eq:d_c-univ_seq}
   d(\xi) \lea \ol\d_{P}(\xi) + K(d) + K(P).
 \end{equation*}
 \end{theorem} 
The proof is similar to the proof Theorem~\ref{thm:univsem} 
on the existence of a universal constructive semimeasure.
It uses the projection technique and a weighted sum.

\subsection{Randomness and complexity}

We hope for a result connecting complexity with randomness similar to
Theorem~\ref{thm:univ-integrable-test}.
Somewhat surprisingly, there is indeed such a result.
This is a surprise since in an infinite sequence, arbitrarily large oscillations
of any quantity depending on the prefixes are actually to be expected.

 \begin{theorem}\label{thm:defic-charac-cpt-seqs}
Let \( X=\Sg^{\dN} \) be the set of infinite sequences.
For all computable measures \( \mu \) over \( X \), we have
 \begin{equation}\label{eq:defic-charac-seq}
 \ol\d_{\mu}(\xi) \eqa \sup_{n}\Paren{-\log \mu(\xi_{1:n}) - K(\xi_{1:n})}.
 \end{equation}
Here, the constant in \( \eqa \) depends on the computable measure \( \mu \).
 \end{theorem}

 \begin{corollary}
Let \( \lg \) be the uniform distribution over the set of infinite binary sequences.
Then
 \begin{equation}
 \ol\d_{\lg}(\xi) \eqa \sup_{n}\Paren{n - K(\xi_{1:n})}.
 \end{equation}
 \end{corollary}
In other words, an infinite binary sequence is random (with respect to the
uniform distribution) if and only if
the complexity \( K(\xi_{1:n}) \) of its initial
segments of length \( n \) never decreases below \( n \) by more than an
additive constant.
 \begin{proof}
To prove \( \lea \), define the function
 \[
 f_{\mu}(\xi) = \sum_{s} 1_{sX}(\xi) \frac{\m(s \mvert \mu)}{\mu(s)}
  = \sum_{n} \frac{\m(\xi_{1:n} \mvert \mu)}{\mu(\xi_{1:n})}
  \ge \sup_{n} \frac{\m(\xi_{1:n} \mvert \mu)}{\mu(\xi_{1:n})}.
 \]
The function \( \xi \mapsto f_{\mu}(\xi) \) 
is lower semicomputable with \( \mu^{\xi} f_{\mu}(\xi) \le 1 \), and hence 
 \[
  \ol\d_{\mu}(\xi) \gea \log f(\xi) \gea 
   \sup_{n}\Paren{-\log\mu(\xi_{1:n}) - K(\xi_{1:n} \mvert \mu)}.
 \]

The proof of \( \lea \) reproduces the proof of Theorem 5.2
of~\cite{GacsExact80}.
Let us abbreviate:
 \begin{equation*}%\label{eq:1y}
   1_{y}(\xi) = 1_{yX}(\xi).
 \end{equation*}
Since \( \cei{d(\xi)} \) only takes integer values and is lower semicomputable, 
there are computable
sequences \( y_{i} \in \Sg^{*} \) and \( k_{i} \in \dN \) with
 \begin{align*}
   2^{\cei{d(\xi)}} &= \sup_{i}\; 2^{k_{i}} 1_{y_{i}}(\xi) \ge 
  \frac{1}{2}\sum_{i} 2^{k_{i}} 1_{y_{i}}(\xi)
 \end{align*}
with the property that if \( i<j \) and \( 1_{y_{i}}(\xi)=1_{y_{j}}(\xi)=1 \) then
\( k_{i}<k_{j} \).
The inequality follows from the fact that for any finite sequence 
\( n_{1}<n_{2}<\dots \), \( \sum_{j} 2^{n_{j}} \le 2 \max_{j} 2^{n_{j}} \).
The function \( \gm(y)=\sum_{y_{i}=y}2^{k_{i}} \) is lower semicomputable.
With it, we have
\begin{align*}
\sum_{i} 2^{k_{i}} 1_{y_{i}}(\xi)  &= \sum_{y\in\Sg^{*}}1_{y}(\xi)\gm(y).
 \end{align*}
Since \( \mu 2^{\cei{d}} \le 2 \), we have \( \sum_{y}\mu(y)\gm(y) \le 2 \), hence
\( \mu(y)\gm(y)\lem\m(y) \), that is \( \gm(y)\lem \m(y)/\mu(y) \).
It follows that
 \[
   2^{d(\xi)} \lem \sup_{y \in \Sg^{*}}\; 1_{y}(\xi)\frac{\m(y)}{\mu(y)}
  = \sup_{n} \frac{\m(\xi_{1:n})}{\mu(\xi_{1:n})}.
 \]
Taking logarithms, we obtain the \( \lea \) part of the theorem.
 \end{proof}

\subsection{Universal semimeasure, algorithmic probability}

Universal semimeasures exist over infinite sequences just as in the discrete
space.

\begin{definition}
  A function \( \mu:\Sg^{*}\to\cR_{+} \) is a \df{semimeasure} if it satisfies
 \begin{align*}
     \mu(x)      &\ge\sum_{s\in\Sg}\mu(xs),
\\ \mu(\Lg) &\le 1.
 \end{align*}
If it is lower semicomputable we will call the semimeasure \df{constructive}.
\end{definition}
These requirements imply, just as for measures,
the following generalization to an arbitrary prefix-free
set \( A\sbsq \Sg^{*} \):
 \begin{align*}
   \sum_{x\in A}\mu(x)\le 1.
 \end{align*}

Standard technique gives the following:

\begin{theorem}[Universal semimeasure]
  There is a universal constructive semimeasure \( \mu \), that is a
constructive semimeasure with the property that for every other constructive
semimeasure \( \nu \) there is a constant \( c_{\nu}>0 \) with
\( \mu(x)\ge c_{\nu}\nu(x) \).
\end{theorem}

\begin{definition}
  Let us fix a universal constructive semimeasure over \( X \) and denote it by
\( M(x) \). 
\end{definition}

There is a graphical way to represent a universal semimeasure, via monotonic
Turing machines, which can be viewed a generalization of self-delimiting
machines. 

\begin{definition}
A Turing machine \( \cT \) is \df{monotonic} if it
has no input tape or output tape, but can ask repeatedly for input symbols and can emit
repeatedly output symbols.
The input is assumed to be an infinite binary string, but it is not assumed that
\( \cT \) will ask for all of its symbols, even in infinite time.
If the finite or 
infinite input string is \( p=p_{1}p_{2}\dotsm \), the (finite or infinite) output string is
written as \( T(p)\in\Sg^{*}\cup\Sg^{\dN} \).
\end{definition}

Just as we can generalize self-delimiting machines to obtain the notion of
monotonic machines, we can generalize algorithmic probability, defined for a
discrete space, to a version defined for infinite sequences.
Imagine the monotonic Turing machine \( \cT \) computing the function \( T(p) \)
as in the definition, and assume that its input symbols come from coin-tossing.

\begin{definition}[Algorithmic probability over infinite sequences]
For a string \( x\in\Sg^{*} \), let \( P_{T}(x) \) be the probability of \( T(\pi)\postfix x \),
when \( \pi=\pi_{1}\pi_{2}\dotsm \) is the infinite coin-tossing sequence.
\end{definition}

We can compute
 \begin{align*}
  P_{T}(x)=\sum_{T(p) \postfix x, p\text{ minimal}} 2^{-\len{p}},
 \end{align*}
where ``minimal'' means that no prefix \( p' \) of \( p \) gives \( T(p')\postfix x \).
This expression shows that \( P_{T}(x) \) is lower semicomputable.
It is also easy to check that it is a semimeasure, so we have a constructive
semimeasure.
The following theorem is obtained using a standard construction:

\begin{theorem}
Every constructive semimeasure \( \mu(x) \) can be represented as
\( \mu(x)=P_{T}(x) \) with the help of an appropriate a monotonic machine \( \cT \).
\end{theorem}

This theorem justifies the following.

\begin{definition}
From now on we will call the universal semimeasure \( M(x) \) also
the \df{algorithmic probability}.
A monotonic Turing machine \( \cT \) giving \( P_{T}(x)=M(x) \) will be called an
\df{optimal machine}.  
\end{definition}

We will see that \( -\log M(x) \) should also be considered a kind of
description complexity:

  \begin{definition}
Denote
 \begin{align*}
   KM(x) =-\log M(x).
 \end{align*}
  \end{definition}

How does \( KM(x) \) compare to \( K(x) \)?

\begin{proposition}\label{prop:KmHCompare}
We have the bounds
 \begin{align*}
 KM(x) \lea K(x) \lea KM(x) + K(\len{x}).
 \end{align*}
\end{proposition}
\begin{proof}
To prove \( KM(x)\lea K(x) \) define the lower semicomputable function
 \begin{align*}
   \mu(x) = \sup_{S}\sum_{y\in S}\m(xy),
 \end{align*}
By its definition this is a constructive semimeasure and it is 
\( \ge \m(x) \).
This implies \( M(x)\gem \m(x) \), and thus \( KM(x)\lea K(x) \).

For the other direction, define the following lower semicomputable function \( \nu \)
over \( \Sg^{*} \): 
 \begin{align*}
   \nu(x)=\m(\len{x})\cdot M(x).
 \end{align*}
To show that it is a semimeasure compute:
 \begin{align*}
   \sum_{x\in\Sg^{*}}\nu(x)=\sum_{n}\m(n)\sum_{x\in\Sg^{n}}M(x)\le\sum_{n}\m(n)\le 1.
 \end{align*}
It follows that \( \m(x)\gem\nu(x) \) and hence \( K(x) \lea KM(x)+K(\len{x}) \).  
\end{proof}

\subsection{Randomness via algorithmic probability}

Fix a computable measure \( P(x) \) over the space \( X \).
Just as in the discrete case, the universal semimeasure gives rise to a
Martin-L\"of test.

\begin{definition}
Denote
 \begin{align*}
 \d'_{P}(\xi) = \log\sup_{n}\frac{M(\xi_{1:n})}{P(\xi_{1:n})}.
 \end{align*}
\end{definition}

By Proposition~\ref{prop:KmHCompare}, this function is
is \( \gem\ol\d_{P}(\xi) \), defined in Theorem~\ref{thm:univ-integrable-test-seq}.
It can be shown that it is no longer expectation-bounded.
But the theorem below shows that it is still a Martin-L\"of
(probability-bounded) test, so it defines the same infinite random sequences.

\begin{theorem}
The function \( \d'_{P}(\xi) \) is a Martin-L\"of test.
\end{theorem}
\begin{proof}
Lower semicomputability follows from the form of definition. 
By standard methods, 
produce a list of finite strings \( y_{1},y_{2},\dots \) with: 
\begin{alphenum}
\item \( M(y_{i})/P(y_{i})>2^{m} \)
\item The cylinder sets \( y_{i}X \) cover of \( G_{m} \), and are disjoint (thus the
  set \( \{y_{1},y_{2},\dots\} \) is prefix-free).
\end{alphenum}
We have
 \begin{align*}
   P(G_{m})=\sum_{i}P(y_{i})<2^{-m}\sum_{i}M(y_{i})\le 2^{-m}.
 \end{align*}
\end{proof}

Just as in the discrete case, the universality of \( M(x) \) implies
 \begin{align*}
   KM(x) \lea -\log P(x) +K(P).
 \end{align*}
      On the other hand,
 \begin{align*}
   \sup_{n} -\log P(\xi_{1:n})-KM(\xi_{1:n})
 \end{align*}
is a randomness test.
So for random \( \xi \), the value \( KM(\xi_{1:n}) \) remains within a constant of
\( -\log P(x) \): it does not oscillate much.

%%% Local Variables: 
%%% mode: latex
%%% TeX-master: "ait-notes"
%%% End: 

%%% Local Variables: 
%%% mode: latex
%%% TeX-master: "ait-notes"
%%% End: 

% \include{infin} % \input into rand

\chapter{Information}

\section{Information-theoretic relations} 

In this subsection, we use some material
from~\cite{GacsTrompVitAlgorStat00}.

\subsection{The information-theoretic identity}

Information-theoretic entropy is related to complexity in both a formal and a
more substantial way.

\subsubsection{Classical entropy and classical coding}

Entropy is used in information theory in order to characterize the
compressibility of a statistical source.

 \begin{definition}
The \df{entropy} of a discrete probability distribution \( P \) is defined as 
 \[
   \cH(P) = -\sum_{x} P(x)\log P(x).
 \]
 \end{definition}
A simple argument shows that this quantity is non-negative.
It is instructive to recall a more general inequality, taken
from~\cite{CsiszarKornerBook81}, implying it:

 \begin{theorem}\label{thm:CsiszarConc}
Let \( a_{i}, b_{i} > 0 \), \( a = \sum_{i} a_{i} \), \( b = \sum_{i} b_{i} \).
Then we have
 \begin{equation}\label{eq:CsiszarConc}
  \sum_{i} a_{i} \log \frac{a_{i}}{b_{i}} \ge a \log \frac{a}{b},
 \end{equation}
with equality only if \( a_{i}/b_{i} \) is constant.
In particular, if \( \sum_{i}a_{i}=1 \) and \( \sum_{i}b_{i}\le 1 \) then we have
\( \sum_{i} a_{i} \log \frac{a_{i}}{b_{i}}\ge 0 \).
 \end{theorem}

\begin{proof} Exercise, an application of Jensen's inequality to the
concave function \( \log x \).
\end{proof}

If \( P \) is a probability distribution and \( Q(x) \ge 0 \), 
\( \sum_{x} Q(x) \le 1 \) then this inequality implies
 \[
   K(P) \le \sum_{x} P(x) \log Q(x).
 \]
We can interpret this inequality as follows.
We have seen in Subsection~\ref{subsec:prefix-codes} that if 
\( x \mapsto z_{x} \) is a prefix code, mapping the objects \( x \) into binary
strings \( z_{x} \) that are not prefixes of each other then
\( \sum_{x} 2^{-\len{z_{x}}} \le 1 \).
Thus, substituting \( Q(x) = 2^{-\len{z_{x}}} \) above we obtain Shannon's result:
 
 \begin{theorem}[Shannon]\label{thm:Shannon-coding-lb}
If \( x \mapsto z_{x} \) is a prefix code then for its expected codeword length
we have the lower bound:
 \[
   \sum_{x} P(x) \len{z_{x}} \ge K(P).
 \]
 \end{theorem}

\subsubsection{Entropy as expected complexity}
Applying Shannon's theorem~\ref{thm:Shannon-coding-lb}
to the code obtained by taking \( z_{x} \) as the
shortest self-delimiting description of \( x \), we obtain the inequality
 \[
   \sum_{x} P(x)K(x) \ge \cH(P)
 \]
On the other hand, since \( \m(x) \) is a universal semimeasure, we have
\( \m(x) \gem P(x) \): more precisely, \( K(x) \lea -\log P(x) + K(P) \), 
leading to the following theorem.

 \begin{theorem} For a computable distribution \( P \) we have
 \begin{equation}\label{eq:exp-entr-compl}
   \cH(P) \le \sum_{x} P(x)K(x) \lea \cH(P) + K(P).
 \end{equation}
 \end{theorem}

Note that \( \cH(P) \) is the entropy of the distribution \( P \) while \( K(P) \) is
just the description complexity of the function \( x \mapsto P(x) \), it has
nothing to do directly
with the magnitudes of the values \( P(x) \) for each \( x \) as real numbers.
These two relations give
 \[ 
K(P) \eqa \sum_{x} P(x)K(x) - K(P),
 \]
the entropy is within an additive constant equal to the
\df{expected complexity}.
Our intended interpretation of \( K(x) \) as information content of the
individual object \( x \) is thus supported by a tight quantitative
relationship to Shannon's statistical concept.

\subsubsection{Identities, inequalities}
The statistical entropy obeys meaningful identities which immediately
follow from the definition of conditional entropy. 

 \begin{definition}
Let \( X \) and \( Y \) be two discrete random variables with a joint distribution.
 The \emph{conditional entropy} of \( Y \) with respect to \( X \) is defined as 
 \[
     \cH(Y \mvert X) =
   -\sum_{x} P[X=x]\sum_y P[Y=y \mvert X=x]\log P[Y=y \mvert X=x].
 \]
The \df{joint entropy} of \( X \) and \( Y \) is defined as the entropy of
the pair \( (X,Y) \). 
 \end{definition}

The following \df{additivity property} is then verifiable by simple calculation:
 \begin{equation}\label{eq:entraddit} 
   \cH(X,Y) = \cH(X) + \cH(Y \mvert X).
 \end{equation}
  Its meaning is that the amount of information in the pair \( (X,Y) \) is
equal to the amount of information in \( X \) plus the amount of our
residual uncertainty about \( Y \) once the value of \( X \) is known.

Though intuitively expected, the following identity needs proof:
 \begin{align*}
   \cH(Y\mvert X) \le \cH(Y).
 \end{align*}
We will prove it after we define the difference below.

 \begin{definition}
The \df{information} in \( X \) about \( Y \) is defined by
 \[
 \cI(X : Y) = \cH(Y) - \cH(Y \mvert X).
 \]
\end{definition}

This is the amount by which our knowledge of \( X \) decreases our 
uncertainty about \( Y \).

The identity below comes from equation~\eqref{eq:entraddit}.
 \[
 \cI(X : Y) = \cI(Y : X) = \cH(X) + \cH(Y) - \cH(X,Y)
 \]
The meaning of the symmetry relation is 
that the \emph{quantity} of information in \( Y \) about
\( X \) is equal to the quantity of information in \( X \) about \( Y \). 
Despite its simple derivation, this fact is not very intuitive; especially
since only the quantities are equal, the actual contents are generally
not (see~\cite{GacsKornComnInf73}). 
We also call \( \cI(X : Y) \) the \df{mutual information} of \( X \) and \( Y \).
We can also write
 \begin{align}\label{eq:info-via-exp-val}
   \cI(X : Y) = \sum_{x,y}P[X=x,Y=y]\log\frac{P[X=x,Y=y]}{P[X=x]P[Y=y]}.
 \end{align}
\begin{proposition}
  The information \( \cI(X:Y) \) is nonnegative and is 0 only if \( X,Y \) are independent.
\end{proposition}
\begin{proof}
Apply Theorem~\ref{thm:CsiszarConc} to~\eqref{eq:info-via-exp-val}.
This is 0 only if \( P[X=x,Y=y]=P[X=x]P[Y=y] \) for all \( x,y \).
\end{proof}

\subsubsection{The algorithmic addition theorem}
For the notion of algorithmic entropy \( K(x) \) defined in 
Section~\ref{sec:codingth}, it is a serious test of fitness to ask whether it
has an additivity property analogous to the identity~\eqref{eq:entraddit}. 
Such an identity would express some deeper
relationship than the trivial identity~\eqref{eq:entraddit} since the
conditional complexity is defined using the notion of 
\emph{conditional computation}, and not by an algebraic formula involving
the unconditional definition.

A literal translation of the identity~\eqref{eq:entraddit} turns out to
be true for the function \( K(x) \) as well as for Kolmogorov's
complexity \( C(x) \) with good approximation. 
But the exact additivity property of algorithmic entropy is subtler. 

\begin{theorem}[Addition] We have
 \begin{equation}\label{eq:Haddit} 
	K(x, y) \eqa K(x) + K(y \mvert x, K(x)). 
 \end{equation}
 \end{theorem}
 \begin{proof}
  We first prove \( \lea \). 
Let \( F \) be the following s.d.~interpreter.
The machine computing \( F(p) \) tries to parse \( p \) into \( p = uv \) so that
\( T(u) \) is defined, then outputs the pair 
\( \tup{ T(u), T(v,\ang{ T(u), \len{u} } )} \). 
If \( u \) is a shortest description of \( x \) and \( v \) a
shortest description of \( y \) given the pair \( \tup{ x,K(x)} \) then the
output of \( F \) is the pair \( \tup{ x,y}  \). 

Now we prove \( \gea \). 
Let \( c \) be a constant (existing in force of 
Lemma~\ref{l:projmeas}) such that for all \( x \) we have 
\( \sum_y2^{-K(x,y)} \le 2^{-K(x)+c} \). 
Let us define the semicomputable function \( f(z,y) \) by 
 \[	
  f(\tup{ x,k} ,y) = 2^{k - K(x,y) - c}.
 \]
  Then for \( k=K(x) \) we have \( \sum_y f(z,y) \le 1 \), hence the
generalization~\eqref{eq:gencodth} of the coding theorem implies 
 \(  K(y \mvert x, K(x)) \lea K(x,y) - K(x)  \).
 \end{proof}

Recall the definition of \( x^{*} \) as the first shortest description of \( x \).

 \begin{corollary}
 We have
 \begin{equation}\label{eq:star-joint}
  K(x^{*}, y) \eqa K(x, y).
 \end{equation}
 \end{corollary}

The Addition Theorem implies 
 \[ 	K(x,y) \lea K(x) + K(y \mvert x) \lea K(x) + K(y) 
 \]
  while these relations do not hold for Kolmogorov's complexity
\( C(x) \), as pointed out in the discussion of Theorem~\ref{thm:addit}. 
For an ordinary interpreter, extra information is needed to separate
the descriptions of \( x \) and \( y \).  
The self-delimiting interpreter
accepts only a description that provides the neccesary information to
determine its end. 

The term \( K(x) \) cannot be omitted from the condition in \( K(y \mvert x, K(x)) \) 
in the identity~\eqref{eq:Haddit}. 
Indeed, it follows from
equation~\eqref{eq:xHx} that with \( y=K(x) \), we have \( K(x,y)-K(x) \eqa 0 \). 
But \( K(K(x) \mvert x) \) can be quite large, as shown in the following
theorem given here without proof.

\begin{theorem}[see~\protect\cite{GacsSymm74}]\label{thm:compl-compl}
 There is a constant \( c \) such that for all \( n \) 
there is a binary string \( x \) of length \( n \) with 
 \[  K(K(x) \mvert x) \ge \log n - \log\log n -c.
 \]
 \end{theorem}
We will prove this theorem in Section~\ref{sec:compl-compl}.
For the study of information relations, let us introduce yet another
notation.

 \begin{definition}
For any functions \( f,g,h \) with values in \( \dS \), let 
 \[  f \preceq g \bmod h
 \] 
  mean that there is a constant \( c \) such that \( K(f(x) \mvert g(x),h(x))
\le c \) for all \( x \). 
We will omit \( \bmod h \) if \( h \) is not in the condition. 
The sign \( \asymp \) will mean inequality in both directions.
 \end{definition}
  
The meaning of \( x \asymp y \) is that the objects \( x \) and \( y \) hold
essentially the same information. 
Theorem~\ref{thm:rectransf}
implies that if \( f\asymp g \) then we can replace \( f \) by \( g \) wherever
it occurs as an argument of the functions \( H \) or \( K \).  
As an illustration of this definition, let us prove the relation
 \[ 	x^{*} \asymp \ang{ x, K(x) } 
 \]
  which implies \( K(y \mvert x, K(x)) \eqa K(y \mvert x^{*}) \), and hence 
 \begin{equation} 	K(x,y) - K(x) \eqa K(y \mvert x^{*}).
\label{eq:Haddit.alt} 
 \end{equation}
  The direction \( \preceq \) holds because we can compute \( x=T(x^{*}) \)
and \( K(x)=\len{x^{*}} \) from \( x^{*} \). 
On the other hand, if we have \( x \)
and \( K(x) \) then we can enumerate the set \( E(x) \) of all shortest
descriptions of \( x \). 
By the corollary of Theorem~\ref{thm:Hstat}, the
number of elements of this set is bounded by a constant. 
The estimate~\eqref{eq:condHupbd} implies \( K(x^{*} \mvert x,K(x)) \eqa 0 \). 

Though the string \( x^{*} \) and the pair \( \tup{ x, K(x) } \) contain the
same information, they are not equivalent under some stronger
criteria on the decoder. 
To obtain \( x \) from the shortest description
may take extremely long time. 
But even if this time \( t \) is
acceptable, to find any description of length \( K(x) \) for \( x \) might
require about \( t2^{K(x)} \) steps of search. 

The Addition Theorem can be made formally analogous to its 
information-theoretic couterpart using Chaitin's notational trick. 

 \begin{definition}
Let \( K^{*}(y \mvert x) = K(y \mvert x^{*}) \).
 \end{definition}

Of course, \( K^{*}(x) = K(x) \). 
Now we can formulate the relation~\eqref{eq:Haddit.alt} as
 \[	K^{*}(x, y) \eqa K^{*}(x) + K^{*}(y \mvert x).
 \]
  However, we prefer to use the function \( K(y \mvert x) \) since the
function \( K^{*}(y \mvert x) \) is not semicomputable. 
Indeed, if there was
a function \( h(x,y) \) upper semicomputable in \( \ang{ x,y} \) such that
\( K(x,y) - K(x) \eqa h(x,y) \) then by the argument of the proof of the 
Addition Theorem, we would have \( K(y \mvert x) \lea h(x,y) \) which is not
true, according to remark after the proof of the Addition Theorem.

\begin{definition}
We define the algorithmic mutual information with the same formula as classical
information: 
 \[
   I(x : y) = K(y) - K(y \mvert x).
 \]
\end{definition}

Since the quantity \(  I(x : y) \) is not quite equal to 
\( K(x) - K(x \mvert y) \) (in fact, Levin showed (see~\cite{GacsSymm74}) that
they are not even asymptotically equal), we prefer to define information also by a 
symmetrical formula.

 \begin{definition}\label{def:I*}
The mutual information of two objects \( x \) and \( y \) is 
 \[
	I^{*}(x : y) = K(x) + K(y) - K(x, y).
 \]
 \end{definition}

The Addition Theorem implies the identity
 \[
 	I^{*}(x : y) \eqa I(x^{*} : y) \eqa I(y^{*} : x).
 \]
Let us show that classical information is indeed the expected value of
algorithmic information, for a computable probability distribution \( p \):

\begin{lemma}
Given a computable joint probability mass distribution \( p(x,y) \) over \( (x,y) \)
we have 
\begin{align}\label{eq:eqamipmi}
\cI(X : Y) - K(p) & \lea  \sum_{x} \sum_y p(x,y) I^{*}(x : y) 
\\& \lea \cI(X : Y) + 2 K(p) , 
\nonumber
\end{align}
where \( K(p) \) is the length of the shortest prefix-free program that computes 
\( p(x,y) \) from input \( (x,y) \).
\end{lemma}
\begin{proof}
We have
\begin{align*}
\sum_{x,y} p(x,y) I^{*}(x : y) &\eqa 
\sum_{x,y} p(x, y)  [K(x) + K(y) - K(x, y)].
\end{align*}
 Define \( \sum_y p(x,y) = p_{1}(x) \) and \( \sum_{x} p(x,y) = p_{2}(y) \)
to obtain
\begin{align*}
\sum_{x,y} p(x,y) I(x:y) &\eqa
 \sum_{x} p_{1} (x) K(x) + \sum_y p_{2} (y) K(y) - \sum_{x,y} p(x,y) K(x, y).
\end{align*}
The distributions \( p_{i} \) (\( i=1,2 \)) are computable.
We have seen in~\eqref{eq:exp-entr-compl} that
\( \cH(q) \lea \sum_{x} q(x) K(x) \lea \cH(q) + K(q) \).

Hence, \( \cH(p_{i}) \lea \sum_{x} p_{i} (x) K(x) \lea \cH(p_{i}) + K(p_{i}) \).
(\( i=1,2 \)), and \( \cH(p) \lea \sum_{x,y} p(x,y) K(x,y) \lea \cH(p) + K(p) \).
On the other hand, the probabilistic mutual information
is expressed in the entropies by \( \cI(X : Y) = \cH(p_{1}) + \cH(p_{2}) - \cH(p) \).
By construction of the \( q_{i} \)'s above,
we have \( K(p_{1}), K(p_{2}) \lea K(p) \). 
Since the complexities are positive, substitution establishes the lemma.
\end{proof}

 \begin{remark}\label{rem:inf-def-rand}
The information \( I^{*}(x : y) \) can be written as
 \[
   I^{*}(x : y) \eqa \log \frac{\m(x, y)}{\m(x)\m(y)}.
 \]
Formally, \( I^{*}(x : y) \) looks like 
\( \ol\d_{\m\times\m}((x, y) \mvert \m \times \m) \) with the function 
\( \ol\d_{\cdot}(\cdot) \) introduced in~\eqref{eq:d_c}.
Thus, it looks like \( I^{*}(x : y) \) measures the deficiency of randomness
with respect to the distribution \( \m \times \m \).
The latter distribution expresses our ``hypothesis'' that \( x,y \) are
``independent'' from each other.
There is a serious technical problem with this interpretation: the function
\( \ol\d_{P}(x) \) was only defined for computable measures \( P \).
Though the definition can be extended, it is not clear at all that the
expression \( \ol\d_{P}(x) \) will also play the role of universal test
for arbitrary non-computable distributions.
Levin's theory culminating in~\cite{LevinRandCons84} develops this idea,
and we return to it in later parts of these notes.
 \end{remark}

% Can we get rid of the \( K(p) \) error term? 
% Yes; by putting \( p(\cdot) \) in the conditional we even get rid of 
% the computability requirement.

% \begin{lemma}
% Given a joint probability mass distribution \( p(x,y) \) over \( (x,y) \)
% (not necessarily computable) we have
%  \[ 
%  \cI(X : Y)  \eqa  \sum_{x,y} p(x,y) I(x : y \mvert p) ,
%  \]
% where the auxiliary \( p \) means that we can directly access the
% values \( p(x,y) \) on the
% auxiliary conditional information tape of the reference
% universal prefix machine.
% \end{lemma}

% \begin{proof}
% The lemma follows from the definition of conditional 
% algorithic mutual information~\eqref{eq:cond-inf} if we show 
% \( \sum_{x} p(x) K(x \mvert p) \eqa \cH(p) \).

% Equip the reference universal prefix machine, with an \( O(1) \) length
% program to compute a Shannon-Fano code from the auxiliary table
% of probabilities.
% Then, given an input \( r \), it can determine 
% whether \( r \) is the Shannon-Fano code word for some \( x \).
% Such a code word has length \( \eqa - \log p(x) \).
% If this is the case, then the machine
% outputs \( x \), otherwise it halts without output. Therefore,
% \( K(x \mvert p) \lea - \log p(x) \).
% This shows the upper bound on the expected prefix complexity. 
% The lower bound follows as usual from the Noiseless Coding Theorem.
% \end{proof}

Let us examine the size of the
defects of naive additivity and information symmetry.

 \begin{corollary} For the defect of additivity, we have
 \begin{equation}
  K(x) + K(y \mvert x) - K(x, y) \eqa I(K(x) : y \mvert x)
 \lea K(K(x) \mvert x).
  \end{equation}
For information asymmetry, we have
 \begin{equation}
 I(x : y) - I(y : x) \eqa I(K(y) : x \mvert y) - I(K(x) : y \mvert x).
 \end{equation}
 \end{corollary}
 \begin{proof}
Immediate from the addition theorem.
 \end{proof}

 \begin{theorem}[Levin]
The information \( I(x : y) \) is not even asymptotically symmetric.
 \end{theorem}
 \begin{proof}
For some constant \( c \), assume
\( |x| = n \) and \( K(K(x) \mvert x) \ge \log n - \log\log n - c \).
Consider \( x, x^{*} \) and \( K(x) \).
Then we can check immediately the relations
 \[
  I(x : K(x)) \lea 3 \log \log n 
              < \log n - \log\log n - c \lea I(x^{*} : K(x)).
 \]
and
 \[
    I(K(x) : x) \eqa I(K(x) : x^{*}).
 \]
It follows that symmetry breaks down in an exponentially great measure
either on the pair \( (x, K(x)) \) or the pair \( (x^{*}, K(x)) \).
 \end{proof}

\subsubsection{Data processing}
The following useful identity is also classical, and is called the
\df{data processing identity}:
 \begin{equation}\label{eq:data-proc-ident}
   \cI(Z : (X, Y)) = \cI(Z : X) + \cI(Z : Y \mvert X).
 \end{equation}
Let us see what corresponds to this for algorithmic information.

Here is the correct conditional version of the addition theorem:
 \begin{theorem}
We have
 \begin{equation}\label{eq:cond-addit}
  K(x, y \mvert z) \eqa K(x \mvert z) + K(y \mvert x, K(x \mvert z), z).
 \end{equation}
 \end{theorem}
 \begin{proof}
The same as for the unconditional version.
 \end{proof}

 \begin{remark}
The reader may have guessed the inequality
 \[
  K(x, y \mvert z) \eqa K(x \mvert z) + K(y \mvert x, K(x), z),
 \]
but it is incorrect: taking \( z = x \), \( y = K(x) \),
the left-hand side equals \( K(x^{*} \mvert x) \), and the right-hand side
equals \( K(x \mvert x) + K(K(x) \mvert x^{*}, x) \eqa 0 \).
 \end{remark}

The following inequality is a useful tool.
It shows that conditional complexity is a kind of one-sided distance, and
it is the algorithmic analogue of the well-known (and not difficult
classical inequality)
 \[
   \cH(X \mvert Y) \le \cH(Z \mvert Y) + \cH(X \mvert Z).
 \]

 \begin{theorem}[Simple triangle inequality]\label{thm:simple-triangle}
We have
 \begin{equation}\label{eq:simple-triangle}
  K(z \mvert x) \lea K(y, z \mvert x) \lea K(y \mvert x) + K(z \mvert y).
 \end{equation}
 \end{theorem}
 \begin{proof}
  According to~\eqref{eq:cond-addit}, we have
\[
  K(y, z \mvert x) \eqa K(y \mvert x) + K(z \mvert y, K(y \mvert x), x).
\]
The left-hand side is \( \gea K(z \mvert x) \), the second term of the
right-hand side is \( \lea K(z \mvert y) \).
 \end{proof}

Equation~\eqref{eq:cond-addit} justifies the following.
 \begin{definition}
 \begin{equation}\label{eq:cond-inf}
 I^{*}(x : y \mvert z) = K(x \mvert z) + K(y \mvert z) - K(x, y \mvert z). 
 \end{equation}
 \end{definition}

Then we have
 \begin{align*}
    I^{*}(x : y \mvert z) &\eqa K(y \mvert z) - K(y \mvert x, K(x \mvert z), z)
\\                      &\eqa K(x \mvert z) - K(x \mvert y, K(y \mvert z), z).
 \end{align*}

 \begin{theorem}[Algorithmic data processing identity]
 We have
 \begin{equation}\label{eq:alg-data-proc-ident}
   I^{*}(z : (x, y)) \eqa I^{*}(z : x) + I^{*}(z : y \mvert x^{*}).
 \end{equation}
 \end{theorem}
 \begin{proof}
We have
 \begin{align*}
    I^{*}(z : (x, y)) &\eqa K(x, y) + K(z) - K(x, y, z)
\\  I^{*}(z : x)      &\eqa K(x)    + K(z) - K(x, z)
\\ I^{*}(z : (x, y)) - I^{*}(z : x) 
            &\eqa K(x, y) + K(x, z) - K(x, y, z) - K(x)
\\          &\eqa K(y \mvert x^{*}) + K(z \mvert x^{*}) - K(y, z \mvert x^{*})
\\          &\eqa I^{*}(z : y \mvert x^{*}),
 \end{align*}
where we used additivity and the definition~\eqref{eq:cond-inf} repeatedly.
 \end{proof}
 \begin{corollary}
The information \( I^{*} \) is monotonic:
 \begin{equation}\label{eq:inf-mon}
   I^{*}(x : (y, z)) \gea I^{*}(x : y).
 \end{equation}
  \end{corollary}

The following theorem is also sometimes useful.

 \begin{theorem}[Starry triangle inequality]\label{lem:triangle}
For all \( x,y,z \), we have
 \begin{equation}\label{eq:starry-triangle}
  K(x \mvert y^*) \lea K(x, z \mvert y^{*}) 
                \lea K(z \mvert y^*) + K(x \mvert z^*).
 \end{equation}
 \end{theorem}
 \begin{proof}
Using the algorithmic data-processing
identity~\eqref{eq:alg-data-proc-ident} and the definitions, we have
 \begin{align*}
 K(z) & - K(z \mvert y^{*}) + 
 K(x \mvert z^{*}) - K(x \mvert y, K(y \mvert z^{*}), z^{*})  
\\ &\eqa I^{*}(y : z) + I^{*}(y : x \mvert z^{*}) \eqa I^{*}(y : (x, z)) 
\\ &\eqa K(x, z) - K(x, z \mvert y^{*}),
 \end{align*}
which gives, after cancelling \( K(z) + K(x \mvert z^{*}) \) on the
left-hand side with \( K(x, z) \) on the right-hand side, to which it is equal
according to additivity, and changing sign:
 \[
  K(x, z \mvert y^{*}) \eqa
  K(z \mvert y^{*}) + K(x \mvert y, K(y \mvert z^{*}), z^{*}).
 \]
From here, we proceed as in the proof of the simple triangle
inequality~\eqref{eq:simple-triangle}
\end{proof}

\subsection{Information non-increase}

In this subsection, we somewhat follow the exposition  
of~\cite{GacsTrompVitAlgorStat00}.

The classical data processing identity has a simple but important
application.
Suppose that the random variables \( Z, X, Y \) form a Markov chain in this
order.
Then \( I(Z : Y \mvert X) = 0 \) and hence \( I(Z : (X, Y)) = I(Z : X) \).
In words: all information in \( Z \) about \( Y \) is coming through \( X \): a Markov
transition from \( X \) to \( Y \) cannot increase the information about \( Z \).
Let us try to find the algorithmic counterpart of this theorem.

Here, we rigorously show that this is the case in the algorithmic statistics
setting:
the information in one object about another
cannot be increased by any deterministic algorithmic method
by more than a constant. 
With added randomization this holds with overwhelming probability.
For more elaborate versions see~\cite{Levin74,LevinRandCons84}.

Suppose we want to obtain information about a certain object \( x \). 
It does not seem to be a good policy to guess blindly. 
This is confirmed by the following inequality.
 \begin{equation}\label{eq:info-guess}
	\sum_y \m(y)2^{I(x : y)} \lem 1 
 \end{equation}
  which says that for each \( x \), the expected value of \( 2^{I(x : y)} \) is 
small, even with respect to the universal constructive semimeasure 
\( \m(y) \). 
The proof is immediate if we recognize that by the Coding 
Theorem, 
 \[	2^{I(x : y)} \eqm \frac{2^{-K(y \mvert x)}}{\m(y)}
 \]
  and use the estimate~\eqref{eq:condHsum}.

We prove a strong version of the information non-increase law
under deterministic processing (later we need the attached corollary):

\begin{theorem}
Given \( x \) and \( z \), let \( q \) be a program computing \( z \) from \( x^{*} \).
Then we have
 \begin{equation*}
   I^{*}(y : z) \lea I^{*}(y : x) + K(q).
 \end{equation*}
\end{theorem}
\begin{proof}
By monotonicity~\eqref{eq:inf-mon} and
the data processing identity~\eqref{eq:alg-data-proc-ident}:
\[
 I^{*}(y : z) \lea
I^{*}(y : (x, z)) \eqa I^{*}(y : x) + I^{*}(y : z \mvert x^{*}).   
\]
By definition of information and the definition of \( q \):
 \[
   I^{*}(y : z \mvert x^{*}) \lea K(z \mvert x^{*}) \lea K(q).
 \]
\end{proof}

Randomized computation can increase information only with negligible
probability.
Suppose that \( z \) is obtained from \( x \) by some randomized computation.
The probability \( p(z \mvert x) \) of obtaining \( z \) from \( x \) is a semicomputable
distribution over the \( z \)'s.
The information increase \( I^{*}(z : y) - I^{*}(x : y) \) satisfies the
theorem below. 

 \begin{theorem} For all \( x, y, z \) we have
 \[
  \m(z \mvert x^{*}) 2^{I^{*}(z : y) - I^{*}(x : y)} 
  \lem \m(z \mvert x^{*}, y, K(y \mvert x^{*})).
 \]
Therefore it is upperbounded by \( \m(z \mvert x) \lem \m(z \mvert x^{*}) \).
 \end{theorem}

\begin{remark}
The theorem implies 
 \begin{equation}\label{eq:inf-incre-exp}
   \sum_{z} \m(z \mvert x^{*}) 2^{I^{*}(z : y) - I^{*}(x : y)} \lem 1.
 \end{equation}
This says that
the \( \m(\cdot \mvert x^{*}) \)-expectation of the exponential of the increase
is bounded by a constant.
It follows that for example, 
the probability of an increase of mutual information
by the amount \( d \) is \( \lem 2^{-d} \).
\end{remark}

 \begin{proof}
We have, by the data processing inequality:
 \[
 I^{*}(y : (x, z)) - I^{*}(y : x) 
   \eqa I^{*}(y : z \mvert x^{*}) 
   \eqa K(z \mvert x^{*}) - K(z \mvert y, K(y \mvert x^{*}), x^{*}).
 \]
Hence, using also \( I^{*}(y: (x, z)) \gea I^{*}(y : z) \) by monotonicity:
 \[
 I^{*}(y : z) - I^{*}(y : x) - K(z \mvert x^{*}) 
   \lea - K(z \mvert y, K(y \mvert x^{*}), x^{*}).
 \]
Putting both sides into the exponent gives the statement of the theorem.
 \end{proof}

\begin{remark}
The theorem on information non-increase and its proof look similar to
the theorems on randomness conservation.
There is a formal connection, via the observation made in
Remark~\ref{rem:inf-def-rand}.
Due to the difficulties mentioned there, we will explore the connection
only later in our notes.
\end{remark}

% We will use \(  \x = \N^\N, \x_{r} = \dZ_{r}^\N  \) 
% for the sets of infinite strings of natural numbers and \( r \)-ary 
% digits respectively. We will write \( \b=\x_{2} \).
% For a subset \( H \) of some set \( A^* \) of strings and some set 
% \( C \subset A^*\cup A^\dN \), let \(  HC = \setof{ xy : y \in C }  \). 
% Also, let \( xC = \{x\}C \) for \( x \in A^* \). 
% For an \( r \)-ary sequence \( p\in \x_{r} \), let \( [p]_{r} \) denote the 
% real number in the interval \( [0,1] \) which \( 0.p \) denotes in the \( r \)-ary 
% number system.  For \( A \subset \dS_{r} \), let  
% \(  [A]_{r} = \setof{ [\og]_{r} : \og \in A\b } \). For \( p \in \dS_{r} \), let 
% \(  [p]_{r} = [ \{ p \} ]_{r}  \). We omit the subscript 2 from \( [\cdot]_{2} \). 

% For \( E,F \subset \dS \), let 
% \(  E' = \setof{ x\in E : \all y\in E \; y\subset x \;\impl\; y=x }  \)
% and 
% \(  E \le F \;\iff\; \all x\in E\; \some y \in F\; y \subset x  \). 
% The relation \( E \le F \) implies \( E\x \subset F\x \) but the converse is not 
% true. 

\section{The complexity of decidable and enumerable sets}\label{sec:halting}

If \( \og = \og(1)\og(2)\dotsb \) is an infinite binary sequence, then we may
be interested in how the complexity of the initial segments
 \[
  \og(\frto{1}{n}) = \og(1) \dotsb \og(n)
 \]
grows with \( n \).
We would guess that if \( \og \) is ``random'' then \( C(\og(\frto{1}{n})) \) is
approximately \( n \); this question will be studied satisfactorily in later
sections.
Here, we want to look at some other questions.

Can we determine whether the sequence \( \og \) is computable, by just looking
at the complexity of its initial segments?
It is easy to see that if \( \og \) is computable then 
\( C(\og(\frto{1}{n})) \lea \log n \).
But of course, if \( C(\og(\frto{1}{n})) \lea \log n \) then it is not necessarily true
yet that \( \og \) is computable.
Maybe, we should put \( n \) into the condition.
If \( \og \) is computable then \( C(\og(\frto{1}{n}) \mvert n) \lea 0 \).
Is it true that \( C(\og(\frto{1}{n}) \mvert n) \lea 0 \) for all \( n \) then indeed
\( \og \) is computable?
Yes, but the proof is quite difficult (see either its original, attributed
to Albert Meyer in~\cite{Loveland69}, or in~\cite{ZvLe70}).
The suspicion arises that when
measuring the complexity of starting segments of infinite sequences,
neither Kolmogorov complexity nor its prefix version are the most natural
choice.
Other examples will support the suspicion, so let us introduce, following
Loveland's paper~\cite{Loveland69}, the following variant. 

\begin{sloppypar}
 \begin{definition}
Let us say that a program \( p \) \df{decides} a string \( x = (x(1),x(2),\dotsc) \)
on a Turing machine \( T \) up to \( n \) if for all \( i \in \{1,\dots,n\} \) we have
 \[
   T(p,i) = x(i).
 \]
Let us further define the \df{decision complexity}
 \[
   C_{T}(x; n) = \min\setof{\len{p} : p \txt{ decides \( x \) on \( T \) up to \( n \)}}.
 \]
 \end{definition}
  \end{sloppypar}

As for the Kolmogorov complexity, an invariance theorem holds, and there is
an optimal machine \( T_{0} \) for this complexity.
Again, we omit the index \( T \), assuming that such an optimal machine has
been fixed.

The differences between \( C(x) \), \( C(x \mvert n) \) and \( C(x ; n) \) are
interesting to explore; intuitively, the important difference is that the
program achieving the decision complexity of \( x \) does not have to offer
precise information about the length of \( x \).
We clearly have
 \[
   C(x \mvert n) \lea C(x ; n) \lea C(x).
 \]
Examples that each of these inequalities can be strict, are left to
exercises.

 \begin{remark}
If \( \og(\frto{1}{n}) \) is a segment of an infinite sequence \( \og \) then we can write
 \[
   C(\og ; n)
 \]
instead of \( C(\og(\frto{1}{n}) ; n) \) without confusion, since there is only one way
to understand this expression.
 \end{remark}

Decision complexity offers an easy characterization of decidable infinite
sequences.

 \begin{theorem}
Let \( \og = (\og(1), \og(2),\dotsc) \) be an infinite sequence.
Then \( \og \) is decidable if and only if
\begin{equation}\label{eq:dec-compl-0}
   C(\og ; n) \lea 0.
\end{equation}
 \end{theorem}
 \begin{sloppypar}
 \begin{proof}
If \( \og \) is decidable then~\eqref{eq:dec-compl-0} is obviously true.
Suppose now that~\eqref{eq:dec-compl-0} holds:
there is a constant \( c \) such that for all \( n \) we have \( C(\og(\frto{1}{n});n) < c \).
Then there is a program \( p \) of length \( \le c \) with the property that for
infinitely many \( n \), decides \( x \) on the optimal machine \( T_{0} \) up to \( n \).
Then for all \( i \), we have \( T_{0}(p,i) = \og(i) \), showing that \( \og \) is
decidable.
 \end{proof}
 \end{sloppypar}

Let us use now the new tool to prove a somewhat more surprising result.
Let \( \og \) be a 0-1 sequence that is the indicator function of a recursively
enumerable set.
In other words, the set \( \setof{i: \og(i) = 1} \) is recursively enumerable.
As we know such sequences are not always decidable, so \( C(\og(\frto{1}{n}) ; n) \) is
not going to be bounded.
How fast can it grow?
The following theorem gives exact answer, showing that it grows
only logarithmically.

 \begin{theorem}[Barzdin, see~\cite{Barzdin68}]\
 \begin{alphenum}
  \item\label{i:Barzdin-upper}
Let \( \og \) be the indicator function of a recursively enumerable set \( E \).
Then we have
 \[
   C(\og ; n) \lea \log n.
 \]
  \item\label{i:Barzdin-lower}
The set \( E \) can be chosen such that for all \( n \) we have 
\( C(\og ; n) \ge \log n \).
 \end{alphenum}
 \end{theorem}
 \begin{proof}
Let us prove~\eqref{i:Barzdin-upper} first.
Let \( n' \) be the first power of 2 larger than \( n \).
Let \( k(n) = \sum_{i \le n'} \og(i) \).
For a constant \( d \), let \( p = p(n, d) \) be a program
that contains an initial segment \( q \) of size \( d \) followed by a
the number \( k(n) \) in binary notation, and padded to \( \cei{\log n} + d \) 
with 0's.
The program \( q \) is self-delimiting, so it sees where \( k(n) \) begins.

The machine \( T_{0}(p, i) \) works as follows, under the control of \( q \).
From the length of the whole \( p \), it determines \( n' \).
It begins to enumerate \( \og(\frto{1}{n}') \) until it found \( k(n) \) 1's.
Then it knows the whole \( \og(\frto{1}{n}') \), so it outputs \( \og(i) \).

Let us prove now~\eqref{i:Barzdin-lower}.
Let us list all possible programs \( q \) for the machine \( T_{0} \) as 
\( q = 1, 2, \dotsc \).
Let \( \og(q) = 1 \) if and only if \( T_{0}(q, q) = 0 \).
The sequence \( \og \) is obviously the indicator function of a recursively
enumerable set.
To show that \( \og \) has the desired property, assume that for some
\( n \) there is a \( p \) with \( T_{0}(p, i) = \og(i) \) for all \( i \le n \).
Then \( p > n \) since \( \og(p) \) is defined to be different from \( T_{0}(p, p) \).
It follows that \( \len{p} \ge \log n \).
 \end{proof}

Decision complexity has been especially convenient for the above theorem.
Neither \( C(\og(\frto{1}{n})) \) nor \( C(\og(\frto{1}{n}) \mvert n) \) would have been suitable to
formulate such a sharp result.
To analyze the phenomenon further, we introduce some more concepts.

 \begin{definition}
Let us denote, for this section, by \( E \) the set of those binary strings \( p \) on
which the optimal prefix machine \( T \) halts:
 \begin{equation}\label{eq:E^{t}}
 \begin{split}
     E^{t} &= \setof{p : T(p) \txt{ halts in \( < t \) steps} },
\\   E     &= E^{\infty}.
 \end{split}
 \end{equation}
Let 
 \begin{equation}\label{eq:chi}
  \chi = \chi_{E}
 \end{equation}
be the infinite sequence that is the indicator function of the
set \( E \), when the latter is viewed as a set of numbers.
 \end{definition}

It is easy to see that the set \( E \) is complete among recursively enumerable
sets with respect to many-one reduction.
The above theorems show that though it contains an infinite amount of
information, this information is not stored in the sequence \( \chi \) densely at
all: there are at most \( \log n \) bits of it in the segment \( \chi(\frto{1}{n}) \).
There is an infinite sequence, though, in which the same information is
stored much more densely: as we will see later, maximally densely.

 \begin{definition}[Chaitin's Omega]
Let
 \begin{equation}\label{eq:Og}
 \begin{split}
   \Og^{t} &= \sum_{p \in E^{t}} 2^{-\len{p}},
\\ \Og     &= \Og^{\infty}.
 \end{split}
 \end{equation}
 \end{definition}

Let \( \Og(\frto{1}{n}) \) be the sequence of the first \( n \) binary digits in the 
expansion of \( \Og \), and let it also denote the binary number 
\( 0.\Og(1)\dotsb\Og(n) \).
Then we have
 \[
    \Og(\frto{1}{n}) < \Og < \Og(\frto{1}{n}) + 2^{-n}.
 \]
  
 \begin{theorem}
The sequences \( \Og \) and \( \chi \) are Turing-equivalent.
 \end{theorem}
 \begin{proof}
Let us show first that given \( \Og \) as an oracle, a Turing machine can
compute \( \chi \).
Suppose that we want to know for a given string \( p \) of length \( k \) 
whether \( \chi(p) = 1 \) that is whether \( T(p) \) halts.
Let \( t(n) \) be the first \( t \) for which \( \Og^{t} > \Og(\frto{1}{n}) \).
If a program \( p \) of length \( n \) is not in \( E^{t(n)} \) then it is not in \( E \)
at all, since \( 2^{-\len{p}} = 2^{n} < \Og - \Og^{t} \).
It follows that \( \chi(1 : 2^{n}) \) can be computed from \( \Og(\frto{1}{n}) \).

To show that \( \Og \) can be computed from \( E \), let
us define the recursively enumerable set
 \[
   E' = \setof{r : r \txt{ rational, } \Og > r}.
 \]
The set \( E' \) is reducible to \( E \) since the latter
is complete among recursively enumerable sets
with respect to many-one reduction.
On the other hand, \( \Og \) is obviously computable from \( E' \).
 \end{proof}

The following theorem shows that \( \Og \) stores the information more densely.

 \begin{theorem}\label{thm:Og-complex}
 We have \( K(\Og(\frto{1}{n})) \gea n \).
 \end{theorem}
 \begin{proof}
Let \( p_{1} \) be a self-delimiting program outputting \( \Og(\frto{1}{n}) \).
Recall the definition of the sets \( E^{t} \) in~\eqref{eq:E^{t}}
and the numbers \( \Og^{t} \) in~\eqref{eq:Og}.
Let \( \Og_{1} \) denote the real number whose binary digits after 0 are given
by \( \Og(\frto{1}{n}) \), and let \( t_{1} \) be the first \( t \) with \( \Og^{t}>\Og_{1} \).
Let \( x_{1} \) be the first string \( x \) such that \( T(p)\ne x \) for any
\( p\in E^{t_{1}} \).

We have \( K(x_{1})\ge n \).
On the other hand, we just computed \( x_{1} \) from \( p_{1} \), so
\( K(x_{1}\mvert p_{1})\lea 0 \).
We found \( n \le K(x_{1}) \lea K(p_{1}) + K(x_{1}\mvert p_{1}) \lea K(p_{1}) \).
 \end{proof}

\section{The complexity of complexity}\label{sec:compl-compl}

\subsection{Complexity is sometimes complex}

This section is devoted to a quantitative estimation of the uncomputability
of the complexity function \( C(x) \).
Actually, we will work with the prefix complexity \( K(x) \), but the results
would be very similar for \( C(x) \).
The first result shows that the value
\( K(x) \) is not only not computable from \( x \), but
its conditional complexity \( K(K(x) \mvert x) \) given \( x \) is sometimes quite
large.
How large can it be expected to be?
Certainly not much larger than \( \log K(x) + 2\log\log K(x) \), since we can
describe any value \( K(x) \) using this many bits.
But it can come close to this, as shown by Theorem~\ref{thm:compl-compl}.
This theorem says that for all \( n \), there exists \( x \) of length \( n \) with
 \begin{equation}\label{eq:compl-compl}
  K(K(x) \mvert x) \gea \log n - \log \log n.
 \end{equation}

 \begin{proof}[Proof of Theorem~\protect\ref{thm:compl-compl}]
Let \( U(p, x) \) be the optimal self-delimiting machine for which
\( K(y \mvert x) = \min_{U(p, x) = y} \len{p} \).
Let \( s \le \log n \) be such that if \( \len{x} = n \) then 
a \( p \) of length \( x \) can be found for which \( U(p, x) = K(x) \).
We will show
 \begin{equation}\label{eq:compl-compl.s-lb}
   s \gea \log n - \log \log n.
 \end{equation}
Let us say that \( p \in \{0, 1\}^{s} \) is \df{suitable} for 
\( x \in \{0,1\}^{n} \) if there exists a \( k = U(p, x) \) and a
\( q \in \{0,1\}^{k} \) with \( U(p, \Lambda) = x \).
Thus, \( p \) is suitable for \( x \) if it produces the length of some program of
\( x \), not necessarily a shortest program.

Let \( M_{i} \) denote the set of those \( x \) of length \( n \)
for which there exist at least \( i \)
different suitable \( p \) of length \( s \).
We will examine the sequence
 \[
   \{0,1\}^{n} = M_{0} \spsq M_{1} \spsq \dots 
   \spsq M_{j} \spsq M_{j+1} = \emptyset,
 \]
where \( M_{j} \ne \emptyset \).
It is clear that \( 2^{s} \ge j \).
To lowerbound \( j \), we will show that the sets \( M_{i} \) decrease only slowly:
 \begin{equation}\label{eq:compl-compl.slow-decr}
  \log \card{M_{i}} \lea \log\card{M_{i+1}} + 4 \log n.
 \end{equation}
We can assume 
 \begin{equation}\label{eq:compl-compl.assume}
   \log\card{M_{i} \xcpt M_{i+1}} \ge \log\card{M_{i}} - 1,
 \end{equation}
otherwise~\eqref{eq:compl-compl.slow-decr} is trivial.
We will write a program that finds an element \( x \) of 
\( M_{i} \xcpt M_{i+1} \) with the property \( K(x) \ge \log\card{M_{i}} - 1 \).
The program works as follows.
 \begin{itemize}

  \item It finds \( i,n \) with the help of descriptions of length
\( \log n + 2\log\log n \), and \( s \) with the help of a description of length
\( 2 \log\log n \).

  \item It finds  \( \card{M_{i+1}} \) with the help of a description of length
\( \log\card{M_{i+1}} + \log n + 2\log\log n \).

  \item From these data, it determines the set \( M_{i+1} \), and then begins
to enumerate the set \( M_{i} \xcpt M_{i+1} \) as \( x_{1}, x_{2},\dotsc \).
For each of these elements \( x_{r} \), it knows that there are exactly \( i \)
programs suitable for \( x_{r} \), find all those, and find the shortest
program for \( x \) produced by these.
Therefore it can compute \( K(x_{r}) \).

  \item According to the assumption~\eqref{eq:compl-compl.assume},
there is an \( x_{r} \) with \( K(x_{r}) \ge \log\card{M_{i}} - 1 \).
The program outputs the first such \( x_{r} \).

 \end{itemize}
The construction of the program shows that its length is
\( \lea \log\card{M_{i+1}} + 4\log n \), hence for the \( x \) 
we found 
 \[
   \log\card{M_{i}} - 1 \le K(x) \lea \log\card{M_{i+1}} + 4\log n,
 \]
which proves~\eqref{eq:compl-compl.slow-decr}.
This implies \( j \ge n/(4 \log n) \), and hence~\eqref{eq:compl-compl.s-lb}.
 \end{proof}

\subsection{Complexity is rarely complex}

Let
 \[
  f(x) = K(K(x) \mvert x).
 \]
We have seen that \( f(x) \) is sometimes large.
Here, we will show that the sequences \( x \) for which this happens are rare.
Recall that we defined \( \chi \) in~\eqref{eq:chi} as
the indicator sequence of the halting problem.

 \begin{theorem}\label{thm:co-rarely-co} 
We have
 \[
   I(\chi : x) \gea f(x) - 2.4 \log f(x).
 \]
 \end{theorem}

In view of the inequality~\eqref{eq:info-guess}, this shows that such
sequences are rare, even in terms of the universal probability, so 
they are certainly rare if we measure them by any
computable distribution.
So, we may even call such sequences ``exotic''.

In the proof, we 
start by showing that the sequences in question are rare.
Then the theorem will follow when we make use of the fact that 
\( f(x) \) is computable from \( \chi \).

We need a lemma about the approximation of one measure by another from
below.

  \begin{lemma}\label{lem:measure-approx}
For any semimeasure \( \nu \) and measure \( \mu \le \nu \), let
\( S_{m} = \setof{x : 2^{m} \mu(x) \le \nu(x)} \).
Then \( \mu(S_{m}) \le 2^{-m} \).
If \( \sum_{x} (\nu(x) - \mu(x)) < 2^{-n} \) then \( \nu(S_{1}) < 2^{-n+1} \).
  \end{lemma}

The proof is straightforward.

Let 
 \[
     X(n) = \setof{x : f(x) = n}.
 \]

 \begin{lemma}\label{lem:co-co-rare} We have
 \[
  \m(X(n)) \lem n^{1.2} 2^{-n}.
 \]
 \end{lemma}
 \begin{proof}
Using the definition of \( E^{t} \) in~\eqref{eq:E^{t}}, let
\begin{align*}
      \m^{t}(x) &= \sum\setof{p : T(p) = x \txt{ in \( < t \) steps} },
\\    K^{t}(x)  &= -\log \m^{t}(x).
\end{align*}
Then according to the definition~\eqref{eq:m(x)} we have
\( \m(x) = \m^{\infty}(x) \), and \( K(x) \eqa K^{\infty}(x) \).
Let
 \begin{align*}
     t(k) &= \min\setof{t : \Og(1 : k) < \Og^{t}},
\\ \mu    &= \m^{t(k)}.
 \end{align*}
Clearly, \( K^{t(k)}(x) \gea K(x) \).
Let us show that \( K^{t(k)}(x) \)
is a good approximation for \( K(x) \) for most \( x \).
Let
 \[
   Y(k) = \setof{x : K(x) \le K^{t(k)}(x) - 1}.
 \]
By definition, for \( x \notin Y(k) \) we have
 \[
   |K^{t(k)}(x) - K(x)| \eqa 0.
 \]
On the other hand, applying Lemma~\ref{lem:measure-approx} with
\( \mu = \m^{t(k)} \), \( \nu = \m \), we obtain 
 \begin{equation}\label{eq:Yk-ub}
   \m(Y(k)) < 2^{-k + 1}.
 \end{equation}
Note that
 \[
   K(K^{t(k)}(x) \mvert x, \Og(1 : k)) \eqa 0,
 \]
therefore for \( x \notin Y(k) \) we have
 \begin{align*}
  K(K(x) \mvert x, \Og(1 : k)) &\eqa 0,
\\   K(K(x) \mvert x) &\lea k + 1.2 \log k.
 \end{align*}
If \( n = k + 1.2\log k \) then \( k \eqa n - 1.2\log n \), and hence,
if \( x \notin Y(n - 1.2\log n) \) then \( K(K(x) \mvert x) \lea n \).
Thus, there is a constant \( c \) such that 
 \[
   X(n) \sbsq Y(n - 1.2\log n - c).
 \]
Using~\eqref{eq:Yk-ub} this gives the statement of the lemma.
 \end{proof}

 \begin{proof}[Proof of Theorem~\protect\ref{thm:co-rarely-co}]
Since \( f(x) \) is computable from \( \Og \), the function
 \[
   \nu(x) = \m(x)2^{f(x)}(f(x))^{-2.4}
 \]
is computable from \( \Og \).
Let us show that it is a semimeasure (within a multiplicative constant).
Indeed, using the above lemma:
 \[
   \sum_{x} \nu(x) = \sum_{k} \sum_{x \in X(k)} \nu(x)
                 = \sum_{k} 2^{k} k^{-2.4} \m(X(k))
                 = \sum_{k} k^{-1.2} \lem 1.
 \]
Since \( \m(\cdot \mvert \Og) \) is the universal semimeasure relative to \( \Og \)
we find \( \m(x \mvert \Og) \gem \nu(x) \), hence
 \begin{align*}
   K(x \mvert \Og) &\lea -\log \nu(x) = K(x) - f(x) + 2.4\log f(x),
\\    I(\Og : x) &\gea f(x) - 2.4 \log f(x).
 \end{align*}
Since \( \Og \) is equivalent to \( \chi \), the proof is complete.
  \end{proof}

%%% Local Variables: 
%%% mode: latex
%%% TeX-master: "ait-notes"
%%% End: 

%%% Local Variables: 
%%% mode: latex
%%% TeX-master: "ait-notes"
%%% End: 

% \include{co-co} % \input into info

\chapter{Generalizations}

 \section{Continuous spaces, noncomputable measures}\label{sec:cont-spaces}

This section starts the consideration of randomness in continuous spaces and
randomness with respect to noncomputable measures.

\subsection{Introduction}

The algorithmic theory of randomness is well developed when the underlying
space is the set of finite or infinite sequences and the underlying
probability distribution is the uniform distribution or a computable
distribution.
These restrictions seem artificial.  
Some progress has been
made to extend the theory to arbitrary Bernoulli distributions by
Martin-L\"of in~\cite{MLof66art}, and to arbitrary distributions, by Levin
in~\cite{LevinRand73,LevinUnif76,LevinRandCons84}.
The paper~\cite{HertlingWeihrauchRand98} by Hertling and Weihrauch 
also works in general spaces, but it is restricted to computable
measures.
Similarly, Asarin's thesis~\cite{Asarin88} defines randomness for sample
paths of the Brownian motion: a fixed random process with computable
distribution.

The exposition here has been inspired mainly by Levin's early 
paper~\cite{LevinUnif76} (and the much more elaborate~\cite{LevinRandCons84}
that uses different definitions): let us summarize part of the content
of~\cite{LevinUnif76}.
The notion of a constructive topological space \( \bX \) and the space of
measures over \( \bX \) is introduced.
Then the paper defines the notion of a uniform test.
Each test is a lower semicomputable function \( (\mu,x) \mapsto f_{\mu}(x) \),
satisfying \( \int f_{\mu}(x) \mu(dx) \le 1 \) for each measure \( \mu \).
There are also some additional conditions.
The main claims are the following.
 \begin{alphenum}

  \item There is a universal test \( \t_{\mu}(x) \), a test 
such that for each other test \( f \) there is a constant \( c > 0 \) with
\( f_{\mu}(x) \le c\cdot \t_{\mu}(x) \).
  The \df{deficiency of randomness} is defined as 
\( \d_{\mu}(x) = \log\t_{\mu}(x) \). 

  \item The universal test has some strong properties of ``randomness
conservation'': these say, essentially,
that a computable mapping or a computable
randomized transition does not decrease randomness.

  \item There is a measure \( M \) with the property that for
every outcome \( x \) we have \( \t_{M}(x) \le 1 \).
In the present work, we will call such measures \df{neutral}.

  \item\label{i:Levin.semimeasure}
Semimeasures (semi-additive measures)
are introduced and it is shown that there is a lower semicomputable
semi-measure that is neutral (let us assume that the \( M \) introduced above
is lower semicomputable).

  \item Mutual information \( I(x : y) \)
is defined with the help of (an appropriate
version of) Kolmogorov complexity, between outcomes \( x \) and \( y \).
It is shown that \( I(x : y) \) is essentially equal to \( \d_{M \times M}(x,y) \).
This interprets mutual information as a kind of ``deficiency of independence''.

 \end{alphenum}
This impressive theory leaves a number of issues unresolved:
  \begin{enumerate}

   \item The space of outcomes is restricted to be a compact topological
space, moreover, a particular compact space: the set of 
sequences over a finite alphabet (or, implicitly in~\cite{LevinRandCons84},
a compactified infinite alphabet).
However, a good deal of modern probability theory happens over
spaces that are not even locally compact:
for example, in case of the Brownian motion, over the
space of continuous functions.

   \item The definition of a uniform randomness test includes some conditions
(different ones in~\cite{LevinUnif76} and 
in~\cite{LevinRandCons84}) that seem somewhat arbitrary.

   \item No simple expression is known for the general
universal test in terms of description complexity.
Such expressions are nice to have if they are available.

  \end{enumerate}

Here, we intend to carry out as much of Levin's program as seems
possible after removing the restrictions.
A number of questions remain open, but we feel that they are worth
to be at least formulated.
A fairly large part of the exposition
is devoted to the necessary conceptual machinery.
This will also allow to carry further some other initiatives started in the
works~\cite{MLof66art} and~\cite{LevinRand73}: the study of tests that test
nonrandomness with respect to a whole class of measures (like the Bernoulli
measures).

Constructive analysis has been developed by several authors,
converging approximately on the same concepts,
We will make use of a simplified  version of the theory introduced  
in~\cite{WeihrauchComputAnal00}.
As we have not found a constructive measure theory in the
literature fitting our purposes, we will develop this theory here,
over (constructive) complete separable metric spaces.
This generality is well supported by standard results in measure
theoretical probability, and is
sufficient for constructivizing a large part of current probability
theory.

The appendix recalls some of the needed topology, measure theory,
constructive analysis and constructive measure theory.
We also make use of the notation introduced there.

\subsection{Uniform tests}

We first define tests of randomness with respect to an arbitrary measure.
Recall the definition of lower semicomputable
real functions on a computable metric space \( \bX \).

\begin{sloppypar}
 \begin{definition}\label{def:mu-test}
Let us be given a computable complete metric space \( \bX = (X, d, D, \ag) \).
For an arbitrary measure \( \mu\in\cM(\bX) \), a \( \mu \)-\df{test of randomness}
is a \( \mu \)-lower semicomputable function \( f:X\to\ol\dR_{+} \) with the property
\( \mu f \le 1 \).
We call an element \( x \) \df{random} with respect to \( \mu \)
if \( f(x)<\infty \) for all \( \mu \)-tests \( f \).
But even among random elements, the size of the tests quantifies
(non-)randomness.

A \df{uniform test} of randomness is a lower semicomputable function
\( f:\bX\times\cM(\bX)\to \ol\dR_{+} \), 
written as \( \tup{x,\mu}\mapsto f_{\mu}(x) \)
such that \( f_{\mu}(\cdot) \) is a \( \mu \)-test for each \( \mu \).
 \end{definition}
\end{sloppypar}

The condition \( \mu f\le 1 \) guarantees that the
probability of those outcomes whose randomness is \( \ge m \) is at most
\( 1/m \).
The definition of tests is in the spirit of Martin-L\"of's tests.
The important difference is in the semicomputability condition:
instead of restricting the measure \( \mu \) to be computable, we require the
test to be lower semicomputable also in its argument \( \mu \).

The following result implies that every \( \mu \)-test can be extended to a uniform
test.

 \begin{theorem}\label{thm:trim}
Let \( \phi_{e} \), \( e=1,2,\dots \) be an enumeration of all lower semicomputable
functions \( \bX\times\bY\times\cM(\bX)\to\ol\dR_{+} \), where \( \bY \) is also a
computable
metric space, and \( s:\bY\times\cM(\bX)\to\ol\dR \) a lower semicomputable
function.
There is a recursive function \( e\mapsto e' \) with the property that
 \begin{alphenum}
   \begin{sloppypar}
  \item For each \( e \), the function \( \phi_{e'}(x,y,\mu) \) is everywhere defined
with \( \mu^{x}\phi_{e'}(x,y,\mu)\le s(y,\mu) \).     
   \end{sloppypar}
  \item For each \( e,y,\mu \), if \( \mu^{x} \phi_{e}(x,y\mu)<s(y,\mu) \) then 
\( \phi_{e'}(\cdot,y,\mu)=\phi_{e}(\cdot,y,\mu) \).
 \end{alphenum}
 \end{theorem}
This theorem generalizes a theorem of Hoyrup and Rojas (in allowing a lower
semicomputable upper bound \( s(y,\mu) \)).
 \begin{proof}
By Proposition~\ref{propo:semicomp-as-limit}, we can represent \( \phi_{e}(x,y,\mu) \) 
as a supremum 
\( \phi_{e}=\sup_{i}h_{e,i} \) where \( h_{e,i}(x,y,\mu) \) is a computable function of
\( e,i \) monotonically increasing in \( i \).
Similarly, we can represent \( s(y,\mu) \) as a supremum \( \sup_{j}s_{j}(y,\mu) \)
where \( s_{j}(y,\mu) \) is a computable function monotonically increasing in \( j \).
The integral \( \mu^{x} h_{e,i}(x,y,\mu) \) is computable as a function of \( (y,\mu) \),
in particular it is upper semicomputable.

\begin{sloppypar}
Define \( h'_{e,i,j}(x,y,\mu) \) as \( h_{e,i}(x,y,\mu) \) for all \( j,y,\mu \)
with \( \mu^{x} h_{e,i}(x,y,\mu)<s_{j}(y,\mu) \), and 0 otherwise.
Since \( s_{j}(y,\mu) \) is computable this 
definition makes the function \( h'_{e,i,j}(x,y,\mu) \) lower semicomputable.
The function \( h''_{e,i}(x,y,\mu)=\sup_{j}h'_{e,i,j}(x,y,\mu) \) is then also lower
semicomputable, with \( h''_{e,i}(x,y,\mu)\le h_{e,i}(x,y,\mu) \),
and \( \mu^{x} h'_{e,i}(x,y,\mu)\le s(y,\mu) \).
Also, \( h''_{e,i}(x,y,\mu) \) is monotonically increasing in \( i \).
The function \( \phi'_{e}(x,y,\mu)=\sup_{i}h''_{e,i}(x,y,\mu) \) is then also lower
semicomputable, and by Fubini's theorem we have
\( \mu^{x}\phi'_{e}(x,y,\mu)\le s(y,\mu) \).
\end{sloppypar}

\begin{sloppypar}
Define \( \phi_{e'}(x,y,\mu)=\phi'_{e}(x,y,\mu) \).
Consider any \( e,y,\mu \) such that \( \mu^{x}\phi_{e}(x,y,\mu)<s(y,\mu) \) holds.
Then for every \( i \) there is a \( j \) with \( \mu^{x}h_{e,i}(x,y,\mu)<s_{j}(y,\mu) \),
and hence \( h'_{e,i,j}(x,y,\mu)=h_{e,i}(x,y,\mu) \).
It follows that \( h''_{e,i}(x,y,\mu)=h_{e,i}(x,y,\mu) \) for all \( i \) and hence
\( \phi'_{e}(x,y,\mu)=\phi_{e}(x,y,\mu) \).
\end{sloppypar}
 \end{proof}

\begin{corollary}[Uniform extension]\label{coroll:trim}
There is an operation \( H_{e}(x,\mu)\mapsto H_{e'}(x,\mu) \) with the property
that \( H_{e'}(x,\mu) \) is a uniform test and if \( 2 H_{e}(\cdot,\mu) \) is a \( \mu \)-test
then \( H_{e'}(x,\mu)=H_{e}(x,\mu) \).
\end{corollary}
 \begin{proof}
In Theorem~\ref{thm:trim} set \( \phi_{e}(x,y,\mu)=\frac{1}{2}H_{e}(x,\mu) \) with
\( s(y,\mu)=1 \).
 \end{proof}

\begin{corollary}[Universal generalized test]\label{coroll:trim-univ}
Let \( s:\bY\times\cM(\bX)\to\ol\dR_{+} \) a lower semicomputable
function.
Let \( E \) be the set of 
lower semicomputable functions \( \phi(x,y,\mu)\ge 0 \) with
\( \mu^{x}\phi(x,y,\mu)\le s(y,\mu) \).
There is a function \( \psi\in E \) that is \emph{optimal} in the sense
that for all \( \phi\in E \) there is a constant \( c_{\phi} \) with 
\( \phi\le 2c_{\phi}\psi \).
\end{corollary}
 \begin{proof}
Apply the operation \( e\mapsto e' \) of theorem~\ref{thm:trim} to the sequence
\( \phi_{e}(x,y,\mu) \) (\( e=1,2,\dots \)) of all lower semicomputable
functions of \( x,y,\mu \).
The elements of the sequence \( \phi'_{e}(x,y,\mu) \), \( e=1,2,\dots \)  are in \( E \) and
and the sequence \( 2\phi'_{e}(x,y,\mu) \), \( e=1,2,\dots \) contains all
elements of \( E \).
Hence 
the function \( \psi(x,y,\mu)=\sum_{e=1}^{\infty}2^{-e} \phi'_{e}(x,y,\mu) \) is in
\( E \) and has the optimality property.
 \end{proof}

 \begin{definition}
A uniform test \( u \) is called \df{universal} if for every other test \( t \) there
is a constant \( c_{t}>0 \) such that for all \( x,\mu \) we have \( t_{\mu}(x) \le c u_{\mu}(x) \). 
 \end{definition}

\begin{theorem}[Universal test,\cite{HoyrupRojasRandomness09}]
There is a universal uniform test.
\end{theorem}
 \begin{proof}
This is a special case of Corollary~\ref{coroll:trim-univ} with \( s(y,\mu)=1 \).
 \end{proof}

\begin{definition}
Let us fix a universal uniform test, called \( \t_{\mu}(x) \).
An element \( x\in X \) is called \df{random} with respect to measure
\( \mu\in\cM(\bX) \) if \( \t_{\mu}(x)<\infty \).

The \df{deficiency of randomness} is defined as \( \d_{\mu}(x) = \log\t_{\mu}(x) \).
\end{definition}

If the space is discrete then typically all elements are random with respect to
\( \mu \), but they will still be distinguished according to their different values
of \( \d_{\mu}(x) \).

\subsection{Sequences}

Let our computable metric space \( \bX=(X,d,D,\ag) \) be the Cantor space
of Example~\ref{example:constr-metric}.\ref{i:x.constr-metric.Cantor}:
the set of sequences over a (finite or countable) alphabet \( \Sg^{\dN} \).
We may want to measure the non-randomness of finite sequences, viewing them
as initial segments of infinite sequences.
Take the universal test \( \t_{\mu}(x) \).
For this, it is helpful to apply the representation of
Proposition~\ref{propo:semicomp-rep.Cantor}, taking into account that adding
the extra parameter \( \mu \) does not change the validity of the theorem:

 \begin{proposition}\label{propo:finite-test}
There is a function
\( g:\cM(\bX)\times\Sg^{*}\to\ol\dR_{+} \) 
with \( \t_{\mu}(\xi) = \sup_{n} g_{\mu}(\xi^{\le n}) \), and 
with the following properties:
 \begin{alphenum}
  \item \( g \) is lower semicomputable.
  \item \( v\prefix w \) implies \( g_{\mu}(v)\le g_{\mu}(w) \).
  \item\label{i:extended.sum} For all integer \( n\ge 0 \) we have 
\( \sum_{w\in\Sg^{n}}\mu(w)g_{\mu}(w) \le 1 \).
 \end{alphenum}
 \end{proposition}
The properties of the function \( g_{\mu}(w) \) clearly imply that 
\( \sup_{n} g_{\mu}(\xi^{\le n}) \) is a uniform test.

The existence of a universal function among the functions \( g \) can be proved by
the usual methods:

 \begin{proposition}\label{propo:finite-univ}
Among the functions \( g_{\mu}(w) \) satisfying the properties listed in
Proposition~\ref{propo:finite-test}, there is one that
dominates to within a multiplicative constant.
 \end{proposition}

These facts motivate the following definition.

 \begin{definition}[Extended test]\label{def:extended-test}
Over the Cantor space, we
extend the definition of a universal test \( \t_{\mu}(x) \) to finite sequences
as follows.
We fix a function \( \t_{\mu}(w) \) with \( w\in\Sg^{*} \)
whose existence is assured by~Proposition~\ref{propo:finite-univ}.
For infinite sequences \( \xi \) we define 
\( \t_{\mu}(\xi) = \sup_{n}\t_{\mu}(\xi^{\le n}) \).
The test with values defined also on finite sequences will be called an
\df{extended test}.
 \end{definition}

We could try to define extended tests also over arbitrary constructive
metric spaces, extending them
to the canonical balls, with the monotonicity property that
\( \t_{\mu}(v)\le\t_{\mu}(w) \) if ball \( w \) is manifestly included in ball \( v \).
But there is nothing simple and obvious corresponding to the 
integral requirement~\eqref{i:extended.sum}.

Over the space \( \Sg^{\dN} \) for a finite alphabet \( \Sg \),
an extended test could also be extracted directly from a test, 
using the following formula (as observed by Vyugin and Shen).

\begin{definition}\label{def:ol-u}
 \begin{align*}
   \ol\u(z) = \inf\setof{\u(\og) : \og\text{ is an infinite extension of } z}.
 \end{align*}
\end{definition}
This function is lower semicomputable, by Proposition~\ref{propo:lsc-min}.

% %Added 2008/01/24
% \subsection{Subspaces}

% Let us be given a computable complete metric space \( \bX = (X, d, D, \ag) \).
% Let \( G\sbsq X \) be a constructive open subset.
% Then the subspace \( \bG \) was defined in
% Definition~\ref{def:constr-metric-subspace}.
% Of course, there is a natural embedding of \( \cM(\bG) \) into \( \cM(\bX) \),
% so probability
% measures over \( \cM(\bG) \) can be also considered measures over \( \bX \).
% But there is also a natural operation going in the other direction.
% For measure \( \mu\in\cM(\bX) \) with \( \mu(G)>0 \) we can define
% the \df{conditional probability}
% \( \mu(A\mvert G)=\mu(A\cap G)/\mu(G) \).

\subsection{Conservation of randomness}

For \( i=1,0 \), let
\( \bX_{i} = (X_{i}, d_{i}, D_{i}, \ag_{i}) \) be computable metric spaces, 
and let \( \bM_{i} = (\cM(\bX_{i}), \sg_{i}, \nu_{i}) \) 
be the effective topological space of
probability measures over \( \bX_{i} \).
Let \( \Lg \) be a computable probability kernel from \( \bX_{1} \) to \( \bX_{0} \) as
defined in Subsection~\ref{subsec:random-trans}.
In the following theorem, the same notation \( \d_{\mu}(x) \) will
refer to the deficiency of randomness with respect to two different spaces,
\( \bX_{1} \) and \( \bX_{0} \), but this should not cause confusion.
Let us first spell out the conservation theorem before interpreting it.

 \begin{theorem}\label{propo:conservation}
For a computable probability kernel \( \Lg \) from \( \bX_{1} \) to \( \bX_{0} \),
we have 
 \begin{equation}\label{eq:conservation}
  \lg_{x}^{y} \t_{\Lg^{*}\mu}(y) \lem \t_{\mu}(x).
 \end{equation}
 \end{theorem}
 \begin{proof}
Let \( \t_{\nu}(x) \) be the universal test over \( \bX_{0} \).
The left-hand side of~\eqref{eq:conservation} can be written as
 \[
  u_{\mu} = \Lg \t_{\Lg^{*}\mu}.
 \]
According to~\eqref{eq:Lg-Lg*}, we have 
\( \mu u_{\mu} = (\Lg^{*}\mu) \t_{\Lg^{*}\mu} \) which is \( \le 1 \) since \( \t \) is
a test.
If we show that \( (\mu,x) \mapsto u_{\mu}(x) \) is lower semicomputable
then the universality of \( \t_{\mu} \) will imply
that \( u_{\mu} \lem \t_{\mu} \).

According to Proposition~\ref{propo:semicomp-as-limit}, as
a lower semicomputable function, \( \t_{\nu}(y) \) can be written as 
\( \sup_{n} g_{n}(\nu, y) \), where \( (g_{n}(\nu, y)) \) is a computable sequence
of computable functions.
We pointed out in Subsection~\ref{subsec:random-trans} that the function 
\( \mu \mapsto \Lg^{*}\mu \) is computable.
Therefore the function \( (n, \mu, x) \mapsto g_{n}(\Lg^{*}\mu, f(x)) \) 
is also a computable.
So, \( u_{\mu}(x) \) is the supremum of a computable sequence of computable
functions and as such, lower semicomputable.
 \end{proof}

It is easier to interpret the theorem first in the special case when
\( \Lg = \Lg_{h} \) for a computable function \( h : X_{1} \to X_{0} \),
as in Example~\ref{example:computable-determ-trans}.
Then the theorem simplifies to the following.

\begin{sloppypar}
 \begin{corollary}\label{coroll:conservation-determ}
For a computable function \( h : X_{1} \to X_{0} \), we have
\( \d_{h^{*}\mu}(h(x)) \lea \d_{\mu}(x) \).
 \end{corollary}
  \end{sloppypar}

Informally, this says that if \( x \) is random with respect to \( \mu \) in
\( \bX_{1} \) then \( h(x) \) is essentially at least as random with respect to the
output distribution \( h^{*}\mu \) in \( \bX_{0} \).
Decrease in randomness
can only be caused by complexity in the definition of the function \( h \).

%Added 2008/01/23
Let us specialize the theorem even more:

\begin{corollary}\label{coroll:conservation-pair}
For a probability distribution \( \mu \) over the space \( X\times Y \)
let \( \mu_{X} \) be its marginal on the space \( X \).
Then we have
 \begin{align*}
   \d_{\mu_{X}}(x) \lea \d_{\mu}(\tup{x,y}).
 \end{align*}
\end{corollary}

This says, informally, that if a pair is random then each of its elements is
random (with respect to the corresponding marginal distribution).

In the general case of the theorem, concerning random transitions,
we cannot bound the randomness of each outcome uniformly.
The theorem asserts that the average nonrandomness, as measured by 
the universal test with respect to the output distribution, does not
increase.
In logarithmic notation:
\( \lg_{x}^{y} 2^{\d_{\Lg^{*}\mu}(y)} \lea \d_{\mu}(x) \), 
or equivalently, 
\( \int 2^{\d_{\Lg^{*}\mu}(y)} \lg_{x}(dy) \lea \d_{\mu}(x) \).

 \begin{corollary}
Let \( \Lg \) be a computable probability kernel from \( \bX_{1} \) to \( \bX_{0} \).
There is a constant \( c \) such that
for every \( x\in\bX^{1} \), and integer \( m > 0 \) we have 
 \[
 \lg_{x}\setof{y : \d_{\Lg^{*}\mu}(y) > \d_{\mu}(x) + m + c} \le 2^{-m}.
 \]
 \end{corollary}

Thus, in a computable random transition,
the probability of an increase of randomness deficiency by
\( m \) units (plus a constant \( c \)) is less than \( 2^{-m} \).
The constant \( c \) comes from the description complexity 
of the transition \( \Lg \).

A randomness conservation result related to 
Corollary~\ref{coroll:conservation-determ} was proved 
in~\cite{HertlingWeihrauchRand98}.
There, the measure over the space \( \bX_{0} \) is not the output
measure of the transformation, but is assumed to
obey certain inequalities related to the transformation.

\section{Test for a class of measures}

\subsection{From a uniform test}

A Bernoulli measure is what we get by tossing a (possibly biased) coin
repeatedly.

\begin{definition}
Let \( X=\dB^{\dN} \) be the set of infinite binary sequences, with the usual
sequence topology.
Let \( B_{p} \) be the measure on \( X \) that corresponds to tossing a coin
independently with probability \( p \) of success: it is called the \df{Bernoulli}
measure with parameter \( p \).
Let \( \cB \) denote the set of all Bernoulli measures. 
\end{definition}

Given a sequence \( x\in X \) we may ask the question whether \( x \) is random with
respect to at least \emph{some} Bernoulli measure.
(It can clearly not be random with respect to two different Bernoulli measures
since if \( x \) is random with respect to \( B_{p} \) then its relative frequencies
converge to \( p \).)
This idea suggests two possible definitions for a test of the property of
``Bernoulliness'': 
 \begin{enumerate}
  \item We could define \( \t_{\cB}(x) = \inf_{\mu\in\cB}(x) \).
  \item We could define the notion of a Bernoulli test as a lower semicomputable
    function \( f(x) \) with the property \( B_{p} f\le 1 \) for all \( p \).
 \end{enumerate}
We will see that in case of this class of measures the two definitions lead to
essentially the same test.

Let us first extend the definition to more general sets of measures, still
having a convenient property.
 
 \begin{definition}[Class tests]\label{def:C-test}
Consider a class \( \cC \) of measures that is effectively
compact in the sense of Definition~\ref{def:eff-compact-general} or (equivalently
for metric spaces) in the sense of Theorem~\ref{thm:eff-compact-metric}.
A lower semicomputable function \( f(x) \) is called a \( \cC \)-\df{test} if for all
\( \mu\in\cC \) we have \( \mu f\le 1 \).
It is a \df{universal} \( \cC \)-\df{test} if it dominates all other \( \cC \)-tests to
within a multiplicative constant.
\end{definition}

\begin{example}
It is easy to show that the class \( \cB \) is effectively compact.
One way is to appeal to the general theorem in
Proposition~\ref{propo:image-of-eff-compact} saying that applying a computable
function to an effectively compact set (in this case the interval
\( \clint{0}{1} \)), the image is also an effectively compact set.
\end{example}

For the case of infinite sequences, we can also define extended tests.

 \begin{definition}[Extended class test]
Let our space \( \bX \) be the Cantor space of infinite sequences \( \Sg^{\dN} \).
Consider a class \( \cC \) of measures that is effectively
compact in the sense of Definition~\ref{def:eff-compact-general} or (equivalently
for metric spaces) in the sense of Theorem~\ref{thm:eff-compact-metric}.
A lower semicomputable function \( f:\Sg^{*}\to\ol\dR_{+} \) 
is called an \df{extended} \( \cC \)-\df{test} if it is monotonic with respect to
the prefix relation and for all \( \mu\in\cC \) and integer \( n\ge 0 \) we have 
 \begin{align*}
  \sum_{x\in\Sg^{n}}\mu(x) f(x) \le 1.   
 \end{align*}
It is \df{universal} if it dominates all other extended \( \cC \)-tests to
within a multiplicative constant.
\end{definition}

The following observation is immediate.

 \begin{proposition}
A function \( f:\Sg^{\dN}\to\ol\dR_{+} \) is a class test if and only if it can be
represented as \( \lim_{n}g(\xi^{\le n}) \) where \( g(x) \) is an extended class test.
 \end{proposition}

The following theorem defines a universal \( \cC \)-test.

\begin{sloppypar}
 \begin{theorem}
Let \( \t_{\mu}(x) \) be a universal uniform test.
Then \( \u(x) = \inf_{\mu\in\cC}\t_{\mu}(x) \) defines a universal \( \cC \)-test.
 \end{theorem}
 \end{sloppypar}
 \begin{proof}
Let us show first that \( \u(x) \) is a \( \cC \)-test.
It is lower semicomputable according to Proposition~\ref{propo:lsc-min}.
Also, for each \( \mu \) we have \( \mu\u \le \mu \t_{\mu} \le 1 \), showing that
\( \u(x) \) is a \( \cC \)-test.

Let us now show that it is universal.
Let \( f(x) \) be an arbitrary \( \cC \)-test.
By Corollary~\ref{coroll:trim} there is a uniform test \( g_{\mu}(x) \) such that for
all \( \mu\in\cC \) we have \( g_{\mu}(x)=f(x)/2 \).
It follows from the universality of the uniform test \( \t_{\mu}(x) \)
that \( f(x) \lem g_{\mu}(x) \lem\t_{\mu}(x) \) for all \( \mu\in\cC \).
But then \( f(x) \lem\inf_{\mu\in\cC}\t_{\mu}(x) = \u(x) \).
 \end{proof}

For the case of sequences, the same statement can be made for extended tests.
(This is not completely automatic since a test is obtained from an
extended test via a supremum, on the other hand a class test is obtained,
according the theorem above, via an infimum.)

\subsection{Typicality and class tests}

The set of Bernoulli measures has an important property shared by many
classes considered in practice: namely that random sequences determine the
measure to which it belongs.

 \begin{definition}
Consider a class \( \cC \) of measures over a computable metric space
\( \bX = (X,d,D,\ag) \).
We will say that a lower semicomputable function
 \begin{align*}
  \s:X\times\cC\to\ol\dR_{+}
 \end{align*}
is a \df{separating test} for \( \cC \) if 
\begin{itemize}
\item \( \s_{\mu}(\cdot) \) is a test for each \( \mu\in\cC \).
\item If \( \mu\ne\nu \) then
\( \s_{\mu}(x)\lor\s_{\nu}(x)=\infty \) for all \( x\in X \).
\end{itemize}
Given a separating test \( \s_{\mu}(x) \) we call an element \( x \)
\df{typical} for \( \mu\in\cC \) if \( \s_{\mu}(x)<\infty \).
 \end{definition}

A typical element determines uniquely the measure \( \mu \) for which
it is typical.
Note that if a separating tests exists for a class then any two different
measures \( \mu_{1},\mu_{2} \) in the class are orthogonal to each other, that is there are
disjoint measureable sets \( A_{1},A_{2} \) with \( \mu_{j}(A_{j})=1 \).
Indeed, let \( A_{j}=\setof{x:s_{\mu_{j}}(x)<\infty} \).

Let us show a nontrivial example: the class of \( \cB \) of Bernoulli measures.
Recall that by Chebyshev's inequality we have
 \begin{align*}
   B_{p}(\setof{x\in\dB^{n}: |\sum_{i}x(i)-n p|\ge\lg n^{1/2}(p(1-p))^{1/2}}) \le \lg^{-2}.
 \end{align*}
Since \( p(1-p)\le 1/4 \), this implies
 \begin{align*}
   B_{p}(\setof{x\in\dB^{n}: |\sum_{i}x(i)-n p>\lg n^{1/2}/2}) < \lg^{-2}.
 \end{align*}
Setting \( \lg=n^{0.1} \) and ignoring the factor \( 1/2 \) gives
 \begin{align*}
   B_{p}(\setof{x\in\dB^{n}: |\sum_{i}x(i)-n p|> n^{0.6}})< n^{-0.2}.
 \end{align*}
Setting \( n=2^{k} \):
 \begin{align}\label{eq:2^k-Cheb}
   B_{p}(\setof{x\in\dB^{2^{k}}: |\sum_{i}x(i)-2^{k}p|> 2^{0.6 k}})< 2^{-0.2k}.
 \end{align}

Now, for the example.

\begin{example}\label{example:Bernoulli-sep}
For a sequence \( \xi \) in \( \bB^{\dN} \), and for \( p\in\clint{0}{1} \) let
 \begin{align*}
   g_{p}(x) = g_{B_{p}}(x) = \sup\setof{k: |\sum_{i=1}^{2^{k}}\xi(i)-2^{k}p|> 2^{0.6k}}.
 \end{align*}
Then we have
 \begin{align*}
   B_{p}^{\xi}g_{p}(\xi) \le \sum_{k} k\cdot 2^{-0.2 k} = c<\infty
 \end{align*}
for some constant \( c \), so \( \s_{p}(\xi)=g_{p}(x)/c \) is a test for each \( p \).
The property \( \s_{p}(\xi)<\infty \) implies that \( 2^{-k}\sum_{i=1}^{2^{k}}\xi(i) \)
converges to \( p \).
For a given \( \xi \) this is impossible for both \( p \) and \( q \) for \( p\ne q \), hence
\( \s_{p}(\xi)\lor\s_{q}(\xi)=\infty \).
\end{example}

The following structure theorem gives considerable insight.

\begin{theorem}\label{thm:test-sep}
Let \( \cC \) be an effectively compact class of measures,
let \( \t_{\mu}(x) \) be the universal uniform test and
let \( \t_{\cC}(x) \) be a universal class test for \( \cC \).
Assume that a separating test \( \s_{\mu}(x) \) exists for \( \cC \).
Then we have the representation
 \begin{align*}
  \t_{\mu}(x) \eqm \t_{\cC}(x)\lor \s_{\mu}(x)
 \end{align*}
for all \( \mu\in\cC \), \( x\in X \).
\end{theorem}
 \begin{proof}
First, we have \( \t_{\cC}(x)\lor\s_{\mu}(x) \lem \t_{\mu}(x) \).
Indeed as we know from the Uniform Extension Corollary~\ref{coroll:trim}, we can
extend \( \s_{\mu}(x)/2 \) to a uniform test, hence \( \s_{\mu}(x)\lem \t_{\mu}(x) \).
Also by definition \( \t_{\cC}(x)\le \t_{\mu}(x) \).

On the other hand, let us show \( \t_{\cC}(x)\lor\s_{\mu}(x) \ge \t_{\mu}(x) \).
Suppose first that 
\( x \) is not random with respect to any \( \nu\in\cC \): then \( \t_{\cC}(x)=\infty \).
Suppose now that \( x \) is random with respect to some \( \nu\in\cC \), 
\( \nu\ne\mu \).
Then \( \s_{\mu}(x)=\infty \).
Finally, suppose \( \t_{\mu}(x)<\infty \).
Then \( \t_{\nu}(x)=\infty \) for all \( \nu\in\cC \), \( \nu\ne\mu \), hence
\( \t_{\cC}(x) = \inf_{\nu\in\cC}\t_{\nu}(x) = \t_{\mu}(x) \), so
the inequality holds again.
 \end{proof}

The above theorem separates the randomness test into two parts.
One part tests randomness with respect to the class \( \cC \), the other one tests
typicality with respect to the measure \( \mu \).
In the Bernoulli example,
 \begin{itemize}
  \item Part \( \t_{\cB}(\xi) \) checks ``Bernoulliness'', that is independence.
It encompasses all the irregularity criteria.
  \item Part \( \s_{p}(\xi) \) checks
(crudely) for the law of large numbers: whether relative frequency converges
(fast) to \( p \).
 \end{itemize}
If the independence of the sequence is taken for granted, we may assume that
the class test is satisfied.
What remains is typicality testing, which is similar to 
ordinary statistical parameter testing.

 \begin{remarks}
   \begin{enumerate}
   \item 
Separation is the only requirement of the test \( \s_{\mu}(x) \),
otherwise, for example in the Bernoulli test case,
no matter how crude the convergence criterion expressed by
\( \s_{\mu}(x) \), the maximum \( \t_{\cC}(x)\lor\s_{\mu}(x) \) 
is always (essentially) the same universal test.
  \item
Though the convergence criterion can be crude, but one still seems to need some
kind of constructive convergence of the relative frequencies if the separation
test is to be defined in terms of relative frequency convergence.
   \end{enumerate}
 \end{remarks}

 \begin{example}
For \( \eps>0 \) let \( P(\eps) \) be the Markov chain \( X_{1},X_{2},\dots \)
with set of states \( \{0,1\} \),
with transition probabilities \( T(0,1)=T(1,0)=\eps \) and \( T(0,0)=T(1,1)=1-\eps \),
and with \( P\evof{X_{1}=0}=P\evof{X_{1}=1}=1/2 \).
Let \( \cC(\dg) \) be the class of all \( P(\eps) \) with \( \eps\ge\dg \).
For each \( \dg>0 \), this is an effectively compact class, and a separating test is easy to
construct since an effective law of large numbers holds for these Markov
chains. 

We can generalize to the set of \( m \)-state stationary Markov chains whose 
eigenvalue gap is \( \ge\eps \).
 \end{example}

This example is in contrast to V'yugin's example~\cite{VyuginErgodic98}
showing that, in the nonergodic case,
in general no recursive speed of convergence can be guaranteed in the
Ergodic Theorem (which is the appropriate generalization of the law of large
numbers)\footnote{This paper of Vyugin also shows that though there is no recursive speed of
convergence, a certain constructive 
proof of the pointwise ergodic theorem still gives rise to a test of
randomness.}.

We can show that if a computable Markov chain is ergodic then
its law of large numbers \emph{does} have a constructive speed of convergence.
Hopefully, this observation can be extended to some interesting 
compact classes of ergodic processes.

\subsection{Martin-L\"of's approach}

Martin-L\"of also gave a definition of Bernoulli tests in~\cite{MLof66art}.
For its definition let us introduce the following notation.

 \begin{notation} The set of sequences with a given frequency of 1's will be
   denoted as follows:
 \begin{align*}
   \dB(n,k) = \setof{x\in\dB^{n}: \sum_{i}x(i)=k}.
 \end{align*}
 \end{notation}

Martin-L\"of's definition translates to the integral constraint version as follows:

 \begin{definition}\label{def:combinat-Bernoulli}
Let \( X=\dB^{\dN} \) be the set of infinite binary sequences with the usual
metrics.
A \df{combinatorial Bernoulli test} is a function
\( f:\dB^{*}\to\ol\dR_{+} \) with the following constraints:
 \begin{alphenum}
  \item It is lower semicomputable.
  \item It is monotonic with respect to the prefix relation.
  \item For all \( 0\le k\le n \) we have
 \begin{align}\label{eq:combinat-Bernoulli-sum}
      \sum_{x\in\dB(n,k)} f(x) \le \binom{n}{k}.
 \end{align}
 \end{alphenum}
 \end{definition}

The following observation is useful.

 \begin{proposition}\label{propo:Bernoulli-extend}
If a combinatorial Bernoulli test \( f(x) \) is given on strings \( x \)
of length less than \( n \), then extending it to longer strings using 
monotonicity we get a function that is still a combinatorial Bernoulli test.
 \end{proposition}
 \begin{proof}
It is sufficient to check the relation~\eqref{eq:combinat-Bernoulli-sum}.
We have (even for \( k=0 \), when \( \dB(n-1,k-1)=\emptyset \)):
 \begin{align*}
  \sum_{x\in\dB(n,k)} f(x) &\le
  \sum_{y\in\dB(n-1,k-1)}f(y)+\sum_{y\in\dB(n-1,k)} f(y) 
\\ &\le \binom{n-1}{k-1} + \binom{n-1}{k} = \binom{n}{k}.
 \end{align*}
 \end{proof}

The following can be shown using standard methods:
  
 \begin{proposition}[Universal combinatorial Bernoulli test]
There is a universal combinatorial Bernoulli test \( f(x) \), that is a combinatorial
Bernoulli test with the property that for every
combinatorial Bernoulli test \( h(x) \) there is a constant \( c_{h}>0 \) such that for
all \( x \) we have \( h(x) \le c_{h} g(x) \).
 \end{proposition}

\begin{definition}
Let us fix a universal combinatorial Bernoulli test \( b(x) \) and extend it to
infinite sequences \( \xi \) by
 \begin{align*}
   b(\xi) = \sup_{n} b(\xi^{\le n}).
 \end{align*}
Let \( \t_{\cB}(\xi) \) be a universal class test for Bernoulli measures, for
\( \xi\in\dB^{\dN} \).
\end{definition}

Let us show that the general class test for Bernoulli measures and
Martin-L\"of's Bernoulli test yield the same random sequences.

 \begin{theorem}\label{thm:combinat-Bernoulli}
With the above definitions, we have \( b(\xi) \eqm \t_{\cB}(\xi) \).
In words: a sequence is nonrandom with respect to all Bernoulli measures if and
only if it is rejected by a universal combinatorial Bernoulli test.
 \end{theorem}
 \begin{proof}
We first show \( b(\xi)\le \t_{\cB}(\xi) \).
Moreover, we show that \( b(x) \) is an extended
class test for the class of Bernoulli measures.
We only need to check the sum condition, namely that for all
\( 0\le p \le 1 \), and all \( n>0 \) the inequality
\( \sum_{x\in\dB^{n}}B_{p}(x)b(x) \le 1 \) holds.
Indeed, we have
 \begin{align*}
 \sum_{x\in\dB^{n}}B_{p}(x)f(x)
     &= \sum_{k=0}^{n}p^{k}(1-p)^{n-k}\sum_{x\in\dB(n,k)}f(x)
\\  &\le \sum_{k=0}^{n}p^{k}(1-p)^{n-k}\binom{n}{k} = 1.
 \end{align*}
On the other hand, let \( f(x)=\t_{\cB}(x) \), \( x\in\dB^{*} \)
be the extended test for \( \t_{\cB}(\xi) \).
For all integers \( N>0 \) let \( n = \flo{\sqrt{N/2}} \).
Then as \( N \) runs through the integers, \( n \) also runs through all integers.
For \( x\in\dB^{N} \) let \( F(x)=f(x^{\le n}) \).
Since \( f(x) \) is lower semicomputable and monotonic with respect to the prefix
relation, this is also true of \( F(x) \).

We need to estimate \( \sum_{x\in\dB(N,K)} F(x) \).
For this, note that for \( y\in\dB(n,k) \) we have
 \begin{align}\label{eq:hyper}
   \card{\setof{x\in\dB(N,K):y\prefix x}} = \binom{N-n}{K-k}.
 \end{align}
Now for \( 0\le K\le N \) we have
 \begin{equation}\label{eq:Bernoulli-estim}
 \begin{aligned}
      \sum_{x\in\dB(N,K)} F(x) &= 
     \sum_{y\in\dB^{n}}f(y)\card{\setof{x\in\dB(N,K):y\prefix x}}
\\ &= \sum_{k=0}^{n}\binom{N-n}{K-k}\sum_{y\in\dB(n,k)}f(y).
 \end{aligned}%
 \end{equation}
Let us estimate \( \binom{N-n}{K-k}\big/\binom{N}{K} \).
If \( K=k=0 \) then this is 1.
If \( k=n \) then it is \( \frac{K\dotsm(K-n+1)}{N\dotsm(N-n+1)} \).
Otherwise, using \( p=K/N \):
\begin{align}\nonumber
\frac{\binom{N-n}{K-k}}{\binom{N}{K}}
 &= \frac{(N-n)(N-n-1)\dotsm(N-K-(n-k)+1)/(K-k)!}{N(N-1)\dotsm(N-K+1)/K!}
\\\nonumber  &= \frac{K\dotsm(K-k+1)\cdot(N-K)\dotsm(N-K-(n-k)+1)}{N\dotsm(N-n+1)}.
\\\label{eq:Bernoulli-estim-1}
  &\le \frac{K^{k}(N-K)^{n-k}}{(N-n)^{n}} = p^{k}(1-p)^{n-k}\Paren{\frac{N}{N-n}}^{n}.
 \end{align}
Thus in all cases the estimate
 \begin{align*}
\binom{N-n}{K-n}\bigg/\binom{N}{K} \le p^{k}(1-p)^{n-k}\Paren{\frac{N}{N-n}}^{n}
 \end{align*}
holds.
We have
 \begin{align*}
\Paren{\frac{N}{N-n}}^{n} =\Paren{1+\frac{n}{N-n}}^{n}
  \le e^{\frac{N}{2(N-n)}} \le e^{2},
 \end{align*}
since we assumed \( 2n^{2}\le N \).
Substituting into~\eqref{eq:Bernoulli-estim} gives
 \begin{align*}
  \sum_{x\in\dB(N,K)} F(x)
  &\le e^{2}\sum_{k=0}^{n}p^{k}(1-p)^{n-k}\sum_{y\in\dB(n,k)}f(y) \le e^{2},
 \end{align*}
since \( f(x) \) is an extended class test for Bernoulli measures.
It follows that \( e^{-2}F(x)\lem b(x) \), hence also \( F(x) \lem b(x) \).
But we have \( \t_{\cC}(\xi)=\sup_{n}f(\xi^{\le n}) = \sup_{n}F(\xi^{\le n}) \), hence
\( \t_{\cC}(\xi) \lem b(\xi) \).
 \end{proof}

\section{Neutral measure}\label{sec:neutral}

Let \( \t_{\mu}(x) \) be our universal uniform randomness test.
We call a measure \( M \) \df{neutral} if
\( \t_{M}(x) \le 1 \) for all \( x \).
If \( M \) is neutral then no experimental outcome \( x \) could 
refute the theory (hypothesis, model)
that \( M \) is the underlying measure to our experiments.
It can be used as ``apriori probability'', in a Bayesian approach to
statistics.
Levin's theorem says the following:

 \begin{theorem}\label{thm:neutral-meas}
If the space \( \bX \) is compact then there is a neutral measure over \( \bX \).
 \end{theorem}
 The proof relies on a nontrivial combinatorial fact,
Sperner's Lemma, which also underlies the proof of the Brouwer fixpoint
theorem.
Here is a version of Sperner's Lemma, spelled out in continuous form:

 \begin{proposition}\label{propo:Sperner}
Let \( p_{1},\dots,p_{k} \) be points of some finite-dimensional space
\( \dR^{n} \).
Suppose that there are closed sets \( F_{1},\dots,F_{k} \)
with the property that for every subset \( 1 \le i_{1} < \dots < i_{j} \le k \)
of the indices, the simplex \( S(p_{i_{1}},\dots,p_{i_{j}}) \) spanned by
\( p_{i_{1}},\dots,p_{i_{j}} \) is covered by
the union \( F_{i_{1}} \cup \dots \cup F_{i_{j}} \).
Then the intersection \( \bigcap_{i} F_{i} \) of all these sets is not empty.
 \end{proposition}

The following lemma will also be needed.

 \begin{lemma}\label{lem:concentr}
For every closed set \( A \sbs \bX \) and measure \( \mu \), if \( \mu(A)=1 \) then
there is a point \( x\in A \) with \( \t_{\mu}(x) \le 1 \).
 \end{lemma}
 \begin{proof}
This follows easily from 
\( \mu\, t_{\mu} = \mu^{x} 1_{A}(x)t_{\mu}(x) \le 1 \).
 \end{proof} 

 \begin{proof}[Proof of Theorem~\protect\ref{thm:neutral-meas}]
For every point \( x \in \bX \), let \( F_{x} \) be the set of measures for which
\( \t_{\mu}(x) \le 1 \).
If we show that for every finite set of points \( x_{1},\dots,x_{k} \), we
have 
 \begin{equation}\label{eq:finite-inters}
  F_{x_{1}}\cap\dots\cap F_{x_{k}} \ne \emptyset,
 \end{equation}
then we will be done.
Indeed, according to Proposition~\ref{propo:measures-compact}, the compactness
of \( \bX \) implies the compactness of the space \( \bM(\bX) \) of measures.
Therefore if every finite subset of the family \( \setof{F_{x} : x \in \bX} \)
of closed sets has a nonempty intersection, then the whole family has a
nonempty intersection: this intersection consists of the neutral
measures.

To show~\eqref{eq:finite-inters}, let \( S(x_{1},\dots,x_{k}) \) be the set of
probability measures concentrated on \( x_{1},\dots,x_{k} \).
Lemma~\ref{lem:concentr} implies that each such measure belongs to one of
the sets \( F_{x_{i}} \).
Hence \( S(x_{1},\dots,x_{k}) \sbs F_{x_{1}} \cup \dots \cup F_{x_{k}} \),
and the same holds for every subset of the indices \( \{1,\dots,k\} \).
Sperner's Lemma~\ref{propo:Sperner} implies 
\( F_{x_{1}} \cap \dots \cap F_{x_{k}} \ne\emptyset \).
 \end{proof}

When the space is not compact, there are generally no neutral probability
measures, as shown by the following example.

 \begin{proposition}\label{thm:no-neutral}
Over the discrete space \( \bX = \dN \) of natural numbers,
there is no neutral measure.
 \end{proposition}
 \begin{proof}
It is sufficient to 
construct a randomness test \( t_{\mu}(x) \) with the property that for
every measure \( \mu \), we have \( \sup_{x} t_{\mu}(x) = \infty \).
Let 
 \begin{equation}\label{eq:no-neutral}
   t_{\mu}(x) = \sup\setof{ k \in \dN: \sum_{y<x}\mu(y) > 1-2^{-k}}.
 \end{equation}
By its construction, this is a lower semicomputable function with
\( \sup_{x} t_{\mu}(x) = \infty \).
It is a test if \( \sum_{x}\mu(x)t_{\mu}(x) \le 1 \).
We have
 \[
  \sum_{x} \mu(x) t_{\mu}(x) 
= \sum_{k>0} \sum_{t_{\mu}(x) \ge k} \mu(x)
< \sum_{k>0} 2^{-k} \le 1.
 \]
 \end{proof}

Using a similar construction over the space \( \dN^{\dN} \) of infinite
sequences of natural numbers, we could show that 
for every measure \( \mu \) there is a sequence \( x \) with \( \t_{\mu}(x)=\infty \).

Proposition~\ref{thm:no-neutral} is a little misleading, since \( \dN \) can be
compactified into \( \ol\dN = \dN \cup \{\infty\} \)
(as in Part~\ref{i:compact.compactify} of Example~\ref{example:compact}).
Theorem~\ref{thm:neutral-meas} implies that there is a neutral probability
measure \( M \) over the compactified space \( \ol\dN \).
Its restriction to \( \dN \) is, of course, not a probability measure, since it
satisfies only \( \sum_{x < \infty} M(x) \le 1 \).
We called these functions \df{semimeasures}.

 \begin{remark}\label{rem:compactify}\
 \begin{enumerate}

   \item
 It is easy to see that 
Theorem~\ref{thm:test-charac-discr} characterizing randomness in terms of
complexity holds also for the space \( \ol\dN \).

   \item
The topological space of semimeasures
over \( \dN \) is not compact, and there is no neutral one among them.
Its topology is not the same as what we get when we restrict the topology
of probability measures over \( \ol\dN \) to \( \dN \).
The difference is that over \( \dN \), for example the set of measures
\( \setof{\mu : \mu(\dN) > 1/2} \) is closed, since \( \dN \) (as the whole space)
is a closed set.
But over \( \ol\dN \), this set is not necessarily closed, since
\( \dN \) is not a closed subset of \( \ol\dN \).
 \end{enumerate}
 \end{remark}

Neutral measures are not too simple, even over \( \ol\dN \), as the following
theorem shows.

 \begin{theorem}\label{thm:no-upper-semi-neutral}
There is no neutral measure over \( \ol\dN \) that is upper
semicomputable over \( \dN \) or lower semicomputable over \( \dN \).
 \end{theorem}
 \begin{proof}
Let us assume that \( \nu \) is a measure that is upper semicomputable over
\( \dN \).
Then the set
 \[
   \setof{(x,r) : x \in\dN,\; r\in\dQ,\; \nu(x) < r}
 \]
is recursively enumerable: let \( (x_{i},r_{i}) \) be a particular
enumeration. 
For each \( n \), let \( i(n) \) be the first \( i \) with \( r_{i} < 2^{-n} \), and let
\( y_{n} = x_{i(n)} \).
Then \( \nu(y_{n}) < 2^{-n} \), and at the same time \( K(y_{n}) \lea K(n) \).
As mentioned, in Remark~\ref{rem:compactify}, 
Theorem~\ref{thm:test-charac-discr} characterizing randomness in terms of
complexity holds also for the space \( \ol\dN \).
Thus, 
 \[
  \d_{\nu}(y_{n}) \eqa -\log\nu(y_{n}) - K(y_{n} \mvert \nu) \gea n - K(n).
 \]
Suppose now that \( \nu \) is lower semicomputable over \( \dN \).
The proof for this case is longer.
We know that \( \nu \) is the monotonic limit of a recursive sequence
\( i\mapsto \nu_{i}(x) \) of recursive semimeasures with rational values
\( \nu_{i}(x) \).
For every \( k=0,\dots,2^{n}-2 \), let 
 \begin{align*}
  V_{n,k} &= \setof{\mu \in \cM(\ol\dN) : 
k\cdot 2^{-n} < \mu(\{0,\dots,2^{n}-1\}) < (k+2)\cdot 2^{-n}},
\\       J &= \setof{(n,k): k\cdot 2^{-n} < \nu(\{0,\dots,2^{n}-1\})}.
 \end{align*}
The set \( J \) is recursively enumerable.
Let us define the functions \( j:J\to\dN \) and \( x:J\to\{0,\dots,2^{n}-1\} \)
as follows: \( j(n,k) \) is the smallest \( i \) with 
\( \nu_{i}(\{0,\dots,2^{n}-1\}) > k\cdot 2^{-n} \), and
 \[
   x_{n,k} = \min\setof{y < 2^{n}: \nu_{j(n,k)}(y) < 2^{-n+1}}.
 \]
Let us define the function \( f_{\mu}(x,n,k) \) as follows.
We set \( f_{\mu}(x,n,k)=2^{n-2} \) if the following conditions hold:
 \begin{alphenum}
  \item\label{i:no-upper-semi-neutral.mu-global} \( \mu \in V_{n,k} \);
  \item\label{i:no-upper-semi-neutral.mu-upper} \( \mu(x) < 2^{-n+2} \);
  \item\label{i:no-upper-semi-neutral.unique}
 \( (n,k) \in J \) and \( x=x_{n,k} \).
 \end{alphenum}
Otherwise, \( f_{\mu}(x,n,k)=0 \).
Clearly, the function \( (\mu,x,n,k) \mapsto f_{\mu}(x,n,k) \) is lower
semicomputable.
Condition~\eqref{i:no-upper-semi-neutral.mu-upper} implies
 \begin{equation}\label{eq:no-upper-semi-neutral.n-k-test}
 \sum_{y} \mu(y) f_{\mu}(y,n,k) \le
    \mu(x_{n,k})f_{\mu}(x_{n,k},n,k) < 2^{-n+2}\cdot 2^{n-2} = 1.
 \end{equation}
Let us show that \( \nu \in V_{n,k} \) implies
 \begin{equation}\label{eq:found-bad}
 f_{\nu}(x_{n,k},n,k) = 2^{n-2}.
 \end{equation}
Consider \( x=x_{n,k} \).
Conditions~\eqref{i:no-upper-semi-neutral.mu-global} 
and~\eqref{i:no-upper-semi-neutral.unique} are satisfied by definition.
Let us show that condition~\eqref{i:no-upper-semi-neutral.mu-upper} is also 
satisfied.
Let \( j=j(n,k) \).
By definition, we have \( \nu_{j}(x) < 2^{-n+1} \).
Since by definition \( \nu_{j}\in V_{n,k} \) and \( \nu_{j} \le \nu \in V_{n,k} \),
we have 
 \[
  \nu(x) \le \nu_{j}(x) + 2^{-n+1} < 2^{-n+1} + 2^{-n+1} = 2^{-n+2}.
 \]
Since all three conditions~\eqref{i:no-upper-semi-neutral.mu-global},
\eqref{i:no-upper-semi-neutral.mu-upper}
and~\eqref{i:no-upper-semi-neutral.unique} are satisfied, we have 
shown~\eqref{eq:found-bad}.
Now we define
 \[
   g_{\mu}(x) = \sum_{n\ge 2}\frac{1}{n(n+1)}\sum_{k}f_{\mu}(x,n,k).
 \]
Let us prove that \( g_{\mu}(x) \) is a uniform test.
It is lower semicomputable by definition, so we only need to prove
\( \sum_{x} \mu(x) f_{\mu}(x) \le 1 \).
For this, let \( I_{n,\mu} = \setof{k: \mu\in V_{n,k}} \).
Clearly by definition, \( \card{I_{n,\mu}}\le 2 \).
We have, using this last fact and the test 
property~\eqref{eq:no-upper-semi-neutral.n-k-test}:
 \[
  \sum_{x} \mu(x) g_{\mu}(x) = 
  \sum_{n\ge 2}\frac{1}{n(n+1)}
  \sum_{k\in I_{n,\mu}} \sum_{x}\mu(x) f_{\mu}(x,n,k)
       \le  \sum_{n\ge 2}\frac{1}{n(n+1)}\cdot 2 \le 1.
 \]
Thus, \( g_{\mu}(x) \) is a uniform test.
If \( \nu\in V_{n,k} \) then we have
 \[
 \t_{\nu}(x_{n,k}) 
\gem g_{\nu}(x_{n,k}) \ge \frac{1}{n(n+1)}f_{\mu}(x_{n,k},n,k) \ge 
  \frac{2^{n-2}}{n(n+1)}.
 \]
Hence \( \nu \) is not neutral.
 \end{proof}

 \begin{remark}
In~\cite{LevinUnif76} and~\cite{LevinRandCons84},
Levin imposed extra conditions on tests which allow to find a lower
semicomputable neutral semimeasure.
 \end{remark}

The universal lower semicomputable
semimeasure \( \m(x) \) has a certain property similar to neutrality.
According to Theorem~\ref{thm:defic-charac-cpt-seqs} specialized to
one-element sequences, for every computable measure \( \mu \) we have
\( \d_{\mu}(x) \eqa -\log\mu(x) - K(x) \)
(where the constant in \( \eqa \) depends on \( \mu \)).
So, for computable measures, the expression
 \begin{equation}\label{eq:ol-d}
  \ol\d_{\mu}(x) = -\log\mu(x) - K(x)
 \end{equation}
can serve as a reasonable deficiency of randomness.
(We will also use the test \( \ol\t = 2^{\ol\d} \).)
If we substitute \( \m \) for \( \mu \) in \( \ol\d_{\mu}(x) \), we get 0.
This substitution is not justified, of course.
The fact that \( \m \) is not a probability 
measure can be helped, using compactification as above, 
and extending the notion of randomness tests. 
But the test \( \ol\d_{\mu} \) can replace \( \d_{\mu} \)
only for computable \( \mu \), while \( \m \) is not computable.
Anyway, this is the sense in which all outcomes might
be considered random with respect to \( \m \), and the heuristic
sense in which \( \m \) may be considered ``neutral''.

%Added 2008/01/23
 \section{Monotonicity, quasi-convexity/concavity}

 Some people find that \( \mu \)-tests as defined in Definition~\ref{def:mu-test} are
 too general, in case \( \mu \) is a non-computable measure.
 In particular, randomness with respect to computable measures has 
 a certain---intuitively meaningful---monotonicity property:
 roughly, if \( \nu \) is greater than \( \mu \) then if
 \( x \) is random with respect to \( \mu \), it should also be random with respect to
 \( \nu \).

  \begin{proposition}\label{propo:test-computable-mon}
 For computable measures \( \mu,\nu \) we have for all rational \( c>0 \):
  \begin{align}\label{eq:test-computable-mon}
   2^{-k}\mu\le \nu \imp \d_{\nu}(x)\lea \d_{\mu}(x)+ k + K(k).
  \end{align}
 Here the constant in \( \lea \) depends on \( \mu,\nu \), but not on \( k \).
  \end{proposition}
 \begin{proof}
 We have \( 1\ge\nu\t_{\nu}\ge 2^{-k}\mu\t_{\nu} \),
 hence \( 2^{-k}\t_{\nu} \) is a \( \mu \)-test.
 Using the method of Theorem~\ref{thm:trim} in finding universal tests, 
 one can show that the sum
  \begin{align*}
    \sum_{k: 2^{-k}\mu\t_{\nu}<1} 2^{-k-K(k)}\t_{\nu}
  \end{align*}
 is a \( \mu \)-test, and hence \( \lem \t_{\mu} \).
 Therefore this is true of each member of the sum, which is just what the theorem
 claims.
 \end{proof}

 There are other properties true for tests on computable measures that we may
 want to require for all measures.
 For the following properties, let us define quasi-convexity, which is a
 weakening of the notion of convexity.

 \begin{definition}\label{def:quasi-convex}
 A function \( f:V\to\dR \) defined on a vector space \( V \) is called
 \df{quasi-convex} if for every real number \( x \) the set \( \setof{v: f(v)\le x} \)
 is convex.
 It is  \df{quasi-concave} if \( -f \) is quasi-convex.
 \end{definition}

 It is easy to see that quasi-convexity is equivalent to the inequality
  \begin{align*}
  f(\lg u + (1-\lg)v)\le f(u)\lor f(v)
  \end{align*}
 for all \( u,v \) and \( 0<\lg<1 \), while quasi-concavity is equivalent to
  \begin{align*}
  f(\lg u + (1-\lg)v)\ge f(u)\land f(v)
  \end{align*}

The uniform test with respect to computable measures is approximately 
both quasi-convex and quasi-concave.
Let \( \nu=\lg\mu_{1}+(1-\lg)\mu_{2} \).

Quasi-convexity means, roughly, that if \( x \) is random with respect to both
\( \mu_{1} \) and \( \mu_{2} \) then it is also random with respect to \( \nu \).
This property strengthens monotonicity in the cases where it applies.

  \begin{proposition}\label{propo:test-computable-conv}
 Let \( \mu_{1},\mu_{2} \) be computable measures and \( 0<\lg<1 \) computable,
with \( \nu=\lg\mu_{1}+(1-\lg)\mu_{2} \).
 Then we have
  \begin{align*}
   \d_{\nu}(x) \lea\d_{\mu_{1}}(x)\lor\d_{\mu_{2}}(x) + K(\lg).
  \end{align*}
  \end{proposition}
 \begin{proof}
The relation \( 1\ge\nu\t_{\nu}=\lg\mu_{1}\t_{\nu}+(1-\lg)\mu_{2}\t_{\nu} \)
implies \( 1\ge \mu_{i}\t_{\nu} \) for some \( i\in\{1,2\} \).
Then \( \d_{\nu}\lea\d_{\mu_{i}}+K(\lg) \) (since \( \lg \) was used to define \( \nu \) and
thus \( \t_{\nu} \)).
 \end{proof}

Quasi-concavity means, roughly, that if \( x \) is non-random with respect to both
\( \mu_{1} \) and \( \mu_{2} \) then it is also nonrandom with respect to \( \nu \):

  \begin{proposition}\label{propo:test-computable-conc}
 Let \( \mu_{1},\mu_{2} \) be computable measures and \( 0<\lg<1 \) arbitrary
(not even necessarily computable), with \( \nu=\lg\mu_{1}+(1-\lg)\mu_{2} \).
 Then we have
  \begin{align*}
   \d_{\nu} \gea\d_{\mu_{1}}\land\d_{\mu_{2}}.
  \end{align*}
  \end{proposition}
  \begin{proof}
 The function \( \t_{\mu_{1}}\land\t_{\mu_{2}} \) 
is lower semicomputable, and is a \( \mu_{i} \)-test for each \( i \).
Therefore it is also a \( \nu \)-test, and as such is \( \lem\t_{\nu} \).
Here, the constant in the \( \lem \) depends only on (the programs for)
\( \mu_{1},\mu_{2} \), and not on \( \lg \).
  \end{proof}

 These properties do not survive for arbitrary measures and arbitrary constants.

  \begin{example}\label{example:monotonicity-counter}\
 Let measure \( \mu_{1} \) be uniform over the interval \( \clint{0}{1/2} \), 
let \( \mu_{2} \) be uniform over \( \clint{1/2}{1} \).
For \( 0\le x \le 1 \) let
 \begin{align*}
      \nu_{x} &= (1-x)\mu_{1} + x\mu_{2}.
  \end{align*}
Then \( \nu_{1/2} \) is uniform over \( \clint{0}{1} \).
Let \( \phi_{x} \) be the uniform distribution over \( \clint{x}{x+1/2} \),
and \( \psi_{x} \) the uniform distribution over
\( \clint{0}{x}\cup\clint{x+1/2}{1} \).

Let \( p<1/2 \) be random with respect to the uniform
distribution \( \nu_{1/2} \).

 \begin{enumerate}
  \item The relations
 \begin{align*}
  \nu_{1/2}\le p^{-1}\nu_{p},
  \quad \d_{\nu_{1/2}}(p)<\infty,\quad\d_{\nu_{p}}(p)=\infty
 \end{align*} 
show that any statement analogous to the monotonicity property of
Proposition~\ref{propo:test-computable-mon}
fails when the measures involved are not required to be computable.

  \item
For rational \( r \) with \( 0<r<p \) the relations
%  (1- p)((1-r)\mu_{1}+r\mu_{2})
%  p(r\mu_{1}+(1-r)\mu_{2})
%   p(1-r) + (1-p)r = p + r - 2pr
%   pr + (1-p)(1-r) = 1 - p - r + 2pr
 \begin{align*}
   \nu:=(1-p)\nu_{r} + p\nu_{1-r},
\quad \d_{\nu}(p)=\infty,
\quad\d_{\nu_{r}}(p)<\infty,
\quad \d_{\nu_{1-r}}(p)<\infty 
 \end{align*}
provide a similar counterexample to Proposition~\ref{propo:test-computable-conv}.

  \item
The relations
 \begin{align*}
 \nu_{1/2} = (\nu_{p}+\nu_{1-p})/2, 
 \quad \d_{\nu_{1/2}}(p)<\infty,\quad \d_{\nu_{p}}(p)=\d_{\nu_{1-p}}(p)=\infty
 \end{align*}
provide a similar counterexample to Proposition~\ref{propo:test-computable-conc}.

The following counterexample relies on a less subtle effect:
 \begin{align*}
\nu_{1/2}=(\phi_{p}+\psi_{p})/2,
\quad\d_{\nu_{1/2}}(p)<\infty,
\quad\d_{\phi_{p}}(p)=\d_{\psi_{p}}(p)=\infty,
 \end{align*}
since
as a boundary point of the support, \( p \) is computable from both \( \phi_{p} \) and
\( \psi_{p} \) in a uniform way.
 \end{enumerate}
For a complete proof, uniform tests must be provided for each of the cases: this
is left as exercise for the reader.
  \end{example}

The non-monotonicity example could be used to argue that the 
we allowed too many \( \mu \)-tests, that the test \( \t_{\mu}(x) \)
should not be allowed to depend on properties of \( \mu \) that exploit the
computational properties of \( \mu \) so much stronger than its quantitative
properties.

\section{Algorithmic entropy}

Some properties of description complexity make it a good
expression of the idea of individual information content.

\subsection{Entropy}
The entropy of a discrete probability distribution $\mu$ is defined as
 \[
   \cH(\mu) = - \sum_{x} \mu(x) \log \mu(x).
 \]
To generalize entropy to continuous distributions the
\df{relative entropy} is defined as follows.
Let $\mu,\nu$ be two measures, where $\mu$ is taken (typically, but not
always), to be a probability measure, and $\nu$ another measure, that can
also be a probability measure but is most frequently not.
We define the \df{relative entropy} $\cH_{\nu}(\mu)$ as follows.
If $\mu$ is not absolutely continuous with respect to $\nu$ then 
$\cH_{\nu}(\mu) = -\infty$.
Otherwise, writing 
 \[
  \frac{d\mu}{d\nu} = \frac{\mu(dx)}{\nu(dx)} =: f(x)
 \] 
for the (Radon-Nikodym) derivative 
(density) of $\mu$ with respect to $\nu$, we define
 \[
   \cH_{\nu}(\mu) = - \int \log\frac{d\mu}{d\nu} d\mu 
     = - \mu^{x} \log\frac{\mu(dx)}{\nu(dx)} = -\nu^{x} f(x) \log f(x).
 \]
Thus, $\cH(\mu) = \cH_{\#}(\mu)$ is a special case.

 \begin{example}
 Let $f(x)$ be a probability density function for the distribution
$\mu$ over the real line, and let $\lg$ be the Lebesgue measure there.
Then
 \[
   \cH_{\lg}(\mu) = -\int f(x) \log f(x) d x.
 \]
 \end{example}

In information theory and statistics, when both $\mu$ and $\nu$ are
probability measures, then $-\cH_{\nu}(\mu)$ is also denoted
$D(\mu \parallel \nu)$, and called (after Kullback) the
information divergence of the two measures.
It is frequently used in the role of a distance between $\mu$ and $\nu$.
It is not symmetric, but can be shown to obey the triangle inequality, and
to be nonnegative.
Let us prove the latter property: in our terms, it says that relative
entropy is nonpositive when both $\mu$ and $\nu$ are probability measures.

 \begin{proposition}\label{propo:Kullback-pos}
Over a space $X$, we have
 \begin{equation}\label{eq:Kullback-pos}
   \cH_{\nu}(\mu) \le -\mu(X) \log\frac{\mu(X)}{\nu(X)}.
 \end{equation}
 In particular, if $\mu(X) \ge \nu(X)$ then $\cH_{\nu}(\mu) \le 0$.
 \end{proposition}
 \begin{proof}
The inequality $- a \ln a \le -a\ln b + b-a$
expresses the concavity of the logarithm function.
Substituting $a = f(x)$ and $b = \mu(X)/\nu(X)$ 
and integrating by $\nu$:
 \begin{align*}
 (\ln 2) \cH_{\nu}(\mu) &= 
 -\nu^{x} f(x) \ln f(x) \le -\mu(X) \ln\frac{\mu(X)}{\nu(X)}
  + \frac{\mu(X)}{\nu(X)} \nu(X) - \mu(X)
\\ &= -\mu(X) \ln\frac{\mu(X)}{\nu(X)},
 \end{align*}
giving~\eqref{eq:Kullback-pos}.
 \end{proof}

\subsection{Algorithmic entropy}

Le us recall some facts on description complexity.
Let us fix some (finite or infinite) 
alphabet $\Sg$ and consider the discrete space $\Sg^{*}$.

The universal lower semicomputable
semimeasure $\m(x)$ over $\Sg^{*}$ was defined 
in Definition~\ref{def:m}.
It is possible to turn $\m(x)$ into a measure, by compactifying the
discrete space $\Sg^{*}$ into 
 \[
  \ol{\Sg^{*}}=\Sg^{*}\cup\{\infty\}
 \]
(as in part~\ref{i:compact.compactify} of Example~\ref{example:compact};
this process makes sense also for a constructive discrete space),
and setting $\m(\infty) = 1-\sum_{x\in\Sg^{*}} \m(x)$.
The extended measure $\m$ is not quite lower semicomputable since
the number $\mu(\ol{\Sg^{*}} \xcpt \{0\})$ is not necessarily
lower semicomputable.

 \begin{remark}
A measure $\mu$ is computable over $\ol{\Sg^{*}}$ if and only if the
function $x \mapsto \mu(x)$ is computable for $x \in \Sg^{*}$.
This property does not imply that the number
 \[
 1 - \mu(\infty) = \mu(\Sg^{*}) = \sum_{x\in\Sg^{*}} \mu(x)
 \]
 is computable.
 \end{remark}

Let us allow, for a moment, measures $\mu$ that are not probability
measures: they may not even be finite.
Metric and computability can be extended to this case, the universal test
$\t_{\mu}(x)$ can also be generalized.
The Coding Theorem~\ref{thm:coding}
and other considerations suggest the introduction of
the following notation, for an arbitrary measure $\mu$:

 \begin{definition}
We define the \df{algorithmic entropy} of a point $x$ with respect to measure
$\mu$ as
 \begin{equation}\label{eq:alg-ent}
   H_{\mu}(x) = -\d_{\mu}(x) = -\log\t_{\mu}(x).
 \end{equation}
 \end{definition}
Then, with $\#$ defined as the counting measure over the discrete set
$\Sg^{*}$ (that is, $\#(S) = \card{S}$), we have
 \[
   K(x) \eqa H_{\#}(x).
 \]
This allows viewing $H_{\mu}(x)$ as a generalization of description
complexity.

The following theorem generalizes an earlier known theorem stating that
over a discrete space, for a computable measure,
entropy is within an additive constant the same as ``average complexity'':
$\cH(\mu) \eqa \mu^{x} K(x)$.

 \begin{theorem}
Let $\mu$ be a probability measure.
Then we have
 \begin{equation}\label{eq:entropy-less-avg-algentr}
  \cH_{\nu}(\mu) \le \mu^{x} H_{\nu}(x \mvert \mu).
 \end{equation}
If $X$ is a discrete space then the following estimate also holds:
  \begin{equation}\label{eq:entropy-gea-avg-algentr}
  \cH_{\nu}(\mu) \gea \mu^{x} H_{\nu}(x \mvert \mu).
  \end{equation}
 \end{theorem}
 \begin{proof} % entropy-charact
Let
$\dg$ be the measure with density $\t_{\nu}(x \mvert \mu)$ with respect to
$\nu$: $\t_{\nu}(x \mvert \mu) = \frac{\dg(dx)}{\nu(dx)}$.
Then $\dg(X) \le 1$.
It is easy to see from the maximality property of $\t_{\nu}(x \mvert \mu)$
that $\t_{\nu}(x \mvert \mu) > 0$, therefore according to
Proposition~\ref{propo:density-props}, we have
$\frac{\nu(dx)}{\dg(dx)} = \Paren{\frac{\dg(dx)}{\nu(dx)}}^{-1}$.
Using Proposition~\ref{propo:density-props} and~\ref{propo:Kullback-pos}: 
 \begin{align*}
     \cH_{\nu}(\mu)     &= - \mu^{x} \log\frac{\mu(dx)}{\nu(dx)},
\\   - \mu^{x} H_{\nu}(x \mvert \mu) &=  \mu^{x} \log \frac{\dg(dx)}{\nu(dx)}
              = - \mu^{x} \log \frac{\nu(dx)}{\dg(dx)},
\\    \cH_{\nu}(\mu) - \mu^{x} H_{\nu}(x \mvert \mu)
     &= - \mu^{x} \log \frac{\mu(dx)}{\dg(dx)} 
  \le -\mu(X) \log\frac{\mu(X)}{\dg(X)} \le 0.
 \end{align*}
This proves~\eqref{eq:entropy-less-avg-algentr}.

Over a discrete space $X$, the function
$(x,\mu,\nu) \mapsto \frac{\mu(dx)}{\nu(dx)} = \frac{\mu(x)}{\nu(x)}$
is computable, therefore by the maximality property of
$H_{\nu}(x \mvert \mu)$ we have 
$\frac{\mu(dx)}{\nu(dx)} \lem \t_{\nu}(x \mvert \mu)$,
hence $\cH_{\nu}(\mu) = -\mu^{x} \log \frac{\mu(dx)}{\nu(dx)} 
 \gea \mu^{x} H_{\nu}(x \mvert \mu)$.
 \end{proof} % entropy-charact

\subsection{Addition theorem}
The Addition Theorem~\eqref{eq:Haddit}
can be generalized to the algorithmic entropy $H_{\mu}(x)$
introduced in~\eqref{eq:alg-ent} (a somewhat similar generalization appeared
in~\cite{VovkVyugin93}).
The generalization, defining $H_{\mu,\nu} = H_{\mu\times\nu}$, is
 \begin{equation}\label{eq:addition-general}
   H_{\mu,\nu}(x,y)\eqa
  H_\mu(x \mvert \nu)+ H_\nu(y \mvert x,\; H_\mu(x \mvert \nu),\; \mu).
 \end{equation}
Before proving the general addition theorem, we establish a few useful
facts.

 \begin{proposition}\label{propo:int.H.of.xy}
 We have
 \[
  H_{\mu}(x \mvert \nu) \lea -\log \nu^{y} 2^{-H_{\mu,\nu}(x, y)}.
 \]
 \end{proposition}
 \begin{proof}
The function $f(x,\mu,\nu)$ that is the right-hand side, is upper
semicomputable by definition, and obeys $\mu^{x}2^{-f(x,\mu,\nu)} \le 1$.
Therefore the inequality follows from the minimum property of
$H_{\mu}(x)$.
 \end{proof}

Let $z \in \dN$, then the inequality
 \begin{equation}\label{eq:H.x.cond.z}
  H_{\mu}(x) \lea K(z) + H_{\mu}(x \mvert z)
 \end{equation}
 will be a simple consequence of the general addition theorem.
The following lemma, needed in the proof of the theorem,
generalizes this inequality somewhat:

  \begin{lemma}\label{lem:H.x.cond.z}
 For a com\-put\-able func\-tion $(y,z) \mapsto f(y,z)$ over $\dN$,
we have
 \[
  H_{\mu}(x \mvert y) \lea K(z) + H_{\mu}(x \mvert f(y,z)).
 \]
 \end{lemma}
 \begin{proof}
  The function
 \[
  (x,y,\mu) \mapsto g_{\mu}(x, y)=\sum_{z} 2^{-H_{\mu}(x \mvert f(y,z))-K(z)}
 \]
 is lower semicomputable,
and $\mu^{x} g_{\mu}(x, y) \le \sum_{z} 2^{-K(z)} \le 1$.
Hence $g_{\mu}(x, y) \lem 2^{-H_{\mu}(x \mvert y)}$.
The left-hand side is a sum, hence the inequality holds for each
element of the sum: just what we had to prove.
 \end{proof}

The following monotonicity property will be needed:

\begin{lemma}\label{lem:mon}
 For $i < j$ we have
 \[
   i + H_\mu(x\mvert i) \lea j + H_\mu(x\mvert j) .
 \]
\end{lemma}
\begin{proof}
  From Lemma~\ref{lem:H.x.cond.z}, with $f(i, n)=i + n$ we have
 \[
   H_{\mu}(x \mvert i) - H_{\mu}(x \mvert j) \lea K(j-i) \lea j-i.
 \]
\end{proof}

Let us generalize the minimum property of $H_{\mu}(x)$.

 \begin{proposition}\label{propo:univ-test-gener}
Let $(y,\nu)\mapsto F_{\nu}(y)$ be an upper semicomputable
function with values in $\ol\dZ=\dZ\cup\{-\infty,\infty\}$.
Then by Corollary~\ref{coroll:trim-univ} among the lower semicomputable functions
$(x,y,\nu) \mapsto g_{\nu}(x,y)$ with $\nu^{x}g_{\nu}(x,y)\le 2^{-F_{\nu}(y)}$
there is one that is maximal to within a multiplicative constant.
Choosing $f_{\nu}(x,y)$ as such a function we have
for all $x$ with $F_{\nu}(y) > -\infty$:
 \begin{align*}
  f_{\nu}(x,y) &\eqm 2^{-F_{\nu}(y)}\t_{\nu}(x \mvert y, F_{\nu}(y)),
 \end{align*}
or in logarithmic notation 
$-\log f_{\nu}(x,y) \eqa F_{\nu}(y)+H_{\nu}(x \mvert y, F_{\nu}(y))$.
 \end{proposition}
 \begin{proof}
  To prove the inequality $\gem$, define
 \[
   g_{\nu}(x,y,m) =\max_{i\ge m}2^{-i}\t_{\nu}(x \mvert y, i).
 \]
  Function $g_{\nu}(x,y,m)$ is lower semicomputable and decreasing
in $m$.
Therefore 
 \[
  g_{\nu}(x,y)  = g_{\nu}(x,y,F_{\nu}(y))
 \]
is also lower semicomputable since it is obtained by substituting an
upper semicomputable function for $m$ in $g_{\nu}(x,y,m)$.
The multiplicative form of Lemma~\ref{lem:mon} implies
 \begin{align*}
       g_{\nu}(x,y,m) &\eqm 2^{-m}\t_{\nu}(x \mvert y, m),
\\   g_{\nu}(x,y) &\eqm 2^{-F_{\nu}(y)}\t_{\nu}(x \mvert y, F_{\nu}(y)).
 \end{align*}
We have, since $\t_{\nu}$ is a test:
 \begin{align*}
   \nu^{x}2^{-m}\t_{\nu}(x \mvert y, m) &\le 2^{-m},
\\   \nu^{x}g_{\nu}(x,y) &\lem 2^{-F_{\nu}(y)},
 \end{align*}
implying $g_{\nu}(x,y)\lem f_{\nu}(x,y)$ by the optimality of $f_{\nu}(x,y)$.

To prove the upper bound, consider all lower semicomputable functions
$\phi_{e}(x,y,m,\nu)$ ($e=1,2,\dots$).
By Theorem~\ref{thm:trim}, there is a recursive mapping $e\mapsto e'$ with the
property that $\nu^{x}\phi_{e'}(x,y,m,\nu)\le 2^{-m+1}$, and for each $y,m,\nu$ if
$\nu^{x}\phi_{e}(x,y,m,\nu)< 2^{-m+1}$ then $\phi_{e}=\phi_{e'}$.
Let us apply this transformation to the function
$\phi_{e}(x,y,m,\nu)=f_{\nu}(x,y)$.
The result is a lower semicomputable function 
$f'_{\nu}(x,y,m)=\phi_{e'}(x,y,m,\nu)$
with the property that $\nu^{x}f'_{\nu}(x,y,m)\le 2^{-m+1}$, further
$\nu^{x}f_{\nu}(x,y)\le 2^{-m}$ implies $f'_{\nu}(x,y,m)=f_{\nu}(x,y)$.
Now $(x,y,m,\nu)\mapsto 2^{m-1}f'_{\nu}(x,y,m)$ is a uniform test
of $x$ conditional on $y,m$ and hence it is $\lem\t_{\nu}(x \mvert y, m)$.
Substituting $F_{\nu}(y)$ for $m$ the relation
$\nu^{x}f_{\nu}(x,y)\le 2^{-m}$ is satisfied and hence we have
 \begin{align*}
 f_{\nu}(x,y) =
 f'_{\nu}(x,y,F_{\nu}(y)) \lem 2^{-F_{\nu}(y)+1}\t_{\nu}(x\mvert y,F_{\nu}(y)).
 \end{align*}
 \end{proof}

As mentioned above, the
theory generalizes to measures that are not probability measures.
Taking $f_{\mu}(x,y)=1$ in Proposition~\ref{propo:univ-test-gener} gives
the inequality
 \[
   H_{\mu}(x \mvert \flo{\log \mu(X)}) \lea \log\mu(X),
 \]
with a physical meaning when $\mu$ is the phase space measure.
Using~\eqref{eq:H.x.cond.z}, this implies
 \begin{equation}\label{eq:unif.ub}
  H_\mu(x)\lea \log\mu(X) + K(\flo{\log\mu(X)}).
  \end{equation}

 \begin{theorem}[General addition]\label{thm:addition-general}
The following inequality holds:
 \[
   H_{\mu,\nu}(x,y) \eqa
  H_{\mu}(x \mvert \nu)+ H_{\nu}(y \mvert x,\; H_\mu(x \mvert \nu),\; \mu).
 \]
 \end{theorem}
 \begin{proof}
To prove the inequality $\lea$ define
$G_{\mu,\nu}(x,y,m) =\min_{i\ge m}\;i +  H_{\nu}(y \mvert x, i, \mu)$.
This function is upper semicomputable and increasing
in $m$.
Therefore function
  \[
   G_{\mu,\nu}(x,y)  = G_{\mu,\nu}(x, y, H_{\mu}(x \mvert \nu))
  \]
is also upper semicomputable since it is obtained by substituting an
upper semicomputable function for $m$ in $G_{\mu,\nu}(x,y,m)$.
Lemma~\ref{lem:mon} implies 
  \begin{align*}
     G_{\mu,\nu}(x,y,m) &\eqa m +  H_{\nu}(y \mvert x, m, \mu),
 \\  G_{\mu,\nu}(x,y)   &\eqa H_{\mu}(x \mvert \nu) + 
    H_{\nu}(y \mvert x, H_{\mu}(x \mvert \nu), \mu).
  \end{align*}
   Now, we have
  \begin{align*}
    \nu^{y} 2^{-m - H_{\nu}(y \mvert x, m, \mu)} &\le 2^{-m},
 \\   \nu^{y} 2^{-G_{\mu,\nu}(x,y)} &\lem 2^{-H_{\mu}(x \mvert \mu)}.
  \end{align*}
Integrating over $x$ by $\mu$ gives $\mu^{x}\nu^{y}2^{-G} \lem 1$, implying
$H_{\mu,\nu}(x,y) \lea G_{\mu,\nu}(x,y)$ 
by the minimality property of $H_{\mu,\nu}(x,y)$.
This proves the $\lea$ half of the theorem.

To prove the inequality $\gea$ let $f_{\nu}(x,y,\mu) = 2^{-H_{\mu,\nu}(x,y)}$.
Proposition~\ref{propo:int.H.of.xy} implies that there is a constant $c$ with
$\nu^{y} f_{\nu}(x,y,\mu) \le 2^{-H_{\mu}(x \mvert \nu)+c}$.
Let 
 \begin{align*}
  F_{\nu}(x, \mu) = H_{\mu}(x \mvert \nu).
 \end{align*}
Proposition~\ref{propo:univ-test-gener} gives (substituting $y$ for $x$ and
  $(x,\mu)$ for $y$):
 \begin{align*}
  H_{\mu,\nu}(x,y) = -\log f_{\nu}(x,y,\mu) 
   \gea F_{\nu}(x,\mu)+ H_{\nu}(y \mvert x,F_{\nu}(x,\mu),\mu),
 \end{align*}
which is what needed to be proved.
 \end{proof}

  The function $H_\mu(\og)$ behaves quite differently for different
kinds of measures $\mu$.
Recall the following property of complexity:
 \begin{equation} \label{eq:compl.of.fun}
  K(f(x)\mvert y)\lea K(x\mvert g(y)) \lea K(x) .
 \end{equation}
for any computable functions $f,g$.
This implies
 \[
   K(y)\lea K(x,y).
 \]
  In contrast, if $\mu$ is a probability measure then
 \[
    H_\nu(y) \gea H_{\mu,\nu}(\og,y).
 \]
  This comes from the fact that $2^{-H_\nu(y)}$ is a test for
$\mu\times\nu$.

\subsection{Information}

Mutual information has been defined in Definition~\ref{def:I*}
as $I^{*}(x : y) = K(x) + K(y) - K(x,y)$.
By the Addition theorem, we have
$I^{*}(x:y) \eqa K(y) - K(y \mvert x,\, K(x)) \eqa K(x) - K(x \mvert y,\,K(y))$.   
The two latter expressions show that in some sense, $I^{*}(x:y)$ is the
information held in $x$ about $y$ as well as the information held in $y$
about $x$.  
(The terms $K(x)$, $K(y)$ in the conditions are
logarithmic-sized corrections to this idea.)
Using~\eqref{eq:ol-d}, it is interesting to view
mutual information $I^{*}(x : y)$ as a deficiency of randomness of
the pair $(x,y)$
in terms of the expression $\ol\d_{\mu}$, with respect to $\m \times \m$:
 \[
   I^{*}(x : y) = K(x) + K(y) - K(x,y) = \ol\d_{\m \times \m}(x, y).
 \]
Taking $\m$ as a kind of ``neutral'' probability, even if it is not quite
such, allows us to view $I^{*}(x:y)$ as a ``deficiency of independence''.
Is it also true that $I^{*}(x:y) \eqa \d_{\m \times \m}(x)$?
This would allow us to deduce, as Levin did, ``information conservation''
laws from randomness conservation laws.\footnote{We cannot use the 
test $\ol\t_{\mu}$ for this, since it can be shown easily
that it does not to obey randomness conservation.}

Expression $\d_{\m \times \m}(x)$ must be understood
again in the sense of compactification, as in
Section~\ref{sec:neutral}.
There seem to be two reasonable ways to compactify the space
$\dN\times\dN$: we either compactify it directly, by adding a symbol
$\infty$, or we form the product $\ol\dN \times \ol\dN$.
With either of them, preserving
Theorem~\ref{thm:test-charac-discr}, we would
have to check whether $K(x,y \mvert \m \times \m) \eqa K(x,y)$.
But, knowing the function $\m(x)\times\m(y)$ we 
know the function $x \mapsto \m(x) \eqm \m(x) \times \m(0)$,
hence also the function $(x,y)\mapsto\m(x,y) = \m(\ang{x,y})$.
Using this knowledge, it is possible to develop an argument similar to
the proof of Theorem~\ref{thm:no-upper-semi-neutral}, showing that
$K(x,y \mvert \m \times \m) \eqa K(x,y)$ does not hold.

 \begin{question} 
Is there a neutral measure $M$ with the property that
$I^{*}(x:y) = \d_{M\times M}(x,y)$?
Is this true maybe for all neutral measures $M$?
If not, how far apart are the expressions $\d_{M\times M}(x,y)$ and
$I^{*}(x:y)$ from each other?
 \end{question}

%%% Local Variables: 
%%% mode: latex
%%% TeX-master: "ait-notes"
%%% End: 

 \section{Randomness and complexity}

We have seen in the discrete case that complexity and randomness are
closely related.
The connection is more delicate technically in the continuous case, but its
exploration led to some nice results.

\subsection{Discrete space}
It is known that for computable \( \mu \), the test \( \d_{\mu}(x) \) can be
expressed in terms of the description complexity of \( x \)
(we will prove these expressions below).
Assume that \( \bX \) is the (discrete) space of all binary strings.
Then we have 
\begin{equation}\label{eq:test-charac-fin-cpt}
   \d_{\mu}(x) = -\log \mu(x) - K(x) + O(K(\mu)).
\end{equation}
The meaning of this equation is the following.  
The expression
\( -\log\mu(x) \) is an upper bound (within \( O(K(\mu)) \)) of the complexity
\( K(x) \), and nonrandomness of \( x \) is measured by
the difference between the complexity and this upper bound.
Assume that \( \bX \) is the space of infinite binary sequences.
Then equation~\eqref{eq:test-charac-fin-cpt} must be replaced with
\begin{equation}\label{eq:test-charac-infin-cpt}
   \d_{\mu}(x) = 
 \sup_{n}\Paren{-\log \mu(x^{\le n}) - K(x^{\le n}) + O(K(\mu))}.
\end{equation}
For noncomputable measures, we cannot replace \( O(K(\mu)) \) in
these relations with anything finite, as shown in
the following example.
Therefore however attractive and simple, 
\( \exp(-\log \mu(x) - K(x)) \) is not a universal uniform test of randomness. 

 \begin{proposition}\label{propo:test-charac-counterexample}
There is a measure \( \mu \) over the discrete space \( \bX \) of binary strings
such that for each \( n \), there is an \( x \) with
\( \d_{\mu}(x) = n - K(n) \) and \( -\log \mu(x) - K(x) \lea 0 \).
 \end{proposition}
 \begin{proof}
Let us treat the domain of our measure \( \mu \) as a set of pairs \( (x,y) \).
Let \( x_{n} = 0^{n} \), for \( n=1,2,\dotsc \).
For each \( n \), let \( y_{n} \) be some binary string of length \( n \)
with the property \( K(x_{n}, y_{n}) > n \).
Let \( \mu(x_{n},y_{n})=2^{-n} \), and 0 elsewhere.
Then \( - \log \mu(x_{n},y_{n}) - K(x_{n},y_{n}) \le n - n = 0 \).
On the other hand, let \( t_{\mu}(x,y) \) be the test nonzero only on pairs of strings
\( (x,y) \) of the form \(( x_{n},y) \):
 \[
 t_{\mu}(x_{n}, y) = \frac{\m(n)}{\sum_{z \in \cB^{n}} \mu(x_{n}, z)}.
 \]
The form of the definition ensures semicomputability and we also have
 \[
 \sum_{x,y} \mu(x,y) t_{\mu}(x,y) \le \sum_{n} \m(n) < 1,
 \]
therefore \( t_{\mu} \) is indeed a test.
Hence \( \t_{\mu}(x,y) \gem t_{\mu}(x,y) \).
Taking logarithms, \( \d_{\mu}(x_{n}, y_{n}) \gea n - K(n) \).
 \end{proof}

 The same example shows that the test defined as
\( \exp(-\log\mu(x) - K(x)) \) over discrete sets,
does not satisfy the randomness conservation property.

 \begin{proposition}
The test defined as \( f_{\mu}(x) = \exp(-\log\mu(x) - K(x)) \) 
over discrete spaces \( \bX \) does
not obey the conservation of randomness.
 \end{proposition}
 \begin{proof}
   \begin{sloppyenv}
Let us use the example of Proposition~\ref{propo:test-charac-counterexample}.
Consider the function \( \pi : (x,y) \mapsto x \).
The image of the measure \( \mu \) under the projection is 
\( (\pi\mu)(x) = \sum_{y} \mu(x,y) \).
Thus, \( (\pi\mu)(x_{n}) = \mu(x_{n},y_{n}) = 2^{-n} \).
Then we have seen that \( \log f_{\mu}(x_{n},y_{n}) \le 0 \).
On the other hand,
 \[ 
 \log f_{\pi\mu} (\pi(x_{n},y_{n})) = -\log(\pi\mu)(x_{n}) - K(x_{n})
 \eqa n - K(n). 
 \]
Thus, the projection \( \pi \) takes a random pair \( (x_{n},y_{n}) \) into
an object \( x_{n} \) that is very nonrandom (when randomness is measured using
the test \( f_{\mu} \)).
\end{sloppyenv}
 \end{proof}
In the example, we have the abnormal situation that a pair is random but
one of its elements is nonrandom.
Therefore even if we would not insist on universality, the test 
\( \exp(-\log\mu(x) - K(x)) \) is unsatisfactory.

Looking into the reasons of the nonconservation in the example,
we will notice that it could only have happened because the 
test \( f_{\mu} \) is too special.
The fact that \( -\log (\pi\mu)(x_{n}) - K(x_{n}) \) is large should show that 
the pair \( (x_{n},y_{n}) \) can be enclosed into the ``simple'' set
\( \{x_{n}\} \times \bY \) of small probability; unfortunately, 
this observation does not reflect on 
\( -\log\mu(x,y) - K(x,y) \) (it does for computable \( \mu \)).
 
It is a natural idea to modify equation~\eqref{eq:test-charac-fin-cpt}
in such a way that the complexity \( K(x) \) is replaced with \( K(x \mvert \mu) \).
However, this expression must be understood properly.
We need to use the definition of \( K(x) \) as \( -\log\m(x) \) directly, and not 
as prefix complexity.

Let us mention the following easy fact:

 \begin{proposition}\label{propo:H/mu-computable}
If \( \mu \) is a computable measure then \( K(x \mvert \mu) \eqa K(x) \).
The constant in \( \eqa \) depends on the description complexity of \( \mu \).
 \end{proposition}

 \begin{theorem}\label{thm:test-charac-discr}
If \( \bX \) is the discrete space \( \Sg^{*} \) then we have
 \begin{equation}\label{eq:test-charac-discr}
   \d_{\mu}(x) \eqa -\log\mu(x) - K(x \mvert \mu).
 \end{equation}
 \end{theorem}
Note that in terms of the algorithmic entropy notation introduced 
in~\eqref{eq:alg-ent}, this theorem can be expressed as
 \[
   H_{\mu}(x) \eqa K(x \mvert \mu) + \log\mu(x).
 \]
 \begin{proof}
In exponential notation, equation~\eqref{eq:test-charac-discr} can be
written as \( \t_{\mu}(x) \eqm \m(x \mvert \mu)/\mu(x) \).
Let us prove \( \gem \) first.
We will show that the right-hand side of this inequality is a test, and
hence \( \lem \t_{\mu}(x) \).
However, 
the right-hand side is clearly lower semicomputable in \( (x, \mu) \) and
when we ``integrate'' it (multiply it by \( \mu(x) \) and sum it), its sum is
\( \le 1 \); thus, it is a test.

Let us prove \( \lem \) now.
The expression \( \t_{\mu}(x)\mu(x) \) is clearly lower semicomputable in
\( (x,\mu) \), and its sum is \( \le 1 \).
Hence, it is \( \lea \m(x \mvert \mu) \).
 \end{proof}

 \begin{remark}\label{rem:oracle-warning}
It important not to consider relative computation with respect to \( \mu \) 
as \emph{oracle computation} in the ordinary sense.
Theorem~\ref{thm:neutral-meas} below will show the existence of a measure with
respect to which every element is random.
If randomness is defined using \( \mu \) as an oracle then we can always find
elements nonrandom with respect to \( \mu \).

For similar reasons, the proof of the Coding Theorem does not transfer to the
function \( K(x\mvert \mu) \) 
since an interpreter function should have the property of \df{intensionality},
depending only on \( \mu \) and not on the sequence representing it.
(It does transfer without problem to an oracle version of \( K^{\mu}(x) \).)
The Coding Theorem still may hold, at least in some cases: this is currently not
known.
Until we know this, we cannot interpret \( K(x\mvert \mu) \) as description complexity
in terms of interpreters and codes.

(Thanks to Alexander Shen for this observation: this remark corrects
an error in the paper~\cite{GacsUnif05}.)
 \end{remark}

\subsection{Non-discrete spaces}
For non-discrete spaces, unfortunately, we can only provide a less
intuitive expression.

 \begin{proposition}\label{thm:test-charac}
let \( \bX=(X, d, D, \ag) \) be a complete
computable metric space, and let \( \cE \) be the enumerated set of bounded
Lipschitz functions introduced in~\eqref{eq:bd-Lip-seq}, 
but for the space \( \bM(\bX) \times \bX \).
The uniform test of randomness \( \t_{\mu}(x) \), can be expressed as
 \begin{equation}\label{eq:test-charac}
  \t_{\mu}(x) \eqm 
  \sum_{f \in \cE}f(\mu,x)\frac{\m(f \mvert \mu)}{\mu^{y} f(\mu, y)}.
 \end{equation}
 \end{proposition}
 \begin{proof}
For \( \gem \),
we will show that the right-hand side of the inequality is a test, and
hence \( \lem \t_{\mu}(x) \).
For simplicity, we skip the notation about the enumeration of \( \cE \) and
treat each element \( f \) as its own name.
Each term of the sum is clearly lower semicomputable in
\( (f, x, \mu) \), hence the sum is lower semicomputable in \( (x, \mu) \).
It remains to show that the \( \mu \)-integral of the sum is \( \le 1 \).
But, the \( \mu \)-integral of the generic term is \( \le \m(f \mvert \mu) \), and
the sum of these terms is \( \le 1 \) by the definition of the function
\( \m(\cdot \mvert \cdot) \).
Thus, the sum is a test.

For \( \lem \), note that \( (\mu,x) \mapsto \t_{\mu}(x) \), 
as a lower semicomputable function, is the supremum of functions in \( \cE \).
Denoting their differences by \( f_{i}(\mu,x) \),
we have \( \t_{\mu}(x) = \sum_{i} f_{i}(\mu,x) \).
The test property implies \( \sum_{i} \mu^{x} f_{i}(\mu,x) \le 1 \).
Since the function \( (\mu,i) \mapsto \mu^{x} f_{i}(\mu,x) \) is lower
semicomputable, this implies \( \mu^{x} f_{i}(\mu,x) \lem \m(i \mvert \mu) \), and
hence 
 \[
 f_{i}(\mu,x) \lem f_{i}(\mu,x) \frac{\m(i \mvert \mu)}{\mu^{x} f_{i}(\mu,x)}.
 \]
It is easy to see that for each \( f\in\cE \) we have
 \[
   \sum_{i : f_{i} = f} \m(i \mvert \mu) \le \mu(f \mvert \mu),
 \]
which leads to~\eqref{eq:test-charac}.
 \end{proof}

 \begin{remark}\label{rem:test-charac-lb}
If we only want the \( \gem \) part of the result, then \( \cE \) can be
replaced with any enumerated computable sequence of bounded computable
functions.
 \end{remark}

  \subsection{Infinite sequences}

In case of the space of infinite sequences and a computable measure,
Theorem~\ref{thm:defic-charac-cpt-seqs} gives a characterization of
randomness in terms of complexity.
This theorem does not seem to transfer to a more general situation, but
under some conditions, at least parts of it can be extended.

For arbitrary measures and spaces, we can say a little less:

 \begin{proposition}\label{propo:test-charac-seq-lb}
For all measures \( \mu\in\cM_{R}(X) \), 
for the deficiency of randomness \( \d_{\mu}(x) \), we have
 \begin{equation}\label{eq:defic-ineq-seq}
 \d_{\mu}(x) \gea \sup_{n}\Paren{-\log \ol\mu(x^{\le n}) - K(x^{\le n} \mvert \mu)}.
 \end{equation}
 \end{proposition}
 \begin{proof}
  Consider the function
 \[
 f_{\mu}(x) = \sum_{s} 1_{\Gg_{s}}(x) \frac{\m(s \mvert \mu)}{\ol\mu(\Gg_{s})}
  = \sum_{n} \frac{\m(x^{\le n} \mvert \mu)}{\ol\mu(x^{\le n})}
  \ge \sup_{n} \frac{\m(x^{\le n} \mvert \mu)}{\ol\mu(x^{\le n})}.
 \]
The function \( (\mu,x) \mapsto f_{\mu}(x) \) 
is clearly lower semicomputable and satisfies
\( \mu^{x} f_{\mu}(x) \le 1 \), and hence 
 \[
  \d_{\mu}(x) \gea \log f(x) \gea 
   \sup_{n}\Paren{-\log\ol\mu(x^{\le n}) - K(x^{\le n} \mvert \mu)}.
 \]
 \end{proof}

\begin{definition}
Let \( \cM_{R}(X) \) be the set of measures \( \mu \) with \( \mu(X)=R \).  
\end{definition}

We will be able to prove the \( \gea \) part of the statement of
Theorem~\ref{thm:defic-charac-cpt-seqs} in a more general
space, and without assuming computability.
Assume that a separating sequence \( b_{1},b_{2},\dots \) is given as defined
in Subsection~\ref{subsec:cells}, along with the set \( X^{0} \).
For each \( x \in X^{0} \), the binary sequence \( x_{1},x_{2},\dots \) 
has been defined.
Let
 \begin{align*}
  \ol\mu(\Gg_{s}) &= R  - \sum\setof{\mu(\Gg_{s'}) : \len{s}=\len{s'},\;s'\ne s}.
 \end{align*}
Then \( (s,\mu)\mapsto \mu(\Gg_{s}) \) is lower semicomputable, and 
\( (s,\mu)\mapsto \ol\mu(\Gg_{s}) \) is upper semicomputable.
And, every time that the functions \( b_{i}(x) \) form a partition
with \( \mu \)-continuity, we have \( \ol\mu(\Gg_{s})=\mu(\Gg_{s}) \) for all \( s \).

 \begin{theorem}\label{thm:defic-charac-compact}
Suppose that the space \( X \) is effectively compact.
Then for all computable measures \( \mu\in\cM_{R}^{0}(X) \), 
for the deficiency of randomness \( \d_{\mu}(x) \), 
the characterization~\eqref{eq:defic-charac-seq} holds.
 \end{theorem}
 \begin{proof}
The proof of part \( \gea \) of the inequality follows directly from
Proposition~\ref{propo:test-charac-seq-lb}, just as in the proof of 
Theorem~\ref{thm:defic-charac-cpt-seqs}.

The proof of \( \lea \) is also similar to the proof of that theorem.
The only part that needs to be reproved is the statement that
for every lower semicomputable function \( f \) over \( X \),
there are computable sequences \( y_{i}\in\dN^{*} \) and
\( q_{i}\in\dQ \) with \( f(x) = \sup_{i} q_{i} 1_{y_{i}}(x) \).
This follows now, since according to
Proposition~\ref{propo:compact-cell-basis}, the cells \( \Gg_{y} \) form a basis of
the space \( X \).
 \end{proof}

\subsection{Bernoulli tests}

In this part, we will give characterize Bernoulli sequences in terms of their
complexity growth.

Recall Definition~\ref{def:combinat-Bernoulli}:
A function \( f:\dB^{*}\to\dR \) is a combinatorial Bernoulli test if
 \begin{alphenum}
  \item It is lower semicomputable.
  \item It is monotonic with respect to the prefix relation.
  \item For all \( 0\le k\le n \) we have \( \sum_{x\in\dB(n,k)} f(x) \le \binom{n}{k} \).
 \end{alphenum}
According to Theorem~\ref{thm:combinat-Bernoulli}, a 
sequence \( \xi \) is nonrandom with respect to all Bernoulli measures if and
only if \( \sup_{n}b(\xi^{\le n})=\infty \), where \( b(x) \) is a combinatorial
Bernoulli test.

We need some definitions.

\begin{definition}
For a finite or infinite sequence \( x \) let \( S_{n}(x)=\sum_{i=1}^{n}x(i) \).

For \( 0\le p\le 1 \) and integers \( 0\le k\le n \), 
denote \( B_{p}(n,k)=\binom{n}{k}p^{k}(1-p)^{n-k} \).

An upper semicomputable function \( D:\dN^{2}\to\dN \), defined for \( n\ge 1 \),
\( 0\le k\le n \) will be called a \df{gap function} if
 \begin{align}\label{eq:dg-bd}
  \sum_{n\ge 1}\sum_{k=0}^{n}B_{p}(n,k)2^{-D(n,k)} \le 1
 \end{align}
holds for all \( 0\le p\le 1 \).
A gap function \( D(n,k) \) is \df{optimal} if for every other gap 
function \( D'(n,k) \) there is a \( c_{D'} \) with \( D(n,k)\le D'(n,k)+c_{D'} \).
\end{definition}

 \begin{proposition}
There is an optimal gap function \( D(n,k)\lea K(n) \).
 \end{proposition}
 \begin{proof}
The existence is proved using the technique of Theorem~\ref{thm:trim}.
For the inequality it is sufficient to note that \( K(n) \) is a gap function.
Indeed, we have
 \begin{align*}
  \sum_{n\ge 1}\sum_{k=0}^{n}B_{p}(n,k)2^{-K(n)} =\sum_{n\ge 1}2^{-K(n)}\le 1.
&\qedhere
 \end{align*}
 \end{proof}

\begin{definition}
Let us fix an optimal gap function and denote it by \( \Dg(n,k) \).
\end{definition}

Now we can state the test characterization theorem for Bernoulli tests.

 \begin{theorem}
Denoting by \( b(\xi) \) the universal class test for the Bernoulli sequences, we
have \( b(\xi) \eqm \ol b(\xi) \),
where
  \begin{align*}
  \log \ol b(\xi) = \sup_{n}\log\binom{n}{k}-K(\xi^{\le n}\mvert n,k,\Dg(n,k))-\Dg(n,k),
 \end{align*}
with \( k=S_{n}(\xi) \).
 \end{theorem}
 \begin{proof}
Let \( \dg(n,k) = 2^{-\Dg(n,k)} \).
 \begin{claim}\label{claim:gap}
Consider lower semicomputable functions \( \gm:\dB^{*}\to\dR_{+} \) such that
for all \( n,k \) we have
 \begin{align*}
   \sum_{y\in\dB(n,k)}\gm(y) \le 2^{-\Dg(n,k)}.
 \end{align*}
Among these functions there is one
that is optimal (maximal to within a multiplicative constant).
Calling it \( \dg(y) \) we have
 \begin{align*}
  \dg(y) \eqm \dg(n,k)\cdot\m(y\mvert n,k,\Dg(n,k)).
 \end{align*}
 \end{claim}
Thus, the right-hand side is equal, within a multiplicative constant,
to a lower semicomputable function of \( y \).
 \begin{proof}
This follows immediately from Proposition~\ref{propo:univ-test-gener}.
with \( \nu=\# \) (the counting measure) over the sets \( \dB(n,k) \).
 \end{proof}
Let us show \( \ol b(\xi)\lem b(\xi) \).
We have with \( k=S_{n}(\xi) \) in the first line:
 \begin{align*}
 \ol b(\xi) &=\sup_{n\ge 1}\binom{n}{k}\m(\xi^{\le n}\mvert n,k,\Dg(n,k))\dg(n,k)
\\ &= \sup_{n\ge 1}\sum_{k=0}^{n}\binom{n}{k}\dg(n,k)
     \sum_{y\in\dB(n,k)}1_{y}(\xi)\m(y\mvert n,k,\Dg(n,k))
\\ &\le \sum_{n\ge 1}\sum_{k=0}^{n}\binom{n}{k}
     \sum_{y\in\dB(n,k)}1_{y}(\xi)\dg(n,k)\m(y\mvert n,k,\Dg(n,k))
\\ &\eqm \sum_{n\ge 1}\sum_{k=0}^{n}\binom{n}{k}
     \sum_{y\in\dB(n,k)}1_{y}(\xi)\dg(y),
 \end{align*}
using the notation of Claim~\ref{claim:gap} above.
Let \( t(\xi) \) denote the right-hand side here, which is thus
a lower semicomputable function.
We have for all \( p \):
 \begin{align*}
   B_{p}^{\xi}t(\xi) &\eqm \sum_{n\ge 1}\sum_{k=0}^{n}B_{p}(n,k)\dg(n,k)
 \sum_{y\in\dB(n,k)}\m(y\mvert n,k,\Dg(n,k)) \le 1,
 \end{align*}
so \( t(\xi)\gem \ol b(\xi) \) is a Bernoulli test, showing \( \ol b(\xi)\lem b(\xi) \).

To show \( b(\xi)\lem \ol b(\xi) \) we will
follow the method of the proof of
Theorem~\ref{thm:defic-charac-cpt-seqs}.
Replace \( b(\xi) \) with a rougher version:
 \[
 b'(\xi) = \frac{1}{2}\max \setof{2^{n} : 2^{n} < b(\xi)},
 \]
then we have \( 2 b'<b \).
There are computable sequences \( y_{i}\in\dB^{*} \) and \( k_{i} \in \dN \) with
\( b'(\xi) = \sup_{i}\; 2^{k_{i}} 1_{y_{i}}(\xi) \)
with the property that if \( i<j \) and \( 1_{y_{i}}(\xi)=1_{y_{j}}(\xi)=1 \) then
\( k_{i}<k_{j} \).
As in the imitated proof, we have
\( 2b'(\xi) \ge \sum_{i} 2^{k_{i}} 1_{y_{i}}(\xi) \).
The function \( \gm(y)=\sum_{y_{i}=y}2^{k_{i}} \) is lower semicomputable.
We have
\begin{align}
 b\ge 2b'(\xi)\ge\sum_{i} 2^{k_{i}} 1_{y_{i}}(\xi)  = \sum_{y\in\dN^{*}}1_{y}(\xi)\gm(y)
 = \sum_{n}\sum_{k=0}^{n}\sum_{y\in\dB(n,k)}1_{y}(\xi)\gm(y).
\label{eq:gamma-expr}
 \end{align}
By Theorem~\ref{thm:combinat-Bernoulli} we can assume
\( \sum_{y\in\dB(n,k)}\gm(y)\le \binom{n}{k} \).
Let
 \begin{alignat*}{3}
           &\dg'(y)    &&= \gm(y)\binom{n}{k}^{-1} \le 1,
\\      &\dg'(n,k) &&=  \sum_{y\in\dB(n,k)}\dg'(y).
 \end{alignat*}
Since \( 1\ge B_{p}b \ge B_{p}(2 b') \) for all \( p \), we have 
 \begin{align*}
      1 &\ge \sum_{n}\sum_{k=0}^{n}\sum_{y\in\dB(n,k)}\gm(y)B_{p}^{\xi}1_{y}(\xi)
            = \sum_{n}\sum_{k=0}^{n}p^{k}(1-p)^{n-k}\sum_{y\in\dB(n,k)}\gm(y)
\\     &= \sum_{n}\sum_{k=0}^{n}B_{p}(n,k)\sum_{y\in\dB(n,k)}\dg'(y)
          = \sum_{n}\sum_{k=0}^{n}B_{p}(n,k)\dg'(n,k).
 \end{align*}
Thus \( \dg'(n,k) \) is a gap function, hence \( \dg'(n,k) \lem 2^{-\Dg(n,k)} \),
and by Claim~\ref{claim:gap} we have 
 \begin{align*}
 \gm(y)\binom{n}{k}^{-1}&=\dg'(y)\lem \dg(y)
   \eqm\dg(n,k)\cdot\m(y\mvert n,k,\Dg(n,k)).
 \end{align*}
Substituting back into~\eqref{eq:gamma-expr} finishes the proof of
\( b(\xi)\lem \ol b(\xi) \).
 \end{proof}

\section{Cells}\label{subsec:cells}

This section allows to transfer some of the results
on sequence spaces to more general spaces, by encoding the elements into
sequences.
The reader who is only interested in sequences can skip this section.

\subsection{Partitions}

The coordinates of the sequences into which we want to encode our elements will
be obtained via certain partitions.

Recall from Definition~\ref{def:continuity-set} that 
a measurable set \( A \) is said to be a \df{\( \mu \)-continuity set} if 
\( \mu(\partial A)=0 \) where \( \partial A \) is the boundary of \( A \).

 \begin{definition}[Continuity partition]
Let \( f:X\to\dR \) be a bounded computable function, and let
\( \ag_{1}<\dots<\ag_{k} \) be rational numbers,
and let \( \mu \) be a computable measure with the property
that \( \mu f^{-1}(\ag_{j})=0 \) for all \( j \).

In this case, we will say that \( \ag_{j} \) are \( \mu \)-\df{continuity points} of \( f \).
Let \( \ag_{0}=-\infty \), \( \ag_{k+1}=\infty \), and for \( j=0,\dots,k \), let
Let \( U_{j} = f^{-1}((j,j+1)) \).
The sequence of disjoint 
computably enumerable open sets \( (U_{0},\dots,U_{k}) \) will be called the
\df{partition generated by} \( f,\ag_{1},\dots,\ag_{k} \).

If we have several partitions \( (U_{i0},\dots,U_{i,k}) \), 
generated by different functions \( f_{i} \) (\( i=1,\dots,m \))
and cutoff sequences \( (\ag_{ij}: j=1,\dots,k_{i}) \)  made up of
\( \mu \)-continuity points of \( f_{i} \)
then we can form a new partition generated by all possible intersections
 \[
  V_{j_{1},\dots,j_{n}} = U_{1,j_{1}}\cap \dots \cap U_{m,j_{m}}.
 \]
A partition of this kind will be called a \df{continuity partition}.
The sets \( V_{j_{1},\dots,j_{n}} \) will be called the \df{cells} of this
partition. 
 \end{definition}

The following is worth noting.

 \begin{proposition}\label{propo:reg-partit-meas-cptable}
In a partition as given above, 
the values \( \mu V_{j_{1},\dots,j_{n}} \) are computable
from the names of the functions \( f_{i} \) and the cutoff points \( \ag_{ij} \).
 \end{proposition}
 \begin{proof}
Straightforward.
 \end{proof}

Let us proceed to defining cells.

 \begin{definition}[Separating sequence]
Assume that a computable sequence of functions 
\( b_{1}(x),b_{2}(x),\dots \) over \( X \) is given,
with the property that for every pair 
\( x_{1},x_{2}\in X \) with \( x_{1}\ne x_{2} \), there is a \( j \)
with \( b_{j}(x_{1})\cdot b_{j}(x_{2}) < 0 \).
Such a sequence will be called a \df{separating sequence}.
Let us give the correspondence between the set \( \dB^{\dN} \)
of infinite binary sequences and elements of the set
 \[
  X^{0} = \setof{x\in X: b_{j}(x)\ne 0,\;j=1,2,\dots}.
 \]
For a binary string 
\( s_{1}\dotsm s_{n} = s\in\dB^{*} \), let
 \[
   \Gg_{s}
 \]
be the set of elements of \( X \) with the property that for
\( j=1,\dots,n \), if \( s_{j}=0 \) then \( b_{j}(\og) < 0 \), otherwise
\( b_{j}(\og)>0 \).

The separating sequence will be called \( \mu \)-\df{continuity} sequence
if \( \mu(X^{0})=0 \).
 \end{definition}

This correspondence has the following properties.
 \begin{alphenum}
  \item \( \Gg_{\Lg}=X \).
  \item For each \( s\in \dB \), the sets \( \Gg_{s0} \) and 
\( \Gg_{s1} \) are disjoint and their union is contained in \( \Gg_{s} \).
  \item For \( x\in X^0 \), we have \( \{x\} = \bigcap_{x\in\Gg_{s}} \Gg_{s} \).
 \end{alphenum}

 \begin{definition}[Cells]
If string \( s \) has length \( n \) then \( \Gg_{s} \) will be called a \df{canonical
\( n \)-cell}, or simply canonical cell, or \( n \)-cell.
From now on, whenever \( \Gg \) denotes a subset of \( X \), it means a
canonical cell.
We will also use the notation
 \[
   \len{\Gg_{s}}=\len{s}.
 \]
 \end{definition}

The three properties above say that if we restrict ourselves to the set
\( X^{0} \) then the canonical cells behave somewhat like binary subintervals:
they divide \( X^{0} \) in half, then each half again in half, etc.  
Moreover, around each point, these canonical cells become ``arbitrarily
small'', in some sense (though, they may not be a basis of neighborhoods).
It is easy to see that if \( \Gg_{s_1},\Gg_{s_2} \) are two canonical
cells then they either are disjoint or one of them contains the other.
If \( \Gg_{s_1}\sbs\Gg_{s_2} \) then \( s_2 \) is a prefix of \( s_1 \).
If, for a moment, we write \( \Gg^0_s=\Gg_s\cap X^0 \) then we have the
disjoint union \( \Gg^0_s=\Gg^0_{s0}\cup\Gg^0_{s1} \).

Let us use the following notation.

 \begin{definition}
For an \( n \)-element binary string \( s \), for \( x \in \Gg_{s} \), we will write 
 \[
           s = x^{\le n} =  x_1\dotsm x_n,
\quad \mu(s) = \mu(\Gg_{s}).
 \]
The \( 2^n \) cells (some of them possibly empty)
of the form \( \Gg_s \) for \( \len{s}=n \) form a partition
 \[
  \cP_n
 \]
  of \( X^0 \).
 \end{definition}

  Thus, for elements of \( X^0 \), we can talk about the \( n \)-th bit \( x_n \)
of the description of \( x \): it is uniquely determined.

 \begin{examples}\label{example:cells}\
 \begin{enumerate}
  \item If \( \bX \) is the set of infinite binary sequences with its usual
topology, the functions \( b_{n}(x) = x_{n}-1/2 \) generate the usual
cells, and \( \bX^{0}=\bX \).
  \item If \( \bX \) is the interval \( \clint{0}{1} \), let 
\( b_{n}(x) = -\sin(2^{n}\pi x) \).
Then cells are open 
intervals of the form \( \opint{k\cdot 2^{-n}}{(k+1)\cdot 2^{n}} \),
the correspondence between infinite binary strings
and elements of \( X^0 \) is just the usual representation of \( x \) as the
binary decimal string \( 0.x_{1}x_{2}\dots \).
 \end{enumerate}
 \end{examples}

When we fix canonical cells,  
we will generally assume that the partition chosen is also
``natural''.
The bits \( x_1,x_2,\ldots \) could contain information about the
point \( x \) in decreasing order of importance from a macroscopic
point of view. 
For example, for a container of gas, the first few bits may
describe, to a reasonable degree of precision, the amount of gas in
the left half of the container, the next few bits may describe the
amounts in each quarter, the next few bits may describe the
temperature in each half, the next few bits may describe again the
amount of gas in each half, but now to more precision, etc.
From now on, whenever \( \Gg \) denotes a subset of \( X \), it means a
canonical cell.
From now on, for elements of \( X^0 \), we can talk about the \( n \)-th
bit \( x_n \) of the description of \( x \): it is uniquely determined.

The following observation will prove useful.

 \begin{proposition}\label{propo:compact-cell-basis}
Suppose that the space \( \bX \) is effectively compact\footnote{It was noted by
  Hoyrup and Rojas that the qualification ``effectively'' is necessary here.}
 and we have a separating sequence \( b_{i}(x) \) as given above.
Then the cells \( \Gg_{s} \) form a basis of the space \( \bX \).
 \end{proposition}
 \begin{proof}
We need to prove that for every ball \( B(x,r) \), the there is a cell
\( x\in\Gg_{s}\sbs B(x,r) \).
Let \( C \) be the complement of \( B(x,r) \).
For each point \( y \) of \( C \), there is an \( i \) such that
\( b_{i}(x)\cdot b_{i}(y) < 0 \).
In this case, let \( J^{0} = \setof{z: b_{i}(z) < 0} \),
\( J^{1} = \setof{z: b_{i}(z) > 0} \).
Let \( J(y)=J^{p} \) such that \( y\in J^{p} \).
Then \( C \sbs \bigcup_{y} J(y) \), and compactness implies that there is a
finite sequence \( y_{1},\dots,y_{k} \) with 
\( C \sbs \bigcup_{j=1}^{k} J(y_{j}) \).
Clearly, there is a cell 
 \[
  x \in \Gg_{s} \sbs B(x,r) \xcpt \bigcup_{j=1}^{k} J(y_{j}).
 \]
 \end{proof}

\subsection{Computable probability spaces}

If a separating sequence is given in advance, we may
restrict attention to the class of measures that make our sequence a
\( \mu \)-continuity sequence:

 \begin{definition}
Let \( \cM^{0}(X) \) be the set of those probability measures \( \mu \)
in \( \cM(X) \) for which \( \mu(X \xcpt X^{0})=0 \).
 \end{definition}

On the other hand, for each computable measure \( \mu \),
a computable separating sequence can be constructed that is a \( \mu \)-continuity
sequence.
Recall that \( B(x,r) \) is the ball of center \( x \) and radius \( r \).
Let \( D=\{s_{1},s_{2},\dots\}=\{\ag(1),\ag(2),\dots\} \) 
be the canonical enumeration of the canonical dense set \( D \).

\begin{theorem}[Hoyrup-Rojas]\label{thm:continuity-balls}
There is a sequence of uniformly computable reals \( (r_n)_{n\in \dN} \) such
that \( (B(s_{i},r_{n}))_{i,n} \) is a basis of balls that are \( \mu \)-continuity sets.
This basis is constructively equivalent to the original one, consisting of all 
balls \( B(s_{i},r) \), \( r\in\dQ \).
\end{theorem}

 \begin{corollary}
There is a computable separating sequence with the \( \mu \)-continuity property.
 \end{corollary}
 \begin{proof}
Let us list all balls \( B(s_{i},r_{n}) \) into a single sequence
\( B(s_{i_{k}},r_{n_{k}}) \).
The functions 
 \begin{align*}
  b_{k}(x)=d(s_{i_{k}},x)-r_{n_{k}}
 \end{align*}
give rise to the desired sequence.
 \end{proof}

For the proof of the theorem, we use some preparation.
Recall from Definition~\ref{def:atom} that an atom is a point with positive
measure.

\begin{lemma}\label{lem:zero_measure_points}
Let \( X \) be \( \dR \) or \( \dR^+ \) or \( \clint{0}{1} \). 
Let \( \mu \) be a computable probability measure on \( X \). 
Then there is a sequence of uniformly computable reals \( (x_{n})_{n} \)
which is dense in \( X \) and contains no atoms of \( \mu \).
\end{lemma}
\begin{proof}
Let \( I \) be a closed rational interval. 
We construct \( x\in I \) with \( \mu(\{x\})=0 \). 
To do this, we construct inductively a nested sequence of
closed intervals \( J_k \) of measure \( <2^{-k+1} \), with \( J_0=I \). 
Suppose \( J_k=\clint{a}{b} \) has been constructed, with \( \mu(J_k)<2^{-k+1} \). 
Let
\( m=(b-a)/3 \): one of the intervals \( \clint{a}{a+m} \) and \( \clint{b-m}{b} \) must
have measure \( <2^{-k} \), and we can find it effectively---let it be \( J_{k+1} \).

From a constructive enumeration \( (I_n)_n \) of all the dyadic intervals, we
can construct \( x_n\in I_n \) uniformly.
 \end{proof}

\begin{corollary}\label{coroll:sequence}
Let \( (\bX,\mu) \) be a computable metric space with a computable measure
and let \( (f_i)_i \) be a sequence of uniformly computable real
valued functions on \( X \). 
Then there is a sequence of uniformly computable reals \( (x_n)_n \) that is
dense in \( \dR \) and such that each \( x_{n} \) is a \( \mu \)-continuity point of each
\( f_{i} \). 
\end{corollary}
\begin{proof}
Consider the uniformly computable measures \( \mu_i=\mu\circ f_i^{-1} \) and define
\( \nu=\sum_i2^{-i}\mu_i \).
It is easy to see that \( \nu \) is a computable measure and then, by
Lemma~\ref{lem:zero_measure_points}, there is a sequence of uniformly computable
reals \( (x_n)_n \) which is dense in \( \dR \) and contains no atoms of \( \nu \). 
Since \( \nu(A)=0 \) iff \( \mu_i(A)=0 \) for all \( i \), we get \( \mu(\{f_i^{-1}(x_n)\})=0 \)
for all \( i,n \).
\end{proof}

\begin{proof}[Proof of Theorem~\protect\ref{thm:continuity-balls}]
Apply Corollary~\ref{coroll:sequence} to \( f_i(x)=d(s_i,x) \).

Since every ball \( B(s_{i},r) \) can be expressed as a computably enumerable union of the balls
of the type \( B(s_{i},r_{n}) \) just constructed, the two bases are constructively
equivalent. 
\end{proof}

%%% Local Variables: 
%%% mode: latex
%%% TeX-master: "ait-notes"
%%% End: 

%%% Local Variables: 
%%% mode: latex
%%% TeX-master: "ait-notes"
%%% End: 

%\include{alg-entr} % \input into tests
%\include{rand-compl} \input into tests

\chapter{Exercises and problems}
 
\begin{exercise}
  Define, for any two natural numbers \( r,s \), a standard encoding \( \cnv_s^r \)
of base \( r \) strings \( x \) into base \( s \) strings with the property
 \begin{equation} \label{eq:cnv}
   |\cnv_s^r(x)| \le |x|\frac{\log r}{\log s} + 1.
 \end{equation} 
 \end{exercise}

\begin{proof}[Solution]
  We will use \(  \x = N^N, \x_{r} = Z_{r}^N  \) for the sets of infinite
strings of natural numbers and \( r \)-ary digits respectively.
  For a sequence \( p\in \x_{r} \), let \( [p]_{r} \) denote the real number in the
interval \( [0,1] \) which \( 0.p \) denotes in the base \( r \) number system.
  For \( p \) in \( \dS_{r} \), let \( [p]_{r} = \setof{ [pq]_{r}: q\in\x_{r}} \).

  For the \( r \)-ary string \( x \), let \( v \) be the size of the largest
\( s \)-ary intervals \( [y]_s \) contained in the \( r \)-ary interval \( [x]_{r} \).
If \( [z]_s \) is the leftmost among these intervals, then let
\( \cnv_s^r(x)=z \).
 This is a one-to-one encoding. We have \( r^{-|x|} < 2sv \), since any \( 2s \)
consecutive \( s \)-ary intervals of length \( v \) contain an \( s \)-ary interval of
length \( sv \).
 Therefore
 \[
 |z| = -\log_s v < |x| \frac{\log r}{\log s}+1+\frac{\log 2}{\log s}
 \]
  hence, since \( 2 \le s \) and \( |x| \) is integer, we have the 
inequality~\eqref{eq:cnv}.
 \end{proof}

\begin{exercise}
  A function \( A \) from \( \dS_{r} \times \dS \) to \( \dS \) is called an \( r \)-ary
interpreter.
  Prove the following generalization of the Invariance
Theorem.
  For any \( s \), there is a p.r.~\( s \)-ary interpreter \( U \) such that for any
p.r.~interpreter \( A \) there is a constant \( c < \infty \) such that for all
\( x,y \) we have
 \begin{equation} \label{eq:cnvopt}
   C_U(x \mvert y) \le C_A(x \mvert y) + c.
 \end{equation}
 \end{exercise}

\begin{proof}[Solution] 
  Let \( V  :  \dZ_s^\ast \times \dZ_s^\ast \times \dS \to \dS \) be a  
partial recursive
function which is \df{universal}: such that for any  p.r.~\( s \)-ary
interpreter \( A \), there is a string \( a \) such that for all \( p,x \), we have 
\( A(p,x) = V(a,p,x) \). 

  The machine computing \( U(p,x) \) tries to decompose \( p \) into \( u^ov \) and
outputs \( V(u,v,x) \). Let us verify that \( U \) is optimal. Let \( A \) be a 
p.r.~\( r \)-ary interpreter, \( B \) an \( s \)-ary  p.r.~interpreter such
that \( B(\cnv_s^r(p),x)=A(p,x) \) for all \( p,x \), \( a \) a binary string
such that \( B(p,x) = U(a,p,x) \) for all \( p,x \).
  Let \( x,y \) be two strings.
  If \( C_A(x \mvert y)=\infty \), then the inequality~\eqref{eq:cnvopt} holds
trivially.
  Otherwise, let \( p \) be a binary string of length \( C_A(x \mvert y)/\log r \) with
\( A(p,y)=x \).
  Then
 \[ U(a^o\cnv_s^r(p), y) = V(a,\cnv_s^r(p),y) 
   = B(\cnv_s^r(p),y) = A(p,y) = x.
 \]
  Since 
 \[ |\cnv_s^r(p)|\le |p|\log r/\log s + 1 = 
	C_A(x \mvert y)/\log s + 1,
 \]
  we have
 \[
   C_U(x \mvert y) \le (2|a| + C_A(x \mvert y)/\log s + 1)\log s
  = C_A(x \mvert y) + (2|a|)+1)\log s. 
 \]
 \end{proof}

\begin{exercise}[\( S^{m}_{n} \) -theorem]
  Prove that there is a binary string \( b \) 
such that \( U(p,q,x) = U(b^o p^o q,x) \) holds for all binary strings 
\( p,q \) and arbitrary strings x. 
 \end{exercise}

\begin{exercise}[Schnorr]
  Notice that, apart from the conversion between representations, what 
our optimal interpreter does is the following.
  It treats the program \( p \) as a pair \( (p(1),p(2)) \) whose first member is
the G\"odel number of some interpreter for the universal p.r.~function
\( V(p_{1},p_{2},x) \), and the second argument as a program for this interpreter.
  Prove that indeed, for any recursive pairing function \( w(p,q) \), if there
is a function \( f \) such that \( |w(p,q)| \le f(p) + |q| \) then \( w \) can be used
to define an optimal interpreter.
 \end{exercise}

 \begin{exercise} \label{xrc.need-prefix}
Refute the inequality \( C(x, y) \lea C(x) + C(y) \).
 \end{exercise}

\begin{exercise} 
 Prove the following sharpenings of Theorem~\ref{thm:addit}.
 \begin{align*} 	C(x,y)	&\lea J(C(x)) + C(y \mvert x,C(x)),
 \\      C(x,y) &\lea C(x) + J(C(y \mvert x,C(x))).
 \end{align*}
  Prove Theorem~\ref{thm:addit} from here.
 \end{exercise}

\begin{exercise}[Schnorr]
 Prove \( C(x+C(x)) \eqa C(x) \). Can you generalize this result?
 \end{exercise}

\begin{exercise}[Schnorr]
 Prove that if \( m < n \) then \( m+C(m) \lea n+C(n) \). 
 \end{exercise}

\begin{exercise}
 Prove 
 \[
     \log\binom{n}{k} \eqa k\log\frac{n}{k} + 
   (n-k)\log \frac{n}{n-k} + \frac{1}{2} \log\frac{n}{k(n-k)}.
 \]
 \end{exercise}
\begin{exercise}[Kamae] \label{xrc.Kamae} 
  Prove that for any natural number \( k \) there is a string \( x \) such that 
for for all but finitely many strings \( y \), we have 
 \[ C(x) - C(x \mvert y) \ge k.
 \]
 In other words, there are some strings \( x \) such that almost any string
\( y \) contains a large amount of information about them.
 \end{exercise}

\begin{proof}[Solution]
  Strings \( x \) with this property are for example the ones which contain
information obtainable with the help of any sufficiently large number.
  Let \( E \) be a r.e.\ set of integers. Let \( e_{0}e_{1}e_{2}\ldots \)
be the infinite string which is the characteristic function of \( E \),
and let \( x(k) = e_{0}\ldots e_{2^k} \). We can suppose that \( C(x(k)) \ge
k \).
  Let \( n_{1},n_{2},\ldots \) be a recursive enumeration of \( E \) without
repetition, and let \( \ag(k)=\max\{i : n_{i} \le 2^k\} \).
  Then for any number \( t \ge \ag(k) \) we have \( C(x(k) \mvert t) \lea \log
k \).
  Indeed, a binary string of length \( k \) describes the number \( k \).
  Knowing \( t \) we can enumerate \( n_{1},\ldots,n_{t} \) and thus learn \( x(k) \).
  Therefore with any string \( y \) of length \( \ge \ag(k) \) we have \( C(x) - C(x
\mvert y) \gea k - \log k \).
 \end{proof}

\begin{exercise}
 \begin{alphenum}
  \item Prove that a real function \( f \) is computable iff there is a
recursive function \( g(x,n) \) with rational values, and 
\( |f(x)-g(x,n)| < 1/n \). 
  \item Prove that a function \( f \) is semicomputable iff there exists a
recursive function with rational values, (or, equivalently, a
computable real function) \( g(x,n) \) nondecreasing in \( n \), with \( f(x) =
\lim_{n\to\infty}g(x,n) \). 
 \end{alphenum}
 \end{exercise}

\begin{exercise} 
   Prove that in Theorem~\ref{thm:nolb} one can write ``semicomputable'' for
``partial recursive''.
 \end{exercise}

\begin{exercise}[Levin]
   Show that there is an upper semicomputable function \( G(p,x,y) \) which for 
different finite binary strings \( p \) enumerates all upper semicomputable
functions \( F(x,y) \) satisfying the inequality~\eqref{eq:Kmajor}.
  Prove
 \[
  C(x\mvert y) \eqa \inf_{p} G(p,x,y)+J(|p|).
 \]
 \end{exercise}

\begin{exercise} Prove 
 \[
 	\sum_{ p \in \dB^{n} } \m(p) \eqm \m(n).
 \]
 \end{exercise}
 
\begin{exercise}[Solovay]
  Show that we cannot find effectively infinitely many places where some
recursive upper bound of \( K(n) \) is sharp.
  Moreover, suppose that \( F(n) \) is a recursive upper bound of \( K(n) \).
  Then there is no recursive function \( {\cD}(n) \) ordering to each \( n \) a
finite set of natural numbers (represented for example as a string) larger than
\( n \) such that for each \( n \) there is an \( x \) in \( {\cD}(n) \) with \( F(x) \le
K(x) \).
  Notice that the function \( \log n \) (or one almost equal to it) has this
property for \( C(n) \).
 \end{exercise}

\begin{proof}[Solution]
 Suppose that such a function exists.
  Then we can select a sequence \( n_{1} < n_{2} < \ldots \) of integers with the
property that the sets \( {\cD}(n_{i}) \) are all disjoint.
  Let
 \[
   a(n) = \sum_{x \in {\cD}(n)} 2^{-F(x)}.
 \]
  Then the sequence \( a(n_{k}) \) is computable and \( \sum_{k} a(n_{k}) < 1 \).
  It is now easy to construct a computable subsequence \( m_{i} = n_{k_{i}} \)
of \( n_{k} \)
and a computable sequence \( b_{i} \) such that \( b_{i}/a(m_{i}) \to\infty \) and
\( \sum_{i} b_{i} < 1 \).
  Let us define the semimeasure \( \mu \) setting
 \[
   \mu(x)=2^{-F(x)}b_{i}/a(m_{i})
 \]
  for any \( x \) in \( {\cD}(m_{i}) \) and 0 otherwise.
  Then for any \( x \) in \( {\cD}(m_{i}) \) we have
 \( K(x) \lea -\log\mu(x) = F(x) - \log c_{i} \) 
where \( c_{i}=b_{i}/a(m_{i})\to\infty \),  so we arrived a contradiction
with the assumption.
 \end{proof}

\begin{exercise}
 Prove that there is a recursive upper bound \( F(n) \) of \( K(n) \) and a 
constant \( c \) with the property that there are infinitely many natural 
numbers \( n \) such that for all \( k > 0 \), the quantity of numbers \( x \le n \) 
with \( K(x) < F(x) - k \) is less than \( cn2^{-k} \). 
 \end{exercise}

\begin{proof}[Solution]  Use the upper bound \( G \) found in 
Theorem~\ref{thm:infsharp} and define \( F(n)=\log n + G(\flo{ \log n }) \).
The property follows from the facts that \( \log n + K(\flo{ \log n }) \) is
a sharp upper bound for \( K(x) \) for ``most'' \( x \) less than \( n \) and that
\( G(k) \) is infinitely often close to \( K(k) \).
 \end{proof}
\begin{exercise}
  Give an example of a computable sequence \( a_{n} > 0 \) of with the property
that \( \sum_{n} a_{n} < \infty  \) but for any other computable sequence
 \( b_{n} > 0 \), if \( b_{n}/a_{n} \to\ infty \) then \( \sum_{n} b_{n} = \infty \).
 \end{exercise}

\begin{proof}[Hint:]
 Let \( r_{n} \) be a recursive, increasing sequence of 
rational numbers with \( \lim_{n} r_{n} = \sum_{x} \m(x) \) and let 
\( a_{n} = r_{n+1} - r_{n} \). 
 \end{proof}

\begin{exercise}[Universal coding, Elias]
 Let \( f(n)=\log_{2} n + 2 \log_{2}\log_{2} n \). 
Show that when \( P \) runs over 
all nonincreasing probability distributions over \( N \) then 
 \[
 \lim_{K(P) \to \infty} K(P)^{-1}\sum_{n} P(n)f(n) = 1.
 \]
  \end{exercise}

\begin{exercise}[T.~Cover]
  Let \( \log_{2}^* n = \log_{2} n + \log_{2}\log_{2} n + \ldots \) (all positive
terms).
  Prove that
 \[
 \sum_{n} 2^{-\log_{2}^* n} < \infty,
 \]
  hence \( K(n) \lea \log_{2}^\ast n \).
  For which logarithm bases does \( K(n) \lea \log_{2}^\ast n \) hold?
 \end{exercise}

\begin{exercise}
 Prove that Kamae's result in Exercise~\ref{xrc.Kamae}
does not hold for \( K(x \mvert y) \).
 \end{exercise}

\begin{exercise}
  Prove that for each \( \eps \) there is an \( m \) such that if \( \cH(P)>m \) then
\( |\sum_{x} P(x)K(x) / \cH(P) - 1| < \eps \).
 \end{exercise}

\begin{exercise}
  If a finite sequence \( x \) of length \( n \) has large complexity then its bits
are certainly not predictable.
Let \( h(k,n) = -(k/n)log(k/n) - (1-(k/n))log(1 - (k/n)) \).
A quantitative relation of this sort is the following.
If, for some \( k > n/2 \), a program of length \( m \) can predict \( x_{i} \) from
\( x_{1},\ldots,x_{i-1} \) for at least \( n-k \) values of \( i \) then
\( C(x) < m + n h(k,n) + o(n) \).
 \end{exercise}

\begin{exercise}
This is a more open-ended problem, having to do with the converse of
the previous exercise.
Show (if true) that for each \( c \) there is a \( d \) such that predictability
of \( x \) at significantly more than \( n/2 + c \) places is equivalent to the
possibility to enclose \( x \) into a set of complexity \( o(n) \) and size
\( n - dn \).
 \end{exercise}

%%% Local Variables: 
%%% mode: latex
%%% TeX-master: "ait-notes"
%%% End: 

\appendix

\chapter{Background from mathematics}

\section{Topology}

In this section, we summarize the notions and results of topology that
are needed in the main text.

\subsection{Topological spaces}\label{subsec:top}

A topological space is a set of points with
some kind of---not quantitatively expressed---closeness relation among them.

 \begin{definition}
A \df{topology} on a set \( X \) is defined by a class \( \tau \)
of its subsets called \df{open sets}.
It is required that the empty set and \( X \) are open, and that arbitrary
union and finite intersection of open sets is open.
The pair \( (X, \tau) \) is called a \df{topological space}.
A set is called \df{closed} if its complement is open.
 \end{definition}

Having a set of open and closed sets allows us to speak about closure
operations.

 \begin{definition}
A set \( B \) is called the \df{neighborhood} of a set \( A \) if \( B \) contains an
open set that contains \( A \).
We denote by 
\( \Cl{A}, \Int{A} \)
the closure (the intersection of all closed sets containing \( A \))
and the interior of \( A \) (the union of all open sets in \( A \)) respectively.
Let
 \[
   \partial A = \Cl{A} \xcpt \Int{A}
 \]
denote the boundary of set \( A \).
\end{definition}

An alternative way of defining a topological space is via a basis.

\begin{definition}
A \df{basis} of a topological space
is a subset \( \bg \) of \( \tau \) such that every open set is the
union of some elements of \( \bg \).
A \df{neighborhood} of a point is a basis element containing it.
A \df{basis of neighborhoods of a point} \( x \) is a set \( N \) of neighborhoods
of \( x \) with the property that each neighborhood of \( x \) contains an element
of \( N \).
A \df{subbasis} is a subset \( \sg \) of \( \tau \) such that every open set is the
union of finite intersections from \( \sg \).
 \end{definition}

\begin{examples}\label{example:topol}\
 \begin{enumerate}

  \item\label{i:x.topol.discr} 
Let \( X \) be a set, and let \( \bg \) be the set of all points of \( X \).
The topology with basis \( \bg \) is the \df{discrete topology} of the
set \( X \).
In this topology, every subset of \( X \) is open (and closed).

  \item\label{i:x.topol.real}
Let \( X \) be the real line \( \dR \), and let \( \bg_{\dR} \) be the set of
all open intervals \( \opint{a}{b} \).
The topology \( \tau_{\dR} \) obtained from this basis is the usual topology of
the real line.
When we refer to \( \dR \) as a topological space without qualification, 
this is the topology we will always have in mind.

  \item Let \( X = \ol\dR = \dR \cup \{-\infty,\infty\} \), 
and let \( \bg_{\ol\dR} \) consist of all open intervals \( \opint{a}{b} \)
and in addition of all intervals of the forms \( \lint{-\infty}{a} \) and
\( \rint{a}{\infty} \).
It is clear how the space \( \ol\dR_{+} \) is defined.

  \item\label{i:x.topol.real-upper-converg}
Let \( X \) be the real line \( \dR \).
Let \( \bg_{\dR}^{>} \) 
be the set of all open intervals \( \opint{-\infty}{b} \).
The topology with basis \( \bg_{\dR}^{>} \) is also a topology of the real
line, different from the usual one.
Similarly, let \( \bg_{\dR}^{<} \) 
be the set of all open intervals \( \opint{b}{\infty} \).

  \item\label{i:x.topol.Cantor}
Let \( \Sg \) be a finite or countable alphabet.
On the space \( \Sg^{\dN} \) of infinite sequences with elements in \( \Sg \), let 
\( \tg_{C} = \setof{A\Sg^{\dN} : A \sbsq \Sg^{*}} \) be called the topology of
the \df{Cantor space} (over \( \Sg \)).
Note that if a set \( E \) has the form \( E=A\Sg^{\dN} \) where \( A \) is finite 
then \( E \) is both open and closed.

 \end{enumerate}
\end{examples}

Starting from open sets, we can define some other kinds of set that are
still relatively simple:

 \begin{definition}
A set is called a \( G_{\dg} \) set if it is a countable intersection of open
sets, and it is an \( F_{\sg} \) set if it is a countable union of closed sets.
 \end{definition}

Different topologies over the same space have a natural partial order
relation among them:

 \begin{definition}
A topology \( \tau' \) on \( X \) is called \df{larger}, or \df{finer} than \( \tau \)
if \( \tau' \spsq \tau \).
For two topologies \( \tau_{1},\tau_{2} \) over the same set \( X \), we define
the topology \( \tau_{1}\lor \tau_{2} = \tau_{1} \cap \tau_{2} \), and
\( \tau_{1} \land \tau_{2} \) as the smallest topology containing 
\( \tau_{1} \cup \tau_{2} \).
In the example topologies of the real numbers above, we have
\( \tau_{\dR} = \tau_{\dR}^{<} \land \tau_{\dR}^{>} \).
 \end{definition}

Most topologies used in practice have some separion property.

 \begin{definition}
A topology is said to have the \( T_{0} \) \df{property} if
every point is determined by the class of open sets containing it.
This is the weakest one of a number of other possible separation
 \end{definition}

Both our above example topologies of the real line have this property.
All topologies considered in this survey will have the \( T_{0} \) property.
A stronger such property is the following:

 \begin{definition}
A space is called a \df{Hausdorff} space, or \( T_{2} \) space,
if for every pair of different
points \( x,y \) there is a pair of disjoint open sets
\( A,B \) with \( x\in A \), \( y\in B \).
 \end{definition}

The real line with topology \( \tau_{\dR}^{>} \) in
Example~\ref{example:topol}.\ref{i:x.topol.real-upper-converg} above is
not a Hausdorff space.
A space is Hausdorff if and only if every open set is the union of closures of
basis elements.

\subsection{Continuity}

We introduced topology in order to be able to speak about continuity:

 \begin{definition}
Given two topological spaces \( (X_{i}, \tau_{i}) \) (\( i=1,2 \)),
a function \( f :X_{1} \to X_{2} \) is
called \df{continuous} if for every open set \( G \sbsq X_{2} \) its inverse
image \( f^{-1}(G) \) is also open.
If the topologies \( \tau_{1},\tau_{2} \) are not clear from the context then
we will call the function \( (\tau_{1}, \tau_{2}) \)-continuous.
Clearly, the property remains the same if we require it only for all
elements \( G \) of a subbasis of \( X_{2} \).

If there are two continuous functions between \( X \) and \( Y \) that
are inverses of each other then the two spaces are called
\df{homeomorphic}.

We say that function \( f \) is continuous 
\df{at point} \( x \) if for every neighborhood
\( V \) of \( f(x) \) there is a neighborhood \( U \) of \( x \) with \( f(U) \sbsq V \).
 \end{definition}

Clearly, function \( f \) is continuous if and only if it is continuous in each point.

There is a natural sense in which every subset of a topological space is
also a topological space:

 \begin{definition}
A \df{subspace} of a topological space \( (X, \tau) \)
is defined by a subset \( Y \sbsq X \), and the topology
\( \tau_{Y} = \setof{G \cap Y : G \in \tau} \), called the \df{induced}
topology on \( Y \).
This is the smallest topology on \( Y \)
making the identity mapping \( x \mapsto x \) continuous.
A partial function \( f :X \to Z \) with \( \dom(f) = Y \)
is continuous iff \( f : Y \to Z \) is continuous.
 \end{definition}

Given some topological spaces, we can also form larger ones using for
example the product operation:

 \begin{definition}
For two topological spaces \( (X_{i}, \tau_{i}) \) (\( i=1,2 \)), 
we define the \df{product topology}
on their product \( X \times Y \): this is the
topology defined by the subbasis
consisting of all sets \( G_{1} \times X_{2} \) and all sets
\( X_{1} \times G_{2} \) with \( G_{i} \in \tau_{i} \).
The product topology is the smallest topology making the projection
functions \( (x,y) \mapsto x \), \( (x,y) \mapsto y \) continuous.
Given topological spaces \( X,Y,Z \) we
call a two-argument function \( f: X \times Y \mapsto Z \) continuous if
it is continuous as a function from \( X \times Y \) to \( Z \).
The product topology is defined similarly for
the product \( \prod_{i \in I} X_{i} \) of an arbitrary number of
spaces, indexed by some index set \( I \).
We say that a function is
\( (\tau_{1},\dots,\tau_{n},\mu) \)-continuous if it is
\( (\tau_{1} \times\dots\times\tau_{n},\mu) \)-continuous.
 \end{definition}

 \begin{examples}\label{example:prod}\
  \begin{enumerate}
  
   \item\label{i:x.prod.real}
The space \( \dR \times \dR \) with the product
topology has the usual topology of the Euclidean plane.

   \item\label{i:x.top-inf-seq}
Let \( X \) be a set with the discrete topology,
and \( X^{\dN} \) the set of infinite sequences with elements from \( X \),
with the product topology.
A basis of this topology is provided by all sets of the form \( uX^{\dN} \)
where \( u \in X^{*} \).
The elements of this basis are closed as well as open.
When \( X = \{0,1\} \) then this topology
is the usual topology of infinite binary sequences.

  \end{enumerate}
 \end{examples}

In some special cases, one of the topologies in the definition of
continuity is fixed and known in advance:

 \begin{definition}
A real function \( f : X_{1} \to \dR \) is called continuous if it is
\( (\tau_{1}, \tau_{\dR}) \)-continuous.
 \end{definition}

\subsection{Semicontinuity}

For real functions, a restricted notion of continuity is often useful.
 
 \begin{definition}\label{def:lower-semicont}
A function \( f:X\to\ol\dR \) is \df{lower semicontinuous} if the set
\( \setof{\tup{x,r}\in X\times\ol\dR: f(x) > r} \) is open.
The definition of upper semicontinuity is similar.
Lower semicontinuity on subspaces is defined similarly to continuity on
subspaces.
 \end{definition}

Clearly, a real function 
\( f \) is continuous if and only if it is both lower and upper semicontinuous.
The requirement of lower semicontinuity of \( f \) is equivalent to saying that
for each \( r \in \dR \), the set \( \setof{x: f(x) > r} \) is open.
This can be seen to be equivalent to the following characterization.

 \begin{proposition}\label{propo:semicont-topol}
A real function \( f \) over \( \bX=(X,\tau) \) is lower semicontinuous if and only
if it is  \( (\tau, \tau_{\dR}^{<}) \)-continuous.
 \end{proposition}

 \begin{example}\label{example:open-set-indic}
The indicator function \( 1_{G}(x) \) of an arbitrary open set \( G \) is lower semicontinuous.
 \end{example}

The following is easy to see.

 \begin{proposition}\label{propo:semicont-sup}
The supremum of any set of lower semicontinuous functions
is lower semicontinuous.
 \end{proposition}

The following representation is then also easy to see.

 \begin{proposition}\label{propo:semicont-rep}
Let \( X \) be a topological space with basis \( \bg \).
The function \( f:X\to\ol\dR_{+} \) is lower semicontinuous if and only if
there is a function \( g:\bg\to\ol\dR_{+} \) with
\( f(x)=\sup_{x\in\bg}g(\bg) \).
 \end{proposition}

 \begin{corollary}\label{coroll:semicont-extend}
Let \( X \) be a topological space with basis \( \bg \) 
and \( f:X\to\ol\dR_{+} \) a lower semicontinuous function defined on a
subset \( D \) of \( X \).
Then \( f \) can be extended in a lower semicontinuous way to the whole space \( X \).
 \end{corollary}
 \begin{proof}
Indeed by the above proposition there is a function \( g:\bg\to\ol\dR_{+} \) with
\( f(x)=\sup_{x\in\bg}g(\bg) \) for all \( x\in D \).
Let us define \( f \) by this same formula for all \( x\in X \).   
 \end{proof}

In the important special case of Cantor spaces, 
the basis is given by the set of finite sequences.
In this case we can also require the function \( g(w) \) to be monotonic in the
words \( w \):

 \begin{proposition}\label{propo:semicont-rep.Cantor}
Let \( X=\Sg^{\dN} \) be a Cantor space as defined in
Example~\ref{example:topol}.\ref{i:x.topol.Cantor}.
Then \( f:X\to\ol\dR_{+} \) is lower semicontinuous if and only if there is a function
\( g:\Sg^{*}\to\ol\dR_{+} \) monotonic with respect to the relation \( u\prefix v \),
with \( f(\xi) = \sup_{u\prefix\xi} g(u) \).
 \end{proposition}

\subsection{Compactness}

There is an important property of topological spaces that, when satisfied,
has many useful implications.

 \begin{definition}
Let \( (X, \tau) \) be a topological space, and \( B \) a subset of \( X \).
An \df{open cover} of \( B \) is a family of open sets whose union contains
\( B \).
A subset \( K \) of \( X \) is said to be \df{compact} if every open cover of \( K \)
has a finite subcover.
 \end{definition}

Compact sets have many important properties: for example, a continuous
function over a compact set is bounded.

 \begin{example}\label{example:compact}\
 \begin{enumerate}

  \item\label{i:compact.compactify}
 Every finite discrete space is compact.
An infinite discrete space \( \bX = (X, \tau) \) is not compact, 
but it can be turned into a compact space \( \ol\bX \) 
by adding a new element called \( \infty \): let
\( \ol X = X\cup\{\infty\} \), and 
\( \ol\tau = \tau\cup\setof{\ol X \xcpt A: A \sbs X\txt{ closed }} \).
More generally, this simple operation can be performed with every space
that is \df{locally compact}, that each of its points has a compact
neighborhood.

  \item In a finite-dimensional Euclidean space, every bounded closed set
is compact.

  \item It is known that if \( (\bX_{i})_{i\in I} \) is a family of compact
spaces then their direct product is also compact.

 \end{enumerate}
 \end{example}

There are some properties that are equivalent to compactness in simple
cases, but not always:

 \begin{definition}
A subset \( K \) of \( X \) is said to be \df{sequentially compact} 
if every sequence in \( K \) has a convergent subsequence with limit in \( K \).
The space is \df{locally compact} if every point has a compact
neighborhood.
 \end{definition}

\subsection{Metric spaces}\label{subsec:metric}

Metric spaces are topological spaces with more structure: in them, the
closeness concept is quantifiable.
In our examples for metric spaces, and later in our treatment of the space
of probability measures, we refer to~\cite{BillingsleyConverg68}.

 \begin{definition}
A \df{metric space} is given by a set \( X \) and a distance function 
\( d : X\times X \to \dR_{+} \) satisfying the
\df{triangle inequality} \( d(x, z) \le d(x, y) + d(y, z) \) and also
property that \( d(x,y) = 0 \) implies \( x = y \).
For \( r \in\dR_{+} \), the sets 
 \[
  B(x,r) = \setof{y : d(x, y) < r},\quad
  \ol B(x,r) = \setof{y : d(x, y) \le r}
 \]
are called the open and closed \df{balls} with radius \( r \) and center \( x \).

A metric space is \df{bounded} when \( d(x,y) \) has an upper bound on \( X \).
 \end{definition}

A metric space is also a topological space, with the basis that is the set
of all open balls.
Over this space, the distance function \( d(x,y) \) is obviously continuous.

Each metric space is a Hausdorff space; moreover, it has the
following stronger property.

 \begin{definition}
A topological space is said to have the \( T_{3} \) \df{property} if
for every pair of different points \( x,y \) there is a continuous function 
\( f : X \to \dR \)  with \( f(x) \ne f(y) \).
 \end{definition}

To see that metric spaces are \( T_{3} \), take \( f(z) =  d(x, z) \).

 \begin{definition}
A topological space is called \df{metrizable} if its topology can be
derived from some metric space.
 \end{definition}

It is known that a topological space is metrizable if and only if it has
the \( T_{3} \) property.

\begin{notation}
For an arbitrary set \( A \) and point \( x \) let 
\begin{align}\nonumber
    d(x, A) &= \inf_{y \in A} d(x,y),
\\\label{eq:Aeps}
  A^{\eps} &= \setof{x : d(x, A) < \eps}.
  \end{align}
\end{notation}

 \begin{examples}\label{example:metric}\
 \begin{enumerate}

 \item\label{example:metric-discrete}
A discrete topological space \( X \) can be turned into a metric space as follows:
\( d(x,y)=0 \) if \( x=y \) and 1 otherwise. 

  \item The real line with the distance \( d(x,y) = |x - y| \) is a metric
space.
The topology of this space is the usual topology \( \tau_{\dR} \) of the real
line.

  \item The space \( \dR \times \dR \) with the Euclidean distance gives the
same topology as the product topology of \( \dR \times \dR \).

  \item An arbitrary set \( X \) with the distance \( d(x,y)=1 \) for all pairs
\( x,y \) of different elements, is a metric space that induces the discrete
topology on \( X \).

  \item\label{i:x.metric-inf-seq} 
Let \( X \) be a bounded metric space, and let
\( Y = X^{\dN} \) be the set of infinite sequences 
\( x = \tup{x_{1}, x_{2}, \dotsc} \)
with distance function \( d^{\dN}(x,y) = \sum_{i} 2^{-i} d(x_{i},y_{i}) \).
The topology of this space is the same as the product topology defined 
in Example~\ref{example:prod}.\ref{i:x.top-inf-seq}.

 \item\label{i:x.metric.Cantor}
   Specializing the above example, if \( \Sg \) is the discrete space defined in
Example~\ref{example:metric-discrete} above then we obtain a metrization of the Cantor
space of Example~\ref{example:topol}\ref{i:x.topol.Cantor}.
For every finite sequence \( x\in\Sg^{*} \) and every infinite sequence 
\( \xi\postfix x \) the ball \( B(\xi,2^{-\len{x}}) \) is equal to a basis element that is
the open-closed cylinder set \( x\Sg^{\dN} \).

  \item\label{i:x.metric-nonsep} 
Let \( X \) be a metric space, and let
\( Y = X^{\dN} \) be the set of infinite bounded sequences 
\( x = \tup{x_{1}, x_{2}, \dotsc} \)
with distance function \( d(x, y) = \sup_{i} d(x_{i}, y_{i}) \).

  \item\label{i:x.C(X)}
Let \( X \) be a topological space, and let
\( C(X) \) be the set of bounded continuous functions over \( X \) with
distance function \( d'(f, g) = \sup_{x} d(f(x), g(x)) \).
A special case is \( C\clint{0}{1} \) where the interval \( \clint{0}{1} \) of real
numbers has the usual topology.

  \item\label{i:x.l2}
Let \( l_{2} \) be the set of infinite sequences \( x = (x_{1}, x_{2}, \dotsc) \)
of real numbers with the property that \( \sum_{i} x_{i}^{2} < \infty \).
The metric is given by 
the distance \( d(x, y) = (\sum_{i} |x_{i} - y_{i}|^{2})^{1/2} \).

 \end{enumerate}
 \end{examples}

In metric spaces, certain previously defined topological objects have richer
properties.

 \begin{examples} Each of the following facts holds in metric spaces and
is relatively easy to prove.
\begin{enumerate}
  \item Every closed set is a
\( G_{\dg} \) set (and every open set is an \( F_{\sg} \) set).
  \item A set is compact if and only if it is sequentially compact.
  \item A set is compact if and only if it is closed and
for every \( \eps \), there is a finite set of
\( \eps \)-balls (balls of radius \( \eps \)) covering it.
\end{enumerate}
 \end{examples}

In metric spaces, the notion of continuity can be strengthened.

 \begin{definition}
A function \( f:X\to Y \) between metric spaces \( X,Y \) is \df{uniformly
continuous} if for each \( \eps>0 \) there is a \( \dg>0 \) such that
\( d_{X}(a,b)<\dg \) implies \( d_{Y}(f(a),f(b))<\eps \).
 \end{definition}

It is known that over a compact metric space, every continuous function is
uniformly continuous.

 \begin{definition}[Lipschitz]
Given a function \( f: X \to Y \) between metric spaces and \( \bg > 0 \),
let \( \Lip_{\bg}(X,Y) \) denote the
set of functions (called the Lipschitz\( (\bg) \) functions, or simply
Lipschitz functions) satisfying
 \begin{equation}\label{eq:Lipschitz}
   d_{Y}(f(x), f(y)) \le \bg d_{X}(x, y).
 \end{equation}
Let \( \Lip(X) = \Lip(X,\dR) = \bigcup_{\bg} \Lip_{\bg} \) 
be the set of real Lipschitz functions over \( X \).
 \end{definition}

As an example, every differentiable real function \( f(x) \) with
\( |f'(x)|\le 1 \) everywhere is a Lipschitz\( (1) \) function.

All these functions are uniformly continuous.

We introduce a certain fixed, enumerated sequence of Lipschitz functions
that will be used frequently as ``building blocks'' of other functions.

 \begin{definition}[Hat functions]\label{def:hat-funcs}
Let
 \begin{equation*}
  g_{u,r,\eps}(x) = |1 - |d(x, u) - r|^{+}/\eps|^{+}.
 \end{equation*}
This is a continuous function that is \( 1 \) in the ball
\( B(u,r) \), it is 0 outside the ball \( B(u, r+\eps) \), and takes intermediate
values in between.
It is clearly a Lipschitz\( (1/\eps) \) function.

If a dense set \( D \) is fixed, let \( \cF_{0}=\cF_{0}(D) \) be the set of functions
of the form \( g_{u,r,1/n} \) where \( u \in D \), \( r \) is rational, \( n = 1, 2, \dots \).
Let \( \cF_{1}=\cF_{1}(D) \) be the maximum of a finite number of elements of
\( \cF_{0}(D) \).
Each element \( f \) of \( \cF_{1} \) is bounded between 0 and 1.
Let 
 \begin{equation}\label{eq:bd-Lip-seq}
  \cE =\cE(D) =  \{g_{1}, g_{2}, \dots \} \sps\cF_{1}
 \end{equation}
 be the smallest set of functions containing \( \cF_{0} \)
and the constant 1, and closed under \( \lor \), \( \land \) and rational linear
combinations.  
For each element of \( \cE \), from its definition
we can compute a bound \( \bg \) such that \( f\in\Lip_{\bg} \).
 \end{definition}

For the effective representation of points in a topological space the
following properties are important.

 \begin{definition}
A topological space has the \df{first countability property} if each point
has a countable basis of neighborhoods.
 \end{definition}

Every metric space has the first countability property since we can
restrict ourselves to balls with rational radius.

\begin{sloppypar}
 \begin{definition}
Given a topological space \( (X, \tau) \) and a sequence \( x = (x_{1}, x_{2},
\dotsc) \) of elements of \( X \), we say that \( x \) \df{converges} to a point \( y \)
if for every neighborhood \( G \) of \( y \) there is a \( k \) such that for all 
\( m > k \) we have \( x_{m} \in G \).
We will write \( y = \lim_{n \to \infty} x_{n} \).
 \end{definition}
  \end{sloppypar}

It is easy to show that if spaces \( (X_{i}, \tau_{i}) \) \( (i=1,2) \)
have the first countability property then a function \( f : X \to Y \) is
continuous if and only if for every convergent sequence \( (x_{n}) \) we have
\( f(\lim_{n} x_{n}) = \lim_{n} f(x_{n}) \).

 \begin{definition}
A topological 
space has the \df{second countability property} if the whole space
has a countable basis.
 \end{definition}

For example, the space \( \dR \) has the second countability property
for all three topologies \( \tau_{\dR} \), \( \tau_{\dR}^{<} \), \( \tau_{\dR}^{>} \).
Indeed, we also
get a basis if instead of taking all intervals, we only take
intervals with rational endpoints.
On the other hand, the metric space of 
Example~\ref{example:metric}.\ref{i:x.metric-nonsep} does not have
the second countability property.

 \begin{definition}
In a topological space \( (X, \tau) \), a set \( B \) of points is called
\df{dense} at a point \( x \) if it intersects every neighborhood of \( x \).
It is called \df{everywhere dense}, or \df{dense}, if it is dense at every
point.
A metric space is called \df{separable} if it has a countable everywhere
dense subset.
 \end{definition}

It is easy to see that a metric space is separable if and only if as a topological
space it has the second countability property.

  \begin{example}\label{example:Cclint{0}{1}}
In Example~\ref{example:metric}.\ref{i:x.C(X)}, for \( X=\clint{0}{1} \), we can
choose as our everywhere dense set the set of all polynomials with rational
coefficients, or alternatively,
the set of all piecewise linear functions whose graph has
finitely many nodes at rational points.

More generally, let \( X \) be a \emph{compact separable} metric space
with a dense set \( D \).
Then it can be shown that
in the metric space \( C(X) \), the set of functions \( \cE(D) \) introduced
in Definition~\ref{def:hat-funcs} is dense, and turns it into a complete
(not necessarily compact!) separable metric space.
  \end{example}

 \begin{definition}
In a metric space, let us call a sequence \( x_{1}, x_{2},\dots \) a \df{Cauchy} sequence if
for all \( i<j \) we have \( d(x_{i},x_{j}) < 2^{-i} \).
 \end{definition}

It is easy to see that if an everywhere dense set \( D \) is given then every
element of the space can be represented as the limit of a Cauchy sequence
of elements of \( D \).
But not every Cauchy sequence needs to have a limit.

 \begin{definition}
A metric space is called \df{complete} if every Cauchy sequence in it has a
limit.
 \end{definition}

For example, if \( X \) is the real line with the point 0 removed
then \( X \) is not complete, since
there are Cauchy sequences converging to 0, but 0 is not in \( X \).

It is well-known that every metric space can be embedded (as an everywhere
dense subspace) into a complete space.

 \begin{example}
Consider the set \( D\clint{0}{1} \) of functions over
\( \clint{0}{1} \) that are right continuous and have left limits everywhere.
The book~\cite{BillingsleyConverg68} introduces two different metrics for
them: the Skorohod metric \( d \) and the \( d_{0} \) metric.
In both metrics, two functions are close if a slight monotonic continuous
deformation of the coordinates makes them uniformly close.
But in the \( d_{0} \) metric, the slope of the deformation must be close to 1.
It is shown that the two metrics give rise to the same topology;
however, the space with metric \( d \) is not complete, and the 
space with metric \( d_{0} \) is.
 \end{example}

We will develop the theory of randomness over separable
complete metric spaces.
This is a wide class of spaces encompassing most spaces of practical
interest.
The theory would be simpler if we restricted it to compact or locally
compact spaces; however, some important spaces like \( C\clint{0}{1} \)
(the set of continuouos functions over the interval \( \clint{0}{1} \), with
the maximum difference as their distance) are not locally compact.

\section{Measures}\label{sec:measures}

For a survey of measure theory, see for example~\cite{PollardUsers01}.

\subsection{Set algebras}

Event in a probability space are members of a class of sets that is
required to be a so-called \( \sg \)-algebra (sigma-algebra).

\begin{definition}
A (Boolean set-) \df{algebra} is a set of subsets of some set \( X \)
closed under intersection and complement (and then, of course, under
union).
It is a \df{\( \sg \)-algebra} (sigma-algebra) if it is also closed 
under countable intersection 
(and then, of course, under countable union).
A \df{semialgebra} is a set \( \cL \)
of subsets of some set \( X \) closed under
intersection, with the property that the complement of every element 
of \( \cL \) is the disjoint union of a finite number of elements of \( \cL \).
\end{definition}

If \( \cL \) is a semialgebra then the set of finite unions of elements of
\( \cL \) is an algebra.

 \begin{examples}\label{example:algebras}\
  \begin{enumerate}

   \item\label{i:x.algebras.l-closed}
The set \( \cL_{1} \) of left-closed intervals of the line (including intervals
of the form \( \opint{-\infty}{a} \)) is a semialgebra.

   \item
The set \( \cL_{2} \) of all intervals of the line
(which can be open, closed, left-closed or
right-closed), is a semialgebra.

  \item\label{i:x.algebras.inf-seq}
In  the set \( \{0,1\}^{\dN} \) of infinite 0-1-sequences, the
set \( \cL_{3} \) of all subsets of the form \( u\{0,1\}^{\dN} \) with
\( u\in\{0,1\}^{*} \), is a semialgebra.

   \item
The \( \sg \)-algebra \( \cB \) generated by \( \cL_{1} \), is the same as the one
generated by \( \cL_{2} \), and is also the same as the one generated by the
set of all open sets: it is called the family of \df{Borel sets} of the
line.
The Borel sets of the extended real line \( \ol\dR \) are defined similarly.

  \item More generally, the class of Borel sets in an arbitrary topological space
    is the smallest \( \sg \)-algebra containing all open sets.
  
   \item
Given \( \sg \)-algebras \( \cA,\cB \) in sets \( X,Y \), the product \( \sg \)-algebra
\( \cA\times\cB \) in the space \( X \times Y \) is the one generated by all
elements \( A \times Y \) and \( X \times B \) for \( A\in\cA \) and \( B\in\cB \).

  \end{enumerate}
 \end{examples}

\subsection{Measures}\label{subsec:measures}

Probability is an example of the more general notion of a measure.

 \begin{definition}
A \df{measurable space} is a pair \( (X, \cS) \) where \( \cS \) is a \( \sg \)-algebra
of sets of \( X \).
A \df{measure} on a measurable space \( (X, \cS) \) is a function 
\( \mu : B \to \ol\dR_{+} \) that is \df{\( \sg \)-additive}:
this means that for every countable family \( A_{1}, A_{2},\dotsc \) of
disjoint elements of \( \cS \) we have 
\( \mu(\bigcup_{i} A_{i}) = \sum_{i} \mu(A_{i}) \).
A measure \( \mu \) is \df{\( \sg \)-finite} if the whole space is the union of a
countable set of subsets whose measure is finite.
It is \df{finite} if \( \mu(X) < \infty \). 
It is a \df{probability measure} if \( \mu(X) = 1 \).
 \end{definition}

 \begin{example}[Delta function]\label{example:delta}
For any point \( x \), the measure \( \dg_{x} \) is defined as follows:
 \begin{align*}
   \dg_{x}(A) = \begin{cases}
                  1 &\text{if } x\in A,
     \\         0 &\text{otherwise}.
                \end{cases}
 \end{align*}
 \end{example}

\begin{definition}\label{def:atom}
If \( \mu \) is a measure, a point \( x \) is called an \df{atom} if \( \mu(x)>0 \).
\end{definition}

Generally, we will consider measures over either a countable set
(a discrete measure for which the union of atoms has total measure)
or an uncountable one, with no atoms.
But mixed cases are possible.

It is important to understand how a measure can be defined in practice.
Algebras are generally simpler to grasp constructively
than \( \sg \)-algebras; semialgebras are yet simpler.
Suppose that \( \mu \) is defined over a semialgebra \( \cL \) and is additive.
Then it can always be uniquely extended to an additive function over
the algebra generated by \( \cL \).
The following is an important theorem of measure theory.

 \begin{proposition}[Caratheodory's extension theorem]\label{propo:Caratheo-extension}
Suppose that a nonnegative set function defined over a semialgebra \( \cL \)
is \( \sg \)-additive.
Then it can be extended uniquely to the \( \sg \)-algebra generated by \( \cL \).
 \end{proposition}

 \begin{examples}\label{example:semialgebra}\
  \begin{enumerate}

   \item Let \( x \) be point and let \( \mu(A) = 1 \) if \( x \in A \) and \( 0 \)
otherwise.
In this case, we say that \( \mu \) is \df{concentrated} on the point \( x \).

   \item\label{i:left-cont} Consider the the line \( \dR \), and the
algebra \( \cL_{1} \) defined 
in Example~\ref{example:algebras}.\ref{i:x.algebras.l-closed}.
Let \( f : \dR \to \dR \) be a monotonic real function.
We define a set function over \( \cL_{1} \) as follows.
Let \( \lint{a_{i}}{b_{i}} \), (\( i=1,\dots,n \)) be a set of disjoint left-closed
intervals.
Then \( \mu(\bigcup_{i} \lint{a_{i}}{b_{i}}) = \sum_{i} f(b_{i}) - f(a_{i}) \). 
It is easy to see that \( \mu \) is additive.
It is \( \sg \)-additive if and only if \( f \) is left-continuous.

  \item\label{i:measure-Cantor} Let \( B = \{0,1\} \), and consider the set 
\( B^{\dN} \) of infinite 0-1-sequences, and the
semialgebra \( \cL_{3} \) of
Example~\ref{example:algebras}.\ref{i:x.algebras.inf-seq}.
Let \( \mu : B^{*} \to \dR_{+} \) be a function.
Let us write \( \mu(uB^{\dN}) = \mu(u) \) for all \( u \in B^{*} \).
Then it can be shown that the following conditions are equivalent:
\( \mu \) is \( \sg \)-additive over \( \cL_{3} \); it is 
additive over \( \cL_{3} \); the equation \( \mu(u) = \mu(u0) + \mu(u1) \) 
holds for all \( u \in B^{*} \).

  \item The nonnegative linear combination of any finite number of measures
is also a measure.
In this way, it is easy to construct arbitrary measures concentrated on a
finite number of points.

  \item Given two measure spaces \( (X,\cA,\mu) \) and \( (Y,\cB,\nu) \) it is
possible to 
define the product measure \( \mu\times\nu \) over the measureable space
\( (X\times Y, \cA\times\cB) \).
The definition is required to satisfy
 \( \mu\times\nu(A\times B) = \mu(A)\times\nu(B) \), and is determined uniquely
by this condition.
If \( \nu \) is a probability measure then, of course,
 \( \mu(A) = \mu\times\nu(A \times Y) \).

  \end{enumerate}
 \end{examples}

Let us finally define measureable functions.

\begin{definition}[Measureable functions]
Let \( (X,\cA) \), \( (Y,\cB) \) be two measureable spaces.
A function \( f:X\to Y \) is called \df{measureable} if and only if
\( f^{-1}(E)\in\cA \) for all \( E\in\cB \).  
\end{definition}

The following is easy to check.

\begin{proposition}
  Let \( (X,\cA) \) be a measureable space and \( (\dR,\cB) \) be the measureable space
  of the real numbers, with the Borel sets.
Then \( f:X\to\dR \) is measureable if and only if all sets of the form
\( f^{-1}(\opint{r}{\infty})=\setof{x: f(x)>r} \) are measureable, where \( r \) is a
rational number.
\end{proposition}

 \begin{remark}\label{rem:measure.step-by-step}
Example~\ref{example:semialgebra}.\ref{i:measure-Cantor} shows a particularly
attractive way to define measures.
Keep splitting the values \( \mu(u) \) in an arbitrary way into
\( \mu(u0) \) and \( \mu(u1) \), and the resulting values on the semialgebra define
a measure.
Example~\ref{example:semialgebra}.\ref{i:left-cont} is less attractive,
since in the process of defining \( \mu \) on all intervals and only keeping
track of finite additivity, we may end up with
a monotonic function that is not left continuous, and thus with a measure
that is not \( \sg \)-additive.
In the subsection on probability measures over a metric space, we will find
that even on the real line, there is a way to define measures in a
step-by-step manner, and only checking for consistency along the way.
 \end{remark} 

A probability space, according to the axioms introduced by Kolmogorov, is
just a measureable space with a normed measure.

 \begin{definition}
A \df{probability space} is a triple \( (X, \cS, P) \) where \( (X, \cS) \) is a
measurable space and \( P \) is a probability measure over it.

Let \( (X_{i}, \cS_{i}) \) (\( i=1,2 \)) be measurable spaces, and let 
\( f : X \to Y \) be a mapping.
Then \( f \) is \df{measurable} if and only if for each element \( B \) of
\( \cS_{2} \), its inverse image \( f^{-1}(B) \) is in \( \cS_{1} \).
If \( \mu_{1} \) is a measure over \( (X_{1}, \cS_{1}) \) then
\( \mu_{2} \) defined by \( \mu_{2}(A) = \mu_{1}(f^{-1}(A)) \) is a measure over
\( X_{2} \) called the measure \df{induced} by  \( f \).
 \end{definition}

\subsection{Integral}\label{subsec:integral}

The notion of integral also generalizes to arbitrary measures, and
is sometimes also used to define measures.

First we define integral on very simple functions.

 \begin{definition}
A measurable function \( f : X \to \dR \) is called a \df{step function} if
its range is finite.

The set of step functions is closed with respect to linear combinations and
also with respect to the operations \( \land,\lor \).
Any such set of functions is called a \df{Riesz space}.
 \end{definition}

%  \begin{remark}
%  More generally (following~\cite{BourbakiIntegr67}), an \df{ordered vector
% space} is a vector space with a partial order, satisfying the conditions
% that \( x \le y \) implies \( x+z \le y+z \) and \( x \ge 0 \), \( \lg > 0 \) implies
% \( \lg x \ge 0 \).
% Such a space is a \df{Riesz space} if it is also a lattice with respect to
% the order.
%  \end{remark}

 \begin{definition}
Given a step function \( f \) which takes values \( x_{i} \) on sets \( A_{i} \), and a
finite measure \( \mu \), we define 
 \[
 \mu(f) = \mu f = \int f\,d\mu = \int f(x) \mu(d x) 
   = \sum_{i} x_{i} \mu(A_{i}).
 \]
 \end{definition}

This is a linear positive functional on the set of step functions.
Moreover, it can be shown that it
is continuous on monotonic sequences: if \( f_{i} \searrow 0 \)
then \( \mu f_{i} \searrow 0 \).
The converse can also be shown:
Let \( \mu \) be a linear positive functional on step functions 
that is continuous on monotonic sequences.
Then the set function \( \mu(A) = \mu(1_{A}) \) is a finite measure.

 \begin{proposition}\label{propo:Riesz-extension}
Let \( \cE \) be any Riesz space of functions with the property that
\( 1 \in \cE \).
Let \( \mu \) be a positive linear functional on \( \cE \) continuous on monotonic
sequences, with \( \mu 1 = 1 \).
The functional \( \mu \) can be extended to the set 
\( \cE_{+} \) of monotonic limits of nonnegative elements of \( \cE \), by
continuity.
In case when \( \cE \) is the set of all step functions, the set \( \cE_{+} \) is
the set of all nonnegative measurable functions.
 \end{proposition}

Now we extend the notion of integral to a wider class of functions.

 \begin{definition}
Let us fix a finite measure \( \mu \) over a measurable space \( (X, \cS) \).
A measurable function \( f \) is called \df{integrable} with respect to \( \mu \)
if \( \mu |f|^{+} < \infty \) and \( \mu |f|^{-} < \infty \).
In this case, we define \( \mu f = \mu |f|^{+} - \mu |f|^{-} \).
 \end{definition}

The set of integrable functions is a Riesz space, and the positive linear
functional \( \mu \) on it is continuous with respect to monotonic sequences.
The continuity over monotonic sequences also implies the following theorem.

 \begin{proposition}[Bounded convergence theorem]
Suppose that functions \( f_{n} \) are integrable and
\( |f_{n}| < g \) for some integrable function \( g \).
Then \( f = \lim_{n} f_{n} \) is integrable and \( \mu f = \lim_{n} \mu f_{n} \).
 \end{proposition}

 \begin{definition}
Two measurable functions \( f,g \) are called \df{equivalent} with respect to
measure \( \mu \) if \( \mu(f - g) = 0 \).
 \end{definition}

For two-dimensional integration, the following theorem holds.

 \begin{proposition}[Fubini theorem]\label{propo:fubini}
Suppose that function \( f(\cdot,\cdot) \) is integrable over
the space \( (X\times Y, \cA\times\cB, \mu\times\nu) \).
Then for \( \mu \)-almost all \( x \), the function \( f(x,\cdot) \) is integrable over
\( (Y,\cB,\nu) \), and the function \( x\mapsto\nu^{y}f(x,y) \) 
is integrable over \( (X,\cA,\mu) \)
with \( (\mu\times\nu) f = \mu^{x}\mu^{y}f \).
 \end{proposition}

To express a continuity property of measures, we can say the following
(recall the definition of \( C(X) \) in Example~\ref{example:metric}.\ref{i:x.C(X)}).

\begin{proposition}\label{propo:measure-cont}
Let \( X \) be a metric space and \( \mu \) a measure.
Then \( \mu \) is a bounded (and thus continuous)
linear functional over the space \( C(X) \).
\end{proposition}

\subsection{Density}

When does one measure have a density function with respect to another?

 \begin{definition}
Let \( \mu, \nu \) be two measures over the same measurable space.
We say that \( \nu \) is \df{absolutely continuous} with respect to \( \mu \), or
that \( \mu \) \df{dominates} \( \nu \), if
for each set \( A \), \( \mu(A) = 0 \) implies \( \nu(A) = 0 \).
 \end{definition}

Every nonnegative integrable function \( f \) defines a new measure \( \nu \) via
the formula \( \nu(A) = \mu(f\cdot 1_{A}) \).
This measure \( \nu \) is absolutely continuous with respect to \( \mu \).
The Radon-Nikodym theorem says that the converse is also true.

 \begin{proposition}[Radon-Nikodym theorem]
If \( \nu \) is dominated by \( \mu \) then there is a nonnegative integrable
function \( f \) such that \( \nu(A) = \mu(f \cdot 1_{A}) \) for all measurable
sets \( A \).
The function \( f \) is defined uniquely to within equivalence with respect to
\( \mu \).
 \end{proposition}

 \begin{definition}
The function \( f \) of the Radom-Nikodym Theorem above
is called the \df{density} of \( \nu \) with respect to \( \mu \).
We will denote it by
 \[
   f(x) = \frac{\mu(dx)}{\nu(dx)} = \frac{d\mu}{d\nu}.
 \]
 \end{definition}

The following theorem is also standard.

 \begin{proposition}[Chain rule and inverse function]
\label{propo:density-props}\
 \begin{enumerate}
  \item
Let \( \mu, \nu, \eta \) be measures such that \( \eta \) is absolutely continuous
with respect to \( \mu \) and \( \mu \) is absolutely continuous with respect to
\( \nu \).
Then the ``chain rule'' holds:
 \begin{equation}\label{eq:chain-rule}
  \frac{d\eta}{d\nu} = \frac{d\eta}{d\mu} \frac{d\mu}{d\nu}.
 \end{equation}

  \item
If \( \frac{\nu(dx)}{\mu(dx)} > 0 \) for all \( x \) then
\( \mu \) is also absolutely continuous with respect to \( \nu \) and
\( \frac{\mu(dx)}{\nu(dx)} = \Paren{\frac{\nu(dx)}{\mu(dx)}}^{-1} \).
 \end{enumerate}
 \end{proposition}

There is a natural distance to be used between measures, though
later we will see that it is not the preferred one in metric spaces.

 \begin{definition}[Total variation distance]\label{def:total-var}
Let \( \mu, \nu \) be two measures, then both are dominated by some measure
\( \eta \) (for example by \( \eta = \mu + \nu \)).
Let their densities with respect to \( \eta \) be \( f \) and \( g \).
Then we define the \df{total variation distance} of the two measures
as
 \[
  D(\mu, \nu)=\eta(|f - g|).
 \]
It is independent of the dominating measure \( \eta \).
 \end{definition}

 \begin{example}
Suppose that the space \( X \) can be partitioned into 
disjoint sets \( A,B \) such that \( \nu(A)=\mu(B) = 0 \).
Then \( D(\mu, \nu) = \mu(A) + \nu(B) = \mu(X) + \nu(X) \).
 \end{example}

\subsection{Random transitions}\label{subsec:transitions}

What is just a transition matrix in case of a Markov chain also
needs to be defined more carefully in the non-discrete cases.
We follow the definition given in~\cite{PollardUsers01}.

 \begin{definition}
Let \( (X, \cA) \), \( (Y, \cB) \) be measureable spaces (defined in
Subsection~\ref{subsec:measures}).
Suppose that a family of probability
measures \( \Lg = \setof{\lg_{x} : x \in X} \) on \( \cB \) is given.
We call it a \df{probability kernel}, (or Markov kernel, or conditional
distribution) if the map \( x \mapsto \lg_{x} B \) is measurable for each
\( B \in \cB \).
 \end{definition} 

When \( X,Y \) are finite sets then \( \lg \) is a Markov transition matrix.
The following theorem shows that \( \lg \) assigns a joint distribution over
the space \( (X \times Y, \cA\times\cB) \) to each input distribution \( \mu \).

 \begin{proposition} For each nonnegative \( \cA\times \cB \)-measureable
function \( f \) over \( X \times Y \),
 \begin{enumerate}
  \item The function \( y \mapsto f(x,y) \) is \( \cB \)-measurable for each fixed \( x \).
  \item The function \( x \mapsto \lg_{x}^{y} f(x, y) \) is \( \cA \)-measurable.
  \item The integral \( f \mapsto \mu^{x} \lg_{x}^{y} f(x, y) \) defines
a measure on \( \cA \times \cB \).
 \end{enumerate}
 \end{proposition}

The above proposition allows to define a mapping over measures.

 \begin{definition}
According to the above proposition, given a probability kernel \( \Lg \),
to each measure \( \mu \) over \( \cA \) corresponds a measure over
\( \cA \times \cB \).
We will denote its marginal over \( \cB \) as
 \begin{equation}\label{eq:Markov-op-meas}
   \Lg^{*} \mu.
 \end{equation}
For every measurable function \( g(y) \) over \( Y \), we can define the measurable 
function \( f(x) = \lg_{x} g = \lg_{x}^{y} g(y) \): we write
 \begin{equation}\label{eq:Markov-op-fun}
   f = \Lg g.
 \end{equation}
 \end{definition}

The operator \( \Lg \) is linear, and monotone with \( \Lg 1 = 1 \).
By these definitions, we have
 \begin{equation}\label{eq:Lg-Lg*}
  \mu(\Lg g) = (\Lg^{*}\mu) g.
 \end{equation}

An example is the simple case of a deterministic mapping:

 \begin{example}\label{example:determ-trans}
Let \( h : X \to Y \) be a measureable function, and
let \( \lg_{x} \) be the measure \( \dg_{h(x)} \) concentrated on the
point \( h(x) \).
This operator, denoted \( \Lg_{h} \) is, in fact, a deterministic transition,
and we have \( \Lg_{h} g = g \circ h \).
In this case, we will simplify the notation as follows:
 \[
   h^{*}\mu =  \Lg_{h}^{*}.
 \]
 \end{example}

\subsection{Probability measures over a metric
space}\label{subsec:measure-metric}

We follow the exposition of~\cite{BillingsleyConverg68}.
Whenever we deal with probability measures on a metric space, we will
assume that our metric space is complete and separable (Polish space).

Let \( \bX = (X, d) \) be a complete separable metric space.
Then \( \bX \) gives rise to a measurable space, where the measurable sets are its
Borel sets.
It can be shown that, if \( A \) is a Borel set and \( \mu \) is a finite measure
then there are sets
\( F \sbsq A \sbsq G \) where \( F \) is an \( F_{\sg} \) set, \( G \) is a \( G_{\dg} \) set,
and \( \mu(F) = \mu(G) \).

It can be shown that a measure is determined by its values on the
elements of any a basis of open sets closed under intersections.
The following proposition follows then essentially from
Proposition~\ref{propo:Caratheo-extension}.

 \begin{proposition}\label{propo:Caratheo-topol}
Let \( \cB \) be a basis of open sets closed under intersections.
Let \( \cB^{*} \) be the set algebra generated by this basis
and let \( \mu \) be any
\( \sigma \)-additive set function on \( \cB^{*} \) with \( \mu(X)=1 \).
Then \( \mu \) can be extended uniquely to a probability measure.
 \end{proposition}

\subsubsection{Weak topology}\label{subsubsec:weak-top}
One can introduce the notion of convergence of measures in a number of
ways.
We have already introduced the total variation distance in
Definition~\ref{def:total-var} above.
But in some cases, the requirement of being close in this distance is too
strong.
Let 
 \[
   \cM(\bX)
 \]
be the set of probability measures on the metric space \( \bX \).
Let \( x_{n} \) be a sequence of points converging to point \( x \) but with 
\( x_{n} \ne x \).
We would like to say that the delta measure \( \dg_{x_{n}} \) 
(concentrated on point \( x_{n} \), see Example~\ref{example:delta}) converges to \( \dg_{x} \).
But the total variation distance \( D(\dg_{x_{n}}, \dg_{x}) \) is 2
for all \( n \).

 \begin{definition}[Weak convergence]
We say that a sequence of probability
measures \( \mu_{n} \) over a metric space \( (X, d) \)
\df{weakly converges} to measure \( \mu \) if for all bounded continuous
real functions \( f \) over \( X \) we have \( \mu_{n} f \to \mu f \).

For a bounded continuous function \( f \) and real numbers \( c \) let
 \[
   A_{f,c} = \setof{\mu : \mu f < c}
 \]
 \end{definition}

A \df{topology of weak convergence} \( (\cM, \tau_{w}) \)
can be defined using a number of different subbases.
The one used in the original definition
is the subbasis consisting of all sets of the form \( A_{f,c} \) above.

We also get a subbasis (see for example~\cite{PollardUsers01}) 
if we restrict ourselves to the set
\( \Lip(X) \) of Lipschitz functions defined in~\eqref{eq:Lipschitz}.
Another possible subbasis giving rise to the same topology
consists of all sets of the form
 \begin{equation}\label{eq:measure-on-open}
   B_{G,c} = \setof{\mu : \mu(G) > c}
 \end{equation}
for open sets \( G \) and real numbers \( c \).
Let us find some countable subbases.
Since the space \( \bX \) is separable, there is a sequence 
\( U_{1}, U_{2},\dotsc \) of open sets that forms a basis of \( \bX \).
Then we can restrict the subbasis of the space of measures to those sets 
\( B_{G, c} \) where \( G \) is the union of a finite number of basis elements
\( U_{i} \) of \( \bX \) and \( c \) is rational.
This way, the space \( (\cM, \tau_{w}) \) itself has the second countability
property.

It is more convenient to define a countable subbasis using bounded
continuous functions \( f \), since the function
\( \mu \mapsto \mu f \) is continuous on such functions, while 
\( \mu \mapsto \mu U \) is typically not continuous when \( U \) is an open set.

 \begin{example}
If \( \bX=\dR \) and \( U \) is the open interval \( \opint{0}{1} \), 
the sequence of probability measures \( \dg_{1/n} \) (concentrated on \( 1/n \))
converges to \( \dg_{0} \), but \( \dg_{1/n}(U)=1 \), and \( \dg_{0}(U)=0 \).
 \end{example}

For some fixed dense set \( D \), let \( \cF_{1}=\F_{1}(D) \) be the set of
functions introduced in~Definition~\ref{def:hat-funcs}.

 \begin{definition}\label{def:continuity-set}
We say that a set \( A \) is a \df{continuity set} of measure \( \mu \) if
\( \mu(\partial A) = 0 \): the boundary of \( A \) has measure 0.
 \end{definition}

 \begin{proposition}\label{propo:Portmanteau}
The following conditions are equivalent:
 \begin{enumerate}
  \item \( \mu_{n} \) weakly converges to \( \mu \).
  \item \( \mu_{n} f \to \mu f \) for all \( f \in \cF_{1} \).
  \item For every Borel set \( A \), that is a continuity set of \( \mu \), we have
\( \mu_{n}(A) \to \mu(A) \).
  \item For every closed set \( F \), \( \lim\inf_{n} \mu_{n}(F) \ge \mu(F) \).
  \item For every open set \( G \), \( \lim\sup_{n} \mu_{n}(G) \le \mu(G) \).
 \end{enumerate}
 \end{proposition}

 \begin{definition}
To define the topological space \( \cM(X) \) of the set of measures over the
metric space \( X \), we choose as subbasis
 \begin{equation}\label{eq:metric-measure-subbasis}
   \sg_{\cM}
 \end{equation}
the sets \( \setof{\mu : \mu f < r} \) and \( \setof{\mu : \mu f > r} \) for all 
\( f \in \cF_{1} \) and \( r \in \dQ \).
 \end{definition}

The simple functions we introduced can also be used to define measure and
integral in themselves.
Recall the definition of the set \( \cE \) in~\eqref{eq:bd-Lip-seq}.T
This set is a Riesz space as defined in Subsection~\ref{subsec:integral}.
A reasoning combining Propositions~\ref{propo:Caratheo-extension} 
and~\ref{propo:Riesz-extension} gives the following.

 \begin{proposition}\label{propo:metric-Riesz-extension}
Suppose that a positive linear functional \( \mu \) with \( \mu 1 = 1 \) is defined
on \( \cE \) 
that is continuous with respect to monotone convergence.
Then \( \mu \) can be extended uniquely to a probability
measure over \( \bX \) with \( \mu f = \int f(x) \mu(dx) \) for all \( f \in \cE \).
 \end{proposition}

Having a topology over the set of measures we can also extend
Proposition~\ref{propo:measure-cont}:
 
 \begin{proposition}\label{propo:measure-func-cont}
Let \( X \) be a complete separable
metric space and \( \cM(X) \) the space of bounded measures over \( X \) with the
weak topology.
The function \( \tup{\mu,f}\mapsto \mu f \) is a continuous function
\( \cM(X) \times C(X) \to \dR \).
 \end{proposition}

As mentioned above, for an open set \( G \) the value \( \mu(G) \) is not a continuous
function of the measure \( \mu \).
We can only say the following:

 \begin{proposition}\label{propo:measure-lscont}
Let \( X \) be a complete separable
metric space, and \( \cM(X) \) the space of bounded measures over \( X \) with the
weak topology, and \( G\sbsq X \) an open set.
The function \( \mu\mapsto \mu(G) \) is lower semicontinuous.
 \end{proposition}

\subsubsection{Distances for the weak topology}\label{subsubsec:Prokh}
The definition of measures in the style of
Proposition~\ref{propo:metric-Riesz-extension}
is not sufficiently constructive.
Consider a gradual definition of the measure \( \mu \), extending it
to more and more elements of \( \cE \), while keeping the positivity and
linearity property.
It can happen that the function \( \mu \) we end up with in the limit, is not
continuous with respect to monotone convergence.
Let us therefore metrize the space of
measures: then an arbitrary measure can be defined as the limit
of a Cauchy sequence of simple meaures.

One metric that generates the topology of weak convergence is the
following.

 \begin{definition}[Prokhorov distance]\label{def:prokh}
The \df{Prokhorov distance} \( \rho(\mu, \nu) \) of two measures is
the infimum of all those \( \eps \) for which, for all Borel sets \( A \) we
have (using the notation~\eqref{eq:Aeps})
 \[
   \mu(A) \le \nu(A^{\eps}) + \eps.
 \]
 \end{definition}

It can be shown that this is a metric and it generates the weak topology.
In computer science, it has been reinvented by the name of ``earth mover
distance''.
The following important result helps visualize it:

 \begin{proposition}[Coupling Theorem, see~\protect\cite{Strassen65}]
\label{propo:coupling}
Let \( \mu,\nu \) be two probability measures over a complete separable metric
space \( \bX \) with \( \rho(\mu, \nu) \le\eps \).
Then there is a probability measure \( \Prob \) on the space \( \bX \times \bX \)
with marginals \( \mu \) and \( \nu \) such that for a pair of random variables
\( \xi,\eta \) having joint distribution \( \Prob \) we have
 \[
   \Pbof{d(\xi,\eta) > \eps} \le \eps.
 \]
 \end{proposition}

Since weak topology has the second countability property, 
the metric space defined by the distance \( \rho(\cdot,\cdot) \) is separable.
This can also be seen directly: let us define a dense set in weak topology.

 \begin{definition}
For each point \( x \), let us define by \( \dg_{x} \) the measure which
concentrates its total weight 1 in point \( x \).
Let \( D \) be a countable everywhere dense set of points in \( X \).
Let \( D_{\cM} \) be the set of finite convex rational combinations of measures of
the form \( \dg_{x_{i}} \) where \( x_{i}\in D \), that is
those probability measures that are concentrated on finitely many points of
\( D \) and assign rational values to them.
 \end{definition}

It can be shown that \( D_{\cM} \) is everywhere dense in the metric space
\( (\cM(X), \rho) \); so, this space is separable.
It can also be shown that \( (\cM(X), \rho) \) is complete.
Thus, a measure can be given as the limit of a Cauchy sequence of elements
\( \mu_{1},\mu_{2},\dots \) of \( D_{\cM} \).

The definition of the Prokhorov distance uses quantification over all Borel sets.
However, in an important simple case, it can be handled efficiently.

 \begin{proposition}\label{propo:simple-Prokhorov-ball}
Assume that measure \( \nu \) is concentrated on 
a finite set of points \( S \sbs X \).
Then the condition \( \rho(\nu,\mu) < \eps \) is equivalent to
the finite set of conditions
 \begin{equation}\label{eq:finite-Prokhorov}
   \mu(A^{\eps}) > \nu(A) - \eps
 \end{equation}
for all \( A \sbs S \).
 \end{proposition}

Another useful distance for measures over a bounded space is the
Wasserstein distance.

\begin{definition}
Over a bounded metric space, we define the Wasserstein distance by
\begin{align*}
   W(\mu,\nu) = \sup_{f\in\Lip_{1}(X)} |\mu f - \nu f|. 
 \end{align*}
\end{definition}

The Wasserstein distance also has a coupling representation:

 \begin{proposition}[See~\protect\cite{Villani2008}]
Let \( \mu,\nu \) be two probability measures over a complete separable metric
space \( \bX \) with \( \rho(\mu, \nu) \le\eps \).
Then there is a probability measure \( \Prob \) on the space \( \bX \times \bX \)
with marginals \( \mu \) and \( \nu \) such that
 \[
  \int d(x,y)\Prob(dx dy)=W(\mu,\nu).
 \]
 \end{proposition}

 \begin{proposition}
The Prokhorov and Wasserstein metrics are equivalent:
the identity function creates a uniformly continuous homeomorphism
between the two metric spaces.
 \end{proposition}
 \begin{proof}
Let \( M=\sup_{x,y\in X} d(x,y) \).
The proof actually shows for \( \eps<1 \):
 \begin{align*}
       \rho(\mu,\nu)\le\eps   &\imp W(\mu,\nu) \le (M+1)\eps,
\\   W(\mu,\nu)\le\eps^{2} &\imp \rho(\mu,\nu) \le \eps.
 \end{align*}
Suppose \( \rho(\mu,\nu)\le\eps<1 \).
By the Coupling Theorem (Proposition~\ref{propo:coupling}), there are random
variables \( \xi,\eta \) over \( X \) with a joint distribution, and having respectively
the distribution \( \mu \) and \( \nu \), such that \( \Pbof{d(\xi,\eta)>\eps}\le\eps \).
Then we have
 \begin{align*}
   |\mu f - \nu f| &= |\Expv f(\xi) - \Expv f(\eta)| \le \Expv |f(\xi) - f(\eta)| 
\\   &\le \eps\Pbof{\rho(\xi,\eta)\le\eps} + M\Pbof{\rho(\xi,\eta)>\eps}
       \le (M+1)\eps.
 \end{align*}
Now suppose \( W(\mu,\nu)\le\eps^{2}<1 \).
For a Borel set \( A \) define \( g_{\eps}^{A}(x)=|1-\rho(x,A)/\eps|^{+} \).
Then we have \( \mu(A)\le \mu(g_{\eps}^{A}) \) and 
\( \nu g_{\eps}^{A} \le \nu(A^{\eps}) \).
Further  \( \eps g_{\eps}^{A}\in\Lip_{1} \), and hence \( W(\mu,\nu)\le\eps^{2} \)
implies \( \eps\mu(g_{\eps}^{A})\le \eps\nu(\eps g_{\eps}^{A})+\eps^{2} \).
This concludes the proof by
 \begin{align*}
 \mu(A) \le \mu(g_{\eps}^{A})\le \nu(\eps g_{\eps}^{A})+\eps \le
 \nu(A^{\eps})+\eps. 
 \end{align*}
 \end{proof}

\subsubsection{Relative compactness}

Convergence over the set of measures, even over a noncompact space, is
sometimes obtained via compactness properties.

%Changed to sequentially compact on 2008/01/26
 \begin{definition}
A set \( \Pi \)
of measures in \( (\cM(X), \rho) \) is called \df{sequentially compact} if every
sequence of elements of \( \Pi \) contains a convergent subsequence.

A set of \( \Pi \) of measures is called \df{tight} if for every \( \eps \) there
is a compact set \( K \) such that \( \mu(K) > 1 - \eps \) for all \( \mu \) in \( \Pi \).
 \end{definition}

The following important theorem is using our assumptions of separability and
completeness of the underlying metric space \( (X,\rho) \).

 \begin{proposition}[Prokhorov]
A set of measures is sequentially compact if
and only if it is tight and if and only if its closure is compact
in the metric space \( (\cM(X), \rho) \).
 \end{proposition}

%Added 2008/01/26
The following, simplest example is interesting in its own right.

\begin{example}
The one-element set \( \{\mu\} \) is compact and therefore
by Prokhorov's theorem tight.
Here, tightness says that for each \( \eps \) a mass of size \( 1-\eps \) of \( \mu \) is
concentrated on some compact set.
\end{example}

The following theorem strengthens the observation of this example.

\begin{proposition}\label{propo:inner-regular}
A finite measure \( \mu \) over a separable complete metric space has the
property
 \begin{align*}
   \mu(B) = \sup\setof{\mu(K) :\txt{ compact } K \sbsq B}
 \end{align*}
for all Borel sets \( B \).
\end{proposition}

In case of compact metric spaces, the following known theorem helps.

  \begin{proposition}\label{propo:measures-compact}
The metric space \( (\cM(\bX), \rho) \) of measures is compact if and
only if the underlying  metric space \( (X, d) \) is compact.
  \end{proposition}

So, if \( (X, d) \) is not compact then the set of measures is not compact.
But still, by Proposition~\ref{propo:inner-regular}, each finite measure
is ``almost'' concentrated on a compact set.

\subsubsection{Semimeasures}

Let us generalize then notion of semimeasure for the case of general Polish
spaces.
We use Proposition~\ref{propo:metric-Riesz-extension} as a starting point.

\begin{definition}
  A function \( f\mapsto \mu f \) defined over the set of all bounded
continuous functions  is called a \df{semimeasure} if it has the following
properties:
\begin{alphenum}
  \item Nonnegative: \( f\ge 0 \) implies \( \mu f \ge 0 \).
  \item Positive homogenous: we have \( \mu(a f) = a \mu f \) for all
nonnegative real \( a>0 \).
  \item Superadditive:  \( \mu(f+g) \ge \mu f + \mu g \).
  \item Normed: \( \mu 1 =1 \).
\end{alphenum}
\end{definition}

\begin{remark}
  The weaker requirement \( \mu 1 \le 1 \) suffices, but if we have \( \mu 1 \le 1 \)
we can always set simply \( \mu 1 = 1 \) without violating the other requirements.
\end{remark}

\subsubsection{Functionals of measures}

For a deeper study of randomness tests, we will need to characterize
certain functionals of finite measures over a Polish space.

\begin{proposition}
Let \( \bX = (X,d) \) be a complete separable metric space with \( \tilde\cM(\bX) \) the
set of finite measures over it.
A weakly continuous linear function \( F:\tilde\cM(\bX)\to\dR \) can always be
written as \( F(\mu) = \mu f \) for some bounded continuous function \( f \) over \( \bX \).
\end{proposition}
\begin{proof}
  Define \( f(x)=F(\dg_{x}) \).
The weak continuity of \( F \) implies that \( f(x) \) is continuous.
Let us show that it is bounded.
Suppose it is not.
Then there is a sequence of distinct points \( x_{n}\in X \), \( n=1,2,\dots \)
with \( f(x_{n})>2^{-n} \).
Let \( \mu \) be the measure \( \sum_{n}2^{-n}\dg_{x_{n}} \), then by linearity we have
\( L(\mu)=\sum_{n}2^{-n}f(x_{n})>\sum_{n} 1=\infty \), which is not allowed.

The function \( \mu\mapsto\mu f \) is continuous and coincides with \( F(\mu) \) on the
dense set of points \( D_{\cM} \), so they are equal.
\end{proof}

\chapter{Constructivity}
\section{Computable topology}

\subsection{Representations}\label{subsec:notation-repr}

There are several equivalent ways that notions of computability can be
extended to spaces like real numbers, metric spaces, measures, and so
on.
We use the 
the concepts of numbering (notation) and representation, as defined 
in~\cite{WeihrauchComputAnal00}.

\begin{notation}
We will denote by \( \dN \) the set of natural numbers and by \( \dB \) the set
\( \{0,1\} \).

Given a set (an alphabet) \( \Sg \) we denote by \( \Sg^{\dN}=\Sg^{\dN} \) the set
of infinite sequences with elements in \( \Sg \).

If for some finite or infinite sequences \( x,y,z,w \), we have
\( z = wxy \) then we write \( w \prefix  z \) (\( w \) is a \df{prefix} of \( z \)) and 
\( x \ltri z \).
After~\cite{WeihrauchComputAnal00}, 
let us define the \df{wrapping function} \( \ig : \Sg^{*} \to \Sg^{*} \) by
 \begin{equation}\label{eq:ig}
   \ig(a_{1}a_{2}\dotsm a_{n}) = 110a_{1}0a_{2}0\dotsm a_{n}011.
 \end{equation}
Note that 
 \begin{equation}\label{eq:ig-len}
   \len{\ig(x)} = (2 \len{x} + 5)\lor 6.
 \end{equation}
For strings \( x,x_{i} \in \Sg^{*} \), \( p, p_{i} \in \Sg^{\dN} \), \( k \ge 1 \),
appropriate tupling functions \( \tupcod{x_{1},\dots,x_{k}} \),
\( \tupcod{x,p} \), \( \tupcod{p,x} \), and so on can be defined with the help of
\( \tupcod{\cdot,\cdot} \) and \( \ig(\cdot) \).
\end{notation}

 \begin{definition}
Given a countable set \( C \), a \df{numbering} (or \df{notation}) 
of \( C \) is a surjective partial mapping \( \dg :\dN \to C \).
Given some finite alphabet \( \Sg \spsq \{0,1\} \) and an arbitrary set \( S \),
a \df{representation} of \( S \) is a surjective partial mapping
\( \chi :\Sg^{\dN} \to S \).
A \df{naming system} is a notation or a representation.
 \end{definition}
Here are some standard naming systems:
 \begin{enumerate}

  \item \( \id \), the identity over \( \Sg^{*} \) or \( \Sg^{\dN} \).

  \item \( \nu_{\dN} \), \( \nu_{\dZ} \), \( \nu_{\dQ} \) for the set of natural
numbers, integers and rational numbers.

  \item \( \Cf : \Sg^{\dN} \to 2^{\dN} \), the \df{characteristic function
representation} of sets of natural numbers, is defined by
\( \Cf(p) = \setof{i : p(i) = 1} \).

  \item \( \En : \Sg^{\dN} \to 2^{\dN} \), the \df{enumeration representation} of
sets of natural numbers, is defined by
\( \En(p) = \setof{n \in \dN : 110^{n+1}11 \ltri p} \).

  \item For \( \Dg \sbsq \Sg \), \( \En_{\Dg} : \Sg^{\dN} \to 2^{\Dg^{*}} \), 
the \df{enumeration representation} of subsets of \( \Dg^{*} \), is defined by
\( \En_{\Dg}(p) = \setof{w \in \Sg^{*} : \ig(w) \ltri p} \).
 \end{enumerate}

Using Turing machines with infinite input tapes, work tapes and output
tapes, one can define names for all
computable functions between spaces that are Cartesian 
products of terms of the kind \( \Sg^{*} \) and \( \Sg^{\dN} \).
(One wonders whether \( \Sg^{*}\cup\Sg^{\dN} \) is not needed, since a Turing
machine with an infinite input tape may still produce only a finite output.
But in this case one can also encode the result into an infinite sequence.)
Then, the notion of computability can be transferred to other spaces as
follows.

 \begin{definition}
Let \( \dg_{i} : Y_{i} \to X_{i} \), \( i=1,0 \) be naming systems of the spaces
\( X_{i} \).
Let \( f : \sbsq X_{1} \to X_{0} \), \( g : \sbsq Y_{1} \to Y_{0} \).
We say that function \( g \) \df{realizes} function \( f \) if
 \begin{equation}\label{eq:realize}
   f(\dg_{1}(y)) = \dg_{0}(g(y))
 \end{equation}
holds for all \( y \) for which the left-hand side is defined.

Realization of multi-argument functions is defined similarly.
We say that a function \( f : X_{1} \times X_{2} \to X_{0} \) 
is \df{\( (\dg_{1},\dg_{2},\dg_{0}) \)-computable} if
there is a computable function \( g : \sbsq Y_{1} \times Y_{2} \to Y_{0} \)
realizing it.
In this case, a name for \( f \) is naturally derived from a name of 
\( g \).\footnote{Any function \( g \) realizing \( f \) via~\eqref{eq:realize}
automatically has a
certain \df{extensionality} property: if \( \dg_{1}(y) = \dg_{1}(y') \) then
\( g(y) = g(y') \).}
 \end{definition}

 \begin{definition}
For representations \( \xi,\eta \), 
we write \( \xi \le \eta \) if there is a computable partial function 
\( f :\Sg^{\dN} \to \Sg^{\dN} \) with \( \xi(x) = \eta(f(x)) \).
In words, we say that \( \xi \) is \df{reducible} to \( \eta \), or that \( f \)
reduces (translates) \( \xi \) to \( \eta \).
There is a similar definition of reduction for notations.
We write \( \xi \equiv \eta \) if \( \xi \le \eta \) and \( \eta \le \xi \).
 \end{definition}

\subsection{Constructive topological space}

Let us start with the definition of topology with the help of a subbasis of
(possibly empty) open sets.

 \begin{definition}
A \df{constructive topological space} \( \bX = (X, \sg, \nu) \)
is a topological space over a set \( X \) with a subbasis \( \sg \) effectively
enumerated (not necessarily without repetitions) as a list 
\( \sg =\set{\nu(1),\nu(2),\dots} \), 
and having the \( T_{0} \) property (thus, every point is determined uniquely
by the subset of elements of \( \sg \) containing it).

By definition, a constructive topological space satisfies the second
countability axiom.
We obtain a basis
  \[
   \sg^{\cap}
 \]
of the space \( \bX \) by taking all possible finite
intersections of elements of \( \sg \).
It is easy to produce an effective enumeration for \( \sg^{\cap} \) from \( \nu \).
We will denote this enumeration by \( \nu^{\cap} \).
This basis will be called the \df{canonical basis}.

For every nonempty subset \( Y \) of the space \( X \), the subspace of \( X \) will
naturally get the same structure \( \bY = (Y,\sg,\nu) \) defined by
\( \setof{V\cap Y: V\in\sg} \).

The \df{product operation} is defined over constructive topological spaces
in the natural way.
 \end{definition}

 \begin{remark}
The definition of a subspace shows that a constructive topological space is
not a ``constructive object'' by itself, since the set \( X \) itself is not
necessarily given effectively.
For example, any nonempty subset of the real line is a constructive topological
space, as a subspace of the real line.
 \end{remark}

 \begin{examples}\label{example:constr-topol}
The following examples will be used later.
 \begin{enumerate}

  \item A discrete topological space, where the underlying
set is finite or countably infinite, with a fixed enumeration.
 
  \item\label{i:x.constr-topol.real}
The real line, choosing the basis to be the open intervals
with rational endpoints with their natural enumeration.
Product spaces can be formed
to give the Euclidean plane a constructive topology.

  \item\label{i:x.constr-topol.upper-real}
The real line \( \dR \), with the subbasis \( \sg_{\dR}^{>} \) defined as
the set of all open intervals \( \opint{-\infty}{b} \) with rational endpoints
\( b \).
The subbasis \( \sg_{\dR}^{<} \), defined similarly, leads to another topology.
These two topologies differ from each other and from
the usual one on the real line, and they are not Hausdorff spaces.
  
  \item\label{i:x.constr-topol.Cantor}
This is the constructive version of Example~\ref{example:topol}.\ref{i:x.topol.Cantor}.
Let \( X \) be a set with a constructive discrete topology,
and \( X^{\dN} \) the set of infinite sequences with elements from \( X \),
with the product topology: a natural enumerated basis is also easy to
define.
 
 \end{enumerate}
 \end{examples}

 \begin{definition}
Due to the \( T_{0} \) property, every point in our space is determined
uniquely by the set of open sets containing it.
Thus, there is a representation \( \gm_{\bX} \) of \( \bX \) defined as follows.
We say that \( \gm_{\bX}(p) = x \) if
\( \En_{\Sg}(p) = \setof{w : x \in \nu(w)} \).
If \( \gm_{\bX}(p) = x \) then we say that the infinite sequence
\( p \) is a \df{complete name} of \( x \):
it encodes all names of all subbasis elements containing \( x \).
From now on, we will call \( \gm_{\bX} \) the \df{complete standard
representation of the space \( \bX \)}.
 \end{definition}

 \begin{remark}
Here it becomes important that representations are allowed to be partial
functions.
 \end{remark}

 \begin{definition}[Constructive open sets]
In a constructive topological space \( \bX = (X, \sg, \nu) \),
a set \( G \sbsq X \) is called \df{constructive open}, or \df{lower semicomputable
open} in set \( B \) if there is a computably enumerable set \( E \) with 
\( G = \bigcup_{w \in E} \nu^{\cap}(w) \cap B \).
It is constructife open if it is constructive open in \( X \).
 \end{definition}

In the special kind of spaces in which randomness has been developed until
now, constructive open sets have a nice characterization:

 \begin{proposition}\label{propo:constr-open-nice-charac}
Assume that the space \( \bX = (X, \sg, \nu) \)
has the form \( Y_{1}\times \dots \times Y_{n} \) where
each \( Y_{i} \) is either \( \Sg^{*} \) or \( \Sg^{\dN} \).
Then a set \( G \) is constructive open iff it is open and the set
\( \setof{\tup{w_{1},\dots,w_{n}} : \bigcap_{i}\nu(w_{i}) \sbs G} \) 
is recursively enumerable.
 \end{proposition}
 \begin{proof}
The proof is not difficult, but it relies on the discrete nature of
the space \( \Sg^{*} \) and on the fact that the space \( \Sg^{\dN} \) is compact
and its basis consists of sets that are open and closed at the same time.
 \end{proof}

It is easy to see that if two sets are constructive open then so is their union.
The above remark implies that a space having
the form \( Y_{1}\times \dots \times Y_{n} \) where
each \( Y_{i} \) is either \( \Sg^{*} \) or \( \Sg^{\dN} \), also the intersection of
two recursively open sets is recursively open.
We will see that this statement holds, more generally, in all computable
metric spaces.

\subsection{Computable functions}

 \begin{definition}
Let \( \bX_{i} = (X_{i}, \sg_{i}, \nu_{i}) \) be constructive topological
spaces, and let \( f : X_{1} \to X_{0} \) be a function.
As we know, \( f \) is continuous iff the inverse image \( f^{-1}(G) \) of each
open set \( G \) is open in its domain.
Computability is an effective version of continuity:
it requires that the inverse image
of basis elements is uniformly constructively open.
More precisely, \( f : X_{1} \to X_{0} \) is \df{computable} if the set 
 \[
  \bigcup_{V \in \sg_{0}^{\cap}} f^{-1}(V) \times \{V\}
 \]
is a constructive open subset of \( X_{1} \times \sg_{0}^{\cap} \).
Here the basis \( \sg_{0}^{\cap} \) of \( \bX_{0} \) is treated as a discrete
constructive topological space, with its natural enumeration.

A partial function is computable if its restriction to the subspace that is
its domain, is computable.
 \end{definition}

The above definition depends on the enumerations \( \nu_{1},\nu_{0} \).
The following theorem
shows that this computability coincides with the one
obtained by transfer via the representations \( \gm_{\bX_{i}} \).
 
 \begin{proposition}[Hertling]\label{propo:hertling-computable}
For \( i=0,1 \), let
\( \bX_{i} = (X_{i}, \sg_{i}, \nu_{i}) \) be constructive topological spaces.
Then a function \( f : X_{1} \to X_{0} \) is
computable iff it is \( (\gm_{\bX_{1}},\gm_{\bX_{0}}) \)-computable for the 
representations \( \gm_{\bX_{i}} \) defined above.
 \end{proposition}

%  \begin{remark}
% As in Proposition~\ref{propo:constr-open-nice-charac},
% it would be nice to have the following statement, at least for total
% functions: 
% ``Function \( f : X_{1} \to X_{0} \) is computable iff the set
%  \[
%  \setof{(v, w) : \nu^{\cap}_{1}(w) \sbs f^{-1}[\nu_{0}(v)] }
%  \]
% is recursively enumerable.''
% But such a characterization seems to require compactness and possibly more.
%   \end{remark}

The notion of computable functions helps define morphisms between constructive
topological spaces.

 \begin{definition}
Let us call two spaces \( X_{1} \) and \( X_{0} \) \df{effectively homeomorphic}
if there are computable maps between them that are inverses of each
other.
In the special case when \( X_{0}=X_{1} \), we say 
that the enumerations of subbases
\( \nu_{0},\nu_{1} \) are \df{equivalent} if the identity
mapping is a effective homeomorphism.
This means that there are recursively enumerable sets \( F,G \) such that
 \[
  \nu_{1}(v) = \bigcup_{(v, w) \in F} \nu_{0}^{\cap}(w) \txt{ for all \( v \)},
\quad
  \nu_{0}(w) = \bigcup_{(w, v) \in G} \nu_{1}^{\cap}(v) \txt{ for all \( w \)}.
 \]
 \end{definition}

\subsection{Computable elements and sequences}
\label{subsec:computable-elements}

Let us define computable elements.

 \begin{definition}
Let \( \bU = (\{0\}, \sg_{0}, \nu_{0}) \) 
be the one-element space turned into a trivial constructive
topological space, and let \( \bX = (X, \sg, \nu) \) be another constructive
topological space.
We say that an element \( x \in X \) is \df{computable} if the function
\( 0 \mapsto x \) is computable.
It is easy to see that this is equivalent to the requirement that
the set \( \setof{u : x \in \nu(u)} \) is recursively enumerable.
Let \( \bX_{j}= (X_{j}, \sg_{j}, \nu_{j}) \), 
for \( i=0,1 \) be constructive topological spaces.
A sequence \( f_{i} \), \( i=1,2,\dots \) of functions with
\( f_{i} : X_{1} \to X_{0} \) is a
\df{computable sequence of computable functions} if 
\( \tup{i, x} \mapsto f_{i}(x) \) is a computable function.
 \end{definition}

It is easy to see that this
statement is equivalent to the statement that there is a recursive function
computing from each \( i \) a name for the computable function \( f_{i} \).
To do this formally, one sometimes using the s-m-n theorem of recursion
theory, or calls the method ``currification''.

The notion of computability can be relativized.

 \begin{definition}
Let \( X_{j} \), \( j=1,2 \) be two constructive topological spaces (or more
general, representations) and let \( x_{j}\in X_{j} \) be given.
We say that \( x_{2} \) is \( x_1 \)-\df{computable} if there is a computable
partial function \( f:X_{1}\to X_{2} \) with \( f(x_{1})=x_{2} \).
 \end{definition}

 \begin{remark}
The requirement that \( f \) computes \( y \) from \emph{every} representation of
\( x \) makes \( x \)-computability a very different requirement
of mere Turing reducibility of \( y \) to \( x \).
We will point out an important implication of this difference later, 
in the notion of \( \mu \)-randomness.
It is therefore preferable not to use the term ``oracle computation''
when referring to computations on representations .
 \end{remark}
 
\subsection{Semicomputability}

Lower semicomputability is a constructive version of lower
semicontinuity, as given in
Definition~\ref{def:lower-semicont}, but the sets that are required to be open
there are required to be constructive open here.
The analogue of Proposition~\ref{propo:semicont-sup} and
Corollary~\ref{coroll:semicont-extend}
holds also: a lower semicomputable function is the supremum of simple constant
functions defined on basis elements, and a lower semicomputable function 
defined on a subset can always be extended to one over the whole space.

 \begin{definition}
Let \( \bX = (X, \sg, \nu) \) be a constructive topological space.
A partial function \( f :X \to \ol\dR_{+} \)  with domain \( D \)
is called \df{lower semicomputable}
if the set \( \setof{(x,r): f(x) > r} \) is a constructive open subset of \( D\times\ol\dR_{+} \).

We define the notion of \( x \)-lower semicomputable similarly to the notion
of \( x \)-computability.
 \end{definition}

Let \( \bY = (\ol\dR_{+}, \sg_{\dR}^{<}, \nu_{\dR}^{<}) \) be the effective
topological space introduced in
Example~\ref{example:constr-topol}.\ref{i:x.constr-topol.real},
in which \( \nu_{\dR}^{>} \) is an enumeration of all open intervals of the
form \( \rint{r}{\infty} \) with rational \( r \).
The following characterization is analogous to Proposition~\ref{propo:semicont-topol}.

 \begin{proposition}\label{propo:semicomp-topol}
A function \( f:X\to\dR \) is lower semicomputable if and only if it 
is \( (\nu,\nu_{\dR}^{>}) \)-computable.
 \end{proposition}

As a name of a computable function, we can use the name of the enumeration
algorithm derived from the definition of computability, or the name
derivable using this representation theorem.

The following example is analogous to Example~\ref{example:open-set-indic}.

 \begin{example}\label{example:re-open-set-indic}
The indicator function \( 1_{G}(x) \) of an arbitrary constructive open
set \( G \) is lower semicomputable.
 \end{example}

The following fact is analogous to Proposition~\ref{propo:semicont-sup}:

 \begin{proposition}\label{propo:semicomp-sup}
Let \( f_{1},f_{2},\dots \) be a computable sequence of lower
semicomputable functions (via their names) over a constructive topological
space \( \bX \).
Then \( \sup_{i} f_{i} \) is also lower semicomputable.
 \end{proposition}

The following fact is analogous to Proposition~\ref{propo:semicont-rep}:

 \begin{proposition}\label{propo:semicomp-rep}
Let \( \bX=(X,\sg,\nu) \) be a constructive
topological space with enumerated basis \( \bg=\sg^{\cap} \)
and \( f:X\to\ol\dR_{+} \) a lower semicomputable function.
Then there is a lower semicomputable function \( g:\bg\to\ol\dR_{+} \) 
(where \( \bg \) is taken with the discrete topology) with
\( f(x)=\sup_{x\in\bg}g(\bg) \).
 \end{proposition}

In the important special case of Cantor spaces, 
the basis is given by the set of finite sequences.
In this case we can also require the function \( g(w) \) to be monotonic in the
words \( w \):

 \begin{proposition}\label{propo:semicomp-rep.Cantor}
Let \( X=\Sg^{\dN} \) be a Cantor space as defined in
Example~\ref{example:constr-topol}.\ref{i:x.constr-topol.Cantor}.
Then \( f:X\to\ol\dR_{+} \) is lower semicomputable if and only if there is a
lower semicomputable function 
\( g:\Sg^{*}\to\ol\dR_{+} \) (where \( \Sg^{*} \) is taken as a discrete space)
monotonic with respect to the relation \( u\prefix v \),
with \( f(\xi) = \sup_{u\prefix\xi} g(u) \).
 \end{proposition}

 \begin{remark}\label{rem:semicomp.cptable-rep}
In the above representation, we could require the function \( g \) to be computable.
Indeed, we can just replace \( g(w) \) with \( g'(w) \) where
\( g'(w) = \max_{v\prefix w}g(v,\len{w}) \), and \( g(v,n) \) is as much of \( g(v) \) as can be
computed in \( n \) steps.
 \end{remark}

Upper semicomputability is defined analogously:

\begin{definition}
We will call a closed set \df{upper semicomputable} (co-recursively enumerable) if
its complement is a lower semicomputable (r.e.) open set.
\end{definition}

The proof of the following statement is not difficult.

 \begin{proposition}\label{propo:one-arg-cpt}
Let \( \bX_{i} = (X_{i}, \sg_{i}, \nu_{i}) \) 
for \( i=1,2,0 \) be constructive topological spaces, and
let \( f: X_{1} \times X_{2} \to X_{0} \), and assume that \( x_{1} \in X_{1} \) is
a computable element.
 \begin{enumerate}

  \item If \( f \) is computable 
then \( x_{2} \mapsto f(x_{1}, x_{2}) \) is also computable.

  \item If \( \bX_{0} = \ol\dR \), and \( f \) is lower semicomputable
then \( x_{2} \mapsto f(x_{1}, x_{2}) \) is also lower semicomputable.

 \end{enumerate}
 \end{proposition}

\subsection{Effective compactness}

Compactness is an important property of topological spaces.
In the constructive case, however, a constructive version of compactness is
frequently needed.

  \begin{definition}\label{def:eff-compact-general}
 A constructive topological space \( X \) with canonical basis \( \bg \)
 has \df{recognizeable covers} if the set 
  \[
    \setof{ S\sbsq\bg : S\txt{ is finite and }\bigcup_{U \in S} U=X}
  \]
 is recursively enumerable.

 A compact space with recognizeable covers will be called
 \df{effectively compact}. 
  \end{definition}

 \begin{example}
Let \( \ag\in\clint{0}{1} \) be a real number such that the set of rationals
less than \( \ag \) is recursively enumerable but the set of rationals larger
than \( \ag \) are not.
(It is known that there are such numbers, for example
 \( \sum_{x\in \dN} 2^{-K(x)} \).)
Let \( X \) be the subspace of the real line on the interval \( \clint{0}{\alpha} \),
with the induced topology.
A basis of this topology is the set \( \bg \) of all nonempty intervals of the form
\( I\cap\clint{0}{\alpha} \), where \( I \) is an open interval with rational endpoints.

This space is compact, but not effectively so.
 \end{example}

 The following is immediate.

 \begin{proposition}
In an effectively compact space, in
every recursively enumerable set of basis elements covering
the space one can effectively find a finite covering.
 \end{proposition}

It is known that every closed subset of a compact space is compact.
This statement has a constructive counterpart:

  \begin{proposition}
Every upper
semicomputable closed subset \( E \) of an effectively compact space \( X \)
is also effectively compact.
  \end{proposition}
  \begin{proof}
\( U_{1},\dots,U_{k} \) be a finite set covering \( E \).
 Let us enumerate a sequence \( V_{i} \) of canonical basis elements
 whose union is the complement of \( E \).
 If \( U_{1}\cup\dots\cup U_{k} \) covers \( E \) then along with the sequence
 \( V_{i} \) it will cover \( X \), and this will be recognized in a finite number
 of steps.
  \end{proof}

On an effectively compact space, computable functions have computable extrema:

 \begin{proposition}\label{propo:lsc-min}
Let \( X \) be an effectively compact space, let \( f(x) \) be
lower semicomputable function on \( X \).
Then the infimum (which is always a minimum)
of \( f(x) \) is lower semicomputable uniformly from the
definition of \( X \) and \( f \).
 \end{proposition}
 \begin{proof}
For any rational number \( r \) we must be able to recognize that the minimum
is greater than \( r \).
For each \( r \) the set \( \setof{x:f(x)>r} \) is a lower semicomputable open
set.
It is the union of an enumerated sequence of canonical basis elements.
If it covers \( X \) this will be recognized in a finite number of steps.
 \end{proof}

It is known that a continuous function map compact sets into compact ones.
This statement also has a constructive counterpart.

 \begin{proposition}\label{propo:image-of-eff-compact}
Let \( X \) be an effectively compact space, and \( f \) a computable function from \( X \)
into another constructive topological space \( Y \).
Then \( f(X) \) is effectively compact.
 \end{proposition}
\begin{proof}
Let \( \bg_{X} \) and \( \bg_{Y} \) be the enumerated bases of the space \( X \)
and \( Y \) respectively.
Since \( f \) is computable, there is a recursively enumerable set 
\( \cE\sbsq\bg_{X}\times\bg_{Y} \)
such that \( f^{-1}(V)=\bigcup_{(U,V)\in E} U \) holds for all \( V\in\bg_{Y} \).
Let \( \cE_{V}=\setof{U: (U,V)\in \cE} \), then \( f^{-1}(V) = \bigcup \cE_{V} \).

Consider some finite cover
\( f(X)\sbsq V_{1}\cup\dots\cup V_{k} \) with some \( V_{i}\in\bg_{Y} \).
We will show that it will be recognized.
Let \( \cU=\bigcup_{i}\cE_{V_{i}} \), then \( X\sbsq \bigcup\cU \), 
and the whole set \( \cU\sbsq\bg_{X} \) is enumerable.
By compactness, there is a finite number of elements \( U_{1},\dots,U_{n}\in\cU \)
with \( X\sbsq\bigcup_{i} U_{i} \).
By effective compactness, every one of these finite covers will be recognized.
And any one such recognition 
will serve to recognize that \( V_{1},\dots,V_{k} \) is a cover of \( f(X) \).
\end{proof}

\subsection{Computable metric space}

Following~\cite{BrattkaPresserMetric03}, we define
a computable metric space as follows.

 \begin{definition}
A \df{constructive metric space} is a tuple \( \bX = (X, d, D, \ag) \) where \( (X,d) \)
is a  metric space, with a countable dense subset \( D \)
and an enumeration \( \ag \) of \( D \).
It is assumed that the real function \( d(\ag(v),\ag(w)) \) is computable.
 \end{definition}

As \( x \) runs through elements of \( D \) and \( r \) through positive rational
numbers, we obtain the enumeration of
a countable basis \( \setof{B(x, r) : x \in D, r\in \dQ} \) (of balls or radius \( r \)
and center \( x \)) of \( \bX \),  
giving rise to a constructive topological space \( \tilde\bX \).

 \begin{definition}
Let us call a sequence \( x_{1}, x_{2},\dots \) a \df{Cauchy} sequence if
for all \( i<j \) we have \( d(x_{i},x_{j}) \le 2^{-i} \).
To connect to the type-2 theory of computability developed above,
the \df{Cauchy-representation} \( \dg_{\bX} \) of the space is defined in a
natural way.
 \end{definition}

It can be shown that as a representation of \( \tilde\bX \), it is equivalent
to \( \gm_{\tilde\bX} \): \( \dg_{\bX} \equiv \gm_{\tilde\bX} \).

\begin{examples}\label{example:constr-metric}\
 \begin{enumerate}

  \item\label{i:x.constr-metric.Cclint{0}{1}} 
Example~\protect\ref{example:Cclint{0}{1}} is a constructive
metric space, with either of the two (equivalent)
choices for an enumerated dense set.

  \item\label{i:x.constr-metric.Cantor} 
Consider the metric space of
Example~\ref{example:metric}.\ref{i:x.metric.Cantor}: the Cantor space \( (X,d) \).
Let  \( s_{0}\in\Sg \) be a distinguished element.
For each finite sequence \( x\in\Sg^{*} \) let us define the infinite sequence
\( \xi_{x}=x s_{0} s_{0}\dots \).
The elements \( \xi_{x} \) form a (naturally enumerated) dense set in the space \( X \),
turning it into a constructive metric space.

 \end{enumerate}
\end{examples}

Let us point out a property of metric spaces that we use frequently.
 
 \begin{definition}
For balls \( B_{j}=B(x_{j},r_{j}) \), \( j=1,2 \) we will say \( B_{1}< B_{2} \) if
\( d(x_{1},x_{2})< r_{2}-r_{1} \).
In words, we will say that \( B_{1} \) is \df{manifestly included} in \( B_{2} \).
This relation is useful since it is an easy way to see \( B_{1}\sbs B_{2} \).
 \end{definition}

The property can be generalized to some constructive topological spaces.

 \begin{definition}
A constructive topological space \( (X,\bg,\nu) \) with basis \( \bg \)
has the \df{manifest inclusion property} if there is a relation
\( b<b' \) among its basis elements with the following properties.
  \begin{alphenum}
   \item \( b<b' \) implies \( b\sbsq b' \).
    \item The set of pairs \( b<b' \) is recursively enumerable.
    \item For every point \( x \), and pair of basis elements \( b,b' \)
      containing \( x \) there is a basis element \( b'' \) containing \( x \) with
      \( b''< b,b' \).
  \end{alphenum}
We express this relation by saying that \( b \) is \df{manifestly included}
in \( b' \).
In such a space, a sequence \( b_{1}>b_{2}>\dots \) with 
\( \bigcap_{i}b_{i}=\{x\} \) is called a \df{manifest representation} of \( x \).
 \end{definition}
Note that if the space has the manifest inclusion property then
for every pair \( x\in b \) there is a manifest representation of \( x \) beginning
with \( b \).

A constructive
metric space has the manifest inclusion property as a topological space,
and Cauchy representations are manifest.

Similarly to the definition of a computable sequence of computable
functions in Subsection~\ref{subsec:computable-elements}, we can
define the notion of a computable sequence of bounded computable functions,
or the computable sequence \( f_{i} \) of computable Lipschitz functions: 
the bound and the Lipschitz constant of \( f_{i} \) are required to be
computable from \( i \).
The following statement shows, in an effective form,
that a function is lower semicomputable if and only if it is the supremum
of a computable sequence of computable functions.

 \begin{proposition}\label{propo:semicomp-as-limit}
Let \( \bX \) be a computable metric space.
There is a computable mapping that to each name of a nonnegative
lower semicomputable
function \( f \) assigns a name of a computable sequence of computable 
bounded Lipschitz functions \( f_{i} \) whose supremum is \( f \).
 \end{proposition}
\begin{proof}[Proof sketch]
Show that \( f \) is the supremum of a computable sequence of
functions \( c_{i} 1_{B(u_{i}, r_{i})} \) where \( u_{i}\in D \) and 
\( c_{i}, r_{i} > 0 \) are rational.
Clearly, each indicator function \( 1_{B(u_{i},r_{i})} \) is the supremum
of a computable sequence of computable functions \( g_{i,j} \).
We have \( f = \sup_{n} f_{n} \) where \( f_{n} = \max_{i \le n} c_{i} g_{i,n} \).
It is easy to see that the bounds on the functions \( f_{n} \) are computable
from \( n \) and that they all are in \( \Lip_{\bg_{n}} \) for a
\( \bg_{n} \) that is computable from \( n \).
 \end{proof}

The following is also worth noting.

 \begin{proposition}
In a computable metric space, the intersection of two constructive open sets is
constructive open.
 \end{proposition}
 \begin{proof}
Let \( \bg = \setof{B(x, r) : x \in D, r\in \dQ} \) be a basis of our space.
For a pair \( (x,r) \) with \( x \in D \), \( r \in \dQ \), let
 \[
  \Gg(x,r) = \setof{(y,s): y\in D,\;s\in \dQ,\; d(x,y)+s < r}.
 \]
If \( U \) is a constructive open set, then there is a computably enumerable set
\( S_{U} \sbs D \times \dQ \) with \( U = \bigcup_{(x,r) \in S_{U}} B(x,r) \).
Let \( S'_{U} = \bigcup\setof{\Gg(x,r) : (x,r) \in S_{U}} \), then we have
\( U = \bigcup_{(x,r) \in S'_{U}} B(x,r) \).
Now, it is easy to see 
 \[
  U\cap V = \bigcup_{(x,r) \in S'_{U} \cap S'_{V}} B(x,r).
 \]
 \end{proof}

The following theorem is very useful.

\begin{theorem}\label{thm:eff-compact-metric}
A computable metric space \( X \) is \df{effectively compact} if and only if
from each (rational) \( \eps \) one can compute a finite set of \( \eps \)-balls
covering \( X \). 
\end{theorem}
 \begin{proof}
Suppose first that the space is effectively compact.
For each \( \eps \), let \( B_{1},B_{2},\dots \) be a
list of all canonical balls with radius \( \eps \).
This sequence covers the space, so already some \( B_{1},\dots,B_{n} \) covers the
space, and this will be recognized.

Suppose now that for every rational 
\( \eps \) one can find a finite set of \( \eps \)-balls covering the space.
Let \( \cS\sbsq\bg \) be a finite set of basis elements (balls) covering the space.
For each element \( G=B(u,r)\in\cS \), let \( G_{\eps}=B(u,r-\eps) \)
be its \( \eps \)-interior, and \( \cS_{\eps} = \setof{G_{\eps}: G\in \cS} \).
Then \( G=\bigcup{\eps>0} G_{\eps} \), and \( X = \bigcup_{\eps>0}\bigcup\cS_{\eps} \).
Compactness implies that there is an \( \eps>0 \) such that already
\( \cS_{eps} \) covers the space.
Let \( B_{1},\dots,B_{n} \) be a finite set of \( \eps/2 \)-balls \( B_{i}=B(c_{i},r_{i}) \)
covering the space that can be computed from \( \eps \).
Each of these balls \( B_{i} \) intersects one of the the sets \( G_{\eps}=B(u,r-\eps) \),
\( d(u,c_{i})\le r-\eps/2 \).
But then \( B_{i}\sbsq B(u,r) \) in a recognizeable way.
Once all the relations \( d(u,c_{i})\le r-\eps/2 \) will be recognized we will also
recognize that \( \cS \) covers the space.
 \end{proof}

We can strengthen now Example~\ref{example:Cclint{0}{1}}:
 \begin{example}[Prove!]\label{example:C(X)-cptable}
Let \( X \) be an effectively compact computable metric space.
Then the metric space \( C(X) \) with the dense set of functions \( \cE(D) \) introduced
in Definition~\ref{def:hat-funcs} is a computable metric space.
 \end{example}

The structure of a constructive metric space will be inherited on certain
subsets:

%Added 2008/01/24
\begin{definition}\label{def:constr-metric-subspace}
Let \( \bX=(X,\d,D,\ag) \) be a constructive metric space, and \( G\sbs X \) a
constructive open subset.
Then there is an enumeration \( \ag_{G} \) of the set \( D\cap G \) that creates a
constructive metric space \( \bG=(G,\d,D,\ag_{G}) \).
\end{definition}

\begin{remark}
An arbitrary subset of a constructive metric space will inherit the
constructive topology.
It will also inherit the metric, but not necessarily the structure of a
constructive metric space.
Indeed, first of all it is not necessarily a complete metric space.
It also does not inherit an enumerated dense subset.
\end{remark}

\killtext{
\section{Enumerative lattices}\label{sec:lattice}

This section, adapted from Hoyrup-Rojas,
is not strictly necessary to follow the rest of the notes.
Its most important implications will also be stated separately.
It allows to see lower semicomputability from a wider perspective.
I added some complications in order to show how an enumerative lattice
also gives rise to a constructive topological space.
 
\begin{definition}
In a lattice, let us call an element different from the minimum element \( \bot \)
\df{discrete} if it is not the supremum of the set of elements smaller than it.

An \df{enumerative lattice} 
is a tuple \( (X,\le,\cP,\nu) \) with the following properties.
 \begin{romanenum}

  \item \( (X,\le) \) is a partial order in which all sets have an upper bound.
 
 \item The (finite or infinite) set \( \cP\sbsq X \) with the property
that every element \( x \) of \( X \) is the supremum of some subset of  \( \cP \).

 \end{romanenum}

An element of an enumerative lattice can be \df{represented} by the
(possibly empty) list of elements of \( \cP \) whose supremum it is.

Let \( \cD \) be the set of discrete elements in \( X \).
We have necessarily \( \cD\sbsq \cP \), otherwise we could not get the elements
of \( \cD \) as suprema of elementes of \( \cP \).
We call an enumerative lattice \df{discrete-aware} if
the set of pairs \( \setof{\tup{p,q}: p\in \cP,\;q\in\cD,\; p\le q} \)  
is recursively enumerable.

An element of the enumerative lattice is called \df{computable} if it is
represented by a computable sequence.
\end{definition}

An enumerative lattice can be turned into a
constructive topological space in a simple way.

 \begin{definition}
From every discrete-aware enumerative lattice \( (X,\le,\cP,\nu) \) we create
a constructive topological space, by defining for each element
\( p\in \cP \) a subbasis element \( V_{p}=\setof{x\in X: p<x} \)
and for each discrete element \( d \) a subbasis element \( W_{d}=\setof{x\in X: d\le x} \).
 \end{definition}

\begin{proposition}
We indeed defined a constructive topological space this way.
\end{proposition}
 \begin{proof}
Only the \( T_{0} \) property must be checked.
Let \( x_{1},x_{2} \) be two distinct elements, let \( Q_{j} \) be the set of
elements of \( \cP \) smaller than \( x_{j} \).
Assume first that none of them is discrete. 
Then \( Q_{1}\ne Q_{2} \) and there will be an element in one of these
distinguishing \( x_{1} \) from \( x_{2} \) topologically.

Suppose now that \( x_{1} \) is discrete and \( x_{2} \) is not.
If \( Q_{1}\ne Q_{2} \) then what we had previously still holds.
If \( Q_{1}=Q_{2} \) then \( x_{1}>x_{2} \), hence \( \setof{z: z\ge x_{1}} \)
distinguishes \( x_{1} \) from \( x_{2} \).

Suppose that both of them are discrete.
If \( x_{1}>x_{2} \) then the previous argument still works.
Otherwise the set \( \setof{z: z\ge x_{2}} \) distinguishes the two elements.
 \end{proof}

Let us show a construction in the other direction.

It is easy to check that this construction gives a special enumerated
lattice.

Recall that an element of a constructive topological space is computable if
the set of subbasis elements containing it is recursively enumerable.

 \begin{proposition}
For a discrete-aware enumerative lattice \( (X,\le,\cP,\nu) \).
an element \( x \) of \( X \) is constructive if and only if it is computable in
the corresponding topological space.
 \end{proposition}
 \begin{proof}
Suppose that \( x \) is constructive.
Let \( B \) be the set of subbasis elements containing \( x \).
If \( x\in \cD \) then the enumerability of \( B \) follows from
discrete-awareness.

If \( x\not\in \cD \) then it is not discrete, there is an enumerated sequence
\( Q \)  of elements of \( \cP \) smaller than \( x \) whose supremum it is, and \( B \) is
the set of subbasis elements determined by these.

Suppose now that \( x \) is computable in the topological space.
If it is in \( \cP \) it is constructive by definition.
If it is not in \( \cP \) then its computability implies that the set of all
elements \( x \) of \( \cP \) with \( x<\cP \) (which defines \( B \)) is recursively
enumerable.
Its supremum is \( x \).
 \end{proof}

\begin{examples}
These examples will be used later.
 \begin{enumerate}
  \item
 \( (\ol\dR,\le,\dQ,\emptyset) \) where \( \ol\dR=\dR\cup\{-\infty,\infty\} \).
The computable elements are the lower semicomputable real numbers.
  \item \( (2^{\dN}, \sbsq, \txt{ singletons}, \txt{ singletons}) \).
So, here all elements of \( \cP \) are discrete.
The computable elements are the recursively enumerable sets.
  \item \( (\set{\bot,\top}, \le, \{\top\}, \{\top\}) \), where \( \bot < \top \).
 \end{enumerate}
\end{examples}

Each of these examples also has the following useful property.

 \begin{definition}
 We say that an enumerative directed lattice \( \dL = (X,\le,\cP,\nu) \) has
 the \df{locality property} if for every infinite subset \( Q \) and element \( p \) of \( \cP \),
  if \( \sup Q>p \) then there is a finite subset \( T \) of \( Q \) with \( \sup T >p \).
 \end{definition}

The following simple proposition does not hold for topological spaces in
general, and one of the reasons to introduce enumerative lattices.

\begin{proposition}
Let \( (X,\le,\cP,\cD) \) be an enumerative lattice.
The computable elements can be enumerated, in such a way that the
representing sequence of each is enumerated uniformly.
\end{proposition}

 \begin{remark}
On complete lattices, it is more usual to introduce the Scott topology.
It seems much finer and less constructive than the topology we
have introduced here.

In particular, the following statement (true with the Scott topology)
does not seem to be true:
Let \( (X_{j},\le,\cP_{j},\nu_{i}) \) for \( j=1,2 \) be discrete-aware enumerative
lattices (say, with the locality property) and \( f : X_{1}\to X_{2} \) a function.
Then \( f \) is continuous in the
corresponding topology if and only if it commutes with suprema.
 \end{remark}

Here is a more general example.

 \begin{example}\label{example:lattice-open-sets}
Let \( (X,\sg,\nu_{0}) \) be a constructive topological space.
Let \( \bg\sg^{\cap} \) be the enumerated basis obtained from \( \sg \) by taking
all finite intersections.
A point of \( x\in X \) is called an \df{atom} if \( \{x\} \) is an open set : in
this case \( \{x\} \) is necessarily in \( \bg \).

Let \( (Y,\le,\cP,\mu) \) be the enumerated lattice defined as follows.
The elements of \( Y \) are the open sets of \( (X,\sg,\nu) \), where \( \le \) denotes
inclusion.
The set \( \cP \) is \( \sg \), and \( \mu \) is the enumeration.

If also an enumeration of all atoms (as one-element sets) is given then
\( (Y,\le,\cP,\mu) \) is discrete-aware and thus can be turned into a constructive
topological space itself.
 \end{example}

This example can be generalized in a useful way.

 \begin{definition}\label{def:enum-C(X,Y)}
Let \( \bX=(X,\bg,\mu) \) be a constructive topological space in which \( \bg \) is an
enumerated basis (not just subbasis) and \( \bY=(Y,\le,\cP,\nu) \)
an enumerative lattice.
For \( b=\mu(i) \) and \( p=\nu(j) \), we define the \df{step function}
 \begin{align*}
   f_{b,p}(x) = \begin{cases}
                          p & \txt{if } x\in b
\\              \bot    & \txt{otherwise}.
                      \end{cases}   
 \end{align*}
This introduces a natural enumeration \( \eta \) on the set \( \cF \) of step functions.
The pointwise ordering will be denoted by \( \ptwise \).
Let \( \cC(\bX,\bY) \) be the set of pointwise suprema of step functions.
 \end{definition}

\begin{proposition}[Check!]
If the enumerative lattice \( \bY=(Y,\le,\cP,\nu) \) is discrete-aware then the set
\( C(X,Y) \) is the set of continuous functions from \( X \) to the topological space
corresponding to \( \bY \).

If \( \bY=\bR^{<} \) is the topological space or the real numbers generated by the
subbasis of intervals of form \( \opint{r}{\infty} \) then \( \bC(\bX,\bY) \) is the set of
lower semicontinuous real functions over \( \bX \).
\end{proposition}

 \begin{proposition}[Check!]\label{propo:enum-C(X,Y)}
The above construction gives an enumerative lattice
\( (\cC(\bX,\bY),\ptwise,\cF,\eta) \).
If the set of atoms of \( \bX \) can be enumerated then the defined
lattice is discrete-aware.
 \end{proposition}

The following proposition is very useful.

 \begin{proposition}\label{propo:extend}
Let \( (X,\bg,\mu) \) be a constructive topological space with the
manifest inclusion property (for example a metric space),
and \( (Y,\le,\cP,\nu) \) an enumerative lattice.
Let \( \cA \) be an algorithm whose input is a manifest representation 
\( \vek{b}=\tup{b_{1},b_{2},\dots} \) of an element 
\( x \) of \( X \), and whose output will be interpreted as a sequence
\( \vek{q}=\tup{q_{1},q_{2},\dots}=\cA(\vek{b}) \) 
of elements of \( \cP \) where \( \sup_{i} q_{i} \) will represent some element of \( Y \).
We can define (uniformly in the program of \( \cA \)) a function
\( f_{\cA}\in\cC(\bX,\bY) \) with the property that for all \( x \) if \( \cA \) is extensional
on \( x \) then \( \cA(\vek{b}) \) represents \( f_{\cA}(x) \).
 \end{proposition}

What is remarkable about this proposition is
that even if the algorithm \( \cA \) defines only a partial function (since it
might not even be extensional for all \( x \)), the function \( f_{\cA} \) extending
this to a total computable function in \( \cC(\bX,\bY) \) always exists.

 \begin{proof}
The program enumerates simultaneously all computations of \( \cA \)
on all finite sequences \( b_{1}>\dots >b_{i} \).
If in step \( n \) an output \( p_{j} \) is produced from \( b_{1}>\dots >b_{i} \)
then it outputs the function \( f_{b_{i},p_{j}} \).
The result \( f_{\cA}(x) \) is the (pointwise) supremum of all these
functions outputted.

It is easy to check that indeed, for all points \( x \) on which the algorithm
\( \cA \) is extensional, it outputs \( f_{\cA}(x) \).
Namely, for every step function \( f_{b,p} \) it outputs, for every \( x\in b \)
the basis element \( b \) is the starting point of some manifest
representation of \( x \).
Therefore if \( \cA \) is extensional on \( x \) then outputting \( f_{b,p} \) never
commits \( f_{\cA} \) to something against the algorithm \( \cA \) would do on a
representation of \( x \).
 \end{proof}

 \begin{corollary} We have the following for \( X,Y \) as above.
 \begin{enumerate}
  \item
The computable elements  of \( \cC(\bX,\bY) \) are exactly the computable
functions from \( X \) to \( Y \) defined by the given representations.
  \item The \( x \)-computable elements of \( \bY \) are exactly the images of \( x \)
by \emph{total} computable functions from \( \bX \) to \( \bY \).
 \end{enumerate}
 \end{corollary}

 The following concept arises frequently in lattices.

 \begin{definition}
A function \( f : \bX\to \bY \) between two complete lattices is called 
\df{Scott continuous} if it is monotonic and interchangeable with suprema.
 \end{definition}

The following proposition is easy to prove but is very useful.

 \begin{proposition}\label{propo:Scott}
Let \( (X,\le,\cP) \) and \( (X',\le,\cP') \) be enumerative lattices.
Let \( f:X\to Y \) be a Scott continuous function.
If \( f \) is computable if and only if it
is uniformly computable on elements of the form
\( \sup\set{p_{1},\dots,p_{k}} \) for \( p_{i}\in\cP \).
 \end{proposition}
}

\section{Constructive measure theory}\label{sec:constr-meas}

The basic concepts and results of measure theory are recalled in
Section~\ref{sec:measures}.
For the theory of measures over metric spaces, see
Subsection~\ref{subsec:measure-metric}.
We introduce a certain fixed, enumerated sequence of Lipschitz functions
that will be used frequently.
Let \( \cE \) be the set of functions introduced in
Definition~\ref{def:hat-funcs}. 
The following construction will prove useful later.

 \begin{proposition}\label{propo:bd-Lip-set}
All bounded continuous functions can be 
obtained as the limit of an increasing sequence of functions from
the enumerated countable set \( \cE \) of bounded computable
Lipschitz functions introduced in~\eqref{eq:bd-Lip-seq}.
 \end{proposition}
The proof is routine.

\subsection{Space of measures}

Let \( \bX = (X, d, D, \ag) \) be a computable metric space.
In Subsection~\ref{subsec:measure-metric}, the space
\( \cM(\bX) \) of measures over \( \bX \) is defined, along with a natural 
enumeration \( \nu = \nu_{\cM} \) for a subbasis \( \sg = \sg_{\cM} \)
of the weak topology.
This is a constructive topological space \( \bM \) which can be metrized by
introducing, as in~\ref{subsubsec:Prokh}, the \df{Prokhorov distance} \( \rho(\mu, \nu) \).
Recall that we defined \( D_{\bM} \) as the set of
those probability measures that are concentrated on finitely many points of
\( D \) and assign rational values to them.
Let \( \ag_{\bM} \) be a natural enumeration of \( D_{\bM} \).
Then 
 \begin{equation}\label{eq:metric-measures}
   (\cM, \rho, D_{\bM}, \ag_{\bM})
 \end{equation}
is a computable metric space whose constructive topology is equivalent to
\( \bM \).
Let \( U=B(x, r) \) be one of the balls in \( \bX \), 
where \( x\in D_{\bX} \), \( r \in \dQ \).
The function \( \mu \mapsto \mu(U) \) is typically not computable,
since as mentioned in~\eqref{subsubsec:weak-top}, it is not even continuous.
The situation is better with \( \mu \mapsto \mu f \).
The following theorem is the computable strengthening of
part of Proposition~\ref{propo:measure-func-cont}: 

 \begin{proposition}\label{propo:computable-integral}
Let \( \bX = (X, d, D, \ag) \) be a computable metric space, and let
\( \bM = (\cM(\bX), \sg, \nu) \) be the constructive topological space of
probability measures over \( \bX \).
If the function \( f : \bX \to \dR \) is bounded and computable
then \( \mu \mapsto \mu f \) is computable.
 \end{proposition}
 \begin{proof}[Proof sketch]
To prove the theorem for bounded Lipschitz functions, we can invoke
the Strassen coupling theorem~\ref{propo:coupling}.

The function \( f \) can be obtained as a limit of a computable
monotone increasing sequence of computable Lipschitz functions \( f^{>}_{n} \),
and also as a limit of a computable monotone decreasing 
sequence of computable Lipschitz functions \( f^{<}_{n} \).
In step \( n \) of our computation of \( \mu f \),
we can approximate \( \mu f^{>}_{n} \) from above 
to within \( 1/n \), and \( \mu f^{<}_{n} \) from below to within \( 1/n \).
Let these bounds be \( a^{>}_{n} \) and \( a^{<}_{n} \).
To approximate \( \mu f \) to within \( \eps \), 
find a stage \( n \) with \( a^{>}_{n} - a^{<}_{n} +2/n < \eps \).
 \end{proof}

Using Example~\ref{example:C(X)-cptable}, we can extend this as follows:

 \begin{proposition}[Prove!]\label{propo:measure-funct-comp}
Let \( \bX \) be an effectively compact
metric space (and thus \( C(\bX) \) is a computable metric space).
Then the mapping \( \tup{\mu,f}\mapsto \mu f \) over \( \cM(\bX)\times C(\bX) \) is 
computable.
 \end{proposition}

Talking about open sets, only a weaker statement can be made.

 \begin{proposition}\label{propo:measure-lscomp}
Let \( G\sbsq X \) be a constructive open set.
The function \( \mu\mapsto \mu(G) \) is lower semicomputable.
 \end{proposition}

 \begin{remark}
   Using the notion of an enumerative lattice defined by Hoyrup and Rojas,
one can make this a statement about the two-argument function
\( \tup{\mu,G}\to \mu(G) \).
 \end{remark}

% the optional section~\ref{sec:lattice} above.
% Recall the definition of the set of functions \( \cC(\bX,\bY) \) from a topological space
% into an enumerative lattice.
% In a computable metric space \( \bX \), let \( \tau \) be the set of open sets, as an
% enumerative lattice.

% \begin{proposition}[Check!]
% Example~\ref{example:lattice-open-sets} gives the structure of an enumerative lattice
% \( \bT \) to the set of open sets of a computable topological space,
% and let \( \cC(\cM(\bX),\clint{0}{1}) \) be the enumerative lattice of lower
% semicontinuous functions from \( \cM(\bX) \) to \( \clint{0}{1} \).
% Then the element of \( \tau\to \cC(\cM(\bX),\clint{0}{1}) \)
% which maps an open set \( G \) and a measure \( \mu \) to \( \mu(G) \) is lower
% semicomputable.

% As said in Example~\ref{example:lattice-open-sets}, if the atoms of \( X \) can be
% enumerated then the set \( \tau \) of open sets can itself be turned into a
% topological space.
% In this case, the function \( \tup{G,\mu}\mapsto\mu(G) \) is a lower
% semicomputable function \( \tau\times \cM(\bX)\to\clint{0}{1} \).
% \end{proposition}

It is known that if our metric space \( \bX \) is compact then so is the space
\( \cM(\bX) \) of measures.
This can be strengthened:

 \begin{proposition}[Prove!]
If a complete computable metric space \( \bX \) is effectively compact
then \( \cM(\bX) \) is also effectively compact.
 \end{proposition}

\subsection{Computable and semicomputable measures}
\label{subsec:computable-meas}
A measure \( \mu \) is called \df{computable} if it is a computable element of the
space of measures.
Let \( \{g_{i}\} \) be the set of bounded Lipschitz functions over \( \bX \)
introduced  in Definition~\ref{def:hat-funcs}.

 \begin{proposition}\label{propo:computable-meas-crit}
Measure \( \mu \) is computable if and only if so is the function
\( i \mapsto \mu g_{i} \).
 \end{proposition}
 \begin{proof}
The ``only if'' part follows from Proposition~\ref{propo:computable-integral}.
For the ``if'' part, note that in order to 
trap \( \mu \) within some Prokhorov neighborhood of size \( \eps \),
it is sufficient to compute \( \mu g_{i} \) within a small
enough \( \dg \), for all \( i\le n \) for a large enough \( n \).
 \end{proof}

A non-computable density function can lead to a computable measure:

 \begin{example}
Let our probability space be the set \( \dR \) of real numbers with its
standard topology.
Let \( a < b \) be two computable real numbers.
Let \( \mu \) be the probability distribution with density function
 \( f(x) = \frac{1}{b-a}1_{\clint{a}{b}}(x) \) 
(the uniform distribution over the interval \( \clint{a}{b} \)).
Function \( f(x) \) is not computable, since it is not even continuous.
However, the measure \( \mu \) is computable: indeed,
\( \mu g_{i} = \frac{1}{b-a} \int_{a}^{b} g_{i}(x) dx \) is a computable
sequence, hence 
Proposition~\ref{propo:computable-meas-crit} implies that \( \mu \) is computable.
 \end{example}

The following theorem compensates somewhat for the fact
mentioned earlier, that the
function \( \mu \mapsto \mu(U) \) is generally not computable.

 \begin{proposition}
Let \( \mu \) be a finite computable measure.
Then there is a computable map \( h \) with the property that for every
bounded computable function \( f \) with \( |f| \le 1 \) with
the property \( \mu(f^{-1}(0))=0 \),
if \( w \) is the name of \( f \) then \( h(w) \) is the
name of a program computing the value \( \mu\setof{x: f(x) < 0} \).
 \end{proposition}
 \begin{proof}
Straightforward.
 \end{proof}

Can we construct a measure just using the pattern of
Proposition~\ref{propo:computable-meas-crit}?
Suppose that there is a computable function \( \tup{i,j}\mapsto m_{i}(j) \)
with the following properties:
 \begin{alphenum}
  \item \( i < j_{1} < j_{2} \) implies \( |m_{i}(j_{1})-m_{i}(j_{2})| < 2^{-j_{1}} \).
  \item For all \( n \), there is a probability measure \( \mu_{n} \) with
\( m_{i}(n) = \mu_{n} g_{i} \) for all \( i < n \).
 \end{alphenum}
Thus, the sequences converge, and for each \( n \), all values \( m_{i}(n) \) for \( i\le n \)
are consistent with coming from a probability
measure \( \nu_{n} \) assigning this value to \( g_{i} \). 
Is there a probability measure \( \mu \) with the property that for each \( i \) we have
\( \lim_{j} m_{i}(j) = \mu g_{i} \)?
Not necessarily, if the space is not compact.

 \begin{example}
Let \( X = \{1,2,3,\dots\} \) with the discrete topology.
Define a probability measure \( \mu_{n} \) with  \( \mu_{n} g_{i} = 0 \) for \( i<n \) and
otherwise arbitrarily.
Since we only posed \( n-1 \) linear conditions on a finite number of variables, it
is easy to see that such a \( \mu_{n} \) exists.
Then define \( m_{i}(n)=\mu_{n}(i) \) for all \( i \).

Now all the numbers \( m_{i}(n) \) converge with \( n \) to 0, but \( \mu=0 \) is not a
probability measure.
 \end{example}

To guarantee that the sequences \( m_{i}(j) \) indeed define a probability
measure, progress must be made, for example, in terms of the narrowing of
Prokhorov neighborhoods.

\subsection{Random transitions}\label{subsec:random-trans}

Consider random transitions now.

 \begin{definition}[Computable kernel]
Let \( \bX,\bY \) be computable metric spaces, giving
rise to measurable spaces with \( \sg \)-algebras \( \cA, \cB \)
respectively.
Let \( \Lg = \setof{\lg_{x} : x \in X} \) be a probability kernel from \( X \) to
\( Y \) (as defined in Subsection~\ref{subsec:transitions}).
Let \( \{g_{i}\} \) be the set of bounded Lipschitz functions over \( Y \) 
introduced in Definition~\ref{def:hat-funcs}.
To each  \( g_{i} \), the kernel assigns a (bounded) measurable function 
 \[
   f_{i}(x) = (\Lg g_{i})(x) =  \lg_{x}^{y} g_{i}(y).
 \]
We will call the kernel \( \Lg \) \df{computable} if so is the assignment 
\( (i, x) \mapsto f_{i}(x) \).
 \end{definition}

When \( \Lg \) is computable, each function \( f_{i}(x) \) is of course continuous. 
The measure \( \Lg^{*}\mu \) is determined by the values 
\( \Lg^{*} g_{i} = \mu (\Lg g_{i}) \), which are computable from \( (i, \mu) \) and
so the mapping \( \mu \mapsto \Lg^{*}\mu \) is computable.

The following example is the simplest case, when the transition is actually
deterministic. 

 \begin{example}[Computable deterministic kernel]
\label{example:computable-determ-trans}
A computable function \( h : X \to Y \) defines an operator 
\( \Lg_{h} \) with \( \Lg_{h} g = g \circ h \) (as in Example~\ref{example:determ-trans}).
This is a deterministic computable transition, in which
\( f_{i}(x) = (\Lg_{h} g_{i})(x) = g_{i}(h(x)) \) is, 
of course, computable from \( (i,x) \).
We define \( h^{*}\mu = \Lg_{h}^{*}\mu \).
 \end{example}

%%% Local Variables: 
%%% mode: latex
%%% TeX-master: "ait-notes"
%%% End: 

%%% Local Variables: 
%%% mode: latex
%%% TeX-master: "ait-notes"
%%% End: 

%\input math
% \include{constr} \input into math

\backmatter

 \bibliographystyle{plain}
 \bibliography{ait,alg,gacs-publ}

\end{document}